\documentclass[final,1p,times,numcompress]{elsarticle}

\usepackage{etoolbox}
\usepackage{lineno,hyperref}
\usepackage{amssymb}
\usepackage{bm}
\usepackage{bbm}
\usepackage{subfigure}
\usepackage{amsmath}
\usepackage{amsthm} 

\newcommand{\nsp}{\hspace{-0.4pt}}
\newcommand{\ssp}{\hspace{0.4pt}}
\newcommand{\Sp}{\,\,\,\,\,\,}

\newcommand{\w}{\mathbbm{w}}

\newcommand{\sx}{\sigma^{x}}
\newcommand{\sy}{\sigma^{y}}
\newcommand{\sz}{\sigma^{z}}
\newcommand{\rr}{\bm{r}}
\newcommand{\ra}{\rangle}
\newcommand{\la}{\langle}
\newcommand{\p}{\partial}
\newcommand{\be}{\begin{equation}}
\newcommand{\ee}{\end{equation}}
\newcommand{\bea}{\begin{eqnarray}}
\newcommand{\eea}{\end{eqnarray}}

\newcommand{\commut}[2]{[\ssp #1\ssp,\,#2\ssp]}
\newcommand{\commutb}[2]{\big[\ssp #1\ssp,\,#2\ssp\big]}

\newcommand{\ket}[1]{|#1\rangle}
\newcommand{\bra}[1]{\langle#1|}
\newcommand{\bk}[1]{\langle#1\rangle}
\newcommand{\kb}[1]{|#1\rangle\langle#1|}
\newcommand{\sket}[1]{|{#1})}
\newcommand{\sbra}[1]{({#1}|}
\newcommand{\sbraket}[2]{({#1}|{#2})}

\newcommand{\one}{1 \hspace{-1.0mm}  {\bf l}}
\newcommand{\Tr}{\textrm{Tr}}
\newcommand{\tr}{ \textrm{tr}}
\newcommand{\diag}{ \textrm{diag}}
\newcommand{\Det}{ \textrm{Det}}
\newcommand{\Cov}{ \textrm{Cov}}
\newcommand{\Ad}{ \textrm{Ad}}
\newcommand{\Pf}{ \textrm{Pf}}
\newcommand{\pr}{ \textrm{Pr}}
\newcommand{\HH}{\mathcal{H}}
\newcommand{\vv}[1]{{#1}}
\newcommand{\V}[1]{ \textrm{\bf vec}({#1})}

\newcommand{\QED}{\hfill\ensuremath{\blacksquare}}

\DeclareRobustCommand\openzero{\leavevmode\hbox{0\kern-.55em0}}

\renewcommand{\Im}{ \textrm{\bf Im}}
\renewcommand{\Re}{ \textrm{\bf Re}}

\newtheorem{Lemma}{Lemma}
\newtheorem{Proposition}{Proposition}

\journal{Physics Reports}

\makeatletter
\def\@mkboth#1#2{}
\newlength\appendixwidth
\preto\appendix{\addtocontents{toc}{\protect\patchl@section}}
\newcommand{\patchl@section}{%
  \settowidth{\appendixwidth}{\textbf{Appendix }}%
  \addtolength{\appendixwidth}{1.5em}%
  \patchcmd{\l@section}{1.5em}{\appendixwidth}{}{\ddt}%
}
\makeatother

\begin{document}

\begin{frontmatter}



\title{Geometry of Quantum Phase Transitions}

\author[1,2]{Angelo Carollo}
\author[1]{Davide Valenti}
\author[1,2,3]{Bernardo Spagnolo}

\address[1]{Department of Physics and Chemistry, University of Palermo, Viale delle Scienze, Ed. 18, I-90128 Palermo, Italy}
\address[2]{Radiophysics Department, Lobachevsky State University of Nizhni Novgorod, 23 Gagarin Avenue, Nizhni Novgorod 603950, Russia}
\address[3]{Istituto Nazionale di Fisica Nucleare, Sezione di Catania,  Via S. Sofia 64, I-90123 Catania, Italy}

\begin{abstract}
In this article we provide a review of geometrical methods employed
in the analysis of quantum phase transitions and non-equilibrium
dissipative phase transitions. After a pedagogical introduction to
 geometric phases and geometric information in the characterisation
of quantum phase transitions, we describe recent developments of
geometrical approaches based on mixed-state generalisation of the
Berry-phase, i.e. the Uhlmann geometric phase, for the investigation
of non-equilibrium steady-state quantum phase transitions
(NESS-QPTs ). Equilibrium phase transitions fall invariably into two
markedly non-overlapping categories: classical phase transitions and
quantum phase transitions, whereas in NESS-QPTs  this distinction may fade off. The approach described in this review, among other things, can quantitatively assess the quantum character of
such critical phenomena. This framework is applied to a paradigmatic
class of lattice Fermion systems with local reservoirs,
characterised by Gaussian non-equilibrium steady states. The
relations between the behaviour of the geometric phase curvature, the
divergence of the correlation length, the character of the 
criticality and the gap - either Hamiltonian or dissipative - are reviewed.
\end{abstract}


\end{frontmatter}
\tableofcontents


\section{Introduction}
Several systems manifest phase transitions as the temperature or
other external parameters are modified, and are characterised by a qualitative changes in the system properties. Examples of phase transitions range from our mundane
experience of ice melting or the loss of ferromagnetism in iron to
the far more exotic superfluid Mott-insulator phase transitions in
optical lattices~\cite{Greiner2002}. 

Thermal phase transitions, i.e. phase transitions that occur at finite temperature, are characterised by a macroscopic order, such as magnetisation, or crystal structure, which gets destroyed by the thermal fluctuations~\cite{Nishimori2010,Mussardo2010,Chaikin1995,Goldenfeld1992,Stanley1987}. At zero temperature, a different type of phase transitions can take place as a non-thermal parameter, such as magnetic field, pressure or chemical composition, is changed. The result is an order in the ground state properties of the system which is destroyed exclusively by quantum fluctuations. Competition between different terms of the system's Hamiltonian, and the ordering thereof, are responsible for such fluctuations, which are a macroscopic manifestation of the Heisenberg uncertainty principle~\cite{Sachdev2011,Sondhi1997,Vojta2003,Belitz2005,Carr2011,Suzuki2013}.

In the past few decades, the studies of phase transitions at zero and finite temperature has been fuelled
by the formidable success of Ginzburg-Landau-Wilson theories,
local ordering, spontaneous symmetry breaking and the renormalisation
group~\cite{Wilson1974,Parisi1988,Zinn-Justin2002,Cardy1996} in
spelling out the nature of a wide variety of critical behaviours.

Such a standard approach has been recently paralleled by the development of a framework emerging from the cross-pollination of quantum information theory, information geometry and metrology. This approach is differential-geometric in essence and relies on two intertwined concepts: \emph{geometric phase} and \emph{fidelity}. The idea is to explore the critical features of many-body systems through the geometrical properties of its equilibrium state. The intuition behind it is relatively straightforward. Phase transitions are characterised by dramatic differences in some system's observable under infinitesimal variation of a control parameter, which are embodied in major structural changes of the state of the system. One could compare the states associated to infinitesimally close values of the parameters by some ``similarity'' function, or distance function. One would then be able to reveal these major structural changes in the state through the singular behaviour of such a distance function. The Uhlmann fidelity~\cite{Uhlmann1976,Alberti1983,Alberti1983a,Alberti1984,Wootters1981,Jozsa1994,Schumacher1995,Fuchs1996} provides a quantitative measure of such a distance, called the \emph{Bures metric}~\cite{Bures1969}.

Alternatively, in a quantum phase transitions (QPT), one may look at these dramatic changes as the manifestation of level crossings, occurring at the thermodynamic limit, which involve ground state and low lying part of the energy spectrum. In much the same way as mass-energy singularity
bends the geometry of the space-time, or a Dirac monopole twists the topology of the field configuration around it~\cite{Nakahara1990}, a level crossing bends the geometry of the ground-state manifold in a singular way. A quantitative account of this effect is provided by the geometric phase~\cite{Berry1984,Berry1989,Wilczek1989,Bohm2003}, or more precisely, its infinitesimal counterpart, the \emph{Berry curvature}. Driving the system close to or around these singularities results in dramatic effects on the state geometry, picked up by a quantum state in the form of geometric phase instabilities~\cite{Carollo2005,Hamma2006,Zhu2008}. Geometric phases are at the core of the characterisation of
topological phase transitions~\cite{Thouless1983,Bernevig2013,Chiu2016}, and have been
employed in the description and detection of QPT, both theoretically~\cite{Carollo2005,Pachos2006,Plastina2006,Hamma2006,Zhu2006,Reuter2007}
and experimentally~\cite{Peng2010}.

Singular curvature and singular metric are two complementary manifestations of the same exceptional behaviour of the quantum state arising across phase transitions~\cite{CamposVenuti2007}. These quantities are a direct expression of the geometry of the state manifold and require no a-priori notions of order parameters, or symmetry breaking. They do not even require a detailed knowledge on the dynamics of the physical model, (e.g. the Hamiltonian), but only on its kinematics (i.e. the state of the system, and its dependence on parameters).  This means that this conceptual framework can be freed from the boundaries of the standard Ginzburg-Landau-Wilson paradigm, and can be used to study a wide range of critical behaviours where a notion of local order does not hold, or cannot be easily identified.

This motivates the employment of geometric phases in the investigation of critical properties in a plethora  of different scenarios~\cite{Carollo2005,Hamma2006,Zhu2006,Pachos2006,Plastina2006,Cui2006,Chen2006,Yi2007,CamposVenuti2007,Reuter2007,Yuan2007,Cui2008,Furtado2008,Hu2008,Nesterov2008,Zhu2008,Paunkovic2008,Contreras2008,Ma2009,Oh2009,Nesterov2009,Cui2009,Quan2009,Wang2010b,Wang2010,Sjoqvist2010,Wang2010b,Basu2010,Lu2010,Li2010,Basu2010a,Zhong2010,Cucchietti2010,Peng2010,Yuan2010,Cheng2010a,Bandyopadhyay2011,Ribeiro2011,Lian2011,Tian2011,Li2011,Li2011a,Castro2011,Patra2011,Zhang2011,Zhang2012,Yuan2012,Requist2012,Shan2012,Tomka2012,Lian2012,Ma2012,Lian2012,Ma2013,Zhang2013,Zhang2013a,Liang2013,Zhang2013,AzimiMousolou2013,Jafari2013,Li2013,Zhang2013c,Ma2013,Sarkar2014,Shan2014,Hickey2014,Lue2015,Zhu2015,Yuan2015,Li2015,Wu2015,Li2015a,Ma2015,Yang2015,Zvyagin2016,Ya2016,Nie2017,Zeng2017,Alvarez2017,Liu2018,Carollo2018,Carollo2019,Zhang2018,Carollo2018a,Henriet2018,Cai2019}.\\
The idea that QPTs  could be explored through the Berry phase properties was first proposed and applied in the prototypical XY spin-1/2 chain~\cite{Carollo2005,Hamma2006,Pachos2006,Zhu2006,Reuter2007,Quan2009,Zhong2010,Patra2011} and extended to many other many-body systems, such as the Dicke model~\cite{Plastina2006,Chen2006}, the Lipkin-Meshkov-Glick model~\cite{Cui2006,Sjoqvist2010,Tian2011}, Yang-Baxter spin-1/2 model~\cite{Hu2008,Wang2010b}, quasi free-Fermion systems~\cite{Cui2008,Cui2009,Ma2009,Ma2013,Ma2015}, interacting Fermion models~\cite{Paunkovic2008,Lin-Cheng2010,Cheng2010a,Li2010,Zhang2011,Yuan2012,Yuan2015}, in ultracold atoms~\cite{Li2011,Li2011a,Li2013}, in spin chains with long range interactions~\cite{Ribeiro2011,Shan2012},  in cluster models~\cite{Nie2017}, in the spin-boson model~\cite{Henriet2018}, in the 1D compass-model~\cite{Wang2010,Lue2015}, and in connection to spin-crossover phenomena~\cite{Nesterov2009}. The critical properties of the geometric phase has also been studied in few-body systems interacting with critical chains~\cite{,Cucchietti2010,Yuan2010,Zhang2012,Zhang2013a,Zhang2013c,Wu2015,Zhu2015},  in non-Hermitian critical systems~\cite{Liang2013,Zhang2013}, in connection to dynamical phase transitions~\cite{Hickey2014,Zvyagin2016}, in the characterisation of topological phase transitions~\cite{Lian2012,Ma2012,Ma2013a,Shan2014,Yang2015,Li2015,Li2015a,Zeng2017,Leonforte2019,Leonforte2019a,Bascone2019,Bascone2019a}, in quenched systems~\cite{Basu2010,Basu2010a,Sarkar2014}, in non-equilibrium phase transitions~\cite{Tomka2012,Requist2012,Zhang2013b,Carollo2018,Cai2019}, in connection with entanglement~\cite{Castro2011,AzimiMousolou2013}, and renormalisation group~\cite{Jafari2013}.

Beyond zero-temperature and thermal phase transitions, a novel type of criticalities have recently emerged in the context of non-equilibrium steady states~(NESSes) of open-driven dissipative systems~\cite{Prosen2008,Diehl2008,DallaTorre2010,Diehl2010a,Heyl2013,LeBoite2013,Carr2013,Ajisaka2014,Marcuzzi2014,Vajna2015,Dagvadorj2015,Weimer2015,Macieszczak2016,Jin2016,Rose2016,Bartolo2016,Maghrebi2016,Sieberer2016,Roy2017,Fink2017,Fitzpatrick2017,Rota2017,Overbeck2017,Foss-Feig2017,Jin2018,Rota2018,Minganti2018,Vicentini2018,Nagy2018,Casteels2018,Rota2019}.
In this context, one considers a quantum many body-system coupled to external reservoirs whose dynamics may be described {\em effectively} in terms of a Liouvillian master Eq.~\cite{Alicki2007,Breuer2007}. Under suitable conditions, the open system dynamics reaches a (possibly unique) {\em non-equilibrium steady state}.  The dynamical and static properties of the NESSes characterise the type of phases in which the systems lies.

In this setting, NESS quantum phase transitions~(NESS-QPTs) are revealed as dramatic structural changes of the Liouvillian steady state upon a slight modification of tuneable external parameters. The analogy with equilibrium phase transitions is straightforward. QPTs  are understood through the properties of the (unique) ground state of the Hamiltonian governing the dynamics of the system. Phase diagrams and criticalities are determined by the low-lying spectrum of excitations of the system's Hamiltonian. Similarly, in dynamics governed by a Liouvillian master equation, observable macroscopical changes in the properties of quantum systems are revealed by non-analytical dependences of the NESS with respect to external variables. Phase diagram and criticalities are thus governed by the dependence of the low lying spectrum of the Liouvillian on these parameters.
Despite the analogy with equilibrium QPTs, a comprehensive picture and characterisation of dissipative NESS-QPTs is lacking. This is partly due to their nature lying in a blurred domain, where features typical of zero temperature QPTs coexist with unexpected properties, some of which reminiscent of thermal phase transitions.

A natural approach to the investigation of such a novel scenario
would be to adapt tools used in the equilibrium settings. We will illustrate how the use of the geometric methods, such as geometric phase~\cite{Bohm2003,Berry1984}, and in particular its
mixed state generalisation, the Uhlmann geometric
phase~\cite{Uhlmann1986}, may provide a novel perspective in the investigation of NESS-QPTs.  Geometric phases and related geometrical tools, such as the Bures
metrics~\cite{Bures1969,Uhlmann1976,Braunstein1994}, have been
successfully applied in the analysis of many equilibrium phase
transitions~\cite{Zanardi2006,Zanardi2007,CamposVenuti2007,Gu2010,Dey2012}.
The Bures metrics have been employed in thermal phase
transitions~\cite{Janyszek1999,Ruppeiner1995,Quan2009a} and QPTs, both
in symmetry-breaking~\cite{Zanardi2006,Zanardi2007,CamposVenuti2007,Zanardi2007a,Gu2010,Dey2012,Kolodrubetz2013} as well as in topological phase transitions~\cite{Yang2008}.

Due to their mixed state nature, the NESSes require the use of a
concept of geometric phase in the density operators domain. Among
all possible
approaches~\cite{Uhlmann1986,Sjoqvist2000,Tong2004,Chaturvedi2004,Marzlin2004,Carollo2005a,Buric2009,Sinitsyn2009},
the Uhlmann geometric phase~\cite{Uhlmann1986} stands out for its
deep-rooted relation to information geometry and
metrology~\cite{Matsumoto1997,Hayashi2017}, whose tools have been
profitably employed in the investigation of QPTs and
NESS-QPTs~\cite{Zanardi2007,Kolodrubetz2013,Banchi2014,Marzolino2017}.
Uhlmann holonomy and geometric phase have been applied to the
characterisation of both topological and symmetry breaking
equilibrium
QPTs~\cite{Paunkovic2008,Huang2014,Viyuela2014,Andersson2016,Viyuela2014a,Budich2015a,Kempkes2016,Mera2017}.
Many proposals to measure the Uhlmann geometric phase have been put
forward~\cite{Tidstrom2003,Aberg2007,Viyuela2016}, and demonstrated
experimentally~\cite{Zhu2011}.

Motivated by this, we will illustrate the role of the mixed state analogue of the Berry curvature, the \emph{Uhlmann curvature}, in the characterisation of dissipative NESS-QPTs.  The infinitesimal local properties of the manifold of density matrices can be captured by an object called \emph{mean Uhlmann curvature} (MUC). The MUC is defined as the Uhlmann geometric phase per unit area of a
density matrix evolving along an infinitesimal loop. It turns out that this quantity has also a
fundamental interpretation in multi-parameter quantum metrology: it
marks the \emph{incompatibility} between independent parameters arising
from the quantum nature of the underlying physical
system~\cite{Ragy2016}. In this sense, one can consider the MUC as a measure of
``quantumness'' in the \emph{multi-parameter} estimation problem,
and its singular behaviour responds only to quantum fluctuations
occurring across a phase transitions. We will review these
ideas, by showing paradigmatic examples of Fermionic quadratic
dissipative Lioviullian models, some of which show rich NESS
features~\cite{Prosen2008,Diehl2008,Eisert2010,Banchi2014,Marzolino2014,Marzolino2017}.

This review focuses on the role of the geometric phases in the investigation of equilibrium and non-equilibrium quantum phase transitions. However, we wouldn't provide a fair account of the subject without touching upon the closely related subject of the fidelity approach~\cite{Zanardi2006,You2007,Zanardi2007,Zhou2008,Schwandt2009}, whose main achievements have been extensively reviewed in~\cite{Gu2010}. The original idea of the fidelity approach was to identify QPTs  from the remarkable sensitivity of the Loschmidt echo~\cite{Quan2006} and the ground-state overlap~\cite{Zanardi2006,Zhou2008} under the slightest changes of external parameters across criticality. This is a phenomenon which is strongly reminiscent of, and indeed related to, the well known Anderson orthogonality catastrophe~\cite{Anderson1967}. Since its early development this approach has been applied successfully to a wide variety of systems, ranging from quasi-free Fermions~\cite{Zanardi2007c,Cozzini2007}, the Dicke model~\cite{Zanardi2006,You2007,Liu2009,Dey2012}, matrix product state systems~\cite{Cozzini2007a}, Bose-Hubbard models~\cite{Buonsante2007,Lacki2014}, Stoner-Hubbard and BCS models~\cite{Paunkovic2008}. Various methods have been developed to evaluate the fidelity, such as exact diagonalisation, density matrix renormalisation group~\cite{Gu2010,Luo2017,Hickey2017}, quantum Monte Carlo methods~\cite{Schwandt2009,Troyer2015,Weber2018}, tensor network algorithms~\cite{Zhao2010,Su2013}. This provides the means to venture into the study of models where analytical solutions are not available~\cite{Gu2010,Luo2017,Hickey2017,Weber2018,Agarwala2019,Giudici2019}.

The fidelity approach has been employed in the investigation of first order phase transitions~\cite{Vicari2018}, in higher-than-second order QPTs ~\cite{Chen2008}, in finite temperature phase transitions~\cite{Zanardi2007b,Zanardi2007a,You2007,Paunkovic2008,Quan2009a} and non-equilibrium phase transitions~\cite{DeGrandi2010,Polkovnikov2011,Banchi2014,Hannukainen2017,Marzolino2017}.
The efficacy of this approach in signalling QPTs  that elude the Landau-Ginzburg-Wilson paradigm has been tested against Berezinskii-Kosterlitz-Thouless phase transitions~\cite{Yang2007,Fjaerestad2008,Wang2010d,Wang2012,Sun2015}, topological phase transitions~\cite{Abasto2008a,Yang2008,Zhao2009,Garnerone2009}, and in strongly correlated models whose critical features are yet to be understood~\cite{Rigol2009,Jia2011,Agarwala2019,Giudici2019}.
The universal scaling properties of quantities related to the fidelity approach, such as the fidelity susceptibility~\cite{Schwandt2009,Albuquerque2010,Rams2011}, the quantum geometric tensor~\cite{CamposVenuti2007,Gu2010} and the Fisher information matrix~\cite{Zanardi2008}, have been investigated and exploited in numerous scenarios.

Finally, one of the relevant features of these geometrical approaches is the potential experimental accessibility of many of their related quantities. Fidelity susceptibility~\cite{Zhang2008,Kolodrubetz2013,Gu2014}, Fisher information~\cite{Zanardi2007,Heyl2013}, geometric phases~\cite{Tidstrom2003,Aberg2007,Viyuela2016} and geometric curvature~\cite{Tran2017,Leonforte2019,Leonforte2019a} can be probed through experimentally viable procedures.  

%

To keep the review as pedagogical as possible, we will mostly focus the discussion on criticalities occurring in transverse field models, namely, Ising and XY models in a
transverse magnetic field, which are paradigmatic examples
exhibiting zero-temperature continuous QPTs . This will also pave the
way to the second part of this review, where we will consider
non-equilibrium versions of similar models. 
More than fifty years since their appearance, it is astounding how
useful the transverse field models continue to be in the understanding of QPTs , non-equilibrium dynamics of quantum critical systems, and their connections to quantum information and geometric information.

\section{Continuous Phase Transition and Universality}
We will consider only continuous QPTs , i.e. phase transitions associated to an order parameter that vanishes continuously at criticality, as opposed to first order QPT, characterised instead by abrupt changes in the order parameter.
According to the standard classification, first order phase transitions are usually signalled by a finite discontinuity in the first derivative of the ground state energy density, whereas, continuous QPTs  are identified by continuous first derivatives but discontinuous second derivatives of the ground state energy density.

Since many of the salient features of continuous phase transitions
are common to classical and quantum phase transitions, let's take a
brief detour into the topic of phase transitions driven by thermal
fluctuations, and review some of the basic notions that will recur
in the discussion on QPTs .

Assume that, at a given temperature $T = T_{c}$, a classical system
with translation invariance symmetry in $d$-spatial dimensions has a
critical point (CP). Any such critical point is characterised by
several critical exponents, which identify most of the remarkable
features of the system behaviour around the
CP~\cite{Amit1984,Cardy1996}. Let $O(\bm{r})$, with $\bm{r}\in
\mathbb{T}^{d}$, denote an order parameter, i.e. a variable whose
thermal expectation value $\bk{O(\rr)}=0$ for $T \ge T_{c}$ and
$\bk{O(\rr)}\neq 0$ for $T  < T_{c}$. If $T  < T_{c}$, the
\emph{two-point correlation function} of the order parameter falls
off exponentially at large distances, i.e.
$\bk{O(\rr_{1})O(\rr_{2})}\sim\exp( -|\bm{r}_{1} -\bm{r}_{2}|/\xi)$,
where $\xi$ is the so called \emph{correlation length}, which is a
function of $T$. As the phase transitions is approached by varying
some parameter, several remarkable phenomena occur. One of most
notable is the divergence of $\xi$ as $T$ approaches $T_{c}$ from
above, namely, as $T\to T_{c}^{+}$, usually as a power law which is
characterised by the critical exponent, $\nu$, of correlation
length, i.e. $\xi\sim (T-T_{c})^{\nu}$. At criticality, i.e. at
$T=T_{c}$, the divergence of the correlation length indicates that
the two point correlation function decays only algebraically with
the distance, namely $\bk{O(\bm{r}_{1})O(\bm{r}_{2})}\sim
|\bm{r}_{1} -\bm{r}_{2}|^{-{d-2+\eta}}$, where $\eta$ is another
critical exponent. Apart from the above long-ranged spacial
correlation, similar scaling can be observed in the temporal
behaviours at criticality. They are characterised by the dynamical
critical exponent, $z$, which defines the scaling law of the system
response-time $\tau_{c}=\omega_{c}^{-1}$ to external perturbations,
which diverges like $\tau\sim \xi^{z}$ as $T \to T_{c}^{+}$, a
phenomenon-known as \emph{critical slowing down}. Correlation length
$\xi$ and response-time $\tau_{c}$ set the only characteristic
length scale and characteristic time scale close to the phase
transitions. Infinite correlation length and time at criticality
entails fluctuations on all length scales and timescales, and the
system is said to be scale-invariant. A direct consequence is that
every observable depends on the external parameters via power laws.
The corresponding exponents are the critical exponents.

Apart from $z$, $\nu$ and $\eta$, other critical exponents are:
$\beta$, which characterises the vanishing of the order parameter as
the critical temperature is approached from below, i.e.
$\bk{O(\rr)}\sim (T_{c}-T)^{\beta}$ as $T\to T_{c}^{-}$; $\alpha$
which defines the divergence of the singular part of the specific
heat $C\sim (T-T_{c})^{-\alpha}$ as $T \to T_{c}^{+}$. By denoting with
$h$ the conjugate field to the order parameter, i.e. the field which
couples to the $O(\rr)$ in the Hamiltonian of the system, one
defines the zero-field susceptibility
$\chi_{c}:=d\bk{O(\rr)}/dh|_{h=0}$, whose critical exponent is
$\gamma$. The latter characterises the divergence $\chi_{c}\sim
(T-T_{c})^{-\gamma}$ as $T \to T_{c}^{+}$. Exactly at $T = T_{c}$,
the order parameter scales with $h$ as $O(\rr)\sim
|h|^{\frac{1}{\delta}}$ as $h \to 0$. The critical exponents, whose
values characterise the type of criticality, are not independent of
each other, many of them are related to one another through a set of
\emph{scaling laws} or \emph{exponent relations}.

Moreover, it turns out that the complete set of critical exponents
not only are mutually dependent, but they are the same across a
whole class of phase transitions, called \emph{universal class}.
Indeed, one of the most astounding and far-reaching hallmark of
continuous phase transitions is the notion of \emph{universality}.
Its most direct implication is that critical exponents are only
determined by the symmetry of the order parameter, the
dimensionality of the system and the nature of the fixed point,
regardless of microscopic details of the Hamiltonian. The divergence
of the correlation length and characteristic time are responsible
for such universal behaviour. Close to criticality, every physical
quantity is averaged over all lengths that are smaller than the
physically relevant scale set by the correlation length. This
suggests that the universal critical behaviour can be satisfactorily
described by an effective theory that keeps only the asymptotic
long-wavelength degrees of freedom of the original Hamiltonian. On
the other hand, specific properties which depends on the microscopic
details of the model, like e.g. the critical temperature, are not
accessible through this coarse graining procedure, and are called
non-universal properties. The paradigm of universality shifts the
focus of the investigation from the specific model reproducing a
peculiar critical phenomenon-occurring in nature to the study of an
entire class of seemingly unrelated problems which are governed by
the same universal features. Beyond its undoubted theoretical
interest, this universality classification has an obvious practical
advantage: relevant informations regarding a specific Hamiltonian
model may be gleaned by exploring a simpler instance of the same
universality class.

The above consideration on finite temperature phase transitions can
be illustrated by turning to a specific example, the classical Ising
model with ferromagnetic next-neighbour interaction, described by
the Hamiltonian
\begin{equation}
H= -J\sum_{<\rr \rr'>}\sz_{\rr}\sz_{\rr'} - h \sum_{\rr}\sz_{\rr}
\end{equation}
where, $\sum_{<\rr \rr'>}$ denotes summation over next-neighbouring
sites of a $d$-dimensional hyper-toric lattice, and $\sz_{\rr}$ is a
Pauli $z$ matrix encoding a classical spin variable (assuming values
$\pm 1$), which lies on the $\rr-$th site of the lattice and which
couples to a magnetic field of intensity $h$. For a vanishing
magnetic field, the model shows a $\mathbb{Z}_{2}$ symmetry which
may be spontaneously broken in a ferromagnetic phase with non-zero
spontaneous magnetisation $m=1/N\langle \sum_{\rr}
\sz_{\rr}\rangle_{h\to 0}$. The system undergoes a phase transitions
as the critical temperature $T_{c}$ is crossed from below, from the
ferromagnetic phase ($T\le T_{c}$) to a $\mathbb{Z}_{2}$ symmetric
paramagnetic phase, with $m=0$ ($T>T_{c}$). In the special cases of
$d\le 2$, this model can be exactly solved with the transfer matrix
method. For $d=1$, no spontaneous magnetisation are exhibited at any
finite temperature $T>0$; whereas, at $T=0$ the system shows
long-range correlation, with a correlation length $\xi$ which
diverges exponentially as $T\to T_{c}=0$. For $d\ge2$, the classical
Ising model exhibits positive critical temperatures $T_{c}>0$. One
can appreciate in this example an instance of a general feature of
classical critical models, i.e. the existence of a lower critical
dimension, in this case $d_{c}^{l}=1$, which denotes the highest
integer dimension for which a model displays a vanishing critical
temperature. In the $d=2$ case, the model can be mapped to a
one-dimensional transverse quantum Ising chain, which can be solved
exactly, and provides a means of calculating all critical exponents.
For example, the correlation length exponent is $\nu=1$, while the
exponent associated to the spontaneous magnetisation
$m=|T-T_{c}|^{\beta}$ is $\beta=1/8$.

\section{Quantum vs Classical Phase Transitions}
Let us turn briefly to a fundamental question: to what extent a
quantum mechanical formulation of a model is necessary in order to
understand its critical phenomena, and to what extent will classical
physics suffice? Are there distinctive features of a quantum
critical phenomenon which cannot be understood in terms of a
classical statistical model?

Within classical statistical mechanics, the statistical properties
of a model are described in terms of its (canonical) partition
function $Z:=\Tr[ e^{- \beta H}]$, with $\beta:=1/k_{B} T$ and
$H=K+U$, which factorises into one term depending on the kinetic
energy $K$ only and another that depends only on the potential
energy $U$, i.e. $Z:=\Tr[ e^{- \beta K}]\Tr[ e^{- \beta U}]$. As a
result, the static properties of the system can be studied
independently of the dynamical ones. Moreover, the kinetic energy,
which is generically expressed as $K:=\sum_{j} p_{j}^{2}/(2m)$ in
terms of the generalised momenta $p_{i}$, contributes to the
partition function $Z$ with a Gaussian factor $e^{-\beta K}$, which
is analytic. Hence, any singular behaviour of $Z$ which results in a
phase transitions can only arise from the potential energy factor. In
particular, the dynamical exponent $z$ is independent of all of
the other critical exponents, and the static critical behaviour can
be studied by means of an effective functional of a time-independent
order parameter. By contrast, in quantum systems the situation
differs quite fundamentally, as in general $[K,U]\neq 0$. One then
sees that the statics and the dynamics of quantum systems are
intrinsically coupled and need to be treated together and
simultaneously. As a consequence, the dynamical exponent $z$ becomes
an integral part of the set of exponents of a given universality
class.

A way of telling whether quantum mechanics is important in the
description of a critical phenomenon is comparing the thermal
fluctuations of the order parameter and its quantum fluctuations,
which are set by the smallest energy scale of the relevant quantum
degrees of freedom. A rule of thumb is contrasting the two most
significant energy scales, namely $\hbar \omega_{c}$, the energy of
long-distance order parameter fluctuations, and the thermal energy,
$k_{B} T$. In practise, one says that the order parameter
fluctuations changes its character from quantum to classical when
$\hbar \omega_{c}$ falls below $k_{B} T$. However, for any thermal
phase transitions, the typical energy scale vanishes as $\hbar
\omega_{c}\sim \xi^{z}\sim |T-T_{c}|^{z \nu}$. Therefore quantum
mechanics becomes necessarily unimportant, and the critical
behaviour at the transitions is dominated by classical fluctuations.
This explains the name ``classical transitions'' for transitions
occurring at finite temperature.

\section{Quantum Phase Transitions}
The situation is different for transitions occurring at $T = 0$,
driven by a set of non-thermal control parameters
$\bm{\lambda}=\{\lambda_{1}\dots \lambda_{N}\}$, like pressure or
magnetic field. In this case the fluctuations in the order parameter
are dominated by quantum mechanics, which therefore justifies
calling this type of criticality ``quantum'' phase
transitions~\cite{Sachdev2011}.

Quantum phase transitions are therefore criticalities arising for
$T=0$, where the system lies in its Hamiltonian ground state. The
ground state is generally \emph{uniquely} determined by the values
of the parameters $\bm{\lambda}$ of its Hamiltonian
$H(\bm{\lambda})$, unless something ``singular'' happens in the
spectrum of the many-body system for some critical values
$\bm{\lambda}_{c}$.  Such a non-analiticity may be due to a simple
level crossing in the many-body ground state. This possibility can
only arise when a $\lambda$ couples to a conserved quantity of the
full Hamiltonian, i.e. $H(\lambda)=H_{0}+\lambda H_{1}$, with
$[H_{0},H_{1}]=0$. This means that, while the eigenvalues will
change as a function of the Hamiltonian parameters, the eigenstates
will be independent of $\lambda$. Hence, a level-crossing may well
occur, creating a point of non-analyticity of the ground state
energy, however, it will not determine critical singularities in the
correlations, and it gives rise to a first-order quantum phase
transitions. This is a type of transitions that also finite-size
systems can exhibit.

A totally different story is what happens for continuous phase
transitions, which are characterised by higher-order singularities
in the energy density. This occurs when a system ground state
energy, whose finite-size spectrum displays an avoided level
crossing, reaches an infinitely sharp transitions in the thermodynamic
limit. This involves infinitely many eigenstates of the many-body
system, and the thermodynamic limit is needed for such a singularity
to arise. In this case, it is the non-commutativity of the Hamiltonian
terms which is responsible for the quantum fluctuations which drive
the systems across the quantum phase transitions. One might think
that phase transitions occurring at zero temperature are not
physically relevant to the actual world. However, one can show that
many finite temperature features of a system can be gleaned through
the properties of its quantum critical point.

From the above considerations, the point of singularity in the
ground state energy density is associated with an energy scale
$\Delta$ which vanishes as $\lambda$ approaches a critical value
$\lambda_{c}$. This energy scale is generally identified by the
energy difference between the ground state and the first excited
state, i.e. the \emph{energy gap}, and its dependence on the system
parameter is generally algebraic in the proximity of the criticality
i.e.
\begin{equation}
\Delta\sim J |\lambda-\lambda_{c}|^{\nu z}.
\end{equation}
Here $J$ is an energy scale associated to the microscopic details of
system couplings, and $z$ and $\nu$ are critical exponents
characteristic of the critical point $\lambda_{c}$, which are
defined as follows. Similarly to continuous thermal phase
transitions, a QPTs may be characterised by an order parameter
$O(\bm{r},t)$, which is an observable whose expectation value
vanishes continuously, as a function $\lambda$, across the critical
point $\lambda_{c}$, going from one phase (the \emph{ordered phase})
to the other (the \emph{disordered phase}). One can define a length
scale $\xi$ which typically characterises the exponential decay of
the equal-time two-point correlation function of the ground state,
\begin{equation} G(\bm{r}-\bm{r}'):=\bk{O(\bm{r},t)
O(\bm{r}',t)}-\bk{O(\bm{r},t)} \bk{O(\bm{r}',t)} \sim
\frac{e^{-|\bm{r}-\bm{r}'|/\xi}}{{|\bm{r}-\bm{r}'| }^{d-2+\eta}}.
\end{equation} Here, $\eta$ is another critical exponent, characterising the
power-law decay of correlations $G(\bm{r})\sim
|\bm{r}|^{-{d-2+\eta}}$ at exactly $\lambda=\lambda_{c}$. In quantum
continuous phase one invariably observes the algebraic divergence of
the correlation length approaching the critical point \begin{equation} \xi\sim
|\lambda-\lambda_{c}|^{-\nu}. \end{equation} Similarly, one can define a time
scale $\tau_{c}$ for the decay of equal-space correlations at
quantum phase transitions 
\begin{equation} \tau_{c} \sim \Delta^{-1} \propto
\xi^z \propto |\lambda-\lambda_c|^{-\nu z}. 
\end{equation}
%


\section{Geometric Phase and Criticality}
A characteristic that all non-trivial geometric evolutions have in
common is the presence of non-analytic points in the energy
spectrum. At these points, the state of the system is not well
defined owing to their degenerate nature. One could say that the
generation of a geometric phase (GP) is a witness of such singular
points. Indeed, the presence of degeneracy at some point is
accompanied by curvature in its immediate neighbourhood and a state
that evolves along a closed path is able to detect it. These points
or regions are of great interest to condensed matter or molecular
physicists as they determine, to a large degree, the behaviour of
complex quantum systems. The geometric phases are already used in
molecular physics to probe the presence of degeneracy in the
electronic spectrum of complex molecules. Initial considerations by
Herzberg \& Longuet-Higgins~\cite{Herzberg1963} revealed a sign
reversal when a real Hamiltonian is continuously transported around
a degenerate point. Its generalization to the complex case was
derived by Stone~\cite{Stone1976} and an optimisation of the real
Hamiltonian case was performed by Johansson \&
Sj\"oqvist~\cite{Johansson2004,Johansson2005}.

Geometric phases have been associated with a variety of condensed
matter and solid-state phenomena. They are at the core of the
characterisation of topological phase
transitions~\cite{Thouless1983,Bernevig2013,Chiu2016}, and have been
employed in the description and detection of QPT, both
theoretically~\cite{Carollo2005,Pachos2006,Plastina2006,Hamma2006,Zhu2006,Reuter2007,Patra2011}
and experimentally~\cite{Peng2010}. However, their connection to
quantum phase transitions has been put forward only in the last decade~\cite{Carollo2005,Pachos2006}. It was further elaborated by Zhu~\cite{Zhu2006}, where the critical exponents were
evaluated from the scaling behaviour of geometric phases, and by
Hamma~\cite{Hamma2006}, who showed that geometric phases can be used
as a topological test to reveal quantum phase transitions. The use
of GP in QPTs can be heuristically understood as follows. As we have
seen, quantum phase transitions are associated by dramatic
structural changes of the system state, resulting from small
variations of control parameters. These critical changes are
accompanied by the presence of degeneracies in the ground state
energy density of the many-body system. The degeneracy are at the
origin of the non-analiticity of the ground-state wave function,
which   characterises the long-range quantum correlations at
criticality~\cite{Carollo2005,Pachos2006,Plastina2006,Hamma2006,Zhu2006,Reuter2007,Patra2011}.

In this section we explore the ability of geometric phase and
related quantities to reveal quantum critical phenomena in many-body
quantum systems. The use of geometric phases provides a new
conceptual framework to understand quantum phase transitions, and at
the same time suggests novel viable approaches to experimentally
probe criticalities. This maybe done through the evolution of the
quantum many-body system in the neighbourhood of a critical point,
in a way that does not take the system directly through a quantum
phase transitions. The latter is hard to physically implement as it
is accompanied by multiple degeneracies that can take the system
away from its ground state.

Moreover, the geometric phase approach is not based on the
identification of  an order parameter - and therefore does not
require a knowledge of symmetry breaking patterns - or more in
general on the analysis of any distinguished observable, e.g.,
Hamiltonian, but it is a purely geometrical characterisation.

This geometric phase approach has been applied in quite a few
explicit
models~\cite{Carollo2005,Plastina2006,Reuter2007,Patra2011}.
However, for the sake of simplicity we will discuss two of them
which stand out for their richness and at the same time for their
simplicity, i.e. the paradigmatic one dimensional Ising, and XY models in
transfer magnetic field~\cite{Carollo2005,Pachos2006} and the Dicke
model~\cite{Plastina2006}, whose geometric phase properties will be
briefly discussed in the next few sections.

The one-dimensional version of the transverse Ising model emerged for the first time in
the resolution of the two-dimensional nearest-neighbor ferromagnetic
classical Ising model. The row-to-row transfer matrix of the
two-dimensional classical model converges to the transverse Ising
chain in a suitable limit~\cite{Lieb1961}, and its exact solution
soon followed~\cite{Katsura1962}. This correspondence represents the
prototypical example of a quantum-to-classical mapping. The model
was employed, shortly later, to reproduce the ferroelectric
transitions in Potassium Dihydrogen Phosphate~\cite{DeGennes1963}.

The Ising model and its generalisation, the XY model, are analytically solvable and they offer enough control parameters to support geometric evolutions. Moreover, its rich
criticality structure the XY model includes the XX critical model. By explicit calculations one can observe that an excitation of the model gains a non-trivial geometric phase if and
only if it circulates a region of criticality, a feature which
expresses the topological origin of the geometric phase.
The generation of this phase can be traced down to the presence of a
conical intersection of the energy levels located at the XX
criticality in an equivalent way used in molecular systems. The
scaling of the geometric phase can be used to obtain the critical
exponents that completely characterise the critical behaviour. It is
not hard to generalise these results to the case of an arbitrary
spin system, which sheds light on the understanding of more general
systems, where analytic solutions might not be available.

\subsection{Berry's phase}\label{sec:BP}
\indent Historically, the definition of geometric phase was
originally introduced by Berry in a context of closed, adiabatic,
Schr\"odinger evolution. What Berry showed in his seminal
paper~\cite{Berry1984} was that a quantum system subjected to a
slowly varying Hamiltonian manifests in its phase a geometric
behaviour due to the structure of the Hilbert space. When the
Hamiltonian of a system evolves cyclically and slowly enough, any
(non-degenerate) eigenstate of the initial Hamiltonian evolves
adiabatically following the instantaneous eigenspace. When the
Hamiltonian returns to its original value, the system is eventually
brought back to the {\em ray} of its initial state, acquiring a
phase that, apart from the usual dynamical phase, is geometrical in
nature. This phase is called \emph{Berry's phase}.

Let's summarise the derivation of the Berry phase and show how this
geometric property comes naturally from the solution of the
Schr\"odinger equation. Consider a Hamiltonian $H(\bm{\lambda})$
depending on some external parameters $\bm{\lambda}=(\lambda_1,
\lambda_2,\dots, \lambda_n)$, and suppose that these parameters can
be varied arbitrarily inside a space $\mathcal{M}$ (the parameter
space). Assume that for each value of $\lambda$ the Hamiltonian has
a completely discrete spectrum of eigenvalues, given by the equation
\begin{equation}\label{eq:spectr}
H(\bm{\lambda})\ket{n(\bm{\lambda})}=\epsilon_n(\bm{\lambda})\ket{n(\bm{\lambda})},
\end{equation}
where $\ket{n(\bm{\lambda})}$ and $\epsilon_n(\bm{\lambda})$ are
smooth concatenations of eigenstates and eigenvalues, respectively,
of $H(\bm{\lambda})$, as functions of the parameters $\bm{\lambda}$.
 Suppose that the values of the parameters evolve smoothly
along a curve $\bm{\lambda}(\tau)\in\mathcal{M}$ $(\tau\in[a,b])$,
respecting the prescription of the adiabatic theorem, i.e that the
rate at which the parameters evolve is low compared to the time
scales of the Bohr frequencies of the system
$(\epsilon_n(\bm{\lambda})-\epsilon_m(\bm{\lambda}))$ (as usual
assume $\hbar=1$). The adiabatic theorem states that, under such a
regime, an eigenstate of the system evolves following ``rigidly" the
transformation of the Hamiltonian: a system, initially prepared in an
eigenstate with eigenvalue $\epsilon_n(\bm{\lambda}(a))$, remains at
any instant $t$ of its evolution in the eigenspace
$\epsilon_n(\bm{\lambda}(t))$, i.e. the eigenspace smoothly
connected with the initial eigenspace $\epsilon_n(\bm{\lambda}(a))$.

Thus, in the simplest case of non-degenerate eigenvalues, since the
eigenspace is one-dimensional, the evolution of any eigenstate is
specified  by the spectral decomposition~(\ref{eq:spectr}) \emph{up
to a phase factor}. The adiabatic approximation introduces a
constraint only on the direction of the vector state at any instant
of time, and the eigenvalue equation (\ref{eq:spectr}) implies no
relation between the phases of the eigenstates
$\ket{n(\bm{\lambda})}$ at different $\bm{\lambda}$'s. Thus, for the
present purpose any (smooth) choice of phases can be made. Then, a
state $\ket{\psi_n}$ initially in the eigenstate
$\ket{n(\bm{\lambda}(a))}$, evolves as
\begin{equation}\label{esp:state}
\ket{\psi(t)_n}\simeq\exp\left\{-i\int_{a}^{t}\epsilon_n(\bm{\lambda}(\tau))d\tau\right\}\exp{i\phi^{B}_n(t)}
\ket{n(\bm{\lambda}(t))},
\end{equation}
where the first phase factor is the usual dynamic one, and the
second one is an additional phase that is introduced to solve the
dynamics of the system.

The novel idea introduced by Berry was to recognise that this
additional phase factor has an \emph{inherent geometrical meaning}.
The crucial point is that this phase $\phi^{B}_n$ is non-integrable,
i.e. it cannot be written as a single valued function of
$\bm{\lambda}$. Its actual value must be determined as a function of
the path followed by the state during its evolution. The most
important thing is that this value depends only on the geometry of
this path, and not on the rate at which it is traversed. Or, in a
more formal way, it is independent of the parameterisation of the
path.

Under the assumption of the adiabatic approximation, this phase can
be determined by requiring that $\ket{\psi(t)_n}$ satisfies the
solution of the Schr\"odinger equation. A direct substitution of the
expression (\ref{esp:state}) into
\begin{equation}
H(\bm{\lambda}(t))\ket{\psi(t)_n}=i\frac{d}{dt}\ket{\psi(t)_n},
\end{equation}
leads to
\begin{equation}\label{esp:dgamma}
\frac{d\phi^{B}_n(t)}{dt}=
i\bra{n(\bm{\lambda}(t))}\frac{d}{dt}\ket{n(\bm{\lambda}(t))}.
\end{equation}
Therefore $\phi^{B}_n(t)$ can be represented by a path integral in
the parameter space $\mathcal{M}$,
\begin{equation}\label{esp:intgam}
\phi^{B}_n(t)=i\int_a^b\bra{n(\bm{\lambda}(t))}\frac{d}{dt}\ket{n(\bm{\lambda}(t))}dt=\int_{\bm{\lambda}(a)}^{\bm{\lambda}(b)}
\bm{A}^{B},
\end{equation}
where $\bm{A}^{B}:=\sum_{\mu} A^{B}_{\mu}d\lambda_{\mu}$ is called
\emph{Berry connection} (one-form), a differential form, whose
elements, with respect to the local coordinates $\{\lambda_{\mu}\}$,
are defined as
\begin{equation}\label{def:1form}
A_{\mu}^{B}=i\bra{n(\bm{\lambda})}\partial_{\mu}\ket{n(\bm{\lambda})},
\end{equation}
where $\partial_{\mu}:=\partial/\partial\lambda_{\mu}$. This phase
$\phi^{B}_n$ becomes physically relevant and non-trivial only when
the parameters are changed along a closed path, such that
$\bm{\lambda}(a)=\bm{\lambda}(b)$. Otherwise, the geometrical phase
can be factored out by choosing a suitable eigenbasis. This
non-trivial phase, \emph{the Berry phase}, is then given by
\begin{equation}\label{eq:BP}
\phi^{B}_n(C)=\oint_C \bm{A}^{B}.
\end{equation}
This is nothing but a line integral of a vector potential,
(analogous to the electromagnetical vector potential) around a
closed path in the parameter space. As this quantity is not
identically zero, it implies that the phase acquired is not
integrable in nature. It is not possible to define the phase as a
single valued function of the parameter space, because {\em the
phase depends on the previous history} of the state, i.e. on the
path that it has followed to arrive to this point. By exploiting the
Stokes theorem in the n-dimensional space $\mathcal{M}$, this
quantity can be written as an integral on the (oriented) surface
$\Sigma(C)$ bounded by the closed curve $C$
\begin{equation}
\phi^{B}_n(C)=\int_{\Sigma(C)} \bm{F}^{B},
\end{equation}
where $\bm{F}^{B}:=d\bm{A}^{B}=\frac{1}{2}\sum_{\mu\nu}
F^{B}_{\mu\nu} d\lambda_{\mu}\wedge d\lambda_{\nu}$ is the
\emph{Berry curvature} (differential two-form), where
$d\lambda_{\mu}\wedge d\lambda_{\nu}$ is the infinitesimal surface
element of $\Sigma(C)$, spanned by the two independent directions
$\lambda_{\mu}$ and $\lambda_{\nu}$ of $\mathcal{M}$, and
$[F^{B}_{\mu\nu}]$ is an antisymmetric tensor field (analogous to
the electromagnetic tensor field), with components
 \begin{equation}\label{eq:2form}
 F^{B}_{\mu\nu}:=\partial_{\mu} A_{\nu} - \partial_{\nu} A_{\mu}=\bk{\partial_{\mu} n(\bm{\lambda})|\partial_{\nu} n(\bm{\lambda})}-\bk{\partial_{\nu} n(\bm{\lambda})|\partial_{\mu} n(\bm{\lambda})},
 \end{equation}
or, in a coordinate independent way,
\begin{equation}
\bm{F}^{B} = \bra{d n(\bm{\lambda})}\wedge \ket{d n(\bm{\lambda})},
\end{equation}
where $\ket{d n}:=d\ket{ n}$. \footnote{Throughout this review we
will sometimes make use of ``$d$'' in the sense of exterior
derivative. In differential geometry, the exterior derivative is the
generalisation of the concept of differential applied to $k$-forms.
For a scalar $f$ (a differential $0$-form), the application of the
exterior derivative yields the usual differential of calculus, i.e.
$df=\sum_{\mu}\partial_{\mu}f d \lambda_{\mu}$. In general, the
exterior derivative is the unique mapping from $k$-forms to
$(k+1)$-forms, satisfying the following properties:
\begin{enumerate}
\item $df$ is the differential of $f$ for smooth functions $f$;
\item $d(df)=0$, (or in short $d^{2}=0$), for any (smooth) form $f$;
\item $d(f\wedge g)=df \wedge g + (-1)^{k} f \wedge dg$, where $f$ is a $k$-form.
\end{enumerate}
This leads to a generalised Stokes' theorem (or the generalised fundamental theorem of calculus), in the following form;\\
if $\mathcal{M}$ is a compact smooth orientable n-dimensional
manifold with boundary $\partial {\mathcal{M}}$, and $\omega$ is an
$(n-1)$-form on $\mathcal{M}$, then~\cite{Nakahara1990}
\[ \int _{\mathcal{M}} d\omega = \int _{\partial {\mathcal{M}}}\omega\,.  \]
} As already pointed out, $\ket{n(\bm{\lambda})}$ is only one choice
out of infinitely many possible concatenation of states that satisfy
the eigenvalue Eq.~(\ref{eq:spectr}). If
$\ket{n(\bm{\lambda})}$ is a solution of~(\ref{eq:spectr}),
$e^{i\alpha(\bm{\lambda})}\ket{n(\bm{\lambda})}$ also satisfies the
same eigenvalue problem. The freedom of choice of a particular
$\ket{n(\bm{\lambda})}$ is often referred to a \emph{gauge freedom},
in analogy to the $U(1)$ gauge freedom of the electromagnetic vector
potential. From the definition (\ref{def:1form}) it follows that the
Berry connection $A$ does depend on the gauge choice; indeed it
transforms according to
\begin{eqnarray}
\ket{n(\bm{\lambda})} & \rightarrow &\ket{n^\prime(\bm{\lambda})}=e^{-i\alpha(\bm{\lambda})}\ket{n(\bm{\lambda})}\nonumber\\
A_\mu^{B} & \rightarrow &
{A^{B}}^\prime_{\mu}=A_\mu^{B}+\partial_{\mu}\alpha(\bm{\lambda})\nonumber
\end{eqnarray}
Notice that such phase transformation reveals a gauge structure of
the Berry connection analogous to the one of an electromagnetic
vector potential. By following the electromagnetic analogy, we can
expect that, although $A$ is gauge dependent, the tensor field
$F^{B}_{\mu\nu}$ should be invariant under this transformations. It
is easily verified that under gauge change, one gets
\begin{equation}\label{eq:gaugeinvF}
    F^{B}_{\mu\nu}\to {F^{B}}^{\prime}_{\mu\nu}=\partial_\mu {A^{B}}^{\prime}_{\nu} - \partial_\nu {A^{B}}^{\prime}_{\nu}=\partial_{\nu} A_\mu -
    \partial^{2}_{\mu\nu}\alpha  - \partial_\nu A^{B}_\mu+\partial^{2}_{\nu\mu} \alpha= F^{B}_{\mu\nu},
\end{equation}
which consequently demonstrates the gauge independence of
$\phi^{B}_{n}(C)$ itself. This reinforces the idea that
$\phi^{B}_n(C)$ is indeed a physical observable effect, independent
of \emph{unessential phase choices}. Moreover, the expression of
$\phi^{B}_n(C)$ as a path integral in the parameter space guarantees
that it \emph{does not depend on the rate of traversal} of the
circuit $C$, provided the adiabatic approximation holds. Therefore,
the Berry phase, a natural consequence of the Schr\"odinger
evolution and the adiabatic approximation, respects the essential
requirements to be a geometric feature: (i) gauge independence, (ii)
parameterisation invariance. As we already mentioned, the eigenvalue
equation (\ref{eq:spectr}) implies no relation between the phases of
the instantaneous eigenstates $\ket{n(\bm{\lambda})}$ at different
$\bm{\lambda}$. It is Eq.~(\ref{esp:dgamma}) which imposes
constraints on the phase acquired by the time dependent eigenstates.
By subtracting the
 dynamical phase, this constraints can be rephrased in a
compact form. By absorbing the phase factor into the definition of
the eigenstates as follows
\begin{equation}\label{}
 \ket{\tilde{n}(t)}=e^{i\phi^{B}(t)}\ket{n(\bm{\lambda}(t))},
\end{equation}
Eq.~(\ref{esp:dgamma}) becomes a condition on the possible
eigenstates $\tilde{n}$ satisfying the time evolution along
$\phi^{B}(t)$
 \begin{equation}\label{eq:ParTraspCond}
  \bra{\tilde{n}(t)}\frac{d}{dt}\ket{\tilde{n}(t)}=0.
\end{equation}
This constraint is the \emph{parallel transport condition}, which
literally requires the time derivative of the instantaneous
eigenvector to have vanishing component along the direction of the
eigenvector itself. The term parallel transport is to be understood
in the sense that neighbouring states along the curve are chosen
``as parallel as possible''. This quantitatively means
that Eq.~(\ref{eq:ParTraspCond}) maximises the scalar product between
infinitesimally closed states
\begin{equation}\label{eq:maximiseParal}
|\bk{\tilde{n}(t)|\tilde{n}(t+dt)}|^{2}\simeq
1-2|\bra{\tilde{n}(t)}\frac{d}{dt}\ket{\tilde{n}(t)}|dt.
\end{equation}
Solving Eq.~(\ref{eq:ParTraspCond}) amounts to choosing a
particular smooth concatenation of eigenstates $\ket{\tilde{n}(t)}$
with a special property: each state and its neighbouring are
\emph{in phase}, i.e. $\arg{\bk{\tilde{n}(t)|\tilde{n}(t+dt)}}\simeq
0$. Although, \emph{locally} the states are in phase, a \emph{global
phase} accumulates as the path is traversed. If compared, the two
endpoints of this chain reveal a relative phase which is the Berry's
phase
\[\langle\tilde{n}(T)\ket{\tilde{n}(0)}=e^{i\phi^{B}(C)}.
\]
This is the original result of Berry: the state of the system, after
a closed adiabatic evolution, returns to a state $\tilde{n}(T)$ that
gains an irreducible part in its phase $\phi^{B}({\cal C})$, in
addition to the dynamical contribution. This phase, analogously to
the definition of phase that we have shown in the previous section,
has an inherent geometrical meaning, since it does not depend on
either the detail of time evolution or unessential phase
transformations.\\

\subsection{A Simple Two-Level System}
Before going into the details of the geometric phase for a many-body
system, it pays to have a brief detour to the simplest, yet
significant example of geometric phase, i.e. the one arising in a
two-level system. As in the previous subsection, we will deal with
geometric phases arising from the adiabatic evolution of the
eigenstate of a Hamiltonian. Consider a $2\times2$ Hamiltonian
$H(\bm{\lambda})$, where $\bm{\lambda}$ is a set of parameters. It
can always be expressed, up to an irrelevant identity matrix, as 
\begin{equation}
 H(\bm{\lambda})=\mathbf{n}(\bm{\lambda})\cdot
\bm{\sigma}=|\mathbf{n}| g(\theta,\varphi) \sz
g^\dagger(\theta,\varphi),\end{equation} where
$\tan{\theta}:=\sqrt{n_x^2+n_y^2}/ n_z$ and
$\tan{\varphi}:=n_y/n_x$. Here

\begin{equation}
  g=e^{-i\sz\varphi/2}e^{-i\sy\theta/2}
\end{equation}
   is a $SU(2)$
transformation which rotates the operator
$\mathbf{n}\cdot\vec{\sigma}$ to the z-direction, and the vector,
$\bm{\sigma}:=(\sx,\sy, \sz)$, of Pauli's operators is given by 
\begin{equation}
 \sx = \left(\begin{array}{cc}
0 & 1 \\
1 & 0
\end{array}\right) ,\Sp
\sy = \left(\begin{array}{cc}
0 & -i \\
i & 0
\end{array}\right) ,\Sp
\sz = \left(\begin{array}{cc}
1 & 0 \\
0 & -1
\end{array}\right).
\end{equation} Using this parameterisation one can represent the Hamiltonian as
a tridimensional vector on a sphere, centred  in the point of
degeneracy of the Hamiltonian ($|\mathbf{n}|=0$), (see
Fig~\ref{fig:Sphere}).

For $\theta=\varphi=0$ we have that $g=\one$ and the two eigenstates
of the system are given by $|+\rangle=(1,0)^T$ and
$|-\rangle=(0,1)^T$ with corresponding eigenvalues
$E_\pm=\pm|\mathbf{n}|/2$. Let us consider the evolution resulting
when a closed path $C$ is spanned adiabatically on the sphere.
Following the general considerations it is easy to show that the
only non-zero component of the Berry connection, $\mathbf{A}^{B}$,
is given by
\[
A^{B}_{\varphi}=\pm \frac{1}{2}\left(1-\cos\theta \right)
\]
that leads to the Berry  phase
\begin{figure}
\begin{center}
\includegraphics[width=0.45\textwidth]{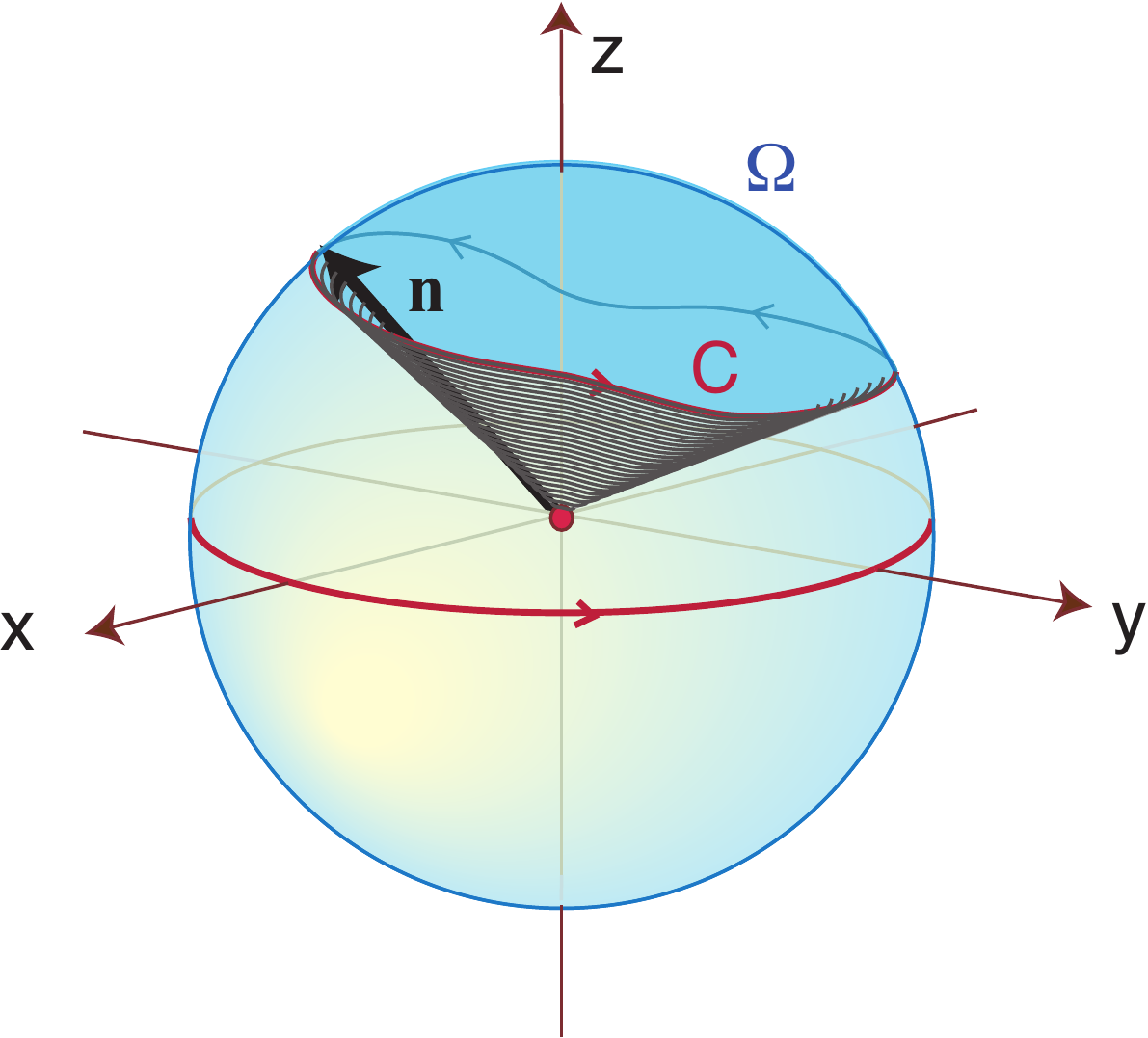}
\caption{The geometric phase is proportional to the solid angle
  spanned by the Hamiltonian with respect of its degeneracy point.}\label{fig:Sphere}
\end{center}
\end{figure}
\begin{equation}\label{spinphase}
\phi^{B}(C)=\oint_{C}\mathbf{A}^{B}\cdot\\dr=\frac{1}{2}
\int_{\Sigma(\theta,\varphi)}\!\!\!\sin{\theta}\;d\theta
d\varphi=\frac{\Omega}{2}.
\end{equation}
Here $\Omega=\int\!\!\int_{\Sigma}sin\theta d\theta d \varphi$ is
the solid angle enclosed by the loop, as seen from the degeneracy
point. In this expression the geometric origin of the Berry phase is
evident. Its value depends only on the path followed by the
parameters and not on the detail of the evolution, and indeed it
depends only on the way in which these parameters are changed in
relation with the degeneracy point of the Hamiltonian.

A particularly interesting case is the one in which the Hamiltonian
can be casted in a real form, corresponding in this case to
$\theta=\pi/2$. In this case the phase~(\ref{spinphase}) becomes
$\pi$ reproducing the sign change of the eigenstate, as this is
circulated around a point of degeneracy ($|\mathbf{n}|=0$), in
agreement with the Longuet-Higgins theorem.

\subsection{The XY Model and Its Criticalities}\label{sec:XY}
The first model that we will employ to illustrate the connection
between geometric phases and critical systems is the spin-1/2 chain
with XY interactions. This is a one-dimensional model where the
spins interact with their nearest neighbours via the Hamiltonian
\begin{equation}
\label{HXYModel} H = -\sum_{j=1}^{n} \left(\frac{1 + \delta}{2}\sx_j
\sx_{j+1} + \frac{1 - \delta}{2}\sy_j \sy_{j+1} +\frac{h}{2}\sz_j
\right),\nonumber
\end{equation}
where $n$ is the number of spins, $\sigma^\mu_j$ are the Pauli
matrices at site $j$, $\delta$ is the x-y anisotropy parameter and
$h$  is the strength of the magnetic field. This model was first
solved explicitly by Lieb, Schultz and Mattis~\cite{Lieb1961} and by
Katsura~\cite{Katsura1962}. Since the XY model is exactly solvable
and still presents a rich structure it offers a benchmark to test
the properties of geometric phases in the proximity of
criticalities.

In particular, we are interested in a generalization of Hamiltonian
(\ref{HXYModel}) obtained by applying to each spin a rotation of
$\varphi$ around the z-direction
\begin{equation}
 \label{eq:HXYphi}
H(\varphi) = g(\varphi)H g^\dag (\varphi)\quad  \textrm{ with }\quad
g(\varphi) = \prod_{j=1}^{n} e^{i\sz_j\varphi},
\end{equation}
in the same way as we did in the case of a single spin-1/2. The
family of Hamiltonians that is parameterised by $\varphi$ is clearly
isospectral and, therefore, the critical behavior is independent of
$\varphi$. This is reflected in the symmetric structure of the
regions of criticality shown in Figure~\ref{criticality}. In
addition, due to the bilinear form of the interaction terms we have
that $H(\varphi)$ is $\pi$-periodic in $\varphi$. The Hamiltonian
$H(\varphi)$ can be diagonalized by a standard
procedure~\cite{Carollo2005}, which can be summarised in  the
following three steps:
\begin{itemize}
\item the Jordan-Wigner transformation, which converts
 spin operators into Fermionic operators via the relations,
\begin{equation}\label{eq:JW} c_l:=\left(\prod_{m<l}\sz_m\right)(\sx_l+i\sy_l)/2, \quad
\{c_{j},c^{\dagger}_{l}\}=\delta_{jl},\quad \{c_{j},c_{l}\}=0; 
\end{equation}
 \item a Fourier transform,
\begin{equation} d_k=\frac{1}{\sqrt{n}}\sum_{l} c_l e^{-i2\pi l k/n}  \textrm{ ,
with } k=-0,\dots, n-1, \end{equation} with , \quad
$\{d_{k},d^{\dagger}_{k'}\}=\delta_{kk'}$, $\{d_{k},d_{k'}\}=0$;
\item a Bogoliubov transformation, which defines the Fermionic
operators, \begin{equation} b_k=d_k \cos\frac{\theta_k}{2}-id_{-k}^\dag
e^{i\phi}\sin\frac{\theta_k}{2}, \end{equation} where the angle $\theta_k$ is
defined as $ \theta_k:=\arccos (\eta_k/\varepsilon_k)$ with
$\eta_k:=\cos{2\pi k \over n}- h$ and
\end{itemize}
\begin{equation} \varepsilon_k:=\sqrt{\eta_k^2+ \delta^2 \sin^2{2\pi k \over n}},
\end{equation} is the energy of the single eigenmode $d_{k}$ of pseudo-momentum
$k$, called \emph{energy dispersion relation}. These procedures
diagonalise the Hamiltonian to a form
\begin{equation}
\label{Hdiagonal} H(\varphi)=\sum_{k=0}^{M}\varepsilon_k b_k^\dag
b_k.
\end{equation}
where either $M=n/2-1$, if $n$ is even, or $M=(n-1)/2$, if $n$ is odd.
The ground state $\ket{g}$ of $H(\varphi)$ is the vacuum
 of the Fermionic modes, $b_k$,
given by
\begin{equation}\label{eq:GSXY}
\ket{g}\!:=\bigotimes_{ k}\Big(\!\cos {\theta_k \over 2}
\ket{0}_{\!k}\ket{0}_{\!\!-k} \!\!-ie^{i\varphi} \sin {\theta_k
\over 2} \ket{1}_{\!k} \ket{1}_{\!\!-k}\!\Big),
\end{equation}
where $\ket{0}_{k}$ and $\ket{1}_k$ are the vacuum and single
excitation of the k-th mode, $d_k$, respectively. The energy gap is
clearly given by the minimum of the energy dispersion relation
$\varepsilon_k$. From (\ref{eq:GSXY}) one can interpret the ground
state as the direct product of $n$ spins, each one oriented along
the direction $(\varphi,\theta_k)$. The critical points in the XY
model are determined by the conditions under which the ground and the first
excited states become degenerate, which in this case amounts to a
vanishing energy dispersion relation $\varepsilon_k$. This is the
only condition in which singularities may arise. There are two
distinct regions of the space diagram that are critical. When
$\delta=0$, we have $\varepsilon_{k}=0$ for $-1\le h\le 1$, which is
a first-order phase transitions with an actual energy crossing and
critical exponents $z=2$ and $\nu=1/2$. The other critical region is
given by $h=\pm 1$ where one finds $\varepsilon_k=0$ for all
$\delta$. These are continuous phase transitions, with finite-size
ground-state spectrum having avoiding-crossing. When $\delta=1$ and
$h=\pm 1$, we obtain the Ising critical model with critical
exponents $z=1$ and $\nu=1$. Finally, let us consider the
criticality behaviour of the rotated $H(\varphi)$. Clearly, the
energy dispersion relation $\varepsilon_k$ does not depend on the
angle $\varphi$, as this parameter is related to an isospectral
transformation. Hence, the criticality region for the rotated
Hamiltonian is obtained just by a rotation around the $h$ axis. This
is illustrated in Figure \ref{criticality}, where the Ising-type
criticality corresponds now to two planes at $h =1 $ and $h=-1$ and
the XX criticality is along the $h$ axis for $|h|<1$.

\begin{figure}
\begin{center}
\includegraphics[width=\textwidth]{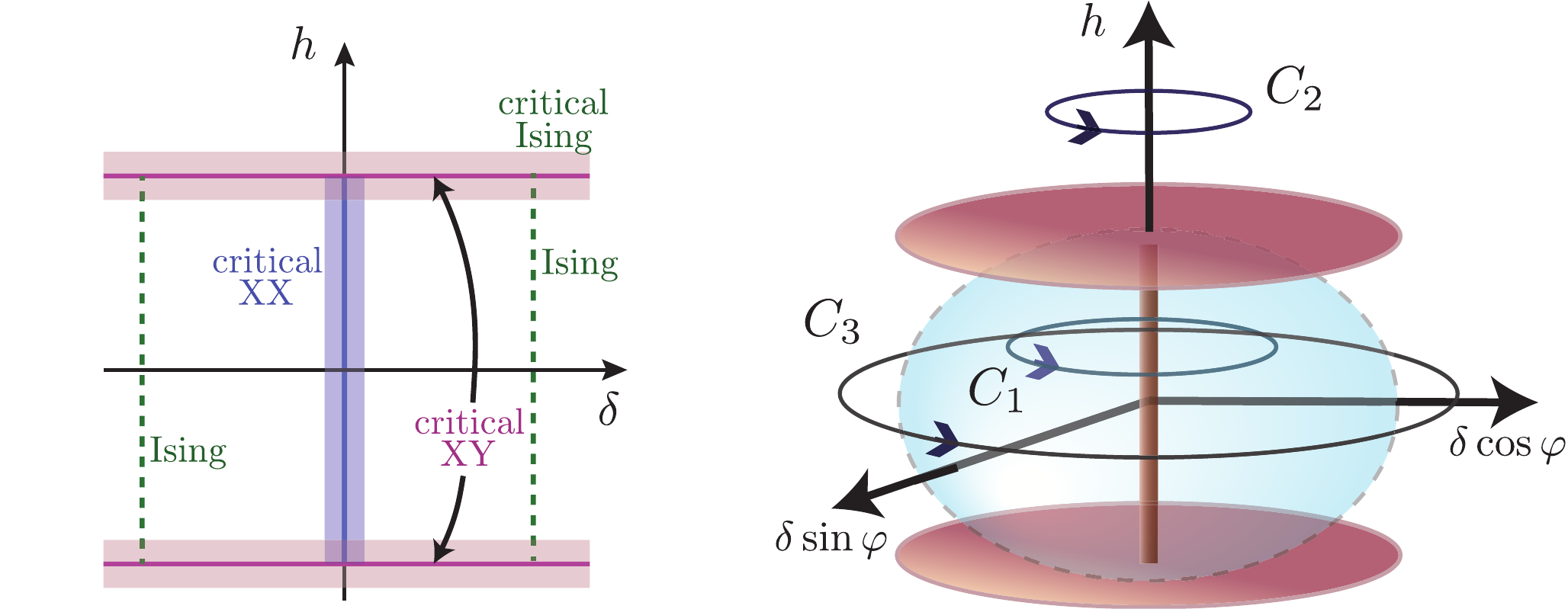}
\vspace{0.5cm} \caption{On the left, the regions of criticality of the $H$
Hamiltonian are presented as a function of the parameters $ h $ and
$\delta$. The corresponding regions of criticality for the Hamiltonian
$H(\varphi)$ are shown on the right, where $\varphi$ parameterises a rotation around the $h$
axis. Possible paths for the geometrical evolutions are depicted,
spanned by varying the parameter $\varphi$.} \label{criticality}
\end{center}
\end{figure}

\subsection{Geometric Phases and XY Criticalities}

Figure 2 depicts the critical points of the XY model. Now, we are
interested in obtaining looping trajectories in the parameter space
described by the Hamiltonian variables $h$, $\delta$ and $\varphi$.
The aim is to determine the geometric evolutions corresponding to
these paths and relate them to regions of criticality. An especially
interesting family of closed paths is the one in which loops
 circulates around the $h$ axis, as the parameter $\varphi$ varies
 from zero to $\pi$. Indeed, these paths may enclose the XX
criticality~\cite{Sachdev2011}, depending on whether $-1<h<1$. It is
possible to evaluate the corresponding geometric phases of the
ground and the first excited states as a function of $h$ and
$\delta$.

\begin{figure}
\begin{center}
\includegraphics[width=\textwidth]
{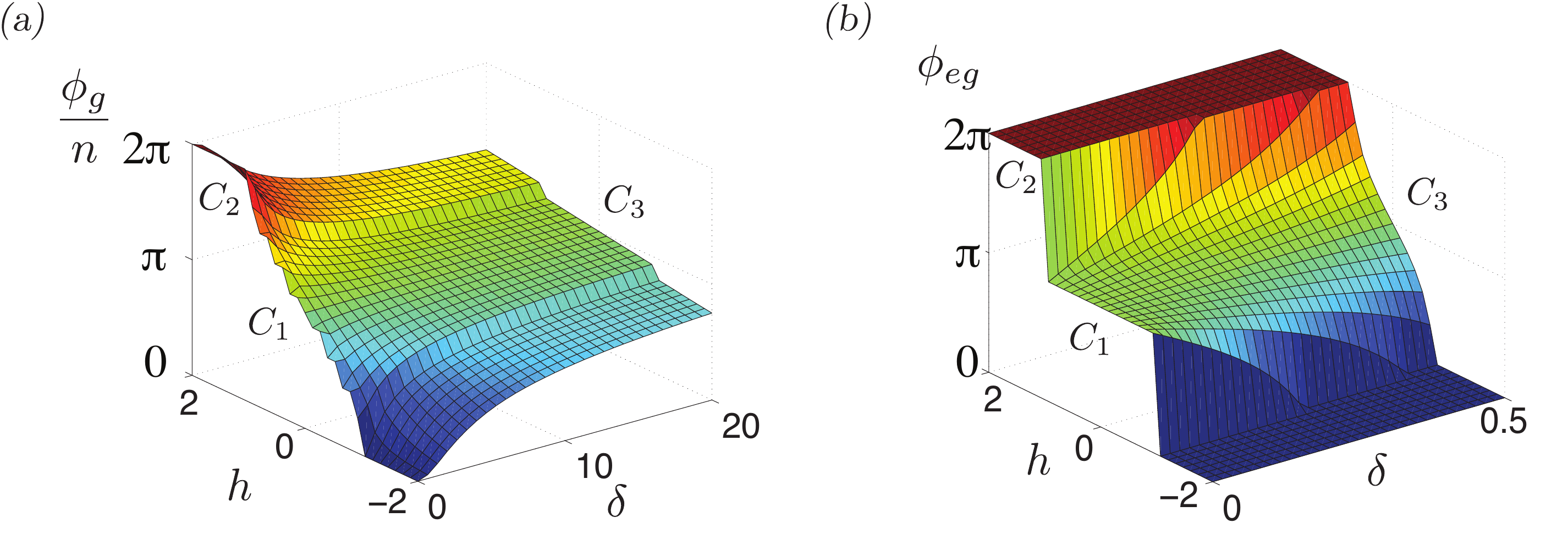} \vspace{0.5cm} \caption{The geometric phase corresponding
to the ground state (a) and the relative one between the ground and
excited state (c) as a function of the path parameters $h $ and
$\delta$. Each point of the surface corresponds to the geometrical
phase for a path that is spanned by varying $\varphi$ from $0$ to
$\pi$ for certain $h$ and $\delta$. The values of the geometric
phase corresponding to the loops $C_1$, $C_2$ and $C_3$ in Figure
\ref{criticality} are also indicated.} \label{Fig:Berry}
\end{center}
\end{figure}

Using the standard formula~\eqref{esp:intgam}, it is easy to show that the Berry phase
of the ground state $\ket{g}$ is given by
\begin{equation}
\label{geomphase} \phi_g =
-i\int_0^{2\pi}\bra{g}\frac{\partial}{\partial\varphi} \ket{g}\, d\varphi =
\sum_{k>0}\pi(1-\cos\theta_k).
\end{equation}
This result can be understood by considering the form of $\ket{g}$,
which is a tensor product of states, each lying in the two
dimensional Hilbert space spanned by $\ket{0}_{k}\ket{0}_{-k}$ and
$\ket{1}_k\ket{1}_{-k}$. For each value of $k(>0)$, the state in
each of these two-dimensional Hilbert spaces can be represented as a
Bloch vector with coordinates $(\varphi,\theta_k)$. A change in the
parameter $\varphi$ determines a rotation of each Bloch vector about
the $z$ direction. A closed circle will, therefore, produce an
overall phase given by the sum of the individual phases as given in
(\ref{geomphase}) and illustrated in Figure \ref{Fig:Berry}(a).

Of particular interest is the relative geometric phase between the
first excited and ground states given by the difference of the Berry
phases acquired by these two states. The first excited state is
given by
\begin{equation}\label{excitedstate}
\ket{e_{k_0}} =\ket{1}_{\!k_0}\ket{0}_{\!\!-k_0}
\bigotimes_{k\neq k_0}\Big(\!\cos {\theta_k \over 2}
\ket{0}_{\!k}\ket{0}_{\!\!-k} \!\!-ie^{i\varphi} \sin {\theta_k
\over 2} \ket{1}_{\!k} \ket{1}_{\!\!-k}\!\Big),
\end{equation}
with $k_0$ corresponding to the minimum value of the energy
dispersion function $\varepsilon_k$. The behavior of this state is
similar to a direct product of only $n-1$ spins oriented along
$(\varphi,\theta_k)$, where the state of the spin corresponding to
momentum $k_0$ does not contribute any more to the geometric phase.
Thus the relative geometric phase between the ground and the excited
state becomes
\begin{equation}
\label{connectionExcited} \phi_{eg} := \phi_e-\phi_g = -\pi (1-\cos
\theta_{k_0})\,.
\end{equation}
In the thermodynamic limit ($N \to \infty$), $\phi_{eg}$ takes the
form
\begin{equation}
\label{GPExcGrndSmallGamma} \phi_{eg}=\left\{\begin{array}{cl}
0,      &    \textrm{for $|h |>1-\delta^2$} \\
-\pi+{\pi h \delta \over
  \sqrt{(1-\delta^2)(1-\delta^2-h^2)}},     &     \textrm{for
  $|h |<1-\delta^2$}
\end{array}\right.
\end{equation}
where the condition $|h | > 1-\delta^2$ constrains the excited state
to be completely oriented along the $z$ direction resulting in a
zero geometric phase. As can be seen from Figure~\ref{Fig:Berry}(b),
the most interesting behavior of $\varphi_{eg}$ is obtained in the
case of $\delta$ small compared to $h$. In this case $\phi_{eg}$
behaves as a step function, giving either $\pi$ or $0$ phase,
depending on whether $|h|<1$ or $|h|>1$, respectively. This
behaviour is precisely determined by the topological property of the
corresponding loop, i.e. whether a critical point is encircled or
not. This property can therefore be used to witness the presence of
a criticality. More precisely, in the $|h |<1-\delta^2$ case, one
can identify the first contribution with a purely topological term,
while the second is an additional geometric
contribution~\cite{Bohm2003}. In other words, this first part gives
rise to phase whose character depends only on whether the loops can
be adiabatically shrunk to a point, i.e. on whether it is
contractible. This is a purely topological character of the
trajectory traced by the $(\varphi,\theta_k)$ coordinates. In
particular if $n$ circulations are performed then the topological
phase is $n\pi$, where $n$ is the winding number. The second term is
geometric in nature and it can be made arbitrarily small by tuning
appropriately the couplings $h$ or $\delta$. This idea is
illustrated in figure~\ref{conical}, where the energy surface of
ground and first excited state is depicted. The point of degeneracy
is the intersection of the two surfaces. This is the point where the
energy density is not analytical. Consider the case of a family of
loops converging to a point. In the trivial case in which the
limiting point does not coincide to a degeneracy, the corresponding
Berry phase converges to zero. If instead, the degeneracy point is
included, the Berry phase tends to a finite value~\cite{Hamma2006}.

To better understand the properties of the relative geometric phase,
we focus on the region of parameters with $\delta\ll 1$. In this
case, it can be shown~\cite{Carollo2005} that the Hamiltonian, when
restricted to its lowest energy modes, can be casted in a {\em real}
form and, for $|h|<1$, its eigenvalues present a {\em conical
intersection} centered at $\delta=0$.
\begin{figure}
\begin{center}
\includegraphics[width=4in]
{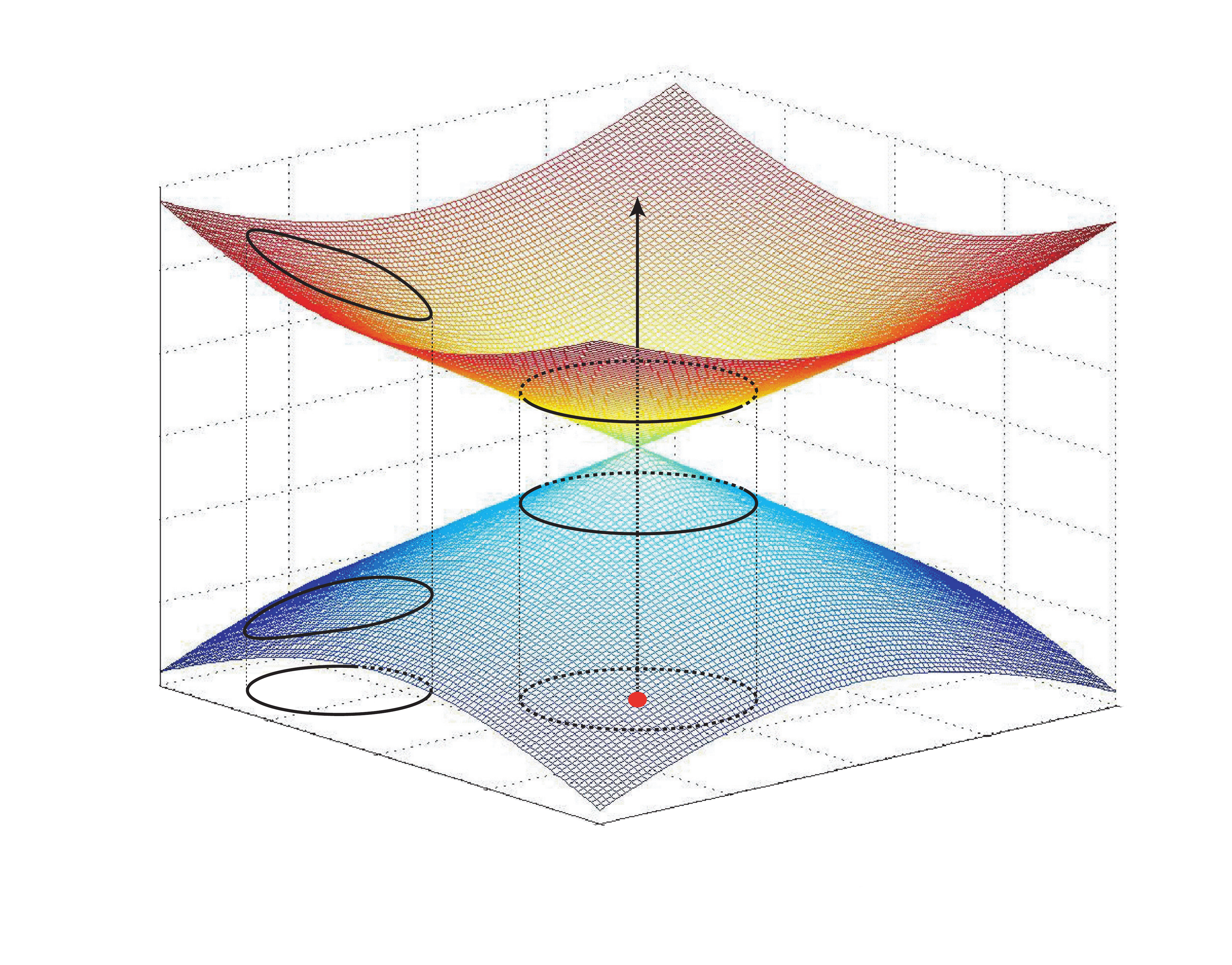} \vspace{0.5cm} \caption{The conical intersection
between the two lowest energy levels of the Hamiltonian as a
function of the parameters. A contractible loop, i.e. a loop that
can be continuously deformed to a point of the domain, produces a
zero geometric phase. A non-trivial geometric phase is obtained for
non-contractible loops.} \label{conical}
\end{center}
\end{figure}
It is well known that, when a closed path is taken within the real
domain of a Hamiltonian, a topological phase shift $\pi$ occurs only
when a conical intersection is enclosed. In the present case, the
conical intersection corresponds to a point of degeneracy where the
XX criticality occurs and it is revealed by the topological term in
the relative geometric phase $\phi_{eg}$. It is worth noticing that
the presence of a conical intersection indicates that the energy gap
scales linearly with respect to the coupling $\gamma$ when
approaching the degeneracy point. This implies that the critical
exponents of the energy, $z$, and of the correlation length, $\nu$,
satisfy the relation $z\nu=1$ which is indeed the case for the XX
criticality~\cite{Sachdev2011}. In the following we will see that
geometric phases are sufficient to determine the exact values of the
critical exponents and thus provide a complete characterization of
the criticality behavior.

\subsection{General Considerations}
One can show that the vacuum expectation value of a Hermitian
operator, $O$, can be written in terms of a geometric phase. The
only requirements posed on $O$ are that it should not commute with
the Hamiltonian and it should be able to transform the ground state
in a cyclic fashion. This is a rather general result that can be
used to study the critical models.

To show that let us extend the initial Hamiltonian, $H_0$, of the
model in the following way \begin{equation} H = H_0 +\lambda O\,. \end{equation} Turning to
the interaction picture with respect to the $O$ term we obtain 
\begin{equation}
 H_{ \textrm{int}(\lambda)} = U(\lambda t)H_0 U^\dagger(\lambda t)\,,
\end{equation} where $U(\lambda t) = \exp(-i \lambda O t)$. From the cyclicity
requirement, there exists a time $T$ such that the unitary rotation
$U(\lambda T)$ brings the ground state $\ket{\psi}$ to itself, i.e.
$U(\lambda T)\ket{\psi} =\ket{\psi}$, producing eventually the
desired cyclic evolution. The Berry phase that results from this
evolution is given by~(\ref{eq:BP}), and thus, one obtains
\begin{equation}
\varphi^{B} =-i\oint \bra{\psi} U^{\dagger}(\lambda t)dU(\lambda
t)\ket{\psi} = 2\pi i \bra{\psi} O \ket{\psi}\,. \label{Relation}
\end{equation}
From this expression, we see that the geometric phase is a simple
function of the vacuum expectation value of the operator that
generates the circular paths. This expression can be easily inverted
to finally give the expectation value of $O$ as a function of the
geometric phase~(\ref{Relation}).

One can easily verify this relation for the simple case of a
spin-1/2 particle in a magnetic field. When the direction of the
magnetic field is changed adiabatically and isospectrally then the
state of the spin is guided in a cyclic path around the
$z$-direction. The generated phase is given by $\varphi_B = i 2\pi
\cos \theta$, where $\theta$ is the fixed direction of the magnetic
field with respect to the $z$ direction. On the other hand, one can
easily evaluate that the expectation value of the operator
$(1-\sigma^z)/2$ that generates the cyclic evolution is given by
$\bra{\psi}(1-\sigma^z)/2\ket{\psi} = (1-\cos\theta)/2$, which
verifies relation (\ref{Relation}), as for example $\lambda T=2\pi$.

This connection has far reaching consequences. It is expected that
intrinsic properties of the state will be reflected in the
properties of the geometric phases. The latter, as they result from
a physical evolution, can be obtained and measured in a conceptually
straightforward way. To illustrate this we shall focus on critical
phenomena in spin systems. Indeed, the presence of critical points
can be detected by the behaviour of specific geometric evolutions
and the corresponding critical exponents can be extracted. This
comes as no surprise as one can choose geometric phases that
correspond to the correlations of the system (expectation values of,
e.g. $\sigma^z_1 \sigma^z_L$) from where the correlation length and
the critical behaviour can be obtained.

Let us apply this idea to the XY model, where the rotations are
generated by the operator $O:=\sum_j \sigma^z_j$. One can easily see
that the resulting geometric phase is proportional to the total
magnetization \begin{equation} M_z = \bra{\psi} \sum_j \sigma^z_j \ket{\psi}. 
\end{equation}
 It is well known~\cite{Sachdev2011} that the magnetization $M_{z}$
can serve as an order parameter from which one can derive all the
critical properties of the XY model just by considering its scaling
behaviour. Indeed, Zhu~\cite{Zhu2006} has considered the scaling of
the ground-state geometric phase of the XY model from where he
evaluated the Ising critical exponents. As it has been shown here,
this is a general property that can be applied to any critical
system.

\subsection{Dicke model and Geometric Phases}\label{sec:DickeGP}
In this section, we illustrate yet another model in which the
properties of a quantum phase transitions can be observed and
investigated through the geometric phase. We will discuss the
thermodynamic and finite size scaling properties of the geometric
phase in the adiabatic Dicke model (DM)~\cite{Plastina2006},
describing the super-radiant phase transitions for an $n$-qubit
register coupled to a slow oscillator mode. One can show that, in
the thermodynamic limit, the Berry phase has a topological feature
similar to the one highlighted in the previous sections for the case
of the XY model. A non-zero geometric phase is obtained only if a
path in parameter space encircles the critical
point. Furthermore, in this context one can show that precursors of
this critical behaviour, for a system with finite size, exist and the
scaling law of the Berry phase can be obtained in the leading order
in the $1/n$ expansion.

Let's consider a system which consists of $n$ two-level systems, a
qubit register or an ensamble of indistinguishable atoms, coupled to
a single bosonic mode. The Hamiltonian is given by, in
unit such that $\hbar=c=1$, is given by
\begin{equation}
    H={\omega} a^{\dagger}a+\Delta S_x+
   \frac{\lambda}{\sqrt{n}}(a^{\dagger}+a)S_z,
    \label{1}
    \end{equation}
where $\Delta$ is the transitions frequency of the qubit, $\omega$ is
the frequency of the oscillator and $\lambda$ is the coupling
strength. The qubit operators are expressed in terms of total spin
components ${S}_{k}=\sum_{j=1}^{n}{\sigma}^{k}_{j}$, where the
${\sigma}^{k}_{j}$ 's ($k=x,y,z$) are the Pauli matrices used to
describe the $j$-th qubit.
A $\pi/2$ rotation around the $y$ axis shows that $H$ is canonically
equivalent to the standard formulation of the Dicke Hamiltonian
\cite{Dicke1954}, including counter-rotating terms.

After the first derivation due to Hepp and Lieb
\cite{Hepp1973,Hepp1973a}, the thermodynamic properties of the DM
have been studied by many
authors\cite{Wang1973,Duncan1974,Gilmore1976,Orszag1977,Sivasubramanian2001,Liberti2004,Liberti2005}.
In the thermodynamic limit ($n\rightarrow \infty$), the system
exhibits a second-order phase transitions at the critical point
$\lambda_c=\sqrt{\Delta\omega/2}$, where the ground state changes
from a normal to a super-radiant phase in which both the field
occupation and the spin magnetization acquire macroscopic values.
The continued interest in DM stems from the fact that it displays a
rich dynamics where many non-classical effects have been predicted
\cite{Schneider2002,Emary2003,Emary2003a,Frasca2004,Hou2004,Busek2005},
and from its broad range of applications \cite{Brandes2005}.
Investigations on the ground state entanglement of the Dicke model
have been also performed~\cite{Lambert2004,Reslen2005,Vidal2006},
pointing out a scaling behavior around the critical point.

In this section we will outline the topological character of the
geometric phase of the Dicke model in the adiabatic regime ($\Delta
\gg \omega$), and illustrate the scaling law of the geometric phase
close to the critical point for a
system with finite size.\\
In order to generate a Berry phase one can change the Hamiltonian by
means of the unitary transformation:
\begin{equation}\label{ut}
    U(\varphi)=\exp{\left(-i\frac{\varphi}{2}S_x\right)},
\end{equation}
where $\varphi$ is a slowly varying parameter, moving from $0$ to $2
\pi$. The transformed Hamiltonian can be written as
\begin{equation}
   {H}(\phi)=U^\dag(\varphi) H U(\varphi)=\frac{\omega}{2}\left[p^2+q^2+\bm{Q}\cdot
   \mathbf{S}\right],
    \label{ht2}
\end{equation}
where $\bm{Q}=\left(D,\frac{L q}{\sqrt{n}}\sin{\varphi},\frac{L
q}{\sqrt{n}}\cos{\varphi}\right)$ is an effective vector field.
Here,
 $D=2\Delta/\omega$ and $L=2\sqrt{2}
\lambda / \omega$ are dimensionless parameters, and the Hamiltonian
of the free oscillator field is expressed in terms of canonical
variables $q =(a^{\dagger}+a)/\sqrt{2}$ and
$p=i(a^{\dagger}-a)/\sqrt{2}$, that obey the quantisation condition
$[q,p]=i$.

In the adiabatic limit \cite{Liberti2006}, where one
assumes a {\it slow} oscillator and work in the regime $D\gg1$, the
Born-Oppenheimer approximation can be employed to write the ground
state of $H(\phi)$ as:
\begin{equation}\label{deco}
    |\Psi_{tot}\rangle=\int d q \, \psi(q) |q\rangle \otimes |\chi(q,\phi)\rangle\,.
\end{equation}
Here, $|\chi (q,\phi)\rangle$ is the state of the ``fast
component''; namely, the lowest eigenstate of the ``adiabatic''
equation for the qubit part, displaying a parametric dependence on
the slow oscillator variable $q$,
\begin{equation}\label{adiaham}
    \bm{Q}\cdot \mathbf{S}|\chi(q,\phi)\rangle=E_l(q)|\chi (q,\phi)\rangle \,.
\end{equation}
As the qubits are indistinguishable, it is easy to prove that the
ground state can be expressed as a direct product of $n$ identical
factors,
\begin{equation}\label{qubitstates}
|\chi(q,\phi)\rangle=|\chi (q,\phi)\rangle_1\otimes |\chi
(q,\phi)\rangle_2\otimes\dots\otimes|\chi (q,\phi)\rangle_n\,.
\end{equation}
Each component can be written as
\begin{equation}\label{gsdgen}
|\chi (q,\phi)\rangle_j=
\sin{\frac{\theta}{2}}|\uparrow\rangle_j-\cos{\frac{\theta}{2}}e^{i\zeta}
|\downarrow\rangle_j, \end{equation} where $|\uparrow\rangle_j$ and
$|\downarrow\rangle_j$ are the eigenstates of $\sigma_j^z$ with
eigenvalues $\pm1$, and where
\begin{eqnarray}\label{gsdgen2}
&& \cos{\theta} := {{\frac{Lq\cos{\varphi}}{\sqrt{n} E(q)}}}\,,\\
&& \zeta := \arctan{\frac{L q \sin{\varphi}}{\sqrt{n} D}}\,.
\end{eqnarray}
Here, $E(q)$ is related to the energy eigenvalue of Eq.
(\ref{adiaham}) as
\begin{equation}\label{eival}
    E_{l}(q)= - n E(q)=- n \sqrt{D^2+\frac{L^2q^2}{n}}\,.
\end{equation}
In the Born-Oppenheimer approach, this energy eigenvalue constitutes
an effective adiabatic potential felt by the oscillator together
with the original square term
\begin{equation} V_{l}(q)=\frac{\omega}{2}\left[q^2- n E(q)\right].\end{equation}
Introducing the dimensionless parameter $ \alpha={L^2}/{2 D}$, one
can show that for $\alpha\leq1$, the potential $V_{l}(q)$ can be
viewed as a broadened harmonic well with minimum at $q=0$ and
$V_{l}(0)=-n\Delta$. For $\alpha>1$, on the other hand, the coupling
with the qubit splits the oscillator potential producing a symmetric
double well with minima at $\pm
q_m=\pm\frac{\sqrt{n}D}{L}\sqrt{\alpha^2-1}$ and $
V_{l}(q_m)=-\frac{n\Delta}{2}\left(\alpha+\frac{1}{\alpha}\right)$.\\\indent
As last step in the Born-Oppenheimer procedure, we need to
evaluate the ground-state wave function for the oscillator,
$\psi_0(q)$, that has to be inserted in Eq. (\ref{deco}) to obtain
the ground state of the composite system. This wave function is the
normalized solution of the one-dimensional time independent
Schr\"odinger equation
\begin{equation}\label{se}
H_{ad}\psi_{0}(q)=\left(-\frac{\omega}{2}\frac{d^2}{dq^2}+
V_l(q)\right)\psi_{0}(q)=\varepsilon_{0}\psi_{0}(q) \, ,
\end{equation}
where $\varepsilon_{0}$ is the lowest eigenvalues of the adiabatic
Hamiltonian defined by the first equality.\\
Once this procedure is carried out for every value of the rotation
angle $\phi$, the Berry phase of the ground state is obtained as
\begin{equation}\label{bp}
    \phi^{B} = i\int_cd\varphi\la\Psi_{tot} |\frac{d}{d\varphi} |\Psi_{tot}\ra
    =\int_{-\infty}^{+\infty}dq\psi_0^2(q)\int_0^{2\pi}d\varphi A^{B}(q,\varphi)\,,
\end{equation} where we introduced a $q$-parametrised Berry connection, given
by \begin{equation}\label{bp2}
    A^{B}(q,\varphi):=i\la\chi_l(q,\varphi)|\frac{d}{d\varphi}|\chi_l(q,\varphi)\ra=
    -{n}\frac{d\zeta}{d\varphi}\cos^2{\frac{\theta}{2}}\nonumber\\
     =-\frac{nD}{2E(q)}\frac{\frac{L q}{\sqrt{n}}\cos{\varphi}}{E(q)-\frac{L q}{\sqrt{n}}\cos{\varphi}}\,.
\end{equation} Substituting this expression into Eq.(\ref{bp}), one finds
\begin{equation}\label{bp3}
    \phi^{B}=n\pi\left(1+\frac{\langle S_x\rangle}{n}\right)\,,
\end{equation}
where the average magnetisation per spin is
\begin{equation}\label{sxm}
\frac{\langle S_x\rangle}{n}=-\int_{-\infty}^\infty \psi_{0}^2(q)
\frac{D}{E(q)} dq \, .
\end{equation}
Notice that Eq.~(\ref{bp3}) holds in general, its validity relying
on the form of the rotation operator $U(\varphi)$ of Eq.~(\ref{ut})
and not being restricted to the adiabatic regime. In the
thermodynamic limit, one can show that
\begin{equation}
 \frac{\langle S_x\rangle}{n}=\left\{%
\begin{array}{ll}
    -1 & \hbox{$(\alpha\leq 1)$} \\
    \hbox{$-\frac{1}{\alpha}$} & \hbox{$(\alpha> 1),$} \\
\end{array}%
\right.
\end{equation}
and thus, for $n\rightarrow \infty$, the BP is given
by~\cite{Liberti2006a}
\begin{equation}
 \frac{\phi^{B}}{n} \Bigr |_{n\rightarrow \infty} =\left\{%
\begin{array}{ll}
    0 & \hbox{$(\alpha\leq 1)$}\\
    \hbox{$\pi(1-\frac{1}{\alpha})$} & \hbox{$(\alpha> 1)$}\,, \\
\end{array}%
\right.
\end{equation}

It is worth stressing once again, that for the thermodynamic limit
this result holds independently of the adiabatic approximation,
whose use is needed here to obtain the finite-size behaviour.
Numerical results for the scaled Berry phase are plotted in
Fig.(\ref{berry}) as a function of the parameter $\alpha$, for
$D=10$ and for different values of $n$, in comparison with the
result for the thermodynamic limit.
\begin{figure}
\begin{center}
 \includegraphics[width=0.7\textwidth]{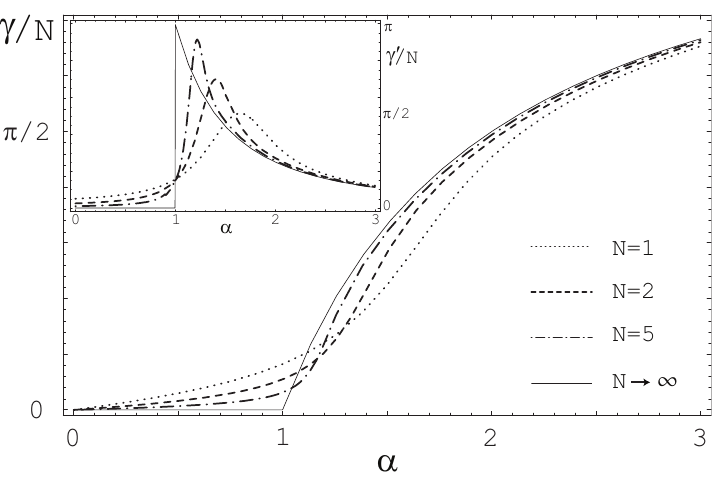}\\
 \caption{\label{berry} Numerical results for the scaled Berry phase as
a function of the parameter $\alpha$, for $D=10$ and for different
values of $n$, in comparison with the result for
$n\rightarrow\infty$. Berry's phase increases with the coupling,
and, in the thermodynamic limit, its derivative becomes
discontinuous at the critical value $\alpha=\alpha_{c}:=1$. The
inset shows the derivative of the Berry phase with respect to
$\alpha$.}
\end{center}
\end{figure}
\begin{figure}
\begin{center}
\includegraphics[width=0.55\textwidth]{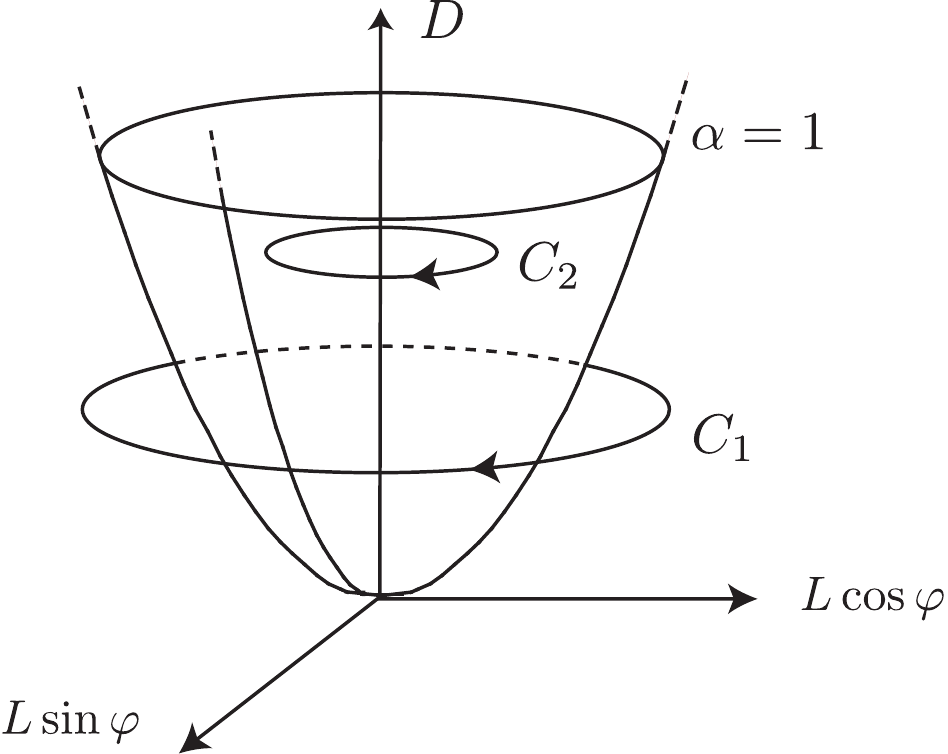}\\
\caption{\label{parab} A qualitative illustration of the paths
followed by the parameters of the Hamiltonian due to the application
of $U(\phi)$. The paraboloid corresponds to the value
$\alpha=L^2/2D=1$, for which the Hamiltonian shows a critical
behavior. If the parameters follow a path, e.g. $C_1$, encircling
the paraboloid, then the system acquires a non-trivial Berry phase,
which tends to $\pi$ for $\alpha\gg 1$. The path $C_2$ gives rise to a zero Berry phase (in
the thermodynamic limit).}
\end{center}
\end{figure}
One can see that the Berry phase increases as the coupling strength
growths and, in the thermodynamic limit, its derivative becomes
discontinuous at the critical value $\alpha=\alpha_c:=1$. This
results agree with the expected behaviour of the geometric phase
across the critical point. Notice that, in the thermodynamic limit,
a non-trivial Berry phase is only obtained when a region of
criticality is encircled, as for the path $C_1$ in Fig.~\ref{parab}.
Indeed, in the enlarged parameter space generated by the application
of the unitary operator $U(\phi)$ of Eq. (\ref{ut}), the critical
point corresponds to the paraboloid $\alpha = \frac{L^2}{2 D} =1$.
As the radius of the path is determined by $\alpha$, one can see
that, in the limit $n\rightarrow \infty$, the Berry phase is zero in
the normal phase ($\alpha \leq 1$) and is non-zero in the
super-radiant phase, i.e. if the path encloses the critical region.
This behaviour is indeed reminiscent of the topological features
displayed by the geometric phase of the XY model described in the
previous sections.

It is worth considering also the finite-scaling behaviour of Berry
phase at the critical point. In order to obtain an analytic
estimation of Berry phase as a function of $n$, one can expand the
adiabatic potential in Eq.(\ref{se}) in powers of $1\over n D$ and
by using the expressions of the perturbation coefficients $c_{k}$,
one obtains an anharmonic oscillator potential
\begin{equation}\label{adpot}
    U_l(q)=\frac{2}{\omega}V_l(q)\simeq -nD+(1-\alpha)q^2+\frac{\alpha^2}{2nD}q^4\,.
\end{equation}
The eigenvalue problem defined by this potential can be solved with
the help of Symanzik scaling \cite{Liberti2006a,Simon1970}. This is
done by rewriting Eq.(\ref{se}) into the equivalent form
\begin{equation}\label{hamscaling}\left[-\frac{d^2}{dx^2}+\mu
x^2+x^4\right]\psi_{0}(x;\mu) =e_{0}\left(\mu\right)\psi_{0}(x;\mu),
\end{equation}
where the scaled position is $x:= q \left ( \frac{\alpha^2}{2 nD}
\right)^{1/6}$, while
$\mu:=\left(\frac{2nD}{\alpha^2}\right)^{2/3}(\alpha_c-\alpha)$.
Finally, the energies in Eq.(\ref{se}) and (\ref{hamscaling}) obey
the scaling relation
\begin{equation}\label{sr}
\frac{2}{\omega}
\varepsilon_{0}(\alpha,nD)=-nD+\left(\frac{\alpha^2}{2nD}\right)^{1/3}e_{0}
\left(\mu\right).
\end{equation} Since $\mu \rightarrow 0$ at
the critical point, we can consider the $x^2$ term to be a
perturbation and employ the Rayleigh-Schr\"{o}dinger perturbation
theory. This yields the expansion $ e_{0}(\mu)=\sum_{k=0}^\infty
c_k\mu^k $, where  the coefficients $c_k$ can be obtained after
solving the equation for a purely quartic oscillator. It is easy to
get $c_0= e_{0}(0)\simeq1.06036$ and $c_1=\int_{-\infty}^\infty
q^2\phi_{0}^2(q;0)dq=e_0^\prime(0)\simeq 0.36203$.
\begin{figure}
\begin{center}
 \includegraphics[width=0.8\textwidth]{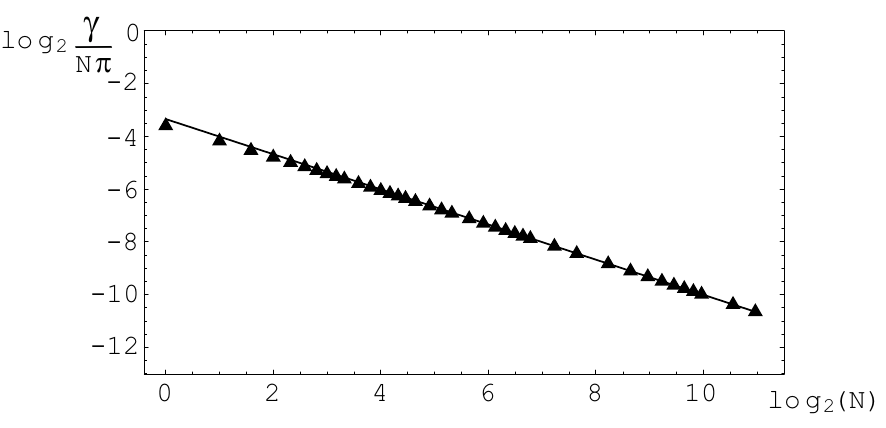}\\
 \caption{\label{gafn} Scaling of the Berry phase
 as a function of $n$ at the critical point $\alpha=1$, for
 $D=10$. For ease of comparison, the continuous plot shows
the analytic expression of Eq. (\ref{BPscal}).}
\end{center}
\end{figure}

It can be shown that the coefficients the expansion of $ e_{0}(\mu)$ completely
determine the average value of every physical observable at the
critical point \cite{Liberti2006a}. In particular, a similar
expansion applied to $\langle S_x \rangle$, allows one to
write~\cite{Plastina2006}
\begin{equation}\label{sxsca}
\frac{\langle
S_x\rangle}{n}\simeq-1+\frac{2c_1}{(2nD)^{2/3}}-\frac{2c_0}{(2nD)^{4/3}}.
\end{equation}
Thus, one obtains the leading orders in the finite size scaling of
the Berry phase as
\begin{equation}\label{BPscal}
    \frac{\phi^{B}}{n}\approx\pi\left[\frac{2 c_1}{(2nD)^{2/3}}-\frac{2
    c_0}{(2nD)^{4/3}}\right].
\end{equation}
This expression shows how the scaled geometric phase goes to zero as
$n$ increases and how the singular thermodynamic behaviour is
approached at $\alpha=\alpha_c=1$. The leading critical behaviour of
the Berry phase, $\phi^{B}/n \sim  n^{-2/3}$ is confirmed in Fig.
(\ref{gafn}) by comparison with the numerical evaluation of the
geometric phase obtained from Eqs. (\ref{bp3})-(\ref{sxm}). In fact,
including also the second order correction, scaling as $n^{-4/3}$,
reproduces the numerical results even for small values of $n$.

Besides the scaling relation at the critical point $\alpha=1$, one
can also obtain the leading $1/n$ correction to the thermodynamic
limit of $\frac{\phi^{B}}{n}$ for small and large values of $\alpha$, i.e., for path of very small and very large
radii~\cite{Plastina2006}. Since the oscillator localises around
$q=0$ for $\alpha \ll1$, while its wave function is split in two
components peaked around $\pm q_m$ for $\alpha \gg 1$, one
gets~\cite{Plastina2006}
\begin{equation}
\frac{\phi^{B}}{n} -  \frac{\phi^{B}}{n} \Bigr |_{n\rightarrow \infty} \approx \left\{%
\begin{array}{ll}
    \hbox{$\frac{\pi \alpha}{2 nD}$} & \hbox{$(\alpha\ll 1)$} \\
    \hbox{$- \frac{\pi}{nD \alpha^2}$} & \hbox{$(\alpha \gg 1)$}\,, \\
\end{array}%
\right.
\end{equation}
Thus, far from the critical point (that is, well inside the normal
and super-radiant phase), the Berry phase per qubit reaches its
thermodynamic limit as $n^{-1}$. Therefore, the topological
behaviour of $\phi^{B}$ can be inferred even for a relatively small
number of qubits. To summarise, this example shows that the
behaviour of the geometric phase in correspondence of the critical
region for the Dicke model, confirms the general connection between
geometric phase and quantum critical phase transitions. Indeed,
geometric phase and QPTs  share the common feature of both appearing
in presence of a singularity in the energy density of the system.
This heuristic argument motivates - once again - the need to explore
the use of Berry phase as a tool to signal and investigate critical
features of certain models. However, strictly speaking,
singularities in the energy density of many body systems only appear
in the thermodynamic limit. It is therefore interesting to notice
that, in a \emph{finite scale} regime, such a connection between
Berry phase and QPTs  can still be drawn, owing to the geometric phase
sensitivity to the increase of the parametric manifold curvature, as
the thermodynamic limit is approached. Studying the Berry phase in
this regime has clearly theoretical interest and obvious
experimental motivations. In the case of the Dicke model, it is
found that the geometric phase start to show, already at finite
sizes, the topological character which is typically manifested at
the thermodynamic limit. Moreover, studying the scaling of the Berry
phase as a function of the system size, one can identify its
critical exponent.

\section{Quantum Phase Transition and Information Geometry}
For the sake of completeness, and to clearly pave the ground to the
last part of this review, we will introduce another approach to the
study of quantum phase transitions, i.e. the so called
\emph{fidelity
approach}~\cite{Zanardi2006,Zanardi2007,You2007,Zhou2008}, which has
been successfully employed in
classical~\cite{Janyszek1999,Ruppeiner1995,Quan2009a} and quantum
phase
transitions~\cite{Zanardi2007c,Zanardi2006,Zanardi2007,CamposVenuti2007,Gu2010,Dey2012},
both in
symmetry-breaking~\cite{Zanardi2006,Zanardi2007,CamposVenuti2007,Zanardi2007a,Gu2010,Dey2012,Kolodrubetz2013}
as well as in topological phase transitions~\cite{Yang2008}. As a
distance measure, the fidelity describes how close two given quantum
states are. Therefore, it is natural to expect that the fidelity can
be used to characterise the drastic changes occurring in quantum
many body states going through QPTs, regardless of what type of
order parameters, if any, characterises the type of phase
transitions. In the fidelity approach, quantum phase transitions are
identified by studying the behaviour of the overlap between two
ground states corresponding to two slightly different set of
parameters; at QPTs  a drop of the fidelity with scaling behaviour is
observed and quantitative information about critical exponents can
be extracted \cite{CamposVenuti2007}. As for the geometric phase
approach, the fidelity approach neither is based on the knowledge of
an order parameter, nor it requires a notions of symmetry breaking.
It is indeed an approach purely based on the distinguishability of
pure states, as well as density matrices, and it does not even
require the knowledge of the Hamiltonian itself. In a sense, it is a
purely kinetic approach, as opposed to the traditional methods which
rely on the information derived from the dynamics of the many-body
systems.

\subsection{Metric on the Hilbert Space: the Fubini-Study metric}\label{sec:FS}
In a given Hilbert space, there is a natural \emph{gauge
invariant metric} that can be defined in terms of elements of the
projective Hilbert space $\mathcal{P}$. Given any two $\psi_{1}$ and
$\psi_{2}$ in the Hilbert space, the distance $d_{FS}$ between any
two points  $p_{1}:=\ket{\psi_{1}}\bra{\psi_{1}}$ and
$p_{2}:=\ket{\psi_{2}}\bra{\psi_{2}}$ in~$\mathcal{P}$ is defined by
\begin{equation}\label{eq:dFS}
d_{FS}(p_{1},p_{2})=\inf_{\alpha_{1},\alpha_{2}}||e^{i\alpha_{1}}\psi_{1}
- e^{i\alpha_{2}}\psi_{2}|| = \sqrt{2 - 2|\bk{\psi_{1}|\psi_{2}}|},
\end{equation}
where minimisation is obtained when $e^{i\alpha_{1}}\psi_{1}$ and
$e^{i\alpha_{2}}\psi_{2}$ are \emph{in phase}, according to
Pancharatnam's criterion, i.e. when
$\bra{e^{i\alpha_{1}}\psi_1}e^{i\alpha_{2}}\psi_2\rangle$ is a real
and positive number.
 Clearly, $d_{FB}(p_{1}, p_{2})\ge0$ with equality holding if and only if $p_{1}=p_{2}$. Also $d_{FS}(p_{1}, p_{2})= d_{FS}(p_{2}, p_{1})$. The triangle inequality, for any $p_{1}=\ket{\psi_{1}}\bra{\psi_{1}}$, $p_{2}=\ket{\psi_{2}}\bra{\psi_{2}}$ and $p_{3}=\ket{\psi_{3}}\bra{\psi_{3}}$  in $\mathcal{P}$, is implied by the following chain of relations
\begin{equation}
    d_{FS}(p_{1},p_{2}) + d_{FS}(p_{2}, p_{3}) = ||\psi_{1}-\psi_{2}|| + ||\psi_{2}-\psi_{3}|| \ge ||\psi_{1}-\psi_{3}|| \ge d_{FS}(p_{1}, p_{3}),
\end{equation}
where $\psi_{2}$ is in phase with $\psi_{1}$, and $\psi_{3}$ with
$\psi_{2}$. Hence $d_{FS}$ is a metric on $\mathcal{P}$, called the
\emph{Fubini-Study distance}~\cite{Fubini1904,Study1905}, which can
be expressed as $d_{FB}(p_{1},
p_{2})^{2}=2(1-\sqrt{\pr(p_{1},p_{2})})$, where
$\pr(p_{1},p_{2}):=|\bk{\psi_{1}|\psi_{2}}|^{2}$ is the so called
\emph{transitions probability}. The latter
quantifies the probability to get an affirmative answer when testing
whether the system is in the state $p_{1}$ if it was actually in
state $p_{2}$, or viceversa. In other words, it quantifies the statistical
distinguishability between pure states. The \emph{Fubini-Study
distance} is indeed a geometrical measure of \emph{statistical
indistinguishability} between pure quantum
states~\cite{Wootters1981,Jozsa1994,Braunstein1994}. In terms of
$d_{FS}(p_{1},p_{2})$, the projective space, corresponding to the
subspace spanned by $\psi_{1}$ and $\psi_{2}$, is described by a
2-sphere with unit radius embedded in a 3-dimensional Euclidian
space. In such a sphere, $d_{FS}(p_{1},p_{2})$ is the straight-line
distance separating $p_{1}$ and $p_{2}$. \\\indent Suppose that
$\psi_{1}$ and $\psi_{2}$ are such that $\Pi(\psi_{1})$ and
$\Pi(\psi_{2})$ are infinitely close in $\mathcal{P}$. Then
Eq.~(\ref{eq:dFS}) defines a Riemannian metric on $\mathcal{P}$
called the \emph{Fubini-Study metric}. To obtains its metric
coefficients, consider a curve $\mathcal{C}$ in $\mathcal{P}$
parameterised in the interval $[s_a,s_b]$, and let $C:=\psi(s)$ be
any of its lift in the Hilbert space. By Taylor expanding,
\begin{equation}
\bk{\psi(s)|\psi(s+ds)}=1 + \bk{\psi |
\dot{\psi}}ds+\frac{1}{2}\bk{\psi | \ddot{\psi}}ds^{2} +
\mathcal{O}(ds^{3}),
\end{equation}
where $\dot{\psi}:=d\psi/ds$. Also, differentiating
$\bk{\psi(s)|\psi(s)}=1$ twice yields
\begin{equation} 
\begin{aligned}
\bk{\psi|\dot{\psi}}+ \bk{\dot{\psi}|\psi}=0,\\
\bk{\psi|\ddot{\psi}}+\bk{\ddot{\psi}|\psi} + 2
\bk{\dot{\psi}|\dot{\psi}} =0,\label{eq:diffnorm}
\end{aligned} 
\end{equation} 
hence,
\begin{equation}
\begin{aligned}
dl^{2}:=&2(1-\sqrt{|\bk{\psi(s)|\psi(s+ds)}|^{2}})=\\=&2-2\sqrt{1+\left(\bk{\dot{\psi}|\psi}+\bk{\psi|\dot{\psi}}\right) ds + \left(\bk{\dot{\psi}|\psi}\bk{\psi|\dot{\psi}} + \frac{1}{2}\bk{\ddot{\psi}|\psi} + \frac{1}{2}\bk{\psi|\ddot{\psi}}\right)ds^{2}}=\\
=&\bra{\dot{\psi}}(\one - \kb{\psi})\ket{\dot{\psi}}ds^{2}\\
=&\sum_{\mu\nu}g_{\mu\nu}d\lambda_{\mu}d\lambda_{\nu}\,,
\end{aligned} 
\end{equation} 
 where $\{\lambda_{\mu}\}\in\mathcal{M}$ are a set of local parameters labelling $\mathcal{P}$ in the neighbourhood of $\Pi(\psi)$, and

\begin{equation}
  g_{\mu\nu}:= \Re\,\,Q_{\mu\nu}
\end{equation}
   is a positive-definite real matrix, where
 \begin{equation}\label{eq:qgt}
Q_{\mu\nu}:= \bra{\partial_{\mu}\psi}(\one -
\kb{\psi})\ket{\partial_{\nu}\psi}
\end{equation}
is a positive semi-definite Hermitian matrix, called the
\emph{quantum geometric tensor}~\cite{Provost1980,Berry1989}. This
quantity, by definition, is gauge invariant, and its imaginary part
coincides, up to a factor $1/2$, to the Berry curvature
 \begin{equation}
\Im\,\,Q_{\mu\nu}= \frac{\bk{\partial_{\mu}\psi|\partial_{\nu}\psi}
- \bk{\partial_{\mu} \psi |\partial_{\nu}\psi}
}{2}=\frac{F_{\mu\nu}^{B}}{2}.
\end{equation}
Notice that, the Fubini-Study metric can also be expressed as
$dl^2_{FB}=\bk{u(s) | u(s)}ds^2$,  where
$\ket{u}=\ket{\dot{\psi}}-\bk{\psi | \dot{\psi}} \ket{\psi}$ is the
component of the tangent vector $\dot{\psi}(s)$ orthogonal to
$\psi(s)$. Alternatively,
\begin{equation}
\ket{u}=\mathcal{D}_{s}\ket{\psi}:=\ket{\dot{\psi}}+ i
A_{s}^{B}\ket{\psi},
\end{equation}
where $\mathcal{D}_{s}:=d/ds + i A_{s}^{B}$ is the covariant
derivative and $A_{s}^{B}=i\bk{\psi|\dot{\psi}}$ is the Berry
connection. Using this metric one can derive the geodesics in the
space $\mathcal{P}$, which amounts to finding the path connecting
two states which minimises the following length:
\begin{equation}
  \label{eq:geodint}
D_{FS} :=\min
\int_{\Pi(\psi(s_a))}^{\Pi(\psi(s_a))}dl_{FS}=\int_{s_a}^{s_b}\sqrt{\bk{u(s)|u(s)}}ds.
\end{equation}
It turns out that the length of the geodesics arc, connecting the two
end points $\psi_{1}$ and $\psi_{2}$, is given by
\begin{equation}\label{eq:FSlength}
D_{FS}(p_{1},p_{2})=\int_{\mathcal{C}} dl_{FS} =
\arccos{|\bk{\psi_{1}|\psi_{2}}|}=\arccos{\sqrt{\pr(p_{1},p_{2})}},
\end{equation}
which is called \emph{Fubini-Study length} or \emph{Fubini-Study
angle}, which is itself a distance on the projective Hilbert space.

\subsection{Fubini-Study distance as a statistical distance}\label{sec:FSstat}
Let's digress onto a specific aspect of the Fubini-Study distance.
The Fubini-Study distance provides a measure of statistical
(in-)distinguishability between pure quantum states. Assume that one
wishes to perform a finite set of experiments to distinguish between
two states $\psi_{1}$ and $\psi_{2}$. To this extent, one needs some
specific set of measurements, or equivalently a set of observables,
and then use the results to define a statistical distance between
the states. However, it is clear that this distance will depends on
the choice of the observable as well as on the states. A solution to
this problem would be to single out an exceptional set of
measurements which maximises the resulting statistical distance. By
definition this will be the distance between the states.
\\\indent Assume that a chosen observable $O$ has $n$ non-degenerate orthogonal eigenstates $\ket{k}$ in terms of which we can expand both states. When the state is $\psi_{i}$, according to the standard Born rule, the probability distribution $P_{i}:=\{p_{i}(k), k=1\dots n\}$ to obtain the $k$-th outcome in the measurement is $p_{i}(k):=|\bk{\psi_{i}|k}|^{2}$. Each state $\psi_{i}$, therefore, results in a distinct outcome probability distribution. The point now is to quantify by means of a statistical distance the degree of distinguishability of these probability distributions in an operationally meaningful way. Two popular choices that accomplish this task are the Bhattacharyya distance,
\begin{equation}
D^{Bha}_{O}(\psi_{1},\psi_{2}) :=\arccos\left(\sum_{k}
\sqrt{p_{1}(k)p_{2}(k)}\right)=\arccos B(P_{1},P_{2})
\end{equation}
and the Hellinger distance
\begin{equation}
D^{H}_{O}(\psi_{1},\psi_{2}) :=\left(\sum_{k} \left(\sqrt{p_{1}(k)}-
\sqrt{p_{2}(k)}\right)^{2}\right)^{\frac{1}{2}}=\sqrt{2-2
B(P_{1},P_{2})}
\end{equation}
both monotonous functions of the Bhattacharyya coefficient, which
can be computed from the square roots of the probabilities,
\begin{equation}
B(P_{1},P_{2}):=\sum_{k} \sqrt{p_{1}(k)p_{2}(k)}= \sum_{k}
|\bk{\psi_{1}| k}| |\bk{k |\psi_{2}} | \le |\bk{\psi_{1}|\psi_{2}}
|.
\end{equation}
There are several optimal measurements that saturate the inequality
above. A solution is found by any observable $O$ having either of
the state $\psi_{i}$ as one of its eigenstates, in which case the
Bhattacharyya distance and the Hellinger distance collapse to the
Fubini-Study length $D_{FB}$~(\ref{eq:FSlength}) and Fubini-Study
distance $d_{FB}$~(\ref{eq:dFS}), respectively. This establishes the
Fubini-Study metric as a measure of the distinguishability of pure
quantum states in the sense of statistical
distance~\cite{Wootters1981}. More precisely, what the Fubini-Study
distance measures is the experimental distinguishability of two
quantum states, assuming no limitations on the type of measurements
one can perform. In practice, a measurement device available to a
laboratory may correspond to a limited subset of observables only,
and this device may be subject to various sources of imperfections.
In this case, the Fubini-Study geometry may not provide the correct
measure of experimental distinguishability, but still it is
relevant, as it provides information on what we can know in general,
without knowledge of the specific physical system.

\subsection{Fubini-Study metric and quantum phase transitions}

Let us consider a smooth family $H(\bm{\lambda})$ of Hamiltonians
labelled by a set of parameters in a manifold
$\,\bm{\lambda}\in{\cal M}$, in the Hilbert-space $\cal H$ of the
system. If $|\Psi_0(\lambda)\rangle\in{\cal H}$ denotes the, unique
for simplicity, ground-state of $H(\lambda),$ one defines a
one-to-one mapping $\Psi_0\colon{\cal M}\rightarrow{\cal
H}/\lambda\rightarrow|\Psi_0(\lambda)\rangle$ associating to each
set of parameters the ground-state of the corresponding Hamiltonian.
This map can be seen also as a map between a point in ${\cal M}$ and
an element of projective space $P{\cal H}$. As already pointed out
in the previous subsections, the projective Hilbert space is
equipped with a metric, the Fubini-Study distance
\begin{equation}
d_{FS}(\psi,\psi'):=  \sqrt{2-2 |\langle\psi,\psi'\rangle|}\, ,
\label{fidelity}
\end{equation}
which quantifies the maximum amount of statistical
distinguishability between the pure states $|\psi\rangle$ and
$|\psi'\rangle.$  It provides the statistical distance between the
probability distributions of the outcomes associated to
$|\psi\rangle$ and $|\psi'\rangle$, optimised over all possible
measurement strategies. Moreover, as we have already seen, this
result easily extends to mixed states, by replacing the Fubini-Study
distance with its natural density matrix generalisation, the
Bures-distance.

These non-trivial notions allow one to identify the projective
Hilbert space geometry with a geometry in the information space, the
larger the distance between $|\psi\ra$ and $|\psi'\ra$ the more
statistically distinguishable two states are. The remarkable
consequences of this simple observation is that \emph{the distance
encodes information on any of the infinitely many possible order
parameters characterising the phase transitions}. At the critical
point, a small difference between the parameters labelling the
Hamiltonian results in a greatly enhanced distinguishability of the
corresponding ground states. This is quantitatively revealed by the
state overlap and in turn by the behaviour of the metric.

Let $\psi$ and $\psi+d\psi$ be two infinitesimally closed states
in the parameter manifold. In section~\ref{sec:FS}, it has been
shown that their elementary distance in the parameter space can be
expressed as: 
\begin{equation}
 d_{FB}^{2}(\psi,\psi+d\psi)=dl^{2}=\sum_{\mu\nu}g_{\mu\nu}d\lambda_{\mu}d\lambda_{nu}
\end{equation} where metric tensor $g_{\mu\nu}=\Re Q_{\mu\nu}$ is the real part
of the {\em quantum geometric tensor} $Q_{\mu\nu}$ introduced in
section~\ref{sec:FS}, i.e.
\begin{equation}
Q_{\mu\nu}=  \langle\partial_\mu \psi|
\partial_\nu \psi\rangle -\langle \partial_\mu
\psi|\psi\rangle\langle\psi|
\partial_\nu \psi\rangle \ .
\label{eq:QGT}
\end{equation}
As explicitly pointed out in section~\ref{sec:FS}, the imaginary
part of the quantum geometric tensor is $\Im  \,
Q_{\mu\nu}=\frac{F_{\mu\nu}}{2}$, where $F_{\mu\nu}$ is nothing but
the Berry curvature 2-form. This provides a unifying framework for
the understanding of the fidelity and geometric phase approaches to
the quantum phase transitions.

A perturbative expansion provides a simple heuristic explanation as
to why one observes a singular behaviour of the quantum geometric
tensor at QPTs . By using the  first order perturbative expansion
\begin{equation}\label{eq:PertExp}
|\Psi_0(\bm{\lambda}+\delta\bm{\lambda})\rangle\sim
|\Psi_0(\bm{\lambda})\rangle +\sum_{n\neq
0}(E_0(\bm{\lambda})-E_n(\bm{\lambda}+d\bm{\lambda}))^{-1}
|\Psi_n(\bm{\lambda})\rangle\langle\Psi_n(\bm{\lambda})|\delta
H|\Psi_0(\bm{\lambda})\rangle, \end{equation} where $\delta H:=
H(\bm{\lambda}+\delta\bm{\lambda})-H(\bm{\lambda})$, one obtains for
the entries of the quantum geometric tensor (\ref{eq:QGT}) the
following expression
\begin{equation}
Q_{\mu\nu}(\bm{\lambda})= \sum_{n\neq 0}\frac{\langle
\Psi_0(\bm{\lambda})|\partial_\mu H|
\Psi_n(\bm{\lambda})\rangle\langle\Psi_n(\bm{\lambda})|\partial_\nu
H| \Psi_0(\bm{\lambda})\rangle} {[E_n(\bm{\lambda})-
E_0(\bm{\lambda})]^2} \ . \label{pert}
\end{equation}

Continuous QPTs  occur when, for some specific values of the
parameters, the energy gap 
\begin{equation}
 \Delta_{n}(\bm{\lambda}_{c}):=E_{n}(\bm{\lambda}_{c})-E_{0}(\bm{\lambda}_{c})\ge0
\end{equation}
  vanishes in the thermodynamic limit.
This amounts to a vanishing denominator in Eq. (\ref{pert}) that may
break down the analyticity of the metric tensor entries.  This
heuristic argument has been first put forward in~\cite{Zanardi2007}
specifically for the Riemannian tensor $g_{\mu\nu}$ and for the
Berry Curvature $F_{\mu\nu}$ in~\cite{Hamma2006}.  An argument based
on more firm grounds can also be formulated in terms of scaling
properties of the quantum geometric tensor~\cite{CamposVenuti2007}.

To get further insight about the physical origin of these
singularities we notice that the metric tensor  (\ref{eq:QGT}) can
be cast in an interesting covariance matrix form \cite{Provost1980}.
Generically, changing the Hamiltonian from $H(\bm{\lambda})$ to
$H(\bm{\lambda}+\delta\bm{\lambda})$ within the same phase no
level-crossings occur. In this case, the unitary operator
$U(\bm{\lambda},
\delta\bm{\lambda}):=\sum_n|\Psi_n(\bm{\lambda}+\delta\bm{\lambda}\rangle\langle\Psi_n(\bm{\lambda})|$
can adiabatically map the eigenspace at the point
$\bm{\lambda}\in\mathcal{M}$ onto those at
$\bm{\lambda}+\delta\bm{\lambda}.$ In terms of the corresponding
Hermitian generators $X_\mu:= i(\partial_\mu U)U^\dagger$, the Fubini-Study metric tensor (\ref{eq:QGT}) takes the form 
\begin{equation}
 g_{\mu\nu}=(1/2) \langle \{ \bar{X}_\mu ,\, \bar{X}_\nu\} \rangle,
\end{equation} where $\bar{X}_\mu:=X_\mu -\langle X_\mu \rangle$. In other words,
$g_{\mu\nu}$ can be identified with the (symmetric) covariance
matrix of the observables $X_{\mu}$. The differential line element
$dl^2$ can be expressed as the variance of the operator-valued
differential one form $\bm{X}:=i(dU)U^\dagger=\sum_{\mu} X_{\mu}
d\lambda_{\mu}$, i.e., $dl^2=\langle \bar{X}^2 \rangle.$ The
operator $\bm{X}$ is the generator of the mapping between sets of
eigenbases corresponding to infinitesimally closed points
$\bm{\lambda}$ and $\bm{\lambda} + d \bm{\lambda} $. The smaller the
difference between these eigenbases for a given parameter variation,
the smaller the variance of $\bm{X}.$ At QPTs  one expects to have the
maximal possible difference between $|\Psi_0(\lambda)\rangle$ and
$|\Psi_0(\lambda+\delta\lambda)\rangle$, i.e., many ``unperturbed''
eigenstates $|\Psi_n(\lambda)\rangle$ are needed to build up the
``new'' GS; accordingly the variance of $\bm{X}$ can get very large,
possibly divergent. In this sense, $\bm{X}$ acquires the
significance of an \emph{order parameter}, and  $dl^2$ can be
interpreted as its susceptibility.

\subsection{Quantum Phase Transition and Super-Extensitivity of the Quantum Geometric Tensor}\label{sec:SEQ}

In this section we will derive a bound useful to establish a
connection between the quantum geometric tensor $Q$ and QPTs . Let's
consider a system of size $L^{d}$ (with dimensionality $d$). Since
$Q(\bm{\lambda})$ is an Hermitian non-negative matrix one has
\begin{equation}\label{eq:Qineq0} 
|Q_{\mu\nu}|\le\| Q\|_{\infty}= \bm{u}^{\dagger}
\cdot Q  \cdot \bm{u}, 
\end{equation} 
where $\|B\|_{\infty}$ stands for the
largest singular value of a matrix $B$, and
$\bm{u}=\{u_{\mu}\}_{\mu=1}^{{\rm {dim}}{\cal M}}$, with
$\bm{u}^{\dagger}\cdot \bm{u}=1$ is the normalised eigenvector of
$Q$ with the largest eigenvalue. One can define the corresponding
combination of parameters
$\bar{\lambda}:=\sum_{\mu}u_{\mu}\lambda_{\mu}$, which, loosely
speaking, is the direction on the parameter manifold encoding the
maximal ``responsiveness'' of the geometry. Let's denote
$\bar{\p}:=\p/\p\bar{\lambda}$, then

\begin{equation}
  \bar{\p} H=\sum_{\mu}(\partial_{\mu}H) u_{\mu}.
\end{equation}
   From Eq.
(\ref{pert}) and the above inequality~(\ref{eq:Qineq0}),
\begin{eqnarray}
|Q_{\mu\nu}| & \le & \sum_{n>0}\Delta_{n}^{-2}|\langle\Psi_{0}|\bar{\p} H|\Psi_{n}\rangle|^{2}\le\Delta_{1}^{-2}\sum_{n>0}|\langle\Psi_{0}|\bar{\p} H|\Psi_{n}\rangle|^{2}\nonumber \\
 & = & \Delta_{1}^{-2}(\langle\bar{\p} H\bar{\p} H^{\dagger}\rangle-|\langle\bar{\p} H\rangle|^{2}),\label{ineq}
 \end{eqnarray}
where the angular brackets denote the average over
$|\Psi_{0}(\bm{\lambda})\rangle$. Now, a crucial assumption is that
the operator $\bar{\p} H$ is \emph{local} i.e., $\bar{\p}
H=\sum_{j}\bar{\p} V_{j}$ where $V_{j}$ are operators with local
support around the site $j$. Then the last term in Eq. (\ref{ineq})
reads \begin{equation} \sum_{i,j}(\langle\bar{\p} V_{i}\bar{\p}
V_{j}^{\dagger}\rangle-\langle\bar{\p} V_{i}\rangle\langle\bar{\p}
V_{j}^{\dagger}\rangle), \end{equation} If the ground state is
\emph{translationally invariant}, this last quantity can be written
as $L^{d}\sum_{r}K(r):=L^{d}K$, where \begin{equation} K(r):=\langle\bar{\p}
V_{i}\bar{\p} V_{i+r}^{\dagger}\rangle-\langle\bar{\p}
V_{i}\rangle\langle\bar{\p} V_{i+r}^{\dagger}\rangle \end{equation} is
independent of $i.$ For gapped systems i.e.,
$\Delta_{1}(\infty):=\lim_{L\to\infty}\Delta_{1}(L)>0$ the
correlation function $G(r)$ is rapidly decaying \cite{Hastings2004}
and therefore $K$ is finite and independent of the system size.
Using (\ref{ineq}) it follows that for these non-critical systems
$|Q_{\mu\nu}|$ {\em cannot grow, as a function of $L,$ more than
extensively}.
Indeed one has that 
\begin{equation}\label{eq:BoundQ} 
\lim_{L\to\infty}|Q_{\mu\nu}|/L^{d}\le
K\Delta_{1}^{-2}(\infty)<\infty. 
\end{equation}
Conversely if 
\begin{equation} \lim_{L\to\infty}|Q_{\mu\nu}|/L^{d}=\infty 
\end{equation}
  i.e.,$|Q_{\mu\nu}|$ grows super-extensively, then either $\Delta_{1}(L)\rightarrow0$
or $K$ cannot be finite. In both cases the system has to be gapless.
Summarizing: {\em a super-extensive behavior of any of the
components of $Q$ for systems with local interaction implies a
vanishing gap in the thermodynamic limit}~\cite{CamposVenuti2007}.

This behaviour has been observed in a variety of
systems~\cite{Zanardi2006,Zanardi2007,You2007,Zhou2008}, and amounts
to a critical fidelity drop at the QPTs . The extensive behaviour of
the fidelity drop within a normal phase is strongly reminiscent of a
well known physical phenomenon, the Anderson {\em orthogonality
catastrophe}~\cite{Anderson1967}. Namely, as the dimension of the
system increases, the fidelity, i.e. the overlap, between two
infinitesimally neighbouring ground states $\Psi_{0}(\bm{\lambda})$
and $\Psi_{0}(\bm{\lambda}+d\bm{\lambda})$ approaches zero, no
matter how small the difference in parameters $\bm{\lambda}$ is, so
that two ground states are mutually orthogonal in the thermodynamic
limit.

This is a well known feature of systems in many-body physics having
infinitely many degrees of freedom~\cite{Anderson1967}. The fact
that two physical states corresponding to two arbitrarily close sets
of parameters, two arbitrarily similar physical situations, must
become orthogonal to each other in the thermodynamic limit, is not
a distinctive feature of a critical point. It is indeed a behaviour
which is present across the whole phase diagram, hence also between
two states belonging to the same phase. Hence, in itself, this
characteristic has little to do with QPTs.  Despite its emphatic
expression, the ``orthogonality catastrophe'' is much less a
dramatic and unusual peculiarity as its name would suggest. From a
quantum information perspective, it is easier to appreciate how
typical such a behaviour must be, given the infinite dimensionality
of the Hilbert-space that many-body states explore.

What is indeed qualitatively different in QPTs  is the rate at which
the fidelity vanishes in the thermodynamic limit. It is only the
presence of a dramatic, large scale change in the ground state
properties of the system which may allow for a \emph{super-extensive
increase} of the metric, and the consequent rate of reduction of the
state overlap. Loosely speaking, a criticality results in an
orthogonality catastrophe that is expressed on a qualitatively
greater scale. The intuition behind this change of scale may be
gleaned as follows. A local perturbation to a many-body Hamiltonian
far from a critical point may only result in \emph{local}
modifications to the state of the system, i.e. modifications which
are within a region of the size of the correlation length $\xi$.
Such local changes contribute to the reduction of the fidelity with
infinitesimal amounts, that, when accrued, are enough to provide an
increase in the total metric with a rate of up to $L^{d}$ in the
system size $L$. A higher rate is only possible when a local
perturbation generates \emph{non-local} changes on the system
states,  i.e. when the correlation length $\xi$ diverges and the
response of the system to a local perturbation brings in
contributions from degrees of freedom at every scale.

In the following section, we will briefly illustrate the above
considerations by using the simple, yet physically relevant many
body Hamiltonian: the $XY$ spin-chain model.

\subsection{ XY Model and Information Geometry}
To illustrate explicitly how divergencies of $g_{\mu\nu}$ may
arise~\cite{Zanardi2007}, let's go back to the XY model already
discussed in section~\ref{sec:XY}, which for the sake of convenience
we will rewrite here \begin{equation} \label{HXYModel} H = -\sum_{j=1}^{n}
\left(\frac{1 + \delta}{2}\sx_j \sx_{j+1} + \frac{1 -
\delta}{2}\sy_j \sy_{j+1} +\frac{h}{2}\sz_j \right), \end{equation} where $n$
is the number of spins, $\sigma^\mu_j$ are the Pauli matrices at
site $j$, $\delta$ is the x-y anisotropy parameter and $h$  is the
strength of the magnetic field. We already pointed out that the $XY$
model may be converted through the Jordan-Wigner
transformation~(\ref{eq:JW}) into the quasi-Free Fermion model, 
\begin{equation}
 \label{HXYModel} H = -\sum_{j=1}^{n} \left[ ( c_j^{\dagger} c_{j+1}
+ \delta c^{\dagger}_j  c^{\dagger}_{j+1} + H.c.) +h (2
c_j^{\dagger}c_{j} -1)\right]. \end{equation} Ground state, and in general
thermal states of quasi-free Fermion models fall within the general
class of  Gaussian Fermion states. We will introduce in the last
chapter a general framework which allows for the derivation of the
main geometric properties of such models. We will not give the
details of the derivation at this stage, and we will only state the
main result, which in this specific case can be derived directly.
For the sake of completeness, we will consider the rotated model
$H(\varphi)$ in~(\ref{eq:HXYphi}). Indeed from~(\ref{eq:QGT}) and
the form of the ground state~(\ref{eq:GSXY}), one gets 
\begin{equation}
 Q_{\mu\nu}=g_{\mu\nu}+ \frac{i}{2}F_{\mu\nu}\,, \end{equation} where 
\begin{equation}
 g_{\mu\nu}= \frac{1}{4}\sum_{k}\left(\partial_{\mu} \theta_{k}
\partial_{\nu} \theta_{k}+\sin^{2}\theta_{k}\partial_{\mu}\varphi
\partial_{\nu}\varphi \right)\,, \end{equation} \begin{equation} F_{\mu\nu}=
\frac{1}{2}\sum_{k}\left(\partial_{\mu}
\theta_{k}\partial_{\nu}\varphi - \partial_{\nu}
\theta_{k}\partial_{\mu}\varphi  \right)\sin\theta_{k}\,, \end{equation} with
the angle $\theta_k$ being defined as $
\theta_k:=\arccos(\eta_k/\varepsilon_k)$, $\eta_k:=\cos{q_{k}}- h$,
$ \varepsilon_k:=\sqrt{\eta_k^2+\delta^2 \sin^2 q_{k}}$, and
$q_k=2\pi k/n$.

One finds that the only non-vanishing derivatives are
$(\partial_{h}\theta_k)=\delta\sin q_k / \varepsilon_k^2$,
$(\partial_{\delta}\theta_k)=
\sin{q_k}(\cos{q_k}-{h})/\Lambda_\nu^2$, and obviously
$\partial_{\varphi}\varphi=1$.

In the thermodynamic limit, $g_{\mu\nu}$ can be calculated
analytically by replacing the discrete variable $q_k$ with a
continuous variable $q$ and substitute the sum with an integral,
i.e., $\sum_{k=1}^M\to[n/(2\pi)]\int_0^\pi\mathrm{d}q$. At critical
points, this cannot be generally done, due to singular behaviour of
terms involved in the sums. The resulting integrals leads to
analytical expressions which differ in the two regions $|{h}|<1$ and
$|{h}|>1$~\cite{Zanardi2007,Kolodrubetz2013}, 
\bea &&
g_{\varphi\varphi}= \frac{n}{8} \left\{
\begin{array}{cc}
\frac{|\delta|}{|\delta|+1}, & |h|<1\\
\frac{\delta^2}{1-\delta^2}\left( \frac{|h|} {\sqrt{h^2-1+\delta^2}}
- 1 \right) , & |h|>1
\end{array}
\right.
\label{eq:g_phiphi}\nonumber\\
&&  g_{hh}= \frac{n}{16} \left\{\begin{array}{cc}
      \frac{1}{|\delta| (1-h^2)}, & |h|<1\\ \frac{|h| \delta^2}
{(h^2-1)(h^2-1+\delta^2)^{3/2}}, & |h|>1
\end{array}\right. \label{eq:g_hh}\nonumber\\
&& g_{\delta \delta}= \frac{n}{16} \left\{\begin{array}{cc}
      \frac{1}{ |\delta| (1+|\delta|)^2}, & |h|<1\\ \left(
    \frac{2}{(1-\delta^2)^2} \Big[ \frac{|h|}{\sqrt{h^2-1+\delta^2}} - 1 \Big] -
     \frac{|h| \delta^2}{(1-\delta^2)(h^2-1+\delta^2)^{3/2}} \right), & |h|>1
\end{array}\right. \label{eq:g_gg}\nonumber\\
&&  g_{h\delta}= \frac{n}{16} \left\{\begin{array}{cc} 0, & |h|<1\\
\frac{-|h| \delta}{h (h^2-1+\delta^2)^{3/2}} , & |h|>1
\end{array}\right. \label{eq:g_hg}.
\eea
The metric as a whole shows a non-analytical behaviour across both
the critical regions $|h|=1$ and $\delta=0$. To visualise more
clearly such singular behaviour, it is convenient to compute the
scalar curvature, which provides a global property of the metric in
each point of the phase diagram. The scalar curvature $R$, which is
the trace of the Ricci curvature tensor \cite{Nakahara1990}, yields
the following expressions 
\bea
R=&&-\frac{8}{n} \frac{1}{|\delta|} \qquad |h|<1\\
R =&& \frac{8}{n}\left[ 4 + \frac{5h}{\sqrt{h^2 + \delta^2 - 1}} - 2
\frac{\left(h^2 + h \sqrt{h^2 + \delta^2 - 1} - 1 \right)}{\delta^2}
\right] \qquad |h|>1. \eea Note that the curvature diverges on the
segment $|{h}|\leq1,\delta=0$ and it is discontinuous on the lines
${h}=\pm1$. Indeed, $\lim_{|{h}|\to1^+}R=8/n(4+ 3h/|\delta|)$,
$\lim_{|{h}|\to1^-}R=-8/n 1/|\delta|$.

\section{Mixed States and Phase Transitions.}
Up to this point, this review has dealt with models at zero temperature, whose description relies on the properties of the unique ground state, i.e. a pure quantum state. A natural extension of the above considerations regards problems which calls for a description of the system in terms of mixed states rather than a pure
quantum state. The simplest example of such a scenario would be the
thermal state of a quantum system, and a natural generalisation of
the above consideration would be to study phase transitions at finite
temperature. As we will discuss the next few sections, the framework
that we will introduce provides a way to analyse the even more
general senario of open system phase transitions, in which the properties of non-equilibrium dynamics yields mixed stable states~\cite{Magazzu2015,Spagnolo2018,Valenti2018,Guarcello2016,Spagnolo2017,Spagnolo2015,Spagnolo2019}. To this end, we
will need to bring in new tools, and we will start by introducing a
mixed state generalisation of the geometric phase, i.e. the Uhlmann
phase.

\subsection{Definition of Mixed Geometric Phase via State Purification}

The first definition of geometric phase for mixed state has been
proposed by Uhlmann~\cite{Uhlmann1986,Dittmann1999}. In his
formulation a mixed state is allowed to perform any kind of
physically admissible evolution. Therefore, it is truly general and
applicable to the most general setting. However, admittedly, this
definition relies on a rather abstract approach which somehow
obscures its physical interpretation. Still, many proposals to
measure it have been already put
forward~\cite{Tidstrom2003,Aberg2007,Viyuela2016}, and demonstrated
experimentally~\cite{Zhu2011,Kiselev2018}. 

The formulation of Uhlmann geometric
phase relies on the concept of ``purification''. According to this
concept, any mixed state $\rho$ can be regarded as the ``reduced
density matrix'' of a pure state lying in an enlarged Hilbert space.
Essentially, one looks for larger, possibly fictitious, quantum
systems from which the original mixed states are seen as reductions
of pure states. For density operators there is a standard way to do
so by use of the Hilbert-Schmidt operators, or by Hilbert-Schmidt
maps from an auxiliary Hilbert space into the original one.

\subsection{Purification}
Let's start with reviewing some basic idea of the purification
procedure. Let $\mathcal{H}$ be a complex Hilbert space of finite
dimension $n$ with the usual scalar product $\langle...\rangle$ and
let ${\cal B}(\mathcal{H})$ be the algebra of linear operators
acting on $\mathcal{H}$. We remind that formally a general
(\emph{mixed} or \emph{pure}) state is defined as positive linear
trace class operator $\rho\in {\cal B}(\mathcal{H})$ such that
$\Tr{\rho}=1$. In this formalism, a \emph{pure state} (or \emph{rank
one density operator}) is any state $\omega\in {\cal
B}(\mathcal{H})$ for which also $\omega^2=\omega$ holds. Using the
standard notation used in section~\ref{sec:BP}, a pure state $\omega$
is denoted with $\ket{\psi}\bra{\psi}$, $\ket{\psi}\in \mathcal{H}$
being the only eigenstate of $\omega$ with eigenvalue 1.

\emph{A purification} of a mixed state $\rho\in {\cal
B}(\mathcal{H})$ is a a \emph{lift} to pure state
$\ket{\psi}\bra{\psi}$ in a larger space ${\cal B}(\mathcal{H}')$
embedding $\rho\in {\cal B}(\mathcal{H})$. To achieve purification,
it is sufficient to consider an auxiliary Hilbert space
$\mathcal{H}_{aux}$, at least of the same dimension $n$, and then
consider the tensor product space:
\begin{equation}\label{eq:TensSpace}
    \mathcal{H}\otimes \mathcal{H}_{aux},\qquad n=dim \mathcal{H} \le dim \mathcal{H}_{aux}.
\end{equation}
A reduction to $\mathcal{H}$ means performing the partial trace over
the auxiliary space. Now, let $\one_{aux}$ be the identity operator
in $\mathcal{H}_{aux}$, then a state $\ket{\psi}\in
\mathcal{H}\otimes \mathcal{H}_{aux}$ is said to \emph{purify}
$\rho$, if for any operator $O \in {\cal B}(\mathcal{H})$
\begin{equation}\label{purif}
    \Tr (O\rho)=\bra{\psi} O\otimes\one_{aux}\ket{\psi},
\end{equation}
or equivalently if $\rho=\Tr_{aux} \ket{\psi}\bra{\psi}$, where
$\Tr_{aux}$ is the partial trace over the auxiliary space
$\mathcal{H}_{aux}$.\footnote{A distinguished way to choose a
purification, called \emph{standard purification}, is to require
\begin{equation}\label{eq:StandPur}
    \mathcal{H}_{aux}=\mathcal{H}, \qquad \bar{\mathcal{H}}=\mathcal{H}\otimes \mathcal{H}_{aux}.
\end{equation}
When not otherwise specified we will consider only standard
purifications.}

However, it can be formally more convenient to work with a
different notation. Indeed, being of finite dimension,
$\bar{\mathcal{H}}=\mathcal{H}\otimes \mathcal{H}$ is canonically
isomorphic to ${\cal B}(\mathcal{H})$. This can be made explicit by
fixing two arbitrarily chosen orthonormal basis
$\phi_{1},\phi_{2},...$ in $\mathcal{H}$ and
$\phi_{1}',\phi_{2}',...$ in $\mathcal{H}_{aux}$. Given any operator
$\w\in {\cal B}(\mathcal{H})$,
\begin{equation}\label{eq:isomorp}
    \ket{\psi_{\w}}=\sum\ket{\phi_{i}}\otimes\ket{\phi_{j}'}\cdot\bra{\phi_{i}}\w\ket{\phi_{j}'}\qquad
    \in \bar{\mathcal{H}}.
\end{equation}
The canonical scalar product in $\bar{\mathcal{H}}$ is equivalent to
the Hilbert-Schmidt scalar product $(\w_{1},\w_{2})$ in ${\cal
B}(\mathcal{H})$
\begin{equation}\label{eq:scalarP}
    (\w_{1},\w_{2}):=Tr\left(\w_{1}\cdot \w_{2}^{\dag}\right)
    =\sum\bra{\phi_{i}}\w_{1}\ket{\phi_{j}'}\bra{\phi_{j}'}\w_{2}^{\dag}\ket{\phi_{i}}
    = \langle \psi_{\w_{1}}|\psi_{\w_{2}}\rangle,
\end{equation}
and the partial trace over the auxiliary space is given by:
\begin{eqnarray}\label{eq:PartialTr}
    \Tr_{aux}\left(\ket{\psi}\bra{\psi}\right)&=&\sum
    \ket{\phi_{i}}\bra{\phi_{k}}\cdot\bra{\phi_{i}}\w\ket{\phi_{j}'}\bra{\phi_{j}'}\w^{\dag}\ket{\phi_{k}}\\
    &=&\w^{\dag}\cdot \w.
\end{eqnarray}
Therefore, given this isomorphism $\w\leftrightarrow \psi_{\w}$, in
the following we will refer as \emph{standard purification} or
\emph{amplitude} of a density matrix $\rho$ either an operator
$\w\in{\cal B}(\mathcal{H})$, for which
\begin{equation}\label{eq:StPurification}
    \rho=\w^{\dag}\cdot \w,
\end{equation}
or its isomorphic counterpart $\psi_{\w} \in \bar{\mathcal{H}}$
defined by Eq.~\eqref{eq:isomorp}.

A crucial point to stress is that, given a mixed state, the
construction of a standard purification is \emph{by no means
unique}. From a formal point of view, it can be easily checked that
any $\w'=U\cdot \w$, for a given unitary operator $U\in U(n)$,
represents a standard purification of the same state $\rho$. This is
somehow expected, as the purification, from a physical point of
view, represents a ``complete information'' on the global system
described by the \emph{global} Hilbert space $\mathcal{H}\otimes
\mathcal{H}_{aux}$, whereas $\rho$ describes only a part of this
compound. Therefore, $\rho$ is expected to contain only that part of
the ``information'' which can be ``locally'' stored in one of the
subsystem $\mathcal{H}$. This becomes physically obvious by
considering that the transformation $U$ in
$\bar{\mathcal{H}}=\mathcal{H}\otimes \mathcal{H}_{aux}$, looks just
like a \emph{local} change of basis in $\mathcal{H}_{aux}$, which by
no means can affect the state $\rho$ in $\mathcal{H}$.

In the next section, I will often stress the implications of this
\emph{``one to many'' relation} between mixed states and their
purifications. Indeed, for what concerns the definition of mixed
state geometric phase it will become crucial to establish a
criterion which diminish such an ambiguity, by selecting distinguished
set of purifications.

\subsection{Parallel Transport of Density Matrices and Uhlmann Geometric Phase}\label{sec:UlmParaTr}
Given this definition of purification, it would be natural to
generalise the concept of geometric phase for a chain of density
matrices, by referring to their purifications. Indeed, as
purifications are, by definition, pure, we could just
straightforwardly apply the definition of Berry phases exploited in
the previous sections. Unfortunately, this \emph{programme} does not
generate an unambiguous value of the geometric phase, on account of
the lack of uniqueness in the purification procedure. The problem,
is, therefore, to select among all possible ones a distinguished set
of purification. A solution to this problem was proposed by
Uhlmann~\cite{Uhlmann1976,Uhlmann1986,Uhlmann1989,Uhlmann1991}. His
idea is based on the concept of parallel transport.

Let's start by considering a path of density operator $\rho(s)$,
$s\in[s_{a},s_{b}]$, and its purified path
\begin{equation}\label{eq:prurif-curve}
\rho(s) \to \w(s)
\end{equation}
i.e. such that $\rho(s)=\w^{\dag}(s)\cdot \w(s)$. By the previous
argument, not only~(\ref{eq:prurif-curve}) represents a purification
but also every unitarily transformed path
\begin{equation}\label{eq:prurif-gauged-curve}
    \w(s) \to U(s)\cdot \w(s).
\end{equation}
By analogy with the idea of gauge transformation used in
section\label{sec:BPadiab} for pure states, it is natural to refer
to~(\ref{eq:prurif-gauged-curve}) as a \emph{gauge transformation}
for mixed states, and in general the freedom in the choice of an
amplitude $\w$ as \emph{gauge freedom}. Notice, that the set of
possible gauge transformations that could be adopted in the case of
the pure states were mere multiplications by complex phase factor's
$e^{i\alpha}\in U(1)$, i.e. a $U(1)$ gauge freedom. This allowed for
the description of the Berry phase in terms of an underlying
\emph{Abelian} gauge theory, which by analogy could be compared with
the usual electromagnetic $U(1)$ gauge field. The much wider choice
of a general unitary operator $U\in U(n)$, in the present case,
calls for the more convoluted $U(n)$ gauge structure. We will show,
that the natural setting underlying the definition of the Uhlmann
geometric phase is within the theory of  \emph{holonomies}, i.e. the
non-Abelian generalisation of the geometric phase.

Let $\w_{1},\w_{2},..,\w_{m}$ be a finite subdivision of the
curve~(\ref{eq:prurif-curve}), i.e. a path ordered subset of
operators~(\ref{eq:prurif-curve}). Notice that these operators have
norm $||\w_{i}||^{2}:=(\w_{i},\w_{i})=1$, due to the normalisation
condition of $\rho_{i}$. Let us consider the discrete chain of pure
states $\psi_{\w_{i}}$
\begin{eqnarray}\label{eq:UhlmGphase}
    \xi&=&(\w_{1},\w_{2})(\w_{2},\w_{3})\dots(\w_{m-1},\w_{m})(\w_{m},\w_{1})\\
    &=&\langle \psi_{\w_{1}}|\psi_{\w_{2}}\rangle\langle
    \psi_{\w_{1}}|\psi_{\w_{3}}\rangle\dots\langle \psi_{\w_{m-1}}|\psi_{\w_{m}}\rangle\langle
    \psi_{\w_{m}}|\psi_{\w_{1}}\rangle.
\end{eqnarray}
One could be tempted to consider the complex argument of the
functional $\xi$, and take the limit from the discrete chain
$\w_{1}\dots \w_{n}$ to a continuous evolution, and identify this
limit with the "mixed state geometric phase" of the path $\rho(s)$.
However, each \emph{gauged transformed}
path~(\ref{eq:prurif-gauged-curve}) generally produces a different
$\tilde{\xi}\neq\xi$. A sensible criterion to diminish this
arbitrariness is needed.\\\indent Uhlmann introduced a parallel
transport criterion, analogous to the parallel transport condition
for pure states, which is able to single out a specific set of
purified paths and uniquely identifies the geometric phase. In fact,
if one tries to purify two density operators, $\rho_{1}$ and
$\rho_{2}$, simultaneously, say with $\psi_{\w_{1}}$ and
$\psi_{\w_{2}}$, the purification ambiguity can be partially lifted
by choosing them to be ``as near as possible'' to each
other~\cite{Bures1969,Uhlmann1976,Araki1982,Alberti1983}. Given
$\ket{\psi_{\w_{1}}}$, there is a $\ket{\psi_{\w_{2}}}$ with maximal
overlap
\begin{equation}\label{eq:MaximalBures}
    |\langle\psi_{\w_{1}}|\psi_{\w_{2}}\rangle|\ge |\langle\psi'_{\w_{1}}|\psi'_{\w_{2}}\rangle|
\end{equation}
or, equivalently, with minimal Fubini-Study distance
$d^{2}_{FS}(\psi_{\w_{1}},\psi_{\w_{2}})=2-2|\langle\psi_{\w_{1}}|\psi_{\w_{2}}\rangle|$
among all pair of vectors, $\ket{\psi'_{\w_{1}}}$,
$\ket{\psi'_{\w_{2}}}$ simultaneously purifying $\rho_{1}$ and
$\rho_{2}$. Ulhmann describes this situation by calling the pair
$\psi_{\w_{1}}$ and $\psi_{\w_{2}}$ \emph{parallel}, as a shorthand
for ``as parallel as possible''~\cite{Uhlmann1986}.
 This criterion, therefore, allows one
to distinguish, within all purifications $\w_{i}$ of the curve $\rho(s)$,
 the exceptional ones $\tilde{\w}_{i}$ for which the overlap
between an element of the purified chain $\tilde{\w}_{i}$ and the
neighbouring one $\tilde{\w}_{i+1}$ is maximised. If this
condition is fulfilled for the whole chain, the remaining
arbitrariness is in a \emph{regauging} $\tilde{\w}_{i}\to
e^{i\alpha_{i}}U\cdot \tilde{\w}_{i}$ of the subdivision, with
$\alpha_{i} \in\mathbb{R}$ and a unitary operator $U$ independent of
the index $i$. This, however, leaves the quantity
\begin{equation}\label{eq:UhlmGphaseGaugeInv}
    \tilde{\xi}=(\tilde{\w}_{1},\tilde{\w}_{2})(\tilde{\w}_{2},\tilde{\w}_{3})...(\tilde{\w}_{m-1},\tilde{\w}_{m})(\tilde{\w}_{m},\tilde{\w}_{1})
\end{equation}
\emph{invariant}. Therefore, $\tilde{\xi}$ is uniquely defined by
the discrete chain of state $\rho_{i}=\w_{i}^{\dag}\w_{i}$ and it is
meaningful to regard $\Phi_{g}=\arg{\tilde{\xi}}$ as a mixed state
generalisation of the geometric phase.

One can also sharpen the condition of parallel
transport~(\ref{eq:MaximalBures}) by making use of the remaining
regauging degree of freedom. One can, indeed, require two
neighbouring purifications to be in phase, in the Pancharatnam
criterion, i.e.
\begin{equation}\label{eq:UhlmGphaseGaugeInv1}
    (\tilde{\w}_{i},\tilde{\w}_{i+1})=\langle\psi_{\w_{i}}|\psi_{\w_{i+1}}\rangle \ge 0.
\end{equation}
For such a parallel purification, the mixed geometric phase becomes:
\begin{equation}\label{eq:discreteInphase}
    \Phi_{g}=\arg{(\tilde{\w}_{N},\tilde{\w}_{1})}\,.
\end{equation}
The condition~(\ref{eq:UhlmGphaseGaugeInv1}) is equivalent to requiring
that
\begin{equation}\label{eq:preLength}
  ||\psi_{\w_{1}}-\psi_{\w_{2}}|| + ||\psi_{\w_{2}}-\psi_{\w_{3}}|| + \dots + ||\psi_{\w_{N-1}}-\psi_{\w_{N}}||,
\end{equation}
attains its minimum. Going to the limit of finer and finer
subdivisions, Eq.~(\ref{eq:preLength}) converges to the length
of the curve of the purification~(\ref{eq:prurif-curve}). Therefore
the purification is parallel if and only if it solves the following
variational problem,
\begin{equation}\label{eq:length}
D_{B} =\int_{\rho(s_{a})}^{\rho(s_{b})} dl_{B}:= \min  \int_{s_a}^{s_b}\sqrt{\bk{\dot{\psi}_{\w(s)}|\dot{\psi}_{\w(s)}}}ds\,, 
\end{equation}
where $\psi_{\w(s)}$ is a purified path of $\rho(s)$, and the dots
denote derivatives with respect to $s$. The resulting minimal length
$D_{B}$ is called \emph{Bures length} or \emph{Bures
angle}~\cite{Bures1969,Uhlmann1992}. Therefore a purification
$\psi_{\w(s)}$ is called ``parallel'' or ``parallel transported''
if, for every gauged purification $\psi'_{\w(s)}$ of $\rho(s)$, it
holds
\begin{equation}\label{eq:UhlmannParallCond1}
\sqrt{\langle\dot{\psi}_{\w(s)}|\dot{\psi}_{\w(s)}\rangle}\le
\sqrt{\langle\dot{\psi'}_{\w(s)}|\dot{\psi'}_{\w(s)}\rangle}
    \quad \forall s.
\end{equation}
It is plain to derive a condition for parallel purification, which
is easier to handle. Suppose that $\psi_{\w(s)}$ is a parallel
purification, then $\psi'_{\w(s)}=U(s)\psi_{\w(s)}$, with
$U(s)=\one\otimes U'$ unitary, is another purification of $\rho(s)$.
Inserting this into Eq.~(\ref{eq:UhlmannParallCond1}) leads to
\begin{equation}\label{eq:UhlmannParallCond2}
0\le \langle\psi_\w|B^{\dag} B|\psi_\w\rangle
+i\left(\langle\dot{\psi}_\w|B|\psi_{\w}\rangle-\langle\psi_{\w}|B|\dot{\psi}_\w\rangle\right),
\end{equation}
where $B$ is the Hermitian generator of $U$, i.e. $B(s):=i\dot{U}(s)
U^{\dagger}(s)$. This inequality is valid if and only if
$\langle\dot{\psi}_\w|B|\psi_{\w}\rangle=\langle\psi_{\w}|B|\dot{\psi}_\w\rangle$
for all Hermitian operators $B=\one\otimes B'$. In the language of
the Hilbert-Schmidt space, this condition becomes $
\Tr\left(\dot{\w}\w^{\dag}B'-\w\dot{\w}^{\dag}B'\right)=0$ for all
$B'$ Hermitian, which essentially means nothing other than
\begin{equation}\label{eq:UhlmannParallCond5}
\dot{\w}\w^{\dag}=\w\dot{\w}^{\dag}.
\end{equation}
This condition, together with the normalisation of $\rho(s)$,
implies that $(\w(s),\w(s+\delta s))\approx1$, thus guaranteeing
that each $\w(s)$ and its neighbour $\w(s+\delta s)$ are in phase in
the Pancharatnam sense. The \emph{Ulhmann mixed geometric phase}
results just in the residual phase difference between initial and
final state, i.e.
\begin{equation}\label{eq:UhlmannGeoPh}
\Phi_{g}=\arg{\left(\w(s_{b}),\w(s_{a})\right)}\,,
\end{equation}
with $s_{a}$ and $s_{b}$ initial and final value of the parameter,
respectively.

\section{Fidelity and Bures Metric}
According to the Uhlmann parallelism, two amplitudes $\w_{1}$ and
$\w_{2}$ are called parallel if they maximise their Hilbert-Schmidt
scalar product, among those simultaneously purifying $\rho_{1}$ and
$\rho_{2}$. A very important byproduct of this maximisation
procedure is the so called fidelity~\cite{Uhlmann1976}, defined as
\begin{equation}\label{eq:fidelitydef}
    \mathcal{F}(\rho_{1},\rho_{2}):= \max_{\w_{1},\w_{2}} |(\w_{1},\w_{2})| = \max_{\psi_{\w_{1}},\psi_{\w_{2}}} |\bk{\psi_{\w_{1}}|\psi_{\w_{2}}}|\,.
\end{equation}
This is a very crucial quantity in quantum information and in
quantum estimation theory. It provides an operationally well-defined
distance between quantum states, in terms of statistical
distinguishability of quantum states. An explicit expression for the
above maximal value has been proven by
Uhlmann~\cite{Uhlmann1976,Jozsa1994}
\begin{equation}\label{eq:fidelity}
\mathcal{F}(\rho_{1},\rho_{2})=\Tr{\sqrt{\sqrt{\rho_{2}}\rho_{1}\sqrt{\rho_{2}}}},
\end{equation}
which readily shows how the fidelity depends on $\rho_{1}$ and
$\rho_{2}$ only. The proof of the above expression relies on the
following simple lemma.

\begin{Lemma}\label{lemma:HS}
\emph{For any operator $B$, and any unitary $U$, $|\Tr{(B
U)}|\le\Tr|B|$, with equality attained for $U=V^{\dagger}$, where
$B=:|B|V$ is the polar decomposition of $B$, with $|B|:=\sqrt{B
B^{\dagger}}$.}
\end{Lemma}
\begin{proof}
The equality follows straightforwardly from the condition stated,
whereas the inequality arises from
\begin{equation}\label{eq:BIneq}
|\Tr{B U}| =
|\Tr(|B|VU)|=|\Tr(|B|^{\frac{1}{2}}|B|^{\frac{1}{2}}VU)|\le\sqrt{\Tr|B|
\Tr(U^{\dagger}V^{\dagger}|B|V U )}=\Tr|B|,
\end{equation}
where the second relation is the Cauchy-Schwartz inequality for the
Hilbert-Schmidt scalar product.\end{proof}

To prove Eq.~(\ref{eq:fidelity}), we define
$\w_{i}=:\sqrt{\rho_{i}}U_{i}$, with $i=(1,2)$ the polar
decompositions of two purifications of $\rho_{1}$ and $\rho_{2}$.
Then, by Lemma~\ref{lemma:HS}
\[
    |\Tr(\w_{1}^{\dagger}\w_{2})|=|\Tr(\sqrt{\rho_{1}}\sqrt{\rho_{2}}U_{2}U_{1}^{\dagger})|\le\Tr|\sqrt{\rho_{1}}\sqrt{\rho_{2}}|=\Tr\sqrt{\sqrt{\rho_{1}}\rho_{2}\sqrt{\rho_{1}}}.
\]
The equality is attained for $U_{2}U_{1}^{\dagger}=V^{\dagger}$,
where
$\sqrt{\rho_{1}}\sqrt{\rho_{2}}=:|\sqrt{\rho_{1}}\sqrt{\rho_{2}}|V$.

Two important limiting case of the fidelity are worth mentioning
explicitly. The first one is when we deal with pure states,
$\rho_{1}=\kb{\psi_{1}}$ and $\rho_{2}=\kb{\psi_{2}}$. In this case,
the fidelity reduces to the standard overlap
$\mathcal{F}(\rho_{1},\rho_{2})=|\bk{\psi_{1}|\psi_{2}}|$. Slighlty
more generally, if just one of the two states is pure,
$\rho_{1}=\kb{\psi_{1}}$, then
$\mathcal{F}(\rho_{1},\rho_{2})^{2}=|\bk{\psi_{1}|\rho_{2}|\psi_{1}}|$,
which is the probability that the state $\rho_{2}$ will score
positively if tested on whether it is in the pure state $\rho_{1}$.
It serves as a figure of merit in many statistical estimation
problems.\\\indent The second example is when $\rho_{1}$ and
$\rho_{2}$ commute, i.e. when they are simultaneously diagonal,
$\rho_{i}=\sum_{k} p_{i}(k) \kb{k}$. In this case, the fidelity
reduces to 
\begin{equation}
 \mathcal{F}(\rho_{1},\rho_{2})=\sum_{k}\sqrt{p_{1}(k)}\sqrt{p_{2}(k)}=B(P_{1},P_{2}),
\end{equation} i.e.  the Bhattacharyya coefficient of the classical statistical
distributions $P_{i}:=\{p_{i}(k),k=1\dots n\}$.

The fidelity also enjoys a number of quite desiderable properties
for a measure of statistical distinguishability~\cite{Jozsa1994}:
\begin{enumerate}
\item $0\le \mathcal{F}(\rho_1,\rho_2)\le 1$;
\item $\mathcal{F}(\rho_1,\rho_2)=1$ iff $\rho_1 =\rho_2$ and $\mathcal{F}(\rho_1,\rho_2)=0$ iff $\rho_1$ and $\rho_2$ have
orthogonal supports;
\item Symmetry, $\mathcal{F}(\rho_1, \rho_2) = \mathcal{F}(\rho_2, \rho_1)$;
\item Strong concavity, $\mathcal{F}(\sum_{j}p_{j}\rho_{j}, \sum_{j}q_j\rho'_j) \ge \sum_{j}\sqrt{p_{j}q_{j}}\mathcal{F}(\rho_{j},\rho'_{j})$;\label{item:SConc}
\item Multiplicativity, $\mathcal{F} (\rho_1 \otimes \rho_2 , \rho_3 \otimes \rho_4 ) = \mathcal{F} (\rho_1 , \rho_3 ) \mathcal{F} (\rho_2 , \rho_4 )$;
\item Unitary invariance, $\mathcal{F}(\rho_1,\rho_2) = \mathcal{F}(U\rho_1 U^{\dagger},U\rho_2U^{\dagger})$;
\item Monotonicity, $\mathcal{F}( \Phi(\rho_1), \Phi(\rho_2)) \le \mathcal{F}(\rho_1,\rho_2)$, where $\Phi$ is a trace preserving CP map.\label{Item:Mono}
\end{enumerate}
Property~\ref{item:SConc} also implies concavity, i.e.
$\mathcal{F}(\sum_{j}p_{j}\rho_{j}, \rho') \ge
\sum_{j}\sqrt{p_{j}q_{j}}\mathcal{F}(\rho_{j},\rho')$.
Property~\ref{Item:Mono} is a key entry, it means that the fidelity
cannot grow under any type of physical process, i.e. unitary
transformations, generalised measurements, stochastic maps and
combinations thereof. This is a crucial demand for any bona-fide
measure of distinguishability. It is, indeed, physically
unacceptable that any stochastic map, including measurements, may
contribute in increasing the distinguishability of two
states.\\\indent To explicitly see in what sense the fidelity is a
measure of statistical distinguishability~\cite{Fuchs1995} let's
follow a similar argument exposed in section~\ref{sec:FSstat}. In
the case of two pure states $\psi_{1}$ and $\psi_{2}$, it has been
shown that the overlap $|\bk{\psi_{1}|\psi_{2}}|$ provides a measure
of the experimental (in)-distinguishability of two quantum states,
assuming no limitations on the type of measurements one can perform.
One can show that the same applies to the fidelity in the case of
two mixed states $\rho_{1}$ and $\rho_{2}$. One assumes a specific
measurement process, and defines a statistical measure of
distinguishability between the two resulting outcome distributions.
This measure clearly depends on the choice of the measurement
process as well as on the states. One can then select the optimal
measurement strategy that maximises the distinguishability
according to some figure of merit, and define the latter as the measure
of distinguishability between the states.\\\indent For simplicity we
will assume both states to be full-rank (i.e. invertible) density
matrices. We will allow for the most general type of measuring
device, i.e. a positive operator valued measurement
(POVM)~\cite{Nielsen2000} $\{E_{k}, k=1\dots n\}$. A given density
matrix $\rho_{i}$ responds to such a measurement process with a
probability distribution $P_{i}:=\{p(\rho_{i},E_{k}),  k=1\dots n
\}$, where $p_{i}(k):=\Tr(\rho_{i}E_{k})$. The optimal POVM is the
one that produces two distribution $P_{1}$ and $P_{2}$ which are the
most statistically distinguishable. As in the case of pure states,
the figure of merit of choice is the Bhattacharyya coefficient
\[
B(P_{1},P_{2})=\sum_{k} \sqrt{p_{1}(k) p_{2}(k)}\,,
\]
which has to be minimised over the POVMs. For a generic unitary operator $U$,
rewriting
\begin{align*}
p_{1}(k):=\Tr\left( (U\sqrt{\rho_{1}} \sqrt{E_{k}})
(U\sqrt{\rho_{1}} \sqrt{E_{k}})^{\dagger} \right)
\end{align*}
yields the following chain of relations,
\begin{align*}
B(P_{1},P_{2})&=\sum_{k}\Tr\left( (U\sqrt{\rho_{1}} \sqrt{E_{k}})
(U\sqrt{\rho_{1}} \sqrt{E_{k}})^{\dagger} \right)^{\frac{1}{2}}
\Tr\left( (\sqrt{\rho_{2}} \sqrt{E_{k}}) (\sqrt{\rho_{2}}
\sqrt{E_{k}})^{\dagger} \right)^{\frac{1}{2}}\\&\ge
\sum_{k}|\Tr\left( (U\sqrt{\rho_{1}} \sqrt{E_{k}}) (\sqrt{\rho_{2}} \sqrt{E_{k}})^{\dagger} \right)| = \sum_{k}|\Tr\left( U\sqrt{\rho_{1}} E_{k} \sqrt{\rho_{2} } \right)|\\
&\ge \Big|\Tr\Big(\sum_{k} U\sqrt{\rho_{1}} E_{k} \sqrt{\rho_{2} }
\Big)\Big|= |\Tr(U\sqrt{\rho_{1}}\sqrt{\rho_{2} })|,
\end{align*}
where the second line is due to the Cauchy-Schwarz inequality for the Hilbert-Schmidt scalar product, with the equality being attained if condition (a): $\sqrt{\rho_{2}} E_{k} \propto U \sqrt{\rho_{1}} E_{k}$ is fullfilled. The remaining relations arise from the linearity of the trace and completeness property of the POVMs, and the second inequality can be saturated only if (b): $\Tr\Big(U\sqrt{\rho_{1}} E_{k} \sqrt{\rho_{2} } \Big)\ge 0$ $\forall k$.\\
Notice that the above inequalities are fulfilled by any unitary $U$.
Therefore, if it has to be a chance of attaining equality in them,
$U$ had better be chosen so as to maximise
$|\Tr(U\sqrt{\rho_{1}}\sqrt{\rho_{2} })|$. From Lemma 1, we know
that this is achieved by
\begin{equation}\label{eq:UOpt}
U= \sqrt{
\sqrt{\rho_{2}}\rho_{1}\sqrt{\rho_{2}}}\rho_{2}^{-\frac{1}{2}}\rho_{1}^{-\frac{1}{2}}.
\end{equation}
It can be checked that, with this unitary operator both condition (a)
and (b) can be satisfied by a set of POVMs $E_{k}$, which are
projective measurements onto the eigenbasis of the following
Hermitian operator
\begin{equation}\label{eq:MOpt}
M:= \rho_{2}^{-\frac{1}{2}}\sqrt{
\sqrt{\rho_{2}}\rho_{1}\sqrt{\rho_{2}}}\rho_{2}^{-\frac{1}{2}}.
\end{equation}
The end result is $B(P_{1},P_{2})=\mathcal{F}(\rho_{1},\rho_{2})$.\\
This finally establishes the interpretation of the fidelity as a
statistical measure of distinguishability. This parallels the
discussion we made regarding the Fubini-Study metric in
section~\ref{sec:FSstat}, when the states to be distinguished were
pure, i.e. $\rho_{i}=\kb{\psi_{i}}$. In that case, it was found that
the Bhattacharyya coefficient distinguishing the probability
distributions for of the optimal measurement apparatus equalled the
overlap $|\bk{\psi_{1}|\psi_{2}}|$. These two solutions are
consistent. However, while for pure states several optimal
measurements are possible, here the observable $M$ providing the
optimal distinguishability is uniquely defined. We have derived its
expression, and it corresponds to the \emph{geometric mean} of the
operators $\rho_{2}^{-1}$ and $\rho_{1}$.

\subsection{The Bures Metric}
The fidelity provides a natural way of defining a distance on the
space of density matrices. The definition~(\ref{eq:fidelitydef}) of
the fidelity is based on a suitably optimised overlap between pure
states. We could therefore borrow the considerations on the
Fubini-Study metric exposed in section~\ref{sec:FSstat}, and apply
them verbatim to purifications. Once the optimisation over the
purification is carried out, the result is the definition of two
Riemannian distances, called \emph{Bures distance}
\begin{equation}\label{eq:BuresDist}
    d_{B}(\rho_{1},\rho_{2}):=\sqrt{2-2\mathcal{F}(\rho_{1},\rho_{2})};
\end{equation}
and \emph{Bures length} or \emph{Bures angle}
\begin{equation}\label{eq:BuresLength}
    D_{B}(\rho_{1},\rho_{2}):=\arccos{\mathcal{F}(\rho_{1},\rho_{2})}.
\end{equation}
Clearly, these are the  generalisations of the Fubini-Study distance
$d_{FS}$ and Fubini-Study length $D_{FS}$, respectively, when the
states $\rho_{1}$ and $\rho_{2}$ are allowed to be mixed. Like in
the case of the Fubini-Study metric, the first distance $d_{B}$
measures the length of a straight chord, while $D_{B}$ measures the
length of a curve within the manifold of density matrices. By
construction, they are Riemannian distances, and are consistent with
the same Riemannian metric. Moreover, notice that they are both
monotonously decreasing functions of the fidelity. This means that
$d_{B}$ and $D_{B}$ can be regarded as distances that measure the
statistical distinguishability between two quantum states. This is
further confirmed by the monotonicity property~\ref{Item:Mono},
which entails that both  $d_{B}$ and $D_{B}$ are non-decreasing
under stochastic maps, i.e.under any physically meaningful quantum
operations.\\\indent With the confidence that we are investigating a
relevant definition of distance, let's turn to the Riemannian metric
induced by the Bures distance, or equivalently by the Bures length.
In the limit of two density matrices infinitesimally apart $\rho(s)$
and $\rho(s+ds)$, both $d_{B}$ and $D_{B}$ converge to the
infinitesimal length
\begin{equation}\label{eq:Buresmetric}
dl^{2}_{B}:=\min \Tr(\dot{\w}^{\dagger}\dot{\w})ds^2,
\end{equation}
in agreement with the expression~(\ref{eq:length}), where the
minimum is attained by the amplitudes $\w(s)$ fulfilling the
parallel transport condition~(\ref{eq:UhlmannParallCond5}). This is
the \emph{Bures metric}. It is easy to check that
condition~(\ref{eq:UhlmannParallCond5}) is fullfilled by the
following ansatz~\cite{Uhlmann1989,Dabrowski1989,Dabrowski1990}
\begin{equation}\label{ansatz1}
\dot{\w} = G \w,\qquad G^{\dagger}=G.
\end{equation}
$G$ can be determined by differentiating $\rho=\w \w^{\dagger}$ and
inserting~(\ref{ansatz1}), which yields:
 \begin{equation}\label{eq:G}
 \dot{\rho}= G\rho+\rho G.
 \end{equation}
The quantity $G$, which may be called the \emph{parallel transport
generator} (PTG) is implicitly defined as the (unique) operator
solution of~(\ref{eq:G}) with the auxiliary requirement that
\begin{equation}\label{eq:AuxG} \bk{\psi | G | \psi} = 0, \textrm{ whenever
}\rho\ket{\psi}=0. \end{equation} In terms of $G$, the Bures metric can be cast
in the following forms
\begin{equation}\label{eq:Buresmetric1}
dl^{2}_{B}:=\Tr(\w^{\dagger}G^{2} \w )ds^2=\Tr(G^{2} \rho
)ds^2=\frac{1}{2}\Tr(G \dot{\rho} )ds^2\,,
\end{equation}
where $g_{\mu\nu}$ is the Bures metric tensor. Assume that the
$\rho(\bm{\lambda})$ is a collection of density matrices labelled by
a set of parameters $\bm{\lambda}:=\{\lambda_{\mu}\}$
belonging to a manifold $\mathcal{M}$. The components of the Bures
metric on the induced manifold are given by
\begin{equation}\label{eq:Buresmetric2}
    dl^{2}_{B}=:\sum_{\mu\nu} g_{\mu\nu} d\lambda_{\mu} d\lambda_{\nu}, \qquad g_{\mu\nu}=\frac{1}{2}\Tr(\{G_{\mu},G_{\nu}\} \rho ),
\end{equation}
where $\{.,.\}$ is the anti-commutator, and $G_{\mu}$ is the
restriction of $G$ along the coordinate $\lambda_{\mu}$, and it is
determined by the analog of Eq.~(\ref{eq:G}), i.e.
$\p_{\mu}{\rho}= G_{\mu}\rho+\rho G_{\mu}$. We can also raise $G$ to
the rank of an operator-valued differential one-form, defined as
$\bm{G}:=\sum_{\mu}G_{\mu} d\lambda_{\mu}$, which clearly obeys
\[
d\rho=\bm{G}\rho+\rho \bm{G}.
\]
This expression can be easily solved in the basis of eigenvalues of
the density matrix $\rho=\sum_{k} p_{k}\kb{k}$ 
\begin{align}\label{eq:GEigen}
\bra{j}\bm{G}\ket{k}=\sum_{p_{j}>0,p_{k}>0}\frac{\bra{j}d\rho\ket{k}}{p_{j}+p_{k}}
\end{align} 
where the restriction $p_{j}>0,p_{k}>0$ on the summation derives
from the auxiliary condition~(\ref{eq:AuxG}). Casting this
expression into Eq.~(\ref{eq:Buresmetric2}) leads to the following
explicit form for the Bures metric
tensor~\cite{Sommers2003,Safranek2017}, \begin{equation}\label{eq:BuresTensor}
g_{\mu\nu}=\frac{1}{2}\sum_{p_{j}>0,p_{k}>0}\frac{\bra{j}\p_{\mu}\rho\ket{k}\bra{k}\p_{\nu}\rho\ket{j}}{p_{j}+p_{k}}\,
. \end{equation} Now, let's cast the expression~(\ref{eq:BuresTensor}) in a
form amenable to interesting considerations. Let us first
differentiate the density matrix
$\p_{\mu}\rho=\sum_{k}(\p_{\mu}p_{k}| k\ra\la k |+p_{k}| \p_{\mu} k
\ra\la k |+p_{k}|k\ra\la \p_{\mu} k|)$ and consider the matrix
element $(\p_{\mu}\rho)_{\mu\nu}$. Notice that $\la
j|k\ra=\delta_{j,k}\Rightarrow\la \p_{\mu} j|k\ra=-\la
i|\p_{\mu}k\ra;$ whence $\la j|
\p_{\mu}\rho|k\ra=\delta_{j,k}\,\p_{\mu}p_{j}+\langle
j|\p_{\mu}k\rangle(p_{k}-p_{j}).$ Plugging this expression back into
(\ref{eq:BuresTensor}) yields 
\begin{equation}
 g_{\mu\nu}=\frac{1}{4}\sum_{p_{k}>0}p_{k}\left(\frac{\p_{\mu}p_{k}}{p_{k}}\right)
\left(\frac{\p_{\mu}p_{k}}{p_{k}}\right)
+\frac{1}{2}\sum_{\substack{p_{j}>0\\p_{k}>0}} \la j |\p_{\mu}k\ra
\la\p_{\nu} k| j \ra\frac{(p_{j}-p_{k})^{2}}{p_{j}+p_{k}}\,
.\label{start} 
\end{equation}

This expression provides an interesting distinction between a
classical and a quantum contribution. Indeed, the first term
in~(\ref{start}) is the so called {\em Fisher-Rao} metric. This is
the metric induced by both the Hellinger distance and the Bhattacharyya
distance between the infinitesimally close probability distributions
$\{ p_{k}\}$ and $\{ p_{k}+dp_{k}\}$. The second term takes
into account the generic non-commutativity of $\rho$ and
$\rho^{\prime}:=\rho+d\rho.$ Thus, one may refer to these two terms
as the classical and non-classical contributions to the metric,
respectively. When $[\rho^{\prime},\,\rho]=0$ the problem reduces to
an effective classical problem and the Bures metric obviously
collapses into the Fisher-Rao metric.

One can draw an interesting connection between the
metric~(\ref{eq:BuresTensor}) and a quantity of quantum information
known as \emph{quantum Chernoff bound}~\cite{Audenaert2007}. This is
the generalisation of the classical Chernoff bound used in
information theory~\cite{Chernoff1952,Cover2006}. Consider an
experimental procedure aiming at distinguishing  two quantum states
$\rho_{1}$ and $\rho_{2}$, where a large number $n$ of copies are
provided, and collective measurement are allowed for. In the limit
of very large $n$, the probability of error in discriminating
$\rho_{1}$ and $\rho_{2}$ decays exponentially as
$P_{err}\sim\exp(-n\xi_{QCB})$, where $\xi_{QCB}$ denotes the
quantum Chernoff bound. The Chernoff bound generates a metric over
the space of quantum states naturally endowed with an operationally
well defined character: \emph{The farther apart two states lie
according to this distance, the smaller is the asymptotic error rate
of a procedure that attempts to tell them apart.}

In \cite{Audenaert2007} it has been proved that 
\begin{equation}
 \exp(-\xi_{QCB})=\min_{0\le s\le1}{\rm
{tr}\left(\rho_{1}^{s}\rho_{2}^{1-s}\right)\le{\cal
F}(\rho_{1},\rho_{2})}\,, \end{equation} which, for infinitesimally close states
$\rho_{1}=\rho$ and $\rho_{2}=\rho+d\rho,$ yields 
\begin{equation}
 dl_{QCB}^{2}:=1-\exp(-\xi_{QCB})=\sum_{\mu\nu}g^{QCB}_{\mu\nu}d\lambda_{\mu}d\lambda_{\nu}\,,
\end{equation} where \begin{equation} g^{QCB}_{\mu\nu}=\frac{1}{2}\sum_{j,k}\frac{\la j|
\p_{\mu}\rho|k\ra \la k|
\p_{\nu}\rho|j\ra}{(\sqrt{p_{j}}+\sqrt{p_{k}})^{2}}.\label{QCB-metric}
\end{equation}
  This expression shows both distinguishability distances,
 the quantum Chernoff bound metric and the Bures metric~(\ref{eq:BuresTensor}), share similar structures. They are identical except for the denominators, where $p_{j}+p_{k}$ are replaced by $(\sqrt{p_{j}}+\sqrt{p_{k}})^{2}.$ The following inequalities $(\sqrt{p_{m}}+\sqrt{p_{n}})^{2}\ge p_{n}+p_{m}$
and $2(p_{j}+p_{k})\ge(\sqrt{p_{j}}+\sqrt{p_{k}})^{2}$ imply the
equivalence of these two metric tensors, i.e. 
\begin{equation}
 \frac{g_{\mu\nu}}{2}\le g_{\mu\nu}^{QCB}\le g_{\mu\nu}.
\label{ineq:metric}\end{equation} Therefore one expects the two
distinguishability measures to convey equivalent information as far
as local properties of the manifold of quantum states are concerned.

\section{Uhlmann Connection and Uhlmann Curvature}\label{sec:UC}
In the closed path  $\rho_{\lambda(s)}$, initial and final
amplitudes are related by a unitary transformation, i.e.
$\w_{\lambda(s_{b})}=\w_{\lambda(s_{a})}V_{\gamma}$. If the path of
amplitudes $\w_{\lambda(s)}$ fullfills the Uhlmann condition,
$V_{\gamma}$ is a \emph{holonomy}, i.e. the non-Abelian
generalisation of Berry phase~\cite{Uhlmann1986}. The holonomy is
expressed as \begin{equation}\label{eq:holonomy}
V_{\gamma}=\mathcal{P}e^{i\oint_{\gamma} \bm{A}}\, , \end{equation} where
$\mathcal{P}$ is the path ordering operator and $\bm{A}$ is the
Uhlmann connection one-form. The Uhlmann connection can be derived
from the following ansatz~\cite{Uhlmann1989,Dittmann1999}
\begin{equation}\label{ansatz2}
d \w = i \w \bm{A} +  \bm{G} \w \,,
\end{equation}
which is the generalisation of~(\ref{ansatz1}) when the parallel
transport condition is lifted. By differentiating $\rho=\w
\w^{\dagger}$ and using the defining property of $\bm{G}$ (see
Eq.~(\ref{eq:G})), it follows that $\bm{A}$ is Hermitian and it is
implicitly defined by the equation 
\begin{equation} \bm{A} \w^{\dagger} \w +
\w^{\dagger} \w \bm{A} =i(d\w^{\dagger} \w -\w^{\dagger} d\w), 
\end{equation}
 with the auxiliary constraint that $\bk{\psi' | A | \psi'} = 0$, for
$\w\ket{\psi'}=0$. From Eq.~(\ref{ansatz2}), it can be  checked that
$\bm{A}$ obeys the expected transformation rule of non-Abelian gauge
potentials, 
\begin{equation} 
\bm{A}\to U^{\dagger}(s) \bm{A} U(s)+i
U^{\dagger}(s)dU(s), \quad \textrm{ under }\quad \w(s)\to \w(s)U(s),
\end{equation} 
and that $\bm{G}$ is gauge invariant. The analog of the Berry
curvature, the Uhlmann curvature two-form, is defined as 
\begin{equation}
 \bm{F}:=d\bm{A}+i \bm{A}\wedge \bm{A}= \frac{1}{2} \sum_{\mu\nu}
F_{\mu\nu}d\lambda_{\mu}\wedge d\lambda_{\nu}. 
\end{equation} 
Its components
$F_{\mu\nu}=\p_\mu A_\nu -\p_\nu
A_\mu+i\left[A_{\mu},A_{\nu}\right]$ can be understood in terms of
the Uhlmann holonomy per unit area associated to an infinitesimal
loop in the parameter space. Indeed, for an infinitesimal
parallelogram $\gamma_{\mu \nu}$, spanned by two independent
directions $\hat{e}_{\mu}\delta_{\mu}$ and
$\hat{e}_{\nu}\delta_{\nu}$ in the manifold, it reads 
\begin{equation}
  F_{\mu\nu}=\lim_{\delta \to 0} i \frac{1-V_{\gamma_{\mu,\nu}}}{\delta_{\mu}\delta_{\nu}},
\end{equation} 
where $\delta \to 0$ is a shorthand of
$(\delta_{\mu},\delta_{\nu})\to (0,0)$. While the Abelian Berry
curvature $\bm{F}^{B}$ is a gauge invariant (as one expect from the
electromagnetic field analogy), the Uhlmann curvature $\bm{F}$ is
only \emph{gauge covariant}, i.e. it transforms as
\begin{align} \bm{F}\to
U^{\dagger}(s) \bm{F} U(s),  \textrm{ under } \w(s)\to \w(s)U(s)\, .
\end{align} 
This is a direct consequence of the gauge covariance of any
non-Abelian holonomy~\cite{Bohm2003,Nakahara1990}, 
\begin{equation}
 V_{\gamma}=\mathcal{P}e^{i\oint_{\gamma} \bm{A}}\to
U^{\dagger}_{t}\mathcal{P}e^{i\oint_{\gamma} \bm{A}}\, U_{t} . 
\end{equation}

The standard approach to the definition of bona-fide observables
associated to non-Abelian gauge fields is to resort to the Wilson
loop $W_{\gamma}:=\Tr{V_{\gamma}}$, i.e. the trace of the holonomy
operator associated to an arbitrary loop. Thanks to the cyclic
property of the trace, the Wilson loop is evidently gauge invariant.
It would then be tempting to define a local gauge invariant
quantity, analogue of Uhlmann curvature, by considering the Wilson
loop per unit area of an infinetimal loop in the parameter space.
This would lead to the trace of curvature, which, however, in the
case of the Uhlmann holonomy is always trivial, i.e.
$\Tr{\bm{F}}=0$.\\\indent Alternatively, one may consider another
gauge invariant quantity, the Uhlmann geometric
phase~(\ref{eq:UhlmannGeoPh}), which in terms of the Uhlmann
holonomy reads,
\begin{equation}\label{Uphi}
\varphi^{U}[\gamma]:=\arg{\Tr{\left[\w_{\lambda(0)}^{\dagger}\w_{\lambda(T)}\right]}}=\arg{\Tr{\left[\w_{\lambda(0)}^{\dagger}\w_{\lambda(o)}V_{\gamma}\right]}}
.
\end{equation}
and evaluate the phase per unit area on an infinitesimal closed
loop, which reads
\[
 \mathcal{U}_{\mu\nu}:=\lim_{\delta \to 0} \frac{\varphi^{U}[\gamma_{\mu\nu}]}{\delta_{\mu}\delta_{\nu}} = \Tr{\left[\w_{\lambda(0)}^{\dagger}\w_{\lambda(0)}F_{\mu \nu}\right] }.
\]
This is by definition a gauge invariant. I will call it \emph{mean
Uhlmann curvature} (MUC), on account of the expression 
\begin{equation}
 \mathcal{U}_{\mu\nu}= \Tr{\left( \rho F_{\mu \nu}\right)}=\langle
F_{\mu\nu}\rangle \end{equation} that $\mathcal{U}_{\mu\nu}$ takes in the
special gauge $\w_{0}=\sqrt{\rho(0)}$.

By taking the external derivative of the expression~(\ref{ansatz2})
and by using the property $d^{2}=0$, it can be shown that
\begin{align*}
0=d^{2}\w &= i \w d\bm{A} +id \w \wedge \bm A + d \bm G \, \w - \bm G \wedge d \w \nonumber\\
&=i \w d\bm{A} +i\big(i\w \bm{A} +  \bm{G} \w\big) \wedge \bm A + d \bm G \, \w - \bm G \wedge  \big(i \w \bm{A} +  \bm{G} \w\big)\nonumber\\
&=i \w \big( d\bm{A} +i \bm{A}\wedge \bm{A}) + \big(d \bm{G} -
\bm{G}\wedge \bm{G}\big)  \w \nonumber
\end{align*}
which leads to~\cite{Uhlmann1991}
\begin{align*}
\w \bm{F} &= i\left(d\bm{G}-\bm{G}\wedge \bm{G}\right) \w,\\
\bm{F} \w^{\dagger} &= -i \w^{\dagger}\left(d\bm{G}+\bm{G}\wedge
\bm{G}\right).
\end{align*}
Multiplying the above expressions by $\w^{\dagger}$ and $\w$,
respectively, and taking the trace yields
\begin{equation}\label{MUC}
\bm{\mathcal{U}} = \Tr(\w \bm{F}\w^{\dagger}) = -i\Tr(\rho
\bm{G}\wedge \bm{G}),
\end{equation}
where $\,\,\bm{\mathcal{U}}:=1/2
\sum_{\mu\nu}\mathcal{U}_{\mu\nu}d\lambda_{\mu}\wedge
d\lambda_{\nu}$ is a real-valued two-form, whose components are
\begin{equation}\label{eq:MUC2} 
\mathcal{U}_{\mu\nu}=-i\Tr(\rho
[G_{\mu},G_{\nu}])\,. 
\end{equation} 
Notice the striking similarity with the
expression~(\ref{eq:Buresmetric2}) with which the mean Uhlmann
curvature shares the same origin and mathematical structure.
Similarly to~(\ref{eq:BuresTensor}), it is easy to derive an
expression for the MUC in the eigenbasis of the density matrix by
making use of the expression~(\ref{eq:GEigen}), i.e.
\begin{equation}\label{eq:UEigen}
\mathcal{U}_{\mu\nu}=-i\sum_{\substack{p_{j}>0\\p_{k}>0}}(p_{j}-p_{k})\frac{\bra{j}\p_{\mu}\rho\ket{k}\bra{k}\p
_{\nu}\rho\ket{j}}{(p_{j}+p_{k})^{2}}\, . 
\end{equation}
From the common mathematical structure of the Bures metric
$g_{\mu\nu}$ in~(\ref{eq:Buresmetric2}) and the mean Uhlmann
curvature~(\ref{eq:MUC2}), it is quite tempting to define the
following positive (semi)-definite Hermitian tensor
\begin{equation}\label{eq:GQGT} Q_{\mu\nu}:= \Tr(\rho G_{\mu}G_{\nu}) =
g_{\mu\nu} + \frac{i}{2}\mathcal{U}_{\mu\nu}. \end{equation} This is clearly a
mixed state generalisation of the quantum geometric
tensor~(\ref{eq:qgt}) defined in section~\ref{sec:FS}. In the
following, we will indistinctly refer to both Eq.~(\ref{eq:qgt}) and
to its mixed state generalisation (see Eq.~(\ref{eq:GQGT})) as quantum
geometric tensor.

\section{Multi-parameter quantum state estimation}\label{sec:Fisher}
This section aims at providing a different physical perspective to
the quantities introduced in the previous sections. It turns out
that objects such as the Bures metric and the mean Uhlmann curvature
are intimately related to a pivotal quantity in quantum parameter
estimation problems, namely the quantum Fisher information.

Estimation theory is the discipline that studies the accuracy by which a given set of physical parameters can be evaluated. When the parameters to be estimated 
belongs to an underlying quantum physical system one falls in the realm of quantum estimation theory, or quantum metrology~\cite{Helstrom1976}. Quantum parameter 
estimation finds applications in a wide variety of fields, from fundamental physics~\cite{Udem2002,Katori2011,Giovannetti2004,Aspachs2010,Ahmadi2014}, to gravitational 
wave interferometry~\cite{Schnabel2010,Aasi2013}, thermometry~\cite{Correa2015,DePasquale2016}, spectroscopy~\cite{Schmitt2017,Boss2017}, 
imaging~\cite{Tsang2016,Nair2016,Lupo2016}, to name a few. Exploiting remarkable features of quantum systems as probes may give an edge over the accuracy of 
classical parameter estimation. Exploring this possibility plays a pivotal role in the current development of quantum technology~\cite{Caves1981,Huelga1997,Giovannetti2006,Paris2009,Giovannetti2011,Toth2014,Szczykulska2016,Pezze2016,Nichols2018,Braun2018}. In multi-parameter 
quantum estimation protocols, several variables are simultaneously evaluated, in a way which may outperform individual estimation strategies with equivalent 
resources~\cite{Humphreys2013,Baumgratz2016}, thereby motivating the use of such protocols in a variety of diverse contexts~\cite{Humphreys2013,Baumgratz2016,Pezze2017,Apellaniz2018}.

The use of peculiar quantum many-body states as probes in quantum metrology can enhance the accuracy in parameter estimation~\cite{Zanardi2008,Braun2018}. Conversely, 
one may think of using quantum metrological tools in the study and characterisation of many-body systems. Noteworthy instances of many-body quantum systems are those 
experiencing quantum phase transitions. Indeed, quantum parameter estimation, with its intimate relation with geometric information, provides a novel and promising approach to 
investigate equilibrium~\cite{Carollo2005,Zhu2006,Hamma2006,Zanardi2006,CamposVenuti2007,CamposVenuti2008,Zanardi2007,Zanardi2007a,Garnerone2009a,Rezakhani2010,Bascone2019,Bascone2019a} 
and out-of-equilibrium~\cite{Magazzu2015,Magazzu2016,Guarcello2015,Spagnolo2015,Spagnolo2017,Spagnolo2018,Valenti2018,Consiglio2016,Spagnolo2019} quantum critical phenomena~\cite{Banchi2014,Marzolino2014,Kolodrubetz2017,Carollo2018,Carollo2018a,Marzolino2017}.

The quantum Fisher information matrix $J(\bm \lambda)$  defines a
figure of merit of the estimation precision of parameters labelling
a quantum state, known as the quantum Cram\'er-Rao (CR) bound
~\cite{Helstrom1976,Holevo2011,Paris2009}. Given a set of locally
unbiased estimators $\{\hat{\bm\lambda}\}$ of the parameters
$\bm\lambda\in\mathcal{M}$, the covariance matrix
$\Cov_{\bm\lambda}[\hat{\bm{\lambda}}] _{\mu\nu}=\langle
(\hat{\lambda}_{\mu}-\lambda_{\mu})(\hat{\lambda}_{\nu}-\lambda_{\nu})\rangle$
is lower bounded, in a matrix sense, as follows
\begin{equation}\label{eq:CRB}
\Cov_{\bm\lambda}[\hat{\bm{\lambda}}] \ge J(\bm\lambda)^{-1}.
\end{equation}

It turns out that such a matrix is mathematically equivalent, except
for pathological case~\cite{Safranek2017}, to the Bures metric tensor
$g$, or precisely 
\begin{equation} 
J_{\mu\nu}(\bm\lambda)=4g_{\mu\nu}. 
\end{equation}

For single parameter estimation, the quantum Cram\'er-Rao
bound~(\ref{eq:CRB}) can always be saturated by a suitable optimal
POVM. However, in a multi-parameter scenario this is not always the
case, the above inequality cannot always be attained. This is due to
the non-commutativity of measurements associated to independent
parameters. It turns out that, within a relatively general setting,
known as \emph{quantum local asymptotyc
normality}~\cite{Hayashi2008,Kahn2009,Gill2013,Yamagata2013}, the
multi-parameter quantum Cram\'er-Rao bound~(\ref{eq:CRB}) is
attainable iff~\cite{Ragy2016} \begin{equation} \mathcal{U}_{\mu\nu}=0 \qquad
\forall \lambda_{\mu}, \lambda_{\nu}~. \end{equation} In this sense,
$\mathcal{U}_{\mu\nu}$ marks the \emph{incompatibility} between
$\lambda_{\mu}$ and $\lambda_{\nu}$, and such an incompatibility arises
from the inherent quantum nature of the underlying physical system.

\indent From a metrological point of view, $\mathcal{U}_{\mu\nu}$ marks the \emph{incompatibility} between $\lambda_{\mu}$ and $\lambda_{\nu}$, where such an incompatibility arises from the inherent quantum nature of the underlying physical system.\\
\indent In particular, one can show that for an \emph{N-parameter estimation model}, the deviation of the \emph{attainable} multi-parameter 
bound from the Cram\'er-Rao bound can be estimated by the quantity $ R:=|| 2i\, J^{-1}\mathcal{U}||_\infty$, where
\begin{align}\label{UBound}
0\le R\le 1.
\end{align}
Indeed, $R$ provides a figure of merit which measures the \emph{amount of incompatibility} within a parameter estimation model. The lower limit condition, $R=0$, is equivalent 
to the compatibility condition, i.e. Eq.~\eqref{MUC}. On the other hand, when the upper bound of Eq.~\eqref{UBound} is saturated, i.e. $R=1$, it maximizes the discrepancy between 
the CR bound, that could be attained in an analogous classical multi-parameter estimation problem, and the actual multi-parameter quantum CR bound. In this sense, this bound 
marks the \emph{condition of maximal incompatibility}. When this condition is met, the indeterminacy arising from the quantum nature of the estimation problem reaches the order 
of $||J^{-1}||_\infty$, i.e. the same order of magnitude of the Cram\'er-Rao bound~(\ref{eq:CRB}). In other words, this implies that the indeterminacy due to the quantum incompatibility 
arises at an order of magnitude which cannot be neglected.

This is particularly relevant, considering that the scope of optimal
schemes is minimising the parameter estimation error. This can only
be done by designing strategies which strive for the higher possible
rate of growth of $J(n)$ with the number $n$ of available resources.
When the condition~(\ref{eq:CRB}) of maximal incompatibility holds,
it implies that the quantum indeterminacy in the parameter
estimation problem remains relevant even in the asymptotic limit
$n\to\infty$.

\subsection{Formulation of the problem}
It is often the case that a physical variable of interest is not
directly accessible, either for experimental limitations or due to
fundamental principles. When this happens one could resort to an
indirect approach, inferring the value of the variable after
measurements on a given probe. This is essentially a parameter
estimation problem whose solution may be found using methods from
classical estimation theory \cite{Cramer1946} or, when quantum
systems are involved, from its quantum counterpart
\cite{Helstrom1976}.

The solution of a parameter estimation problem amounts to find an estimator, {\em i.e} a mapping $\hat{\bm{\lambda}}=\hat{\bm{\lambda}} (x_1,x_2,...)$, from the set 
$\chi$ of measurement outcomes into the space of parameters $\bm\lambda \in\mathcal{M}$.  Optimal unbiased estimators in classical estimation theory are those 
saturating the Cram\'er-Rao (CR) inequality, 
\begin{equation}\label{eq:CCRB}
\Cov_{\bm\lambda}[\hat{\bm{\lambda}}] \geq J^{c} (\bm \lambda)^{-1}, 
\end{equation} 
which poses a lower bound 
on the mean square error $\Cov_{\lambda} [\hat{\bm{\lambda}}]_{\mu\nu} =
E_{\lambda} [(\hat{\lambda} - \lambda)_\mu(\hat{\lambda}-\lambda)_\nu]$ in terms of the Fisher information (FI) 
\begin{equation} J^{c}_{\mu\nu}(\bm\lambda) = \int_\chi 
d\hat{\bm{\lambda}}(x)\, p(\hat{\bm{\lambda}}|\lambda) \partial_\mu \log p(\hat{\bm{\lambda}}|\lambda) \partial_\nu \log p(\hat{\bm{\lambda}}|\lambda)\:. 
\end{equation}
 The expression~(\ref{eq:CCRB}) should be understood as a matrix inequality. In general, one writes
\[
\tr(W\Cov_{\bm\lambda}[\hat{\bm{\lambda}}] )\ge\tr(W J^{c}(\bm\lambda)^{-1}),
\]
where $W$ is a given positive definite cost matrix, which allows the uncertainty cost of different parameters to be weighted.

In the classical estimation problem, both in the single parameter case, and in the multi-parameter one, the CR bound can be attained in the asymptotic limit of an infinite 
number of experiment repetitions using the maximum likelihood estimator \cite{Kay1993}. However, an interesting difference between single and multi-parameter 
metrology arises due to correlations between the variables. Indeed, it may well happen that the off-diagonal elements of the Fisher information matrix are
non-vanishing. Hence, there are statistical correlations between the parameter estimators. In a protocol in which all variables but $\lambda_{\mu}$ are precisely 
known, the single-parameter CR bound implies that the best attainable accuracy in estimating $\lambda_{\mu}$ is given by $\text{Var}(\hat{\bm\lambda}) \geq 1/J^{c}_{\mu\mu}$. 
However, in a scenario in which all parameters are simultaneously estimated, one finds that the ultimate precision is lower bounded by $\text{Var}(\hat{\bm\lambda}) \geq 
(J^{c}(\bm\lambda)^{-1})_{\mu\mu}$. A straightforward calculation shows that, for positive-definite matrices, $(J^{c}(\bm\lambda)^{-1})_{\mu\mu} \geq 1/J^{c}(\bm\lambda)_{\mu\mu}$, 
where the inequality is saturated only for vanishing off-diagonal elements. In the limit of a large number of experiment repetitions the CR bound is attainable. This means 
that the equivalence between the simultaneous and the individual protocols in the asymptotic limit holds only if the Fisher information is a diagonal matrix, i.e. if the estimators 
are not correlated~\cite{Cox1987}.

Obviously, any given real positive definite matrix can be transformed via an orthogonal rotation into a diagonal matrix. This clearly implies that there is always a combination 
of the parameters for which the Fisher information matrix is diagonal. However, this choice should be contrasted with the physical opportunity of performing such a rotation, as the 
choice of the parameters we are interested in may arise as a result of physical considerations and in this sense determines a preference in a specific basis.

The underlying quantities used in the derivation of classical Fisher information are parameter-dependent probability-distributions of the data, whereas the objects involved in 
the quantum estimation problems are density operators $\rho(\bm\lambda)$  labelled by ${\bm\lambda\cal \in M}$. Hence, a further difficulty of quantum estimation protocols 
is devising the optimal measurement strategy which gathers from the density matrix the greatest amount of information on the labelling parameters. For single parameter estimation, 
the solution is quite straightforward. If one maximises the classical Fisher information over all possible quantum measurements, the result is the so-called quantum Fisher information 
(QFI). The key object involved in the calculation of the QFI is the so-called \emph{symmetric logarithmic derivative} (SLD), $L_{\mu}$, a Hermitian operator which is implicitly defined as 
the solution of the equation 
\begin{equation}
 \partial_\mu \rho(\bm\lambda )
= \frac12 \left\{ \rho (\bm\lambda) L_\mu (\bm\lambda) + L_\mu
(\bm\lambda) \rho (\bm\lambda) \right\}\:. \end{equation} 
The above equation,
apart from a factor $1/2$, is identical to the defining
Eq.~(\ref{eq:G}) of the parallel transport generator $G_{\mu}$.
However, a relatively benign difference between $G_{\mu}$ and
$L_{\mu}$ arises from the auxiliary condition~(\ref{eq:AuxG}). This
may cause a sizeable discrepancy between their behaviours in some
pathological cases, which may occur around points of the state
manifold, where $\rho(\bm\lambda)$ undergoes a change of
rank~\cite{Safranek2017}.

The QFI can be calculated using the formula: 
\begin{equation}
 J_{\mu\mu}(\bm{\lambda})=\frac{1}{2}\Tr\rho(\bm{\lambda})\{L_{\mu}(\bm{\lambda}),L_{\mu}(\bm{\lambda})\}\,.
\end{equation} One can always choose the projective measurement in the
eigenbasis of the SLD which yields FI equal to the QFI. Hence, the
QFI determines the ultimate achievable estimation precision of the
parameter on density matrices $\rho(\bm\lambda)$ in the asymptotic
limit of an infinite number of experiment repetitions. In a
multiparameter scenario, a direct generalization of single parameter
CR bound leads to the multiparameter QFI CR
bound~\cite{Helstrom1976,Holevo2011,Paris2009}, that reads
\begin{equation}\label{eq:SCRB}
\tr(W \Cov(\hat{\bm{\lambda}}))\ge\tr(W J^{-1}),
\end{equation}
 where
\begin{equation}
  J_{\mu\nu}=\frac{1}{2}\Tr\rho\{L_{\mu},L_{\nu}\},
\end{equation}
 is the quantum Fisher information matrix (QFIM) and $W$ is the cost
matrix.

\subsection{Multi-parameter Incompatibility: a Measure of Quantumness}\label{sec:quant} Unlike the single parameter case, in the multi-parameter scenario the quantum CR bound cannot 
always be saturated. Intuitively, this is due to the incompatibility of the optimal measurements for different parameters. A sufficient condition for the saturation is indeed $[L_{\mu},L_{\nu}]=0$, 
which is however not a necessary condition. Within the comprehensive framework of quantum local asymptotic normality (QLAN)~\cite{Hayashi2008,Kahn2009,Gill2013,Yamagata2013}, 
a necessary and sufficient condition for the saturation of the multi-parameter CRB is given by $\mathcal{U}_{\mu\nu}=0$ for all $\mu$ and $\nu$~\cite{Ragy2016}. \\
\indent Here, we show explicitly  that the ratio between $\mathcal{U}_{\mu\nu}$ and $J_{\mu\nu}$ provides a figure of merit for the discrepancy between an attainable multi-parameter 
bound and the single parameter CRB quantified by $J^{-1}$. We will confine ourself to the broad framework of QLAN, in which the \emph{attainable} multi-parameter bound is given by 
the so called Holevo Cram\'er-Rao bound (HCRB)~\cite{Helstrom1976,Holevo2011,Paris2009}. For a $N$-parameter model, the HCRB can be expressed as~\cite{Hayashi2008}
\begin{equation}
\tr(W \Cov(\hat{\bm{\lambda}}))\ge C_{H}(W),
\end{equation}
where
\begin{equation}
C_{H}(W):=\min_{\{X_{\mu}\}}\{\tr (W \Re Z) +||(W \Im Z)||_{1}\}.
\end{equation}
The $N\times N$ Hermitian matrix $Z$ is defined as 
\begin{align}\label{defZ} Z_{\mu\nu}:=\Tr (\rho X_{\mu}X_{\nu})
\end{align}
where $\{X_{\mu}\}$ is an array of $N$ Hermitian operators on $\HH$ satisfying the unbiasedness conditions $\Tr(\rho X_{\mu})=0$, $\forall \mu$ and 
$\Tr (X_{\mu} \partial_{\nu}\rho)=\frac{1}{2}\Tr \rho \{X_{\mu}, L_{\nu}\}=\delta_{\mu\nu}$, $\forall \mu,\nu$, and $||B||_{1}$ denotes the sum of all singular 
values of $B$. 
If one chooses for $\{X_{\mu}\}$ the array of operators $\tilde{X}_{\mu}:=\sum_{\nu} [J^{-1}]_{\mu\nu} L_{\nu}$, it yields
\begin{equation}\label{Z}
Z\to\tilde{Z}:=J^{-1} I  J^{-1} =  J^{-1} - i 2 J^{-1} \mathcal{U} J^{-1},
\end{equation}
where $I$ is the matrix of elements $I_{\mu\nu}:=\Tr{(\rho L_{\mu} L_{\nu})}$, and $\mathcal{U}$ is the skew-symmetric real matrix 
$\mathcal{U}_{\mu\nu}=\frac{i}{4}\Tr(\rho[L_{\mu},L_{\nu}])$. If one indicates by $\mathcal{D}(W):=C_{H}(W) - \tr{W J^{-1}}$  the discrepancy between the 
attainable multi-parameter HCRB and the CRB, one can write the following bounds
\begin{equation}\label{R}
0 \le \mathcal{D}(W)\le 2 ||W\, J^{-1} \mathcal{U} J^{-1}||_{1}\le \tr{(W J^{-1})} R,
\end{equation}
where $R=||2i \mathcal{U}J^{-1}||_\infty\, $, and the first inequality is saturated iff $\mathcal{U}=0$~\cite{Ragy2016}.
\\
\indent One can show that~\cite{Carollo2019}
\begin{align}\label{R1}
0 \le R \le 1.
\end{align}
When the upper bound~\eqref{R1} is saturated, i.e. condition $R=1$ is met, it implies that
\begin{equation}\label{AbsIneq}
\mathcal{D}(W) \simeq \tr(W J^{-1}),
\end{equation}
which means that the discrepancy $\mathcal{D}(W)$ reaches the same order of magnitude of the CR bound itself. This limit marks the \emph{condition of 
maximal incompatibility} for the two-parameter estimation problem, arising from the quantum nature of the underlying system. In the opposite limit $R=0$, 
the parameter model is \emph{compatible}, and the discrepancy between the quantum CR bound, and its classical counterpart vanishes. Therefore, $R$ 
provides a figure of merit which quantifies the quantum contribution to the indeterminacy of multi-parameter estimations.\\

\emph{Proof of Eq.~\eqref{R1}.}
The lower bound comes straightforwardly from Eq.~\eqref{R}. For the upper bound, notice that $Z=\{Z_{\mu\nu}\}$ in Eq.~\eqref{defZ} is a positive semi-definite 
matrix, since $\forall \vec{a}=\{a_\mu\}_{\mu=1}^N\in\mathbb{C}^N$, $\vec{a}^\dag\cdot Z\cdot \vec{a} = \Tr(\rho A^\dag A)\ge 0$, with $A:=\sum_\mu a_\mu X_\mu$. 
Then, from Eq.~\eqref{Z}, 
\begin{equation}\label{ZR}
J^{1/2}\tilde{Z}J^{1/2}:= \one - i 2 J^{-1/2} \mathcal{U} J^{-1/2}\ge0.
\end{equation}
Since $i 2 J^{-1/2} \mathcal{U} J^{-1/2}$ is a skew-symmetric Hermitian matrix, its eigenvalues are either zero or real numbers that occur in $\pm$ pairs. Then, 
from Eq.~\eqref{ZR} we deduce that these eigenvalues lie within the interval $\{-1,1\}$. Moreover, $J^{-1} \mathcal{U}$ is a diagonalisable matrix with the same 
eigenvalues of $J^{-1/2} \mathcal{U} J^{-1/2}$. Indeed, if $J^{-1/2} \mathcal{U} J^{-1/2}=U^\dag D U$, with $D$ diagonal, then $J^{-1} \mathcal{U}=S^{-1}DS$, 
where $S=U J^{1/2}$. Hence, $R=||i 2 J^{-1}\mathcal{U}||_\infty=||i 2 J^{-1/2} \mathcal{U} J^{-1/2}||_\infty\le 1$ \QED
\\
For the special case of a two-parameter model, in the eigenbasis of $J$, with eigenvalues $j_{1}$ and $j_{2}$, it holds
\begin{equation}
2iJ^{-1}\mathcal{U}=\left(\begin{array}{cc}j_{1}^{-1}&0 \\0 &
j_{2}^{-1}\end{array}\right)\left(\begin{array}{cc} 0
&\mathcal{U}_{12} \\-\mathcal{U}_{1 2} & 0\end{array}\right)
 = \left(\begin{array}{cc} 0
&2i\frac{\mathcal{U}_{1 2}}{ j_1} \\-2i\frac{\mathcal{U}_{1 2}}{j_2} & 0\end{array}\right).
\end{equation}
It follows that
\begin{equation}\label{AbsIneq}
R=||2i J^{-1} \mathcal{U}||_{\infty}= \sqrt{\frac{\Det \,2 \mathcal{U}}{\Det{J}}}.
\end{equation}
Hence, $\sqrt{\Det \,2 \mathcal{U}/{\Det{J}}}$ provides a figure of merit which measures the \emph{amount of incompatibility} between two independent parameters 
in a quantum \emph{two-parameter} model. \\
\indent For self-adjoint operators $B_{1},\dots,B_{N}$, the Schrodinger-Robertson's uncertainty principle is the inequality~\cite{Robertson1929}
\begin{equation}
\Det\left[\frac{1}{2} \Tr \rho \{B_{\mu},B_{\nu}\}\right]_{\mu,\nu=1}^{N}\ge\Det\left[-\frac{i}{2} \Tr\rho[B_{\mu},B_{\nu}]\right]_{\mu,\nu=1}^{N},
\end{equation}
which, applied to the SLD $L_{\mu}$'s, yields
\begin{equation}\label{DetIneq}
\Det J \ge \Det 2\, \mathcal{U}.
\end{equation}
For $N=2$, the inequality~(\ref{DetIneq}) \label{AbsIneq} is equivalent to the upper-bound of Eq.~\eqref{R1}, and if saturated, it implies the 
\emph{condition of maximal incompatibility} for the two-parameter estimation problem.\\

Another interesting inequality relates the eigenvalues of $J$ (and
hence of $g$) with those of $\mathcal{U}$. The QGT\footnote{We are
deliberately assuming $L_{\mu}=2 G_{\mu}$, thus neglecting possible
discrepancy arising due to variations in the rank of
$\rho$~\cite{Safranek2017}.} \begin{equation} Q_{\mu\nu}:=\Tr\rho G_{\mu}G_{\nu}
=\frac{1}{4} \Tr\rho L_{\mu}L_{\nu}= \frac14
J_{\mu\nu}+\frac{i}{2}\mathcal{U}_{\mu\nu} \end{equation} is a positive
(semi)-definite Hermitian matrix. Hence, by definition $J\ge - 2 i\,
\mathcal{U}$, in a matrix sense. It follows that~\cite{Horn2013}
\begin{equation}
    || J ||_{\infty}\ge 2 || i\, \mathcal{U}||_{\infty}.
\end{equation}

 \section{Thermal States and Classical Phase Transitions.}
In this section, we would like to give a glance at the extension of
the geometric information approach to finite temperature. This
approach is indeed directly linked to the approach developed in
classical phase transitions~\cite{Ruppeiner1995,Brody1995}.
 The natural generalisation of the geometric information approach to finite temperature is done by replacing the Fubini-Study metric with its mixed state generalisation, the Bures metric.
 In temperature driven phase transitions, the density matrices $\rho_{\beta_{1}}$ and $\rho_{\beta_{2}}$ at different temperatures $\beta_{i}=(\kappa_{B}T_{i})^{-1}$ generally commute\footnote{Here one considers Hamiltonians which do not explicitly depend on temperature. For example, effective mean field Hamiltonians with temperature dependent parameters may result in non-commuting density matrices $[\rho_{\beta_{1}},\rho_{\beta_{2}}]\neq 0$. The BCS Hamiltonians are one such case~\cite{Mera2017}}, and the
 Uhlmann fidelity ${\cal F}(\rho_{\beta_{1}},\rho_{\beta_{2}}):= \Tr[\rho_{\beta_{2}}^{1/2}\rho_{\beta_{1}} \rho_{\beta_{2}}^{1/2}]^{1/2}$ reduces to the Bhattacharyya coefficient
${\cal F}(\rho_0,\rho_1)=\sum_k \sqrt{p_k^{1} p_k^{2}}$, where the
$p_k^i$ are the eigenvalues of $\rho_{\beta_{i}}$. Correspondingly,
when two states differing just by two infinitesimally closed
temperatures $\beta_{1}=\beta$ and $\beta_{2}=\beta+d\beta$ are
considered, the Bures metric collapses into the Fisher-Rao metric, 
\begin{equation}
 dl_{B}=\frac{1}{4}\sum_{p_{k}>0}\frac{(\p_{\beta}
p_{k})^{2}}{p_{k}}d\beta^{2}=\frac{1}{4}\left[ \la H^{2}\ra -\la
H\ra^{2}\right]d\beta^{2}\,.\label{eq:FRmetric} 
\end{equation} 
In particular,
when $\rho_\beta=Z^{-1}(\beta) \exp(-\beta H)=Z^{-1}(\beta)
\sum_{k}e^{-\beta E_{k}}|k\rangle\langle k|$, with $Z(\beta):=\Tr
\exp(-\beta H)$,  where $E_{k}$ and $|k\rangle$ are the eigenvalues
and eigenvectors of the Hamiltonian, one easily sees that 
\begin{equation}
 \p_{\beta}p_{k}=\p_{\beta}\left(\frac{e^{-\beta E_{k}}}{Z}\right)=-
p_{k}\Big[E_{k}-(\sum_{j}E_{j}p_{j})\Big]=p_{k}\left[\langle
H\rangle_{\beta}-E_{k}\right]\,, \end{equation} therefore~(\ref{eq:FRmetric})
can be written
as~\cite{Ruppeiner1995,Brody1995,You2007,Zanardi2007}, 
\begin{equation}
 dl_{B}=\frac{1}{4}\sum_{k}p_{k}\left[\langle
H\rangle_{\beta}-E_{k}\right]^{2}d\beta^{2}= \frac{1}{4}\left[ \la
H^{2}\ra_{\beta}-\la
H\ra_{\beta}^{2}\right]d\beta^{2}=\frac{d\beta^{2}}{4}T^{2}c_{V}(\beta)\,.\label{eq:FRmetric1}
\end{equation}
 This is quite a remarkable relation which provides a connection between a geometric information measure of distinguishability and a macroscopic thermodynamic quantity $c_V.$\\
More importantly, it also provides a direct connection with
classical critical phenomena. Thermal phase transitions are signalled
by non-analytical behaviour of the specific heat, which may be
picked up as singularity of the Bures/Fisher-Rao metric. This simple
fact shows that information geometry at large provides a
comprehensive framework under which both quantum and classical
critical phenomena can be readily investigated.

\section{Susceptibility and mean Uhlmann curvature}
The MUC is a geometrical quantity, whose definition relies on a rather formal definition of holonomies of density matrices. In spite of its abstract formalism, the MUC has interesting connections to a physically relevant object which is directly observable in experiments, the susceptibility. In a thermal equilibrium setting, one can rigorously relate the MUC to the dissipative part of the dynamical susceptibility~\cite{Leonforte2019,Leonforte2019a}. Indeed, by using linear response theory, one can consider the most general scenario of a system with a Hamiltonian $\HH_0$, perturbed as follows
\begin{equation}
\label{eq:pert} 
\HH = \HH_0 + \sum_\mu \hat{O}_\mu \lambda_{\mu},
\end{equation}
where $\{\hat{O}_{\mu}\}$ is a set of observables of the system, and
$\{\lambda_{\mu}\}$ is the corresponding set of perturbation
parameters. We are considering the system in thermal equilibrium, i.e. $\rho = \frac{e^{-\beta \HH}}{\mathcal{Z}} $, where $\mathcal{Z} = \Tr [ e^{-\beta \HH} ]$ is the partition function. 
The dissipative part of the dynamical susceptibility, with respect to $\hat{O}_\mu$, is defined as:
\begin{equation}
\chi''_{\mu \nu} (t) = \frac{1}{2 \hbar} \langle [ \hat{O}_{\mu} (t), \hat{O}_{\nu} ] \rangle_0,
\end{equation}
One can show that the Fourier transform of the dissipative part of the dynamical susceptibility has the following expression in the Lehmann representation
\begin{equation} 
\chi_{\mu \nu}''(\omega,\beta)= \frac{\pi}{\hbar}\sum_{i j} (\hat{O}_{\mu})_{i j} (\hat{O}_{\nu})_{j i} (p_i - p_j) \,\,\delta \left(\omega + \frac{E_i - E_j}{\hbar}\right),
\end{equation}
where $p_i$'s are the eigenvalues of the density matrix in the Boltzmann-Gibbs ensemble, i.e. $p_i=e^{-\beta E_i}/Z$, and $E_i$'s are the corresponding Hamiltonian eigenvalues. For thermal states, one can exploit the identity 
\begin{align}
\frac{p_i- p_j}{p_i + p_j}= \int_{-\infty}^{+\infty} d\omega \tanh \left( \frac{\hbar \omega \beta}{2} \right) \,\,\delta \left(\omega + \frac{E_i - E_j}{\hbar}\right),
\end{align}
which leads to the following relation between the $\chi_{\mu \nu}''(\omega,\beta)$ and the MUC,
\begin{equation} 
\label{ulmsuc1} 
\mathcal{U}_{\mu \nu} = \frac{i}{ \hbar \pi}  \int_{-\infty}^{+\infty} \frac{d\omega}{\omega^2} \tanh^2 \left( \frac{\hbar \omega \beta}{2} \right)\chi_{\mu \nu}''(\omega,\beta) ,
\end{equation}
where the set of perturbations $\{\lambda_\mu\}$ in~(\ref{eq:pert})
plays the role of the parameters in the derivation of
$\mathcal{U}_{\mu \nu}$. By means of the fluctuation-dissipation
theorem~\cite{Altland2006}, one can further derive an expression
for Eq.~\eqref{ulmsuc1} in terms of the dynamical structure factor,
$S_{\mu \nu}(\omega,\beta) = \int_{-\infty}^{+\infty} d t e^{i \omega t}
S_{\mu \nu} (t)$ (i.e. the Fourier transform of the correlation
matrix $S_{\mu \nu}(t) = \langle \hat{O}_\mu (t) \hat{O}_\nu (0)\rangle$),
which reads
\begin{equation} \label{MUCCorr1} \mathcal{U}_{\mu \nu} =
\frac{i}{2 \pi \hbar }  \int_{-\infty}^{+\infty} \frac{d\omega}{\omega^2}
\tanh^2 \left( \frac{\hbar \omega \beta}{2} \right) ( S_{\mu \nu}(\omega,\beta)
- S_{\nu \mu}(-\omega,\beta) ).
\end{equation}
Eqns.~(\ref{ulmsuc1}) and~(\ref{MUCCorr1}) provide a means to explore experimentally the geometrical properties of physical systems via the dissipative part of the dynamical susceptibility, and the imaginary part of the (off-diagonal)-dynamical structure factor.  The dynamical susceptibility is a well studied quantity which has been observed experimentally in a wide variety of materials~\cite{Yim2006,Coldea2010,Lake2010,Han2012,Dai2015,Halg2015}. Using such an experimentally accessible quantity provides a robust method to explore the geometric features of criticality in strongly correlated materials.

\section{Dissipative Non-Equilibrium Phase Transitions}
In contrast to equilibrium critical phenomena, less is known in case
of non-equilibrium phase transitions. Indeed, a generalised treatment
is not known, lacking an analog to the equilibrium free energy. Thus
the rich and often surprising variety of non-equilibrium phase
transitions  observed in
physical~\cite{Woodcock1985,Chrzan1992,Schweigert1998,Blythe2002,Whitelam2014,Zhang2015a},
chemical,
biological~\cite{Egelhaaf1999,Marenduzzo2001,Barrett-Freeman2008,Woo2011,Mak2016,Battle2016},
as well as socioeconomic
systems~\cite{Llas2003,Baronchelli2007,Scheffer2012}, has to be
studied for each system individually.

The paradigm of universality, legitimised by its astounding success
in equilibrium phase transitions, does not find an equivalently
comprehensive framework within non-equilibrium phenomena. This
results in a large variety of universality classes without general
tools for their characterisation \cite{Odor2004,Lubeck2004}. For
instance, algebraically decaying correlation functions are not
peculiar of critical phenomena \cite{Prosen2010}. Also the specral
gap of the Liouvillian, an open system generalisation of the
Hamiltonian gap, may vanish in the thermodynamic limit in the whole
phase diagram, with critical points resulting only in a faster
convergence~\cite{Prosen2008,Znidaric2015}.


The specific domain of quantum many-body physics in recent years
has witnessed a growing interest in non-equilibrium phenomena. The
reason can be credited to at least two main causes. On the one hand,
the unprecedented level of accuracy reached in nowadays experimental
techniques provides many mature platforms for the investigation and
manipulation of many-body interacting systems. Indeed, several set of
available tools have enabled the development of Liouvillian
engineering, which in addition to coherent Hamiltonian dynamics also
includes controlled dissipation in many-body quantum systems.
Examples of suitable experimental platforms for the implementation
and simulation of such an open many-body framework range from
ultra-cold atoms in optical lattices~\cite{Diehl2008,Bloch2008}, to
ion traps~\cite{Barreiro2011,Schindler2012}, to cavity
microarrays~\cite{Hartmann2006,Greentree2006,Angelakis2007,Underwood2012,Houck2012,Raftery2014,Fitzpatrick2017}
and Rydberg atoms~\cite{Weimer2010,Dudin2012}. On the other hand,
the study of non-equilibrium quantum phenomena can arguably be
related to important open challenges in many-body physics, ranging
from high temperature super-conductivity to quantum computation in
condensed matter setups.

There are two major framework within which non-equilibrium quantum
many-body-physics are generally investigated. One of the most
popular approach \cite{Eisert2015} considers a large, ideally
infinite, system which is initially kept in an equilibrium state. A
perturbation may then be applied to the system either via a sudden
quench of the system Hamiltonian, or by a periodic application of an
external field, or still by the coupling to a suitably structured
reservoir. The resultant system time evolution is then observed, and
the outset - or the lack thereof - of stationary or metastable
long-time behaviour provides several quantitative and qualitative
information on the non-equilibrium properties of the many-body
system.

Another, more direct approach, which we will consider here, consists
instead in coupling a finite or infinite system to several external
reservoirs which may be described {\em effectively} in terms of a
master equation~\cite{Alicki2007,Breuer2007}. Under suitable
conditions, the open system dynamics reaches a possibly unique
{\em non-equilibrium steady state}.  The dynamical and static
properties of the NESSes to which the system eventually relaxes are
the central object of this second course of investigations.

There are several techniques involving many levels of assumptions
and approximations in deriving the effective system's dynamics of
open systems interacting with a reservoir~
\cite{Alicki2007,Breuer2007}. The standard approach, mostly used in
quantum optical settings, results in a local-in-time Markovian
linear differential equations for the system's density matrix, the
so-called quantum Liouville equation. The most general Markovian
form of such an equation is sometimes referred to as the Redfield
equation. A more mathematically appealing form, which manifestly
preserves the complete positivity of the density matrix, and can be
derived from the Redfield model with the additional {\em secular}
approximation, is the so called Lindblad equation
\cite{Lindblad1976}.

In this setting, non-equilibrium criticalities are identified as
dramatic structural changes of the Liouvillian steady state due to
small modification of tuneable external parameter of the system. The
analogy with equilibrium phase transitions is straightforward. For
zero temperature, QPTs  are understood through the properties of the
(unique) ground state of the Hamiltonian
$\mathcal{H}(\bm{\lambda})$, $\bm{\lambda}\in\mathcal{M}$ governing
the dynamics of the system, \begin{equation} \frac{d\rho}{dt} = -i
[\mathcal{H}(\bm{\lambda}),\rho]\,. \end{equation} Phase diagram and
criticality are determined by the low-lying spectrum of excitations
of the system's Hamiltonian. It is the singular properties of the
ground state, as a function of the Hamiltonian parameters
$\bm{\lambda}\in\mathcal{M}$, which are associated to the macroscopic
observable effects typical of criticality. These manifest themselves
through divergence of correlation length and generally occur if the
gap of the Hamiltonian closes.

Similarly, in a dynamics governed by a Liouvillian master equation
\begin{equation} 
\frac{d\rho}{dt} = \mathcal{L}({\bm{\lambda}})  \rho \,, 
\end{equation}
 which generally depends on a set of external parameters
$\bm{\lambda}\in\mathcal{M}$, the family of (possibly unique) NESSes
$\rho_{s}(\bm{\lambda})$ are themselves labelled by the same
parameters. Observable macroscopical behaviour in the physical
properties of a many-body quantum system are associated to
non-analytical dependences of $\rho(\bm{\lambda})$ in the manifold
$\mathcal{M}$.

From a mathematical point of view, the Liouvillian ${\cal
L}(\bm{\lambda}) $ is a linear {\em non-Hermitian} super-operator
acting on the space of density matrices, and a NESS is defined as
its ``eigen-density-matrix'' with \emph{zero eigenvalue}, 
\begin{equation}
 \mathcal{L}({\bm{\lambda}})  \rho_{s}({\bm{\lambda}}) =0\,. \end{equation} The
spectral resolution of the Liouvillian operator
$\mathcal{L}_{\bm{\lambda}}$ provides information on the uniqueness
of the stable state and on the asymptotic decay rates which governs
the system's relaxation towards the NESS(es). The smallest of such
rates, denoted by $\Delta_{\mathcal{L}}$, is the so called
\emph{Liouvillian spectral gap} and determines the dominant time
scale $\tau_{c}\sim\Delta_{\mathcal{L}}^{-1}$ of the dissipative
dynamics.  This quantity is the closest open system analogue to the
Hamiltonian gap. Pretty much in the spirit of QPTs , NESS criticalities
are accompanied by the vanishing of $\Delta_{\mathcal{L}}$, a
phenomenon known as \emph{critical slow down}. Although generally
accepted as an indicator of dissipative phase transitions at
large~\cite{Diehl2008,Verstraete2009,Diehl2010a,Honing2012,Horstmann2013,Bardyn2013},
a vanishing dissipative gap is not a distinguishing feature of NESS
criticality, and it may be observed across phases characterised by
short range
correlations~\cite{Prosen2008,Prosen2010,Znidaric2015,Marzolino2017,Carollo2018}.

Now we will consider systems whose interaction with an environment leads
to a time evolution governed by a Lindblad master
equation~\cite{Breuer2007},
\begin{equation}\label{eq:Lindblad}
\frac{d\rho}{dt} =\mathcal{L}\rho= -i [\mathcal{H},\rho] +
\sum_{\alpha}(2 \Lambda_{\alpha}\rho \Lambda_{\alpha}^{\dagger}
-\{\Lambda_{\alpha}^{\dagger}\Lambda_{\alpha},\rho \}),
\end{equation}
where $\rho$ is the density matrix of the system, $\mathcal{H}$ is
its Hamiltonian, and the Lindblad operators (or jump operators)
$\Lambda_{\alpha}$ determine the interaction between the system and
the bath.  The dissipative dynamics is completely determined by the
jump operators $\Lambda_{\alpha}$, whose physical origin is prone
to different interpretations: From a microscopic point of view they
can be regarded as the effective action of a full Hamiltonian
dynamics of system and bath, where the degrees of freedom of the
reservoir have been traced out. Here, three major approximations are
needed: System and environment are initially in an uncorrelated
state, system and bath interacts weakly (Born approximation), and
the equilibration time of environment is short compared to other
time scales (Markov approximation). A second more versatile
interpretation originates from the concept of digital
simulators~\cite{Feynman1982,Lloyd1996,Weimer2010,Bloch2012,Blatt2012,Houck2012,Aspuru-Guzik2012},
where a set of arbitrary (local) Lindblad operators
$\Lambda_{\alpha}$ can be explicitly implemented in terms of local,
tailored interactions~\cite{Weimer2010}. From a mathematical point
of view, Eq.~(\ref{eq:Lindblad}) is the most general form of time
evolution described by a quantum dynamical semigroup, i.e., a family
of completely positive trace-preserving maps ${\cal E}_{t}$, which
is strongly continuous and satisfies ${\cal E}_{t_{1}}{\cal
E}_{t_{2}} ={\cal E}_{t_{1}+t_{2}}$~\cite{Lindblad1976,Alicki2007}.

Despite the above mentioned analogies, however, non-equilibrium QPTs 
are of a nature different from the standard QPTs  at zero temperature,
and their investigation requires a substantial change of approach,
both conceptually and methodologically. At a conceptual level,
stationary states are the result of coherent dynamics dominated by
incoherent dissipative processes. The response to a small
perturbation is primarily described by relaxation mechanisms rather
than the result of adiabatic modifications of the (ground) state.
From a technical point of view, the {\em non-Hermitian} nature of
the Liouvillian superoperator ${\cal L} $,  as opposed to pure
eigenvectors of a Hermitian Hamiltonian operator $H$, calls for the
development of alternative strategies, as the usual spectral theorem
and perturbation theory simply do not apply. This is quite a
daunting task to tackle in its full generality.

However, by confining oneself to the physically relevant case of
quadratic Liouvillian models of Fermions and spin lattices one is
able to state results with a significant level of
generality~\cite{Eisert2010,Honing2012,Horstmann2013}. This
parallels the central role that quasi-free models play in the
standard theory of quantum phase transitions. Specific models
belonging to this class indeed display rich non-equilibrium
features, non-trivial topological properties and
NESS-QPTs ~\cite{Diehl2008,Prosen2008,Horstmann2013,Bardyn2013}.

\section{Gaussian Fermionic States}
Let's make a brief detour to introduce the formalism of Gaussian
Fermionic states (GFS). To this end we consider systems of $n$ Fermionic
particles described by creation and annihilation operators
$c_{j}^{\dag}$ and $c_{j}$. These operators obey the canonical
anti-commutation relations, \begin{equation} \{c_{j},c_{k}\}=0\qquad
\{c_{j},c^{\dag}_{k}\}=\delta_{jk}\,. \end{equation} A convenient formulation
for quadratic models is in terms of the Hermitian Majorana
operators, which are defined as 
\begin{equation} w_{2j-1}:=c_{j}+c^{\dagger}_{j}\,,\qquad
w_{2j}:=i(c_{j}-c_{j}^{\dagger})\,, 
\end{equation} 
which, as generators of a
Clifford algebra, satisfy the following anti-commutation relations
\begin{equation} 
\{w_{j},w_{k}\}=2\delta_{jk}\,. 
\end{equation} 
We will consider an Hamiltonian quadratic in the Fermionic operators
\begin{equation}\label{eq:QuadH} 
\mathcal{H}:=\sum_{jk=1}^{2n} H_{jk}w_{j} w_{k}=
\bm{w}^{T} H \bm{w}\,, \qquad H=H^{\dagger}=-H^{T}\,, 
\end{equation} 
where
$\bm{w}:=(w_{1}\dots w_{2n})^{T}$ is an array of Majorana operators
and $H$ is a $2n\times 2n$ Hermitian antisymmetric matrix.
Quadratic Hamiltonian models describe quasi-free Fermions and are
known to be exactly solvable. Their ground states and thermal states
are Gaussian Fermionic states. Indeed, Gaussian Fermionic states are
defined as states that can be expressed as
\begin{equation}\label{eq:GFS1}
\rho=\frac{e^{-\frac{i}{4}\bm{w}^{T} \Omega \bm{w}}}{Z}\,,\qquad
Z:=\Tr[ e^{-\frac{i}{4}\bm{w}^{T} \Omega \bm{w}}]
\end{equation}
 where $\Omega$ is a $2n\times 2n$ real antisymmetric matrix (see~\ref{app:FGS}). The thermal state of the quadratic model is obtained by the simple identification $\Omega=-4i\beta H$. Obviously, the converse is always true, a Gaussian state is the thermal state of a suitable quadratic Hamiltonian, which is sometime referred to as parent Hamiltonian~\cite{Bardyn2013}.
Gaussian states are completely specified by the two-point
correlation function
\begin{equation}
 \Gamma_{jk}:=1/2\Tr{(\rho[w_{j},w_{k}])}\,, \qquad \Gamma=\Gamma^{\dagger}=-\Gamma^{T}\,,
\end{equation}
 where the matrix $\Gamma:=\{\Gamma_{jk}\}_{1}^{2n}$ is a $2n\times 2n$ imaginary antisymmetric matrix.  All higher-order correlation functions of a Gaussian state can be obtained from  $\Gamma$ by Wick's theorem~\cite{Bach1994}.\\
 One can show that $\Gamma$ and $\Omega$ can be simultaneously cast in a canonical form by an orthogonal matrix $Q$, $Q^T=Q^{-1}$,
\begin{equation}\label{eq:Q} \Gamma=Q \bigoplus_{k=1}^{n}\left(\begin{array}{cc}0
& i\gamma_k \\- i\gamma_k & 0\end{array}\right)Q^{T},\quad \Omega=Q
\bigoplus_{k=1}^{n}\left(\begin{array}{cc}0 & \Omega_{k} \\-
\Omega_k & 0\end{array}\right)Q^{T} \end{equation} where $\pm\gamma_{k}$ are
the eigenvalues of $\Gamma$ and $\pm i\Omega_{k}$ are the
eigenvalues of $\Omega$. Indeed, the two matrices are related as,
\begin{align}\label{eq:corrQ}
\Gamma = \tanh\left(i\frac{\Omega}2\right)~.
\end{align} 
and their eigenvalues can be expressed as $\gamma_{k}=\tanh{(\Omega_{k}/2)}$, which implies that $|\gamma_{k}|\le1$. Moreover let 
\begin{equation}
 \bm{z}=(z_{1},\dots,z_{2n})^{T}:=Q\bm{w} 
 \end{equation} 
 be the Majorana
Fermions in the eigenmode representation. With respect to these
Fermions the Gaussian state can be expressed as,
\begin{align} 
  \rho &= \prod_k \frac{1-i \gamma_k\,z_{2k-1} z_{2k}}2~.
  \label{eq:GFS}
\end{align} 
Hence, a Gaussian Fermionic state can be factorised into a tensor
product $\rho=\bigotimes_{k}\rho_{k}$ of density matrices of the
eigenmodes $\rho_{k}:=\frac{1-i \gamma_k\,z_{2k-1} z_{2k}}2$.
Notice that for $\gamma_k=\pm 1$, one has $\Omega_k = \pm\infty$,
making the definition~(\ref{eq:GFS1}) of Gaussian state not well
defined, unlike Eq.~\eqref{eq:GFS}, showing that the latter offers a
regular parameterisation even in those extremal points. Notice that
$|\gamma_{k}|=1$ corresponds to a Fermionic mode
$\tilde{c}_{k}=1/2(z_{2k-1}+z_{2k})$ being in a pure state, as  it
is clear from the following explicit expression for the purity of
the states $\rho_{k}$,
\begin{align}
  \Tr[\rho_{k}^2] =\frac{1+\gamma_{k}^{2}}{2}~.
  \label{e.purity}
\end{align}
This implies the following basis-independent expression of the
purity
\begin{align}
  \Tr[\rho^2] =\prod_{k}\frac{1+\gamma_{k}^{2}}{2}=
  \sqrt{\det\left(\frac{\one+\Gamma^2}2\right)}~.
  \label{e.purity}
\end{align}

 \section{Dissipative Markovian Quadratic Models}~\label{sec:QuadLind}
We will discuss dissipative Fermionic models, described by a
Lindblad master equation
\begin{equation}\label{eq:Lindblad1}
\frac{d\rho}{dt} =\mathcal{L}\rho= -i [\mathcal{H},\rho] +
\sum_{\alpha}(2 \Lambda_{\alpha}\rho \Lambda_{\alpha}^{\dagger}
-\{\Lambda_{\alpha}^{\dagger}\Lambda_{\alpha},\rho \}),
\end{equation}
whose global dynamics is quadratic in Fermion operators. This means
that the Hamiltonian considered will be of the
type~(\ref{eq:QuadH}), and the set of jump operators
$\Lambda_{\alpha}$ will be linear in the Fermion operators, i.e. 
\begin{equation}
 \Lambda_{\alpha}=\bm{l}_{\alpha}^{T}\bm{w}, \textrm{ with
}\bm{l}_{\alpha}:=(l_{1}^{\alpha},\dots,l_{2n}^{\alpha})^{T}~, 
\end{equation}
 where $\bm{l}_{\alpha}$ denotes a set of 2n-dimensional complex
vectors. We assume that $H$ and $\bm{l}_{\alpha}$'s depend on a set
of parameters $\bm{\lambda}\in\mathcal{M}$ which defines the
underlying dissipative model. Due to the quadratic dependence on the
Fermionic operators, the Liouvillian can be diagonalised exactly and
its stable state is Gaussian. This has been proven in full
generality in
Ref.~\cite{Prosen2008a,Prosen2010a,Zunkovic2010,Banchi2014} using a
formalism called ``third quantisation''. This essentially consists
in the development of an operator algebra acting on the space of
density matrices which mimics the algebraic properties of the
Fermionic operators acting on the ordinary Fock space. In the
following, we will review the way in which the stable state
$\rho_{s}$ is obtained within this formalism, which will provide the
natural parameterisation necessary for the subsequent developments.
The Liouvillian~\cite{Prosen2010a} can be written as a quadratic
form in terms of the following set of $2n$ creation and annihilation
super-operators
\begin{equation} \begin{aligned}
\hat{a}^\dagger_j \,\cdot := -\frac{i}2 \,W\,\Big\{w_j,
\cdot\Big\}~,   && \hat{a}_j \,\cdot := -\frac{i}2 \,W\,\Big[w_j,
\cdot\Big]~, \label{e.super}
\end{aligned} \end{equation} 
where $[.,.]$ and $\{.,.\}$ are the usual commutator and
anti-commutator, respectively and 
\begin{equation} W:=i^n\prod_{j=1}^{2n}
w_j\qquad 
\end{equation}
is an Hermitian operator satisfying the following
properties 
\begin{equation} W=W^{\dagger}~,\quad W^{2}=\one~,\quad
\{W,w_{k}\}=0\quad \forall\, k~. 
\end{equation}

From the above properties one can prove that the super-operators
defined in Eq.~\eqref{e.super} satisfy the canonical
anti-commutation relations, \begin{equation} \{\hat{a}_j, \hat{a}_k\} =
0~,\qquad\{\hat{a}_j^\dagger, \hat{a}_k\} = \delta_{jk}~, \end{equation} and
thus they reproduce the algebra of ordinary Fermionic operators.
Notice, however, that the space on which $\hat{a}_{k}$ and
$\hat{a}_{k}^{\dagger}$ act is the Hilbert-Schmidt space of linear
operator $\mathcal{B}(\mathcal{H})$, to which the set of density
matrices $\rho$ belongs.

Let's denote by $\mathcal{R}$ the $4^{n}$-dimensional subspace of
$\mathcal{B}(\mathcal{H})$ spanned by $\prod_j w_j^{s_j}$,
($s_j\in\{0,1\}$). One can regard this subspace as a linear Hilbert
space whose elements, denoted by $\sket{\bm s}$ are normalised with
respect to the Hilbert-Schmidt inner product, i.e. $\sbraket{\bm
s}{\bm s} \equiv \Tr[s^\dagger s]= 1$ for $\sket{\bm s}\in\mathcal
R$. Notice that the vacuum of the Fermionic super-operators, i.e.
the state $\sket{0}$ such that $\hat{a}_{k}\sket{0}=0$, corresponds
to the completely mixed state $\sket{0} \propto\one$. Moreover, one
can verify that the superoperator $a_k^\dagger$ is the Hermitian
conjugate of $a_k$ in $\mathcal R$.

Using the above formalism, a lengthy but straightforward calculation
shows that the Lindbladian equation Eq.~(\ref{eq:Lindblad1}) can be
explicitly expressed in the following bilinear form in terms of
$\hat{a}_{k}$ and $\hat{a}_{k}^{\dagger}$, \begin{equation}\label{eq:QuadL}
 \mathcal L = -\sum_{jk} \left(X_{jk}\,\hat{a}_j^\dagger \hat{a}_k +
  Y_{ij}\,\hat{a}_j^\dagger \hat{a}_k^\dagger/2\right)
\end{equation} where
\begin{align}
X&:=4[i H + \Re{(M)}]~,\qquad X=X^{*}\label{eq:X}\\
Y&:=-i8\Im{(M)}~, \qquad Y=Y^{\dagger}=-Y^{T}\label{eq:Y}
\end{align}
and 
\begin{equation}\label{eq:BathM} 
M_{jk}:=(\sum_{\alpha}
\bm{l}_{\alpha}\otimes\bm{l}_{\alpha}^{\dagger})_{jk}=\sum_{\alpha}
l^{\alpha}_{j}(l^{\alpha}_{k})^{*}~,\qquad M=M^{\dagger}\ge 0~, 
\end{equation}
 is a positive semidefinite matrix called \emph{bath matrix}.

Under certain condition derived in~\cite{Prosen2010a}, the
dissipative dynamics admits a unique non-equilibrium steady state
solution $\rho_{s}$ such that $\mathcal{L}\rho_{s}=0$, and such a
state is Gaussian. The two point correlation function of the NESS is
obtained from the solution of the following (continuous) Lyapunov
equation: \begin{equation}\label{eq:Lyap} X\Gamma + \Gamma X^{T} = Y~. \end{equation} As
shown in the next section the correlation matrix $\Gamma$ plays also
a central role in the diagonalisation of the Liouvillian.

Notice that the real matrix $X$ does not have to be diagonalisable.
However, for convenience one can safely assume that this is the
case. Indeed, in the explicit models that we will consider in the
next sections, this assumption is always satisfied. Nevertheless,
the generality of the consideration that will follow will not be
affected if this condition is lifted, as it will be argued in the~\ref{app:Spec}.

Let $\bm{a}:=(\hat{a}_{1},\dots,\hat{a}_{2n})^{T}$ be the array of
Fermionic annihilation super-operators, and let $U$ be the
invertible matrix that diagonalises $X$, i.e. 
\begin{equation}
 X=U\,D_{X}\,U^{-1}\qquad D_{X}:=\diag(\{x_k\}_{k=1}^{2n}), \end{equation} where
$x_k\in \mathbb{ C}$ are the eigenvalues of $X$. One can show that
the following non-unitary Bogoliubov
transformation~\cite{Blaizot1986}, 
\begin{equation}
 \left\{\begin{array}{l}{\bm b} \,\,=\, U^{-1}( {\bm a}-\Gamma\,{\bm a}^\dagger),\\
{\bm b}^\times \!\!= U^{T}{\bm a}^{\dagger},
\end{array}\right.
\end{equation} 
where $\bm{b}:=(\hat{b}_{1},\dots,\hat{b}_{2n})^{T}$, brings
$\mathcal L$ to the diagonal form 
\begin{equation}\label{eq:LDiag} \mathcal L =
-\sum_k x_k\, \hat{b}_k^\times \hat{b}_k~. 
\end{equation}
Notice that, due to the non-unitarity of the Bogoliubov transformation, the operators
$\hat{b}_j$ and $\hat{b}_j^\times$ satisfy the canonical
anti-commutation relations, however $\hat{b}_j^\times \neq
\hat{b}_j^\dagger$. Pretty much in the spirit of the standard
Bogoliubov transformation, the unnormalized steady state
$\rho_{s}$ is the vacuum of the annihilation super-operators ${\bm
b}$, i.e. $\hat{b}_j \sket{\rho_{s}} = 0$, $\forall j=1,\dots,2n$.
From the operator form of the Bogoliubov
transformation~\cite{Blaizot1986}, one finds
\begin{align}
  \rho_{s}  = e^{-\frac12 {\bm a}^\dagger\cdot \Gamma{\bm a}^\dagger}\!({\one})~
  \label{e.ss}
\end{align}
where, as noted earlier, the identity operator is the vacuum of $\bm
a$. Due to the explicit form of the super-operators \eqref{e.super},
in the next section we will show that the above state is a Gaussian
Fermionic state and that its two point correlation functions $\Tr
\rho_{s} [w_i,w_j]$ are given by $\Gamma_{ij}$, i.e. by the solution
of the continuous Lyapunov Eq.~\eqref{eq:Lyap}.

According to~\cite{Prosen2010a}, the condition of uniqueness of the
steady state is 
\begin{equation} 
\Delta_{\mathcal{L}}:=2\min_{j}{\Re{(x_{j})}}> 0, 
\end{equation}
  where the $x_{j}$'s are the eigenvalues of $X$, and $\Delta_{\mathcal{L}}$ is the Liouvillian spectral gap. When this condition is met, any other state will eventually decay into the NESS in a time scale $\tau\sim1/\Delta_{\mathcal{L}}$.
In the thermodynamic limit $n\to\infty$ a vanishing gap
$\Delta(n)\to 0$ may be accompained, though not necessarily, by
non-differentiable properties of the
NESS~\cite{Prosen2008,Znidaric2015}. For this reason, the scaling of
$\Delta_{\mathcal{L}}(n)$ has been used as an indicator of NESS
criticality~\cite{Prosen2010,Znidaric2011,Horstmann2013,Cai2013,Znidaric2015}.
NESS-QPTs  have been investigated through the scaling of the Bures
metrics~\cite{Zanardi2007a,CamposVenuti2007}, whose
super-extensivity has been connected to a vanishing
$\Delta_{\mathcal{L}}$~\cite{Banchi2014}. Along the same line, it is
also possible to demonstrate that a similar relation exists between
the scaling properties of the dissipative gap and the mean Uhlmann
curvature~\cite{Carollo2018,Carollo2018a,Leonforte2019,Leonforte2019a,Bascone2019}. Essentially, one will develop in
the context of NESS-QPTs  an argument similar to the one established
in section~\ref{sec:SEQ}, for the zero-temperature quantum phase
transitions. There  the expression~(\ref{eq:BoundQ}) defines a
necessary relation between the scaling of the quantum geometric
tensor and a vanishing Hamiltonian gap. Here we will provide a
relation connecting the scaling properties of the (generalised)
quantum geometric tensor and the dissipative gap
$\Delta_{\mathcal{L}}$.

\subsection{Diagonalisation of the Lindblad Equation}
Following the notation introduced in section~\ref{sec:QuadLind}, the
Liouvillian~\eqref{eq:Lindblad1} can be re-expressed as
\begin{equation} \begin{aligned}
  \mathcal L = -\frac12\begin{pmatrix} {\bm a}^\dagger & \bm a \end{pmatrix}
  \,
  \begin{pmatrix}
    X & Y \\ 0 & -X^T
  \end{pmatrix}
  \,
  \begin{pmatrix} \bm a \\ {\bm a}^\dagger \end{pmatrix}  -\frac12 \Tr X~.
\end{aligned} \end{equation} 

Consider the following invertible transformation \begin{equation} T:=
\begin{pmatrix}
    \one & \Gamma\\0&\one
  \end{pmatrix}~,\qquad
  T^{-1}:=  \begin{pmatrix}
    \one & -\Gamma\\0&\one
  \end{pmatrix}~,
\end{equation}
%
If $\Gamma$ is the matrix solution of \eqref{eq:Lyap}, then
 \begin{equation} \begin{aligned}
   \begin{pmatrix}
    X & Y \\ 0 & -X^T
  \end{pmatrix}
  =\,\hat{T}^{-1} \,\begin{pmatrix}
     X & Y-X\Gamma - \Gamma X^{T} \\ 0 & -X^{T}
  \end{pmatrix}\,\hat{T}=
  \hat{T}^{-1}\,
   \begin{pmatrix}
     X & 0 \\ 0 & -X^{\hat{T}}
  \end{pmatrix}\,
  \hat{T}\,.
  \label{e.diaglind}
 \end{aligned} \end{equation} 
Therefore,  one straightforwardly sees that the
matrix~(\ref{e.diaglind}) is diagonalised by the following
transformation
 \begin{equation} \begin{aligned}
      \hat S &:=\,\begin{pmatrix}
    U^{-1} &0\\0&U^T
      \end{pmatrix}\,\hat{T}\,=\,
      \begin{pmatrix}
    U^{-1} & U^{-1}\, \Gamma \\0&U^{T}
      \end{pmatrix}
\label{e.bogohat}
\end{aligned} \end{equation} 
One can show that $\hat S$ is a non-unitary Bogoliubov
transformation~\cite{Blaizot1986}, which amounts to verify that
$\hat S$ fulfils the following condition (see Eq.(2.6) of
\cite{Blaizot1986})
\begin{equation} \begin{aligned}\label{eq:GroupProp}
  \hat S\,\Sigma\,\hat S^{T} = \Sigma \qquad \textrm{ where }\quad \Sigma :=  \begin{pmatrix} 0 &\one \\\one &0 \end{pmatrix}~.
\end{aligned} \end{equation} 
This transformation leads to the definition of a new set of creation
and annihilation super-operators as
\begin{equation} \begin{aligned}
  \begin{pmatrix} \bm b \\ {\bm b}^\times \end{pmatrix} &= \hat S \,
    \begin{pmatrix} \bm a \\ {\bm a}^\dagger \end{pmatrix}~.
      \label{e.bogodiag}
\end{aligned} \end{equation} 
Since $\mathcal S$ is a non-unitary Bogoliubov transformation, the
operators $\hat{b}_j$ and $\hat{b}_j^\times$ satisfy the canonical
anti-commutation relations, but $\hat{b}_j^\times \neq
\hat{b}_j^\dagger$. Moreover, by employing the relation $
\begin{pmatrix} {\bm a}^\dagger & \bm a \end{pmatrix} =
  \begin{pmatrix} \bm a \\ {\bm a}^\dagger \end{pmatrix}^T\,\Sigma^x$, together with the property~(\ref{eq:GroupProp}), one finds that
\begin{equation} \begin{aligned}
  \mathcal L = -\frac12\begin{pmatrix} {\bm b}^\times & \bm b \end{pmatrix}
  \,
  \begin{pmatrix}
    D_{X} & 0 \\ 0 & -D_{X}
  \end{pmatrix}
  \,
  \begin{pmatrix} \bm b \\ {\bm b}^\times \end{pmatrix}  -\frac12 \Tr X~,
\end{aligned} \end{equation} 
i.e.,
\begin{equation} \begin{aligned}
  {\mathcal L = -\sum_j x_j \,\hat{b}^\times_j\, \hat{b}_j}~.
  \label{e.lindiag}
\end{aligned} \end{equation} 
To the canonical transformation \eqref{e.bogodiag} it
corresponds an operator acting on the Fock space, which thanks to
Eq.~(2.16) of \cite{Blaizot1986}, can be written into the form
\begin{equation} \begin{aligned}
  \hat{b}_j &= \mathcal S\, \hat{a}_j\, \mathcal S^{-1}~, &
  \hat{b}_j^\times &= \mathcal S\, \hat{a}_j^\dagger\, \mathcal S^{-1}~,
\end{aligned} \end{equation} 
where
\begin{equation} \begin{aligned}
  \mathcal S = :\exp\left(-\frac12 {\bm a}^\dagger \,\Gamma\, {\bm a}^\dagger +
  {\bm a}^\dagger\,(U-1)\,{\bm a}\right):~,
\end{aligned} \end{equation} 
and $:\exp(\cdot):$ denotes the normal ordering of the exponential.

By exploiting the above operator, one is then able to explicitly
express the vacuum of the Bogoliubov operators  ${\bm b}$, i.e. the
stationary state $\rho_{s}$ of the Liouvillian, in terms of the
original super-operators $\bm a$. Recall that the vacuum of $\bm a$,
i.e. the element $\sket{\bm 0} \in \mathcal R$ such that $\hat{a}_i
\sket{\bm 0} = 0$, $\forall j=1,\ldots2n$, is the completely mixed
state. It also fulfils the property $\sbra{\bm 0}\mathcal L = 0$.

The vacuum of the Bogoliubov operators $\bm b$ can be readily
obtained from the operator $\mathcal S$: $\sket\rho_{s}=\mathcal
S\sket{\bm 0}$. Indeed, as $\hat{a}_j \sket{\bm 0}=0$, one has $
\hat{b}_j\sket{\rho_s} = \mathcal S \hat{a}_j\mathcal S^{-1}\mathcal
S\sket{\bm 0} = 0$. Hence,
\begin{equation} \begin{aligned}
  \sket{\rho_{s}} =
  \mathcal S \sket{\bm 0} = e^{-\frac12 {\bm a}^\dagger\,\Gamma\,{\bm a}^\dagger}
  \sket{\bm 0}~.
  \label{e.superss}
\end{aligned} \end{equation} 
 The state \eqref{e.superss} is exactly \eqref{eq:GFS}, as will be shown in the following.
Thanks to the transformation $Q$ defined in \eqref{eq:Q}, one can
write $\Gamma$ in a canonical form with respect to the set of
Fermion eigenmodes $\bm{z}$. By using the definition
\eqref{e.superss}, one gets 
\begin{equation}
   \frac12 {\bm a}^\dagger \, \Gamma \, {\bm a}^\dagger \rho =\frac18\left(\bm w\cdot \Gamma\bm w\rho + 2 \bm w\cdot \Gamma\rho \bm w +
   \rho \bm w \cdot \Gamma \bm w\right)=\sum_k \mathcal G_k(\rho)~,
\end{equation} where 
\begin{equation}
    \mathcal G_k(\rho):= \frac{i}4 \gamma_k \big[z_{2k-1}z_{2k}\rho
   +z_{2k-1}\rho z_{2k}-z_{2k}\rho z_{2k-1}+
   \rho z_{2k-1}z_{2k}\big]~.
\end{equation} One can verify the following two properties, 
\begin{equation}
 \left\{\begin{array}{l} \mathcal G_k(\one) = i \,\gamma_k
\,z_{2k-1}z_{2k}~,
\\ \mathcal G_k(z_{2k-1}z_{2k}) = 0~,\end{array}\right.
\end{equation} which streighforwardly leads to
\begin{equation} \begin{aligned}
  \rho_s \propto e^{-\frac12 \bm{a}^\dagger \Gamma \bm{a}^\dagger} \sket{0} \propto
  \prod_{k} e^{-\mathcal G_k} \one =
  \prod_k \left(1 - i \,\gamma_k \,z_{2k-1}z_{2k} \right)~,
  \label{e.rhoz}
\end{aligned} \end{equation} 
thus recovering Eq.~\eqref{eq:GFS}.

\subsection{Liouvillian Spectrum}\label{sec:LSpec}
The conditions for the existence and uniqueness of \eqref{e.rhoz}
have been derived in \cite{Prosen2010a}. We now review those
conditions and express them in terms of the spectral gap
$\Delta_{\mathcal{L}}$.

The correlation matrix $\Gamma\in M_{2n}(\mathbb{C})$ is the
matrix solution of Eq.~\eqref{eq:Lyap}. This equation acquires a
familiar linear matrix representation, when expressed through the so
called (non-canonical) ``vectorising"  isomorphism 
\begin{equation}
  \textrm{vec}\colon M_{2n}({\mathbb{C}})\rightarrow
({\mathbb{C}}^{2n})^{\otimes\,2}\,/\, |i\rangle\langle j|\rightarrow
|i\rangle\otimes|j\rangle, \end{equation} which vectorises a matrix.
 This is also a Hilbert-space isomorphism, namely

\begin{equation}
  \langle \V A,\V B\rangle = ( A,\,B):=\tr(A^\dagger B).
 \end{equation}
 One can directly check that

\begin{equation}
  \V{ABC} = \left(\,A\otimes C^T\,\right)\, \V{B}.
\end{equation}
 Applying the vectorising isomorphisim to both sides of continuous
Lyapunov Eq.~\eqref{eq:Lyap} one then gets 
\begin{equation}\label{syl1}
\hat{X} \V{\Gamma}=\V{Y},\qquad \hat{X}:=(X\otimes {\mathbf{1}}
+{\mathbf{1}}\otimes X),  
\end{equation}
where
$\hat{X}\in{\rm{End}}({\mathbb{C}}^{2n})^{\otimes\,2}\cong
M_{4n^2}({\mathbb{C}}).$ There are three distinct operators involved
in the above formalism, and correspondingly there are three
different definitions of spectral gaps, which are described in the
following.
\begin{enumerate}
\item The map $X\colon {\mathbb{C}}^{2n}\rightarrow  {\mathbb{C}}^{2n},$ a $2n\times 2n$ {\em real} diagonalizable matrix. Its spectrum
is $\{x_j\}_{j=1}^{2n}\subset{\mathbb{C}}$ and, because of reality,
is invariant under complex conjugation. From the non-negativity of
the bath matrix $M$ one can prove that $\Re\,x_j\ge0,\forall j.$
(see~\ref{app:Spec}).

\item The map $\hat{X}=X\otimes {\mathbf{1}} +{\mathbf{1}}\otimes X\in{\rm{End}}({\mathbb{C}}^{2n})^{\otimes\,2}\cong M_{4n^2}({\mathbb{C}}),$
 a $4n^2\times 4n^2$ matrix. Since $X$ is assumed diagonalisable, also $\hat{X}$ will be so, and its spectrum is $\{ x_i+x_j\}_{i,j=1}^{2n}\subset{\mathbb{C}}$. The minimum of its modulus is clearly given by $\Delta_{\hat{X}}:=\min_{i,j} |x_i+x_j|.$
Importantly, diagonalisability of $\hat{X}$ straightforwardly
implies
\begin{align}
\Delta_{\hat X}^{-1} = \|\hat{X}^{-1}\|_\infty~.
\label{eq:Deltaxhat}
\end{align} 
Moreover, for the  uniqueness of the steady state we must have
$\hat{X}$ invertible i.e., $\Delta_{\hat{X}}>0.$

\item The Liouvillian ${\cal L}\colon {\rm{End}}(  ({\mathbb{C}}^2)^{\otimes n} )\rightarrow{\rm{End}}(  ({\mathbb{C}}^2)^{\otimes n} ),$
 a $2^{2n}\times 2^{2n}$ matrix. As it can be seen from Eq.~(\ref{e.lindiag}), its spectrum can be defined through the array occupation numbers $\bm{n}:=( n_{1},\dots, n_{2n} )^{T}$, where $n_{k}=0,1$ are the eigenvalues of Bogoliubov number operators $b^{\times}_{k} b_{k}$. Its spectrum is given by
\begin{equation} \begin{aligned}
{\rm{Sp}}({\cal L})= \{-x_{\bm{n}}:=\bm{n}^{T}\cdot \bm{x}\in
\mathbb{C}\} \quad  \textrm{ where }  \quad \bm{x}:=(x_{1},\dots
,x_{2n})^{T},  \textrm{ with } x_k\in {\rm{Sp}}(X). \label{e.spectrum}
\end{aligned} \end{equation} 
Notice that $0\in {\rm{Sp}}({\cal L})$ i.e., $\cal L$ is always
non-invertible and the steady state (e.g., the Gaussian state
${\bm{n}}={\bm{0}}$) is in the kernel of $\cal L$. If this latter
is one-dimensional, i.e. a unique steady state, the gap of  $\cal L$ can be
defined as $ {\Delta}_{\cal L}:=\min_{  {\bm{n}} \neq {\mathbf{0}}
}\, |x_{\mathbf{n}}|. $

\end{enumerate}
Notice that, on account of the stability of the physical system, one
is expected to have $\Re\,x_j\ge0,\forall j.$ Indeed, this is the
quantum equivalent of the classical Lyapunov stability condition,
where the time-scale for convergence $\rho(t)\to\rho(\infty)$ is
dictated by $\tilde{\Delta}^{-1}$, with $\tilde\Delta=
\min_{{\mathbf{n}} \neq {\mathbf{0}}} \Re\, x_{{\mathbf{n}}}.$

It is not hard to show, that the three distinct definitions of
spectral gaps just described effectively collapse into each other.
\begin{Proposition}\label{prop:Delta}
If $\Delta = \min_j 2\Re(x_j) > 0 $ then
\begin{equation}
 \Delta={\Delta}_{\cal L}= \Delta_{\hat{X}}~.
 \label{eq:prop1}
\end{equation}
\end{Proposition}
\begin{proof}
$|x_{\mathbf{n}}|=|  \bm{n}^{T}\cdot \bm{x} |\ge
|\Re(\bm{n}^{T}\cdot \bm{x}) |.$ The first bound can be saturated by
choosing the $n_j$'s in such a way that only a set $P$ of complex
conjugated pairs $x^\pm_p$ of eigenvalues are present. In this case
$ |\Re(\bm{n}^{T}\cdot \bm{x})|=2\sum_{p\in P} \Re\,x_p$, where we
used  the  assumption $(\forall p)\, \Re\,x_p\ge 0.$ Using again
positivity of all the terms, this sum can be made as small as
possible by choosing $|P|=1$ and minimising over $p=1,\ldots,n.$
This shows that ${\Delta}_{\cal L}=\min_{\mathbf{n}}
|x_{\mathbf{n}}|=2 \min \{\Re\,x_p\}_{p=1}^n.$ By a  similar
argument one shows that $\Delta_{\hat{X}}=\min
\{|x_i+x_j|\}_{i,j=1}^{2n}$ is given by the same expression i.e.
${\Delta}_{\cal L}= \Delta_{\hat{X}}$. Finally
$\Delta=2\min_{\mathbf{n}} \Re\,
x_{\mathbf{n}}\equiv2\tilde\Delta=2\min_p \Re\, x_p=\Delta_{\cal
L}$.
\end{proof}

\section{Geometric Properties of the Steady States}\label{sec:GeoNESS}
We would like to transfer to the framework of NESS-QPTs  the insight
that we have learnt from the geometric information approach and the
geometric phase methods that so far have been applied to the
equilibrium case. The idea would be to explore the properties of the
metrics and the properties of the geometric phases pretty much in
the spirit of the equilibrium phase transitions. The natural
candidate for the definition of a metric is clearly the Bures
metric, as the intuition built from QPTs  in open systems would
suggest~\cite{Cozzini2007,Zanardi2007b,Zanardi2007a,Zanardi2007b}.
This has been done in Ref.~\cite{Banchi2014}.

A completely different story applies to the geometric phase, as a
natural candidate in the mixed state domain does not exist. In the
context of mixed quantum states, it is necessary to exploit
unorthodox concepts of geometric phases and many possible
definitions of the mixed state geometric phase have been put
forward~\cite{Uhlmann1986,Sjoqvist2000,Carollo2003,Carollo2004,Tong2004,Chaturvedi2004,Marzlin2004,Carollo2005a,Buric2009,Sinitsyn2009}.
Which definition is best suited in this context depends largely on
the type of information that one wants to pursue. In this context
the Uhlmann GP~\cite{Uhlmann1986} stands out for its deep rooted
relations with the fidelity approach and for its relations with
quantum estimation theory.

Motivated by this, we will concentrate on the mean Uhmann curvature,
which has been already introduced in section~\ref{sec:UC}. Rather
than exploring the geometric phase itself, which provides insight on
the global properties of the mixed state manifold, one can look at the Uhlmann curvature, as it conveys information on the
local geometric structure of the parametric manifold. This choice is
ideally suited to the study of non-equilibrium phase transitions
which are related to local differential properties of the state
space.

The mathematical properties of the MUC makes it an ideal candidate
both at the conceptual and at the technical level. From a physical
point of view the mean Uhlmann curvature is gauge invariant, thus
ensuring that its behaviour is physically relevant and cannot emerge
as an artefact of the gauge choice. Moreover, technically the MUC is
much easier to handle than the full Uhlmann curvature, due to its
scalar nature as opposed to the non-Abelian structure of the
curvature.

We will review the general formulas which unify, within the same
framework, Bures metrics and mean Uhlmann curvature over the set of
Gaussian Fermionic states. This will be needed in order to discuss
in the following sections how the scaling of the metric and the
curvature provide information on the closing of $\Delta_{\mathcal
L}$ and the divergences of the two-point correlations. Finally one
can see this in action, by applying such a theoretical framework to
exactly solvable models. This  demonstrates that NESS
phase diagram can be accurately mapped by studying  the scaling
behaviour and the singularities of the metric tensor $g$ and of the
$\mathcal{U}$: critical lines can be identified and the different
phases distinguished.

Moreover, with joint information of both $g $ and $\mathcal{U}$ and
from insight derived from quantum estimation theory, a concept of
``quantum-ness'' of the NESS-QPTs  can be introduced. The aim of this
approach is to glean an insight into the character of the fluctuations
driving the non-equilibrium phase transitions.

The calculation of both Uhlmann curvature and Bures metrics in large
Hilbert spaces is quite a daunting task. Standard
approaches~\cite{Braunstein1994} are computationally not viable in
many-body setups, where the quest for effective methods to evaluate
both metric and curvature is the subject of active
research~\cite{Ercolessi2013}. We will show that in the physically
relevant class of Gaussian Fermionic states this can be accomplished
in a surprisingly efficient way~\cite{Carollo2018a}.

\subsection{Mean Uhlmann Curvature and Bures Metric of Gaussian Fermionic States}\label{sec:MUCGFS}
Before discussing the geometric properties of Gaussian Fermionic
states, let's recall some basic properties of the correlation
function. For  GFS, all odd-order correlation functions are zero,
and all even-order correlations, higher than two, can be obtained
from covariance matrix $\Gamma$ by Wick's theorem~\cite{Bach1994} , i.e. 
\begin{equation} \Tr(\rho
\omega_{k_{1}}
\omega_{k_{2}}...\omega_{k_{2p}})=\Pf(\Gamma_{k_{1}k_{2}\dots
k_{2p}}), \qquad  1 \le k_{1} < . . . < k_{2p} \le 2n. 
\end{equation}
$\Gamma_{k_{1}k_{2}\dots k_{2p}}$ is the corresponding $2p \times
2p$ submatrix of $\Gamma$ and $\Pf(\Gamma_{k_{1}k_{2}\dots k_{2p}})^{2}
= \det (\Gamma_{k_{1}k_{2}\dots k_{2p}})$ is the Pfaffian. An
especially useful case is the four-point correlation function
\begin{equation}\label{Pf}
    \Tr{(\rho\omega_{j} \omega_{k}\omega_{l}\omega_{m})}= a_{jk}a_{lm}-a_{jl}a_{km}+a_{jm}a_{kl},
\end{equation}
where $a_{jk}:=\Gamma_{jk}+\delta_{jk}$.

To derive a convenient expression for the Uhlmann curvature and the
Bures metric for Gaussian Fermionic states, first let's recall their
expression in terms of the parallel transport
generator~(\ref{eq:G}): \begin{equation} \left\{\begin{array}{l}
g_{\mu\nu}:=\,\,\Re\,Q_{\mu\nu}\,;
\\\mathcal{U}_{\mu\nu}:=2\,\Im\,Q_{\mu\nu}\,; \end{array}\right.
\qquad \qquad Q_{\mu\nu}=\Tr{\rho G_{\mu}G_{\nu}}, \end{equation} where
$Q_{\mu\nu}$ is the generalised quantum geometric
tensor (QGT) (see Eq.~(\ref{eq:GQGT})).

 The starting point is to derive the generator $\bm G$ in terms of the two point correlation matrix $\Gamma$. Due to the quadratic dependence of the density matrix in $\bm{\omega}$ (see Eq.~(\ref{eq:GFS})), and following the arguments of~\cite{Monras2013,Jiang2014,Carollo2018a}, it can be shown that $\bm G$ is a quadratic polynomial in the Majorana Fermions
 \begin{equation}\label{SLDGS}
 \bm{G}_{}=: \frac{1}{4}\bm{\omega}^{T}\cdot \bm{K} \bm{\omega} + \bm{\zeta}_{}^{T}\bm{\omega} +\bm{\eta}_{},
 \end{equation}
 where $\bm{K}=\sum_{\mu}K_{\mu}d\lambda_{\mu}$, and $K_{\mu}:=\{K_{\mu}^{jk}\}_{jk=1}^{2n}$ are a set of $2n\times2n$ Hermitian anti-symmetric matrices, $\bm{\zeta}=\bm{\zeta}_{\mu} d\lambda_{\mu}$, with $\bm{\zeta}_{\mu}$ a $2n$ real vector, and $\bm{\eta}=\eta_{\mu}d\lambda_{\mu}$ a real valued one-form. Note that any odd-order correlation function for a Gaussian Fermionic state vanishes identically, then

\begin{equation}
  \langle  \omega_{k}\rangle = \Tr{(\rho \omega_{k})}=0 \qquad \forall  k=1 \dots 2n\,.
\end{equation}
 By differentiating the above equation, one readily shows that the
linear term in~(\ref{SLDGS}) is identically zero
\begin{equation*}
 0=\Tr{(\omega_{k} d \rho )}= \Tr(\omega_{k}\{\bm{G},\rho\}) = \Tr(\rho\{\bm{\zeta}^{T}\bm{\omega},\omega_{k}\})=\zeta^{k}\,,
\end{equation*}
where $\zeta^{k}$ is the $k$-th component of $\bm{\zeta}$, and in
the third equality one takes into account that the third order
correlations vanish. The quantity $\bm \eta$ can be determined from
the trace preserving condition, i.e. \begin{equation} d\,\Tr\,
\rho=\Tr{(d{\rho})}=2\Tr{(\rho \bm{G})}=0\,, \end{equation} which, after
plugging in Eq.~(\ref{SLDGS}), leads to
\begin{equation}\label{eta}
\bm\eta=-\frac{1}{4}\Tr{(\rho \bm{\omega}^{T}\bm{K}
\bm{\omega})}=\frac{1}{4}\Tr{(\bm{K}_{}\, \Gamma)}.
\end{equation}
 In order to determine $\bm K_{}$, let's take differential of $\Gamma_{jk}=1/2\Tr{(\rho[\omega_{j},\omega_{k}])}$, then
 \begin{equation} \begin{aligned}
 d\Gamma_{jk}=\frac{1}{2}\Tr{(d \rho[\omega_{j},\omega_{k}])}&=\frac{1}{2}\Tr{(\{\rho,\bm G_{}\}[\omega_{j},\omega_{k}])}\nonumber\\
 &= \frac{1}{8}\Tr{(\{\rho,\bm{\omega}^{T}\bm K_{}\bm{\omega}\}[\omega_{j},\omega_{k}])}+ \bm\eta_{}\frac{1}{2}\Tr{(\rho[\omega_{j},\omega_{k}])}\nonumber\\
 &=\frac{1}{8}\sum_{lm}K_{}^{lm}\Tr{(\{\rho,[\omega_{l},\omega_{m}]\}[\omega_{j},\omega_{k}])}+\bm\eta_{} \Gamma_{jk}\nonumber\\
 &=(\Gamma \bm K_{} \Gamma-\bm K_{})_{jk} + \left[\bm \eta - \frac{1}{4}\Tr{(\bm K_{}\, \Gamma)}\right]\Gamma_{jk},
 \end{aligned} \end{equation} 
 where the last equality is obtained with the help of Eq.~(\ref{Pf}) and using the antisymmetry of $\Gamma$ and $ \bm K$ under the exchange of $j$ and $k$. Finally, according to Eq.~(\ref{eta}), the last term vanishes and we obtain the following (discrete time) Lyapunov equation
\begin{equation}\label{DLE}
d\Gamma = \Gamma \bm{K}_{}\Gamma -\bm{K}_{}.
\end{equation}
The above equation can be formally solved by \begin{equation}\label{eq:SolK} \bm
K_{}=(\Ad_{\Gamma}-\one)^{-1}(d \Gamma), \end{equation} where
$\Ad_{\Gamma}(X):=\Gamma X \Gamma^{\dagger}$ is the adjoint action.
In the eigenbasis of $\Gamma$, (i.e.
$\Gamma\ket{j}=\gamma_{j}\ket{j}$) it reads
\begin{equation}\label{Kappa}
\bk{j |\bm K_{}|k}=(\bm
K_{})_{jk}=\frac{(d\Gamma)_{jk}}{\gamma_{j}\gamma_{k}-1}=
-\frac{d\Omega_{k}}{2}\delta_{jk}+\tanh{\frac{\Omega_{j}-\Omega_{k}}{2}}
\bk{j|dk},
\end{equation}
where, in the second equality, we made use of the relation
$\gamma_{k}=\tanh{(\Omega_{k}/2)}$, which yields the following
diagonal $(d\Gamma)_{jj}=(1-\gamma_{j}^{2})d\Omega_{j}$ and
off-diagonal terms $(d\Gamma)_{jk}=(\gamma_{k}-\gamma_{j})\bk{j|
dk}$. This expression is well defined everywhere except for
$\gamma_{j}=\gamma_{k}=\pm 1$, where the Gaussian state $\rho$
becomes singular, i.e. it is not full rank. In this condition, the
expression~(\ref{Kappa}) for the generator $\bm G$ may become
singular. Nevertheless,  the boundness of the function
$|\tanh{\frac{\Omega_{j}-\Omega_{k}}{2}}|\le 1$ in Eq.~(\ref{Kappa})
shows that such a singularity is relatively benign. Thanks to this,
we can show that the condition $\gamma_{j}=\gamma_{k}=\pm 1$
produces, at most, removable singularities in the QGT
(see.~\cite{Safranek2017}). This allows the quantum geometric tensor
to be extended by continuity from the set of full-rank density
matrices to the submanifolds with $\gamma_{j}=\gamma_{k}=\pm 1$.

Knowing the expression for the parallel transport generator $\bm G$,
we can calculate the QGT by plugging
$G_{\mu}=\frac{1}{4}[\bm{\omega}^{T}K_{\mu} \bm{\omega} -
\Tr{(K_{\mu}\cdot \Gamma)]}$ into $Q_{\mu \nu}:=\Tr{(\rho
G_{\mu}G_{\nu})}$. Making use of Eq.~(\ref{Pf}) and exploiting the
antisymmetry of both $\Gamma$ and $\bm K$ under the exchange of
Majorana Fermion indices leads to~\cite{Carollo2018}
\begin{align}\label{eq:QGT1}
Q_{\mu\nu}&=\frac{1}{8}\Tr{[(\one-\Gamma)K_{\mu}(\one+\Gamma)K_{\nu}]}\\
&=\frac{1}{8}\sum_{jk}(1-\gamma_{j})(1+\gamma_{k})K^{jk}_{\mu}K^{kj}_{\nu}\nonumber\\
&=\frac{1}{8}\sum_{jk}\frac{(1-\gamma_{j})(1+\gamma_{k})}{(1-\gamma_{j}\gamma_{k})^{2}}(\partial_{\mu}\Gamma)_{jk}(\partial_{\nu}\Gamma)_{kj},\nonumber
\end{align}
where the last equality is obtained by plugging in
Eq.~(\ref{Kappa}). Let's have a closer look at the QGT in the limit
of $(\gamma_{j},\gamma_{k})\to \pm(1,1)$. The boundness of $\bm
K_{jk}$ and the multiplicative factors $(1\pm \gamma_{j})$
in Eq.~(\ref{eq:QGT1}) cause each term  to
vanish with $|\gamma_{j}|\to 1$. This means that the QGT has a well defined value in the
above limit, and we can safely extend by continuity the QGT to the
sub-manifolds $(\gamma_{j},\gamma_{k})=\pm(1,1)$.

The explicit expression of $Q_{\mu\nu}$ produces the following
results for the Bures metrics
\begin{equation} \begin{aligned}\label{eq:BuresGFS}
g_{\mu\nu}&=\Re({Q_{\mu\nu})}=\frac{1}{8}\Tr{(K_{\mu}K_{\nu}-\Gamma K_{\mu}\Gamma K_{\nu})}\\
&=-\frac{1}{8}\Tr{(\partial_{\mu}\Gamma K_{\nu})}\nonumber\\
&=\frac{1}{8}\sum_{jk}\frac{(\partial_{\mu}\Gamma)_{jk}(\partial_{\nu}\Gamma)_{kj}}{1-\gamma_{j}\gamma_{k}},\nonumber
\end{aligned} \end{equation} 
which in a parameter independent way reads \begin{equation}\label{eq:BuresGFS1}
dl^{2}=\sum_{\mu\nu}g_{\mu\nu}d\lambda_{\mu}d\lambda_{\nu}=\frac{1}{8}\Tr\left[d\Gamma\frac{1}{\one-\Ad_{\Gamma}}(d\Gamma)
\right] \,, \end{equation} 
where $d\Gamma:=\sum_{\mu} {\p_{\mu}\Gamma d\lambda_{\mu}}$.
The Eq.~(\ref{eq:BuresGFS1}) has been obtained by
substituting the formal solution~(\ref{eq:SolK}) of $\bm{K}$ in the
second equality of Eq.~(\ref{eq:BuresGFS})
. The above expression was
derived by a different procedure by Banchi et al.~\cite{Banchi2014}.
For the MUC the explicit expression is
\begin{equation} \begin{aligned}\label{MUCCov}
\mathcal{U}_{\mu\nu}&=2\Im({Q_{\mu\nu})}=\frac{i}{4}\Tr{(\Gamma[K_{\mu},K_{\nu}])}\\
&=\frac{i}{4}\sum_{jk}\frac{\gamma_{k}-\gamma_{j}}{(1-\gamma_{j}\gamma_{k})^{2}}(\partial_{\mu}\Gamma)_{jk}(\partial_{\nu}\Gamma)_{kj}\,.
\end{aligned} \end{equation} 
Also the above expression can be cast in a parameter-independent
way. Exploiting Eq.~(\ref{eq:SolK}) leads to
\begin{equation} \begin{aligned}
\bm{\mathcal{U}}&=\frac{i}{4}\Tr{(\Gamma \bm{K}\wedge
\bm{K})}\\\nonumber &=\frac{i}{4}\Tr{\left[\Gamma
\frac{1}{\one-\Ad_{\Gamma}}(d\Gamma) \wedge
\frac{1}{\one-\Ad_{\Gamma}}(d\Gamma) \right]}.
\end{aligned} \end{equation} 
The above Eq.~(\ref{eq:BuresGFS}) reproduces known results for
the thermal states~\cite{Zanardi2007a} and for pure
states~\cite{Zanardi2007c,Cozzini2007}. On the other hand, no
previous account of a close form expression of the Uhlmann curvature
were known in literature for the case of Fermionic Gaussian states,
they being in the equilibrium or out of equilibrium condition. As
expected, formula~(\ref{MUCCov}) reduces to the correct expression
in the case of pure states, provided that the appropriate matrix
$\Gamma$ is
considered~\cite{Carollo2005,Pachos2006,Hamma2006,Zhu2006}.
\subsection{Super-Extensivity of the (Generalised) Quantum Geometric Tensor}
The above results apply to the general class of Gaussian Fermionic
states. In this section, we will derive some results which are
specific to the quadratic Liouvillian models considered. The aim of
this section is to connect the kinematic properties embodied by the
quantum geometric tensor to the dynamical features of the underlying
physical model. More specifically we will derive a bound similar to
the one obtained in section~\ref{sec:SEQ}, which relates the
super-extensivity of the generalised  quantum geometric tensor to
the dissipative gap.

As in the case of Eq.~(\ref{eq:qgt}) also $Q(\bm{\lambda})$ defined
in Eq.~(\ref{eq:GQGT}) is a Hermitian non-negative matrix. Thus one has
\begin{equation}\label{eq:Qineq} |Q_{\mu\nu}|\le\| Q\|_{\infty}= \bm{u}^{\dagger}
\cdot Q  \cdot \bm{u}, \end{equation} where $\bm{u}=\{u_{\mu}\}_{\mu=1}^{{\rm
{dim}}{\cal M}}$, with $\bm{u}^{\dagger}\cdot \bm{u}=1$ is the
normalised eigenvector of $Q$ with the largest eigenvalue. One can
define the corresponding combination of parameters
$\bar{\lambda}:=\sum_{\mu=1}^{dim{\mathcal{M}}}u_{\mu}\lambda_{\mu}$
and the corresponding directional derivative
$\bar{\p}:=\p/\p\bar{\lambda}$, so that
\begin{equation}
  \bar{\p} \Gamma=\sum_{\mu}(\partial_{\mu}\Gamma) u_{\mu},
\end{equation}
   and from Eq.
(\ref{pert}) and the above inequality~(\ref{eq:Qineq}) we get
\begin{align}\label{eq:QGT2}
\|Q\|_{\infty}
=\frac18\Tr\left[(\one-\Gamma)\frac1{\one-\Ad_{\Gamma}}(\bar{\p}
\Gamma) (\one+\Gamma)\frac1{\one-\Ad_{\Gamma}}(\bar{\p} \Gamma)
\right]\,.
 \end{align}
Let's express Eq.~\eqref{eq:QGT2} in a form amenable to further
manipulations, by employing the vectorization isomorphism. Notice
that, under such an isomorphism 
\begin{equation} 
\Ad_{\Gamma}(A):=\Gamma
A\Gamma^{\dagger}\stackrel{ \textrm{vec}}{\longrightarrow}
\,\left(\Gamma \otimes \Gamma^T\,\right) \V{A} = - \left(\Gamma
\otimes \Gamma\right) \V{A}\,,\nonumber 
\end{equation} 
and Eq.~(\ref{eq:QGT2}) becomes
\begin{align}
\|Q\|_{\infty}=&\,\frac18 \,\,\V{\bar{\p}\Gamma}^{\dagger}\cdot\left(\frac{(\one+\Gamma)\otimes(\one+\Gamma)}{\one+\Gamma\otimes\Gamma} \right) \cdot \V{\bar{\p}\Gamma}\nonumber\\
\le& P_{\Gamma} \| \V{\bar{\p}\Gamma}\|^{2}\nonumber\\
\le& 2n P_\Gamma\,\,\| \bar{\p}\Gamma\|^2_\infty~, \label{eq:Qbound}
\end{align}
where 
\begin{equation} 
P_{\Gamma}:= \frac18 \left
\|\frac{(\one+\Gamma)\otimes(\one+\Gamma)}{\one+\Gamma\otimes\Gamma}
\right\|. 
\end{equation} 
In the first inequality of Eq.~(\ref{eq:Qbound}), we used
the definition of operator norm, while in the second we have
employed the fact that $\|\V A \|=\|A\|_2$ and $\|A\|_2\le\sqrt{2n}
\|A\|_\infty.$

The bound~(\ref{eq:Qbound}) still is not specific to dissipative
quadratic Liouvillian. In order to relate Eq.~(\ref{eq:Qbound}) to
the dynamical properties of the Liouvillian \eqref{e.lindiag} one
could differentiate Eq.~\eqref{syl1}
\begin{equation}
\V {\p_{\mu}\Gamma}= \hat{X}^{-1} \V{\p_{\mu}Y}- \hat{X}^{-1}
\p_{\mu}\hat{X}  \V \Gamma~. \label{eq:dGamma}
\end{equation}
Through the above equation, one realises that $\p_{\mu}{\Gamma}$ is
the solution of a continuous Lyapunov equation, similar
to~(\ref{eq:Lyap}), which provides a convenient way to calculate it
numerically once $\Gamma$, $X$, $Y$ and $\p_{\mu}X$ are known, i.e.
\begin{align}\label{eq:dGLyap}
X\, \left(\p_\mu \Gamma\right) + \left(\p_\mu \Gamma \right)\,X^T =
\partial_\mu Y -\left(\p_\mu X\right)\, \Gamma - \Gamma\,
\left(\p_\mu X^T\right)~.
\end{align}
Taking norms in Eq.~(\ref{eq:dGamma}) leads to \footnote{
$\|O\|_\infty:= \sup_{\|v\|=1}\|Ov\|=$largest singular value of $O;$
Notice that $\|O v\|\le\|O\|_\infty \|v\|.$
$\|O\|^2_2={{\rm{Tr}}(O^\dagger O) }=$sum of the squares of the
singular values of $O.$ }
\begin{align}
\|\p_{\mu}\V{\Gamma}\|  &\le
\|\hat{X}^{-1}\|_\infty(\|\p_\mu\V{Y}\| + \|\p_\mu\hat{X}\|_\infty
\|\V{\Gamma}\|)\nonumber\\ &=
 \|\hat{X}^{-1}\|_\infty(\|\p_\mu{Y}\|_2 + \|\p_\mu\hat{X}\|_\infty \|{\Gamma}\|_2)\nonumber \\ & \le
\sqrt{2n} \|\hat{X}^{-1}\|_\infty(\|\p_\mu{Y}\|_\infty +
\|\p_\mu\hat{X}\|_\infty \|{\Gamma}\|_\infty) \nonumber\\&\le
\sqrt{2n} \|\hat{X}^{-1}\|_\infty(\|\p_\mu{Y}\|_\infty +
\|\p_\mu\hat{X}\|_\infty)~,
\end{align}
where, relations $\|\V A \|=\|A\|_2$, $\|A\|_2\le\sqrt{2n}
\|A\|_\infty$ and $\|\Gamma\|_\infty\le1$ have been employed; the
latter following from Eq.~(\ref{eq:corrQ}).

Essentially, the upper bound for $\p_{\mu}\Gamma$, obtained above,
only depends on the system parameters and their differentials,
i.e., $X,\,dX$ and $Y,\,dY$. Finally one derives the following
bound
\begin{equation}\label{eq:dGammaB}
\|\p_{\mu}\V{\Gamma}\|^2\le 2 n \|\hat{X}^{-1}\|_\infty^2
(\|\p_{\mu}{Y}\|_\infty + 2\|\p_{\mu}{X}\|_\infty)^2,
\end{equation}
where the relation $\|\p_{\mu}\hat{X}\|_\infty=\|
\p_{\mu}X\otimes{\mathbf{1}} +{\mathbf{1}}\otimes
\p_{\mu}X\|_\infty\le 2\|\p_{\mu}X\|_\infty$ has been used.

Now, we finally wrap all the latest results around: by plugging
Eq.~(\ref{eq:dGammaB}) in~(\ref{eq:QGT2}) and by employing
relation~(\ref{eq:Deltaxhat}) and Proposition~\ref{prop:Delta} of
section~\ref{sec:LSpec}, one eventually obtains the following upper
bound which relates the behaviour of $\Delta(n)$ to $|Q_{\mu\nu}|$, i.e.
\begin{align}
  \frac{|Q_{\mu\nu}|}n \le 2\frac{P_{\Gamma}}{\Delta_{\mathcal{L}}^2}\,
  \left(\|dY\|_\infty+2\|dX\|_\infty\right)^2~.
  \label{eq:BoundGQGT}
\end{align}
The latter is the relation that was anticipated earlier: it is the
dissipative analogue of the inequality for zero-temperature QPT
derived in section~\ref{sec:SEQ}, where it was shown that
super-extensivity of the quantum geometric tensor $Q_{\mu\nu}$
implies the vanishing of the energy gap~\cite{CamposVenuti2007} and
the outset of a phase transitions. The above bound connects the
\emph{generalised} QGT to the dynamical feature of the dissipative
Liouvillian model. It is indeed a relation between the kinematics
expressed by the geometry of the NESS and the dynamics, embodied by
the dissipative gap. Specifically, this bound shows that, if
$P_{\Gamma}\simeq \mathcal{O}(1)$, a scaling of $|Q_{\mu\nu}|\propto
n^{\alpha+1}$ entails a dissipative gap that vanishes at least as
$\Delta_{\mathcal{L}}\propto n^{-{\alpha}/2}$, establishing a link
between the dynamical properties of the NESS-QPTs  and the geometric
property $Q_{\mu\nu}$.\\\indent Needless to say, the above bound on
the QGT immediately determines bounds on both the Bures metric and
on the mean Uhlmann geometric phase, \bea
|g_{\mu\nu}| &= |\Re Q_{\mu\nu}|\le |Q_{\mu\nu}|\,,\\
|\mathcal U_{\mu\nu}|&=|\Im Q_{\mu\nu}|\le |Q_{\mu\nu}|\,, \eea
whose scaling properties can thus be related to the NESS-QPTs. 

It is important, however, to stress that $\Delta_{\mathcal{L}}$ is
an entirely different quantity from the Hamiltonian gap, linked to
the scaling of $Q_{\mu\nu}$ for zero-temperature-QPTs.  A complete
understanding of the relation between the Liouvillian gap and the
Hamiltonian gap ruling equilibrium QPTs  is still lacking. Notice,
indeed, that in the non-dissipative case the spectrum $ \textrm{Sp}(X)$ is purely
imaginary. From the perspective of the Liouvillian dynamics,
this implies an identically vanishing dissipative gap
$\Delta_{\mathcal{L}}\equiv 0$. This contrasts with the na\"ive
attempt of formulating a general equilibrium/non-equilibrium QPT
criterion which levels the dissipative gap to the same status of a
Hamiltonian gap in standard QPTs.  Moreover, unlike in equilibrium
QPTs~\cite{Hastings2004}, where super-extensivity is  a {\em
sufficient} condition for criticality at $T=0$, in the dissipative
case $|Q_{\mu\nu}|=\mathcal O(n^{1+\alpha}), \,(\alpha>0)$ only
implies $\Delta_{\mathcal L}=\mathcal O(n^{\alpha/2})$, but it does
not necessarily imply criticality. Indeed, in NESS-QPT, a closure of
the gap generally neither is implied by criticality nor implies
it~\cite{Prosen2008,Carollo2018}.

On the other hand, one can see that in the case of translationally
invariant models, where a notion of criticality in the
thermodynamic limit is easier to
handle~\cite{Eisert2010,Honing2012}, further progress can be done.
There, the problem of relating the geometric properties to the
dynamical features of the model can be bypassed, in favour of a
direct relation between the geometry and the divergence of the
correlation length~\cite{Carollo2018}.

Note that in the non-diagonalisable case a correction to
Eq.~(\ref{eq:Deltaxhat}) should be considered, which adds an extra
polynomial dependence in
\eqref{eq:prop1}\cite{Prosen2010a,Banchi2014}. However, this
variation does not affect the qualitative and quantitative
consequences of the bound \eqref{eq:BoundGQGT}: super-extensivity of
the quantum geometric tensor entails a vanishing
 Liuvillean gap.

\section{Translationally Invariant Models}
Before turning to specific models where the above general
considerations can be exemplified, we would like to draw the
attention to an important subclass of quadratic Liouvillian
Fermionic models, namely those enjoying the translational invariance
symmetry. In a translationally invariant system one can employ the
whole wealth of powerful tools stemming out of the Fourier transform
and work directly in the thermodynamic limit. This enables one to
quantitatively define criticality in terms of singularities in the
quasi-momentum space, thereby secluding the kinematics of the
NESS-QPTs  from the dynamical properties of the model.
The most natural notion of many-body criticality is in terms of diverging correlation length, which in a translationally invariant system is relatively straightforward to handle. This way of defining criticality enables one to bypass the difficulties arising from the ambiguous relation between NESS-QPTs  and the vanishing dissipative gap.\\
The object of investigation is the covariance matrix, which in a translationally invariant system can be conveniently studied through its Fourier components. It is the non-analytical behaviour in the Fourier basis which conveys information on the long-wavelength limit, i.e. on the divergence of the correlation length.\\
Consider an explicit translationally invariant $d$-dimensional
lattice of Fermions located at sites $\vv{r}\in\mathbb{Z}^{d}_{L}$,
and assume finite (or quasi-finite) range interaction. The system
size is $n=L^{d}$, and subsequently, one takes the thermodynamic
limit $L\to\infty$. One can define the covariance matrix over a
discrete quasi-momentum space. However the considerations on the
long-wavelength limit that will follow truly, apply only at the
thermodynamic limit: hence divergences of correlation lengths
manifest genuine quantum many-body effects.

To emphasise the translational property, let us relabel the Majorana
Fermions as 
\begin{equation}\label{eq:MFTI} 
\bm{\omega}_{\vv{r}}=\begin{pmatrix}
\omega_{r,1} \\ \omega_{r,2} \end{pmatrix}, \quad \textrm{ with  }
\left\{\begin{array}{l}\omega_{r,1}=c_{\vv{r}}+c_{\vv{r}}^{\dagger}
\\\omega_{r,2}=i(c_{\vv{r}}-c_{\vv{r}}^{\dagger})\end{array}\right.\,
\end{equation} 
where $\omega_{r,\beta}$, ($\beta=1,2$) are the two flavours of
Majorana Fermions on each site $\vv{r}$, and $c_{\vv{r}}$ and
$c_{\vv{r}}^{\dagger}$ are the annihilation and creation operator,
respectively, of the corresponding ordinary Fermion. Due to
translational invariance, the Hamiltonian may be written as
 \begin{equation}\label{eq:HTI}
 \mathcal{H}=\sum_{\vv{r},\vv{s}}\bm{\omega}^{T}_\vv{r} h(\vv{r}-\vv{s})\bm{\omega}_{\vv{s}},
\end{equation}
 where $h(\vv{r})=h(-\vv{r})^{\dagger}=h(\vv{r})^{*}$ are $2\times2$
imaginary matrices. Similarly the jump operators can be expressed as
 \begin{equation}\label{eq:JumpTI}
\Lambda_{\alpha}(\vv{r})=\sum_{\vv{s}}
\bm{l}_{\alpha}^{T}(\vv{s}-\vv{r})\bm{\omega}_{\vv{s}}, \end{equation} where
$\bm{l}_{\alpha}(\vv{r})$ are $2$-dimensional complex arrays.
Accordingly, the bath matrix are written as 
\begin{equation}
 [M]_{(\vv{r},\beta)(\vv{s},\beta')}=[m(\vv{r}-\vv{s})]_{\beta
\beta'}\qquad (\beta,\beta'=1,2) \end{equation} where
$m(\vv{r})=m^{\dagger}(-\vv{r})$ are the $2\times2$ matrices
$m(\vv{r}):=\sum_{\alpha,\vv{s}}\bm{l}_{\alpha}(\vv{s}-\vv{r})\otimes\bm{l}_{\alpha}^{\dagger}(\vv{s})$.

Since both Hamiltonian and bath matrix are circulant, so it is the
correlation matrix of the unique steady state solution, which reads 
\begin{equation}
 [\Gamma]_{(\vv{r},\beta)(\vv{s},\beta')}=[\gamma(\vv{r}-\vv{s})]_{\beta
\beta'}:=\frac12\Tr\left(\rho
[\omega_{\vv{r},\beta},\omega_{\vv{s},\beta'}]\right),\qquad \qquad
(\beta,\beta'=1,2). 
\end{equation} 
The latter can be conveniently expressed in
terms of its Fourier component, called the covariance symbol, as
\begin{align*}
\tilde{\gamma}{(\vv{\phi})} := \sum_{\vv{r}}
\gamma{(\vv{r})}e^{-i\vv{\phi} \cdot \vv{r}},
\end{align*}
where $\vv{\phi}\in [-\pi,\pi)$. In terms of the symbol functions,
the continuous Lyapunov equation reduces to a set of $2\times2$
matrix equations
\begin{equation}\label{CLEphi}
    \tilde{x}(\vv{\phi}) \tilde{\gamma}(\vv{\phi})+\tilde{\gamma}(\vv{\phi})\tilde{x}^{T} (-\vv{\phi})= \tilde{y}(\vv{\phi}),
\end{equation}
 where  $\tilde{x}(\vv{\phi})=2[ 2i \tilde{h}(\vv{\phi})+\tilde{m}(\vv{\phi})+\tilde{m}^{T}(-\vv{\phi})]$ and $\tilde{y}(\vv{\phi})=-4[\tilde{m}(\vv{\phi})-\tilde{m}^{T}(-\vv{\phi})]$ are the symbol functions of $X$ and $Y$, respectively, and $\tilde{h}(\phi)$, $\tilde{m}(\phi)=\sum_{\alpha} \tilde{\bm{l}}_{\alpha}\otimes \tilde{\bm{l}}_{\alpha}^{\dagger}$ and $\tilde{\bm{l}}_{\alpha}(\phi)$ are the Fourier components of $h(r)$, $m(r)$ and $\bm{l}_{\alpha}(r)$, respectively. Notice that $\tilde{m}(\phi)=\tilde{m}(\phi)^{\dag}=\sum_{\alpha} \tilde{\bm{l}}_{\alpha}\otimes \tilde{\bm{l}}_{\alpha}^{\dagger}\ge 0$  (positive semidefinite matrix).

 The spatial correlations between Majorana Fermions are then recovered from the inverse Fourier transform of the covariance symbol
 \begin{equation}\label{Corr}
 \gamma(\vv{r})=\frac{1}{(2\pi)^{d}}\int_{\mathbb{T}^{d}}d^{d}\phi \tilde{\gamma}(\vv{\phi})e^{i\vv{\phi} \cdot \vv{r}}.
 \end{equation}

Following~\cite{Eisert2010,Honing2012}, here we will define
criticality by the divergence of correlation length, which is
defined as \begin{equation}\label{eq:CorLen} \xi^{-1} := - \lim_{|r|\to \infty}
\frac{\ln {||\gamma(r)||}}{|r|}. \end{equation} In the thermodynamic limt,
the divergence may only arise as a consequence of the non-analytical
dependence of $\gamma(r)$ on the system parameters. Let's confine
ourselves to the case of a one-dimensional Fermionic chain.
 In order to derive informations on the large distance behaviour of the correlations, it is convenient to express the above integral~(\ref{Corr}) in the complex plane through the analytical continuation $e^{i\phi}\to z$. This results in the following expression for the correlation function

\begin{equation}
  \gamma(\vv{r})= \sum_{\bar{z}\in S_{1}}  \textrm{Res}_{\bar{z}}[z^{r-1}\tilde{\gamma}{(z)}],
\end{equation}
  where $  \textrm{Res}_{\bar{z}}$ indicates the residues of the poles inside the unit circle $S_{1}:=\{ z,\textrm{such that } |z|\le1\}$. Since $\tilde{\gamma}(z)$ is the solution of a finite dimensional matrix Eq.~(\ref{CLEphi}), it may only possess simple poles. Thus, the above expression may become singular only when an isolated pole of $\tilde{\gamma}(z)$ approaches the unit circle from the inside~\cite{Eisert2010,Honing2012}. This may happen for some specific critical values $\lambda=\lambda_{0}\in\mathcal{M}$. As $\lambda$ approaches $\lambda_{0}$ the correlation length $\xi$ diverges. One can show that the long wave-length behaviour is governed by the pole closest to unit circle $|\bar{z}_{0}|$, and indeed the correlation length is given by
\begin{equation}
\xi=\ln|\bar{z}_{0}|. 
\end{equation}

\subsection{Mean Uhlmann Curvature and Criticality in Translationally Invariant Models}
Let's now turn to the geometric properties of translationally
invariant models at criticality. In particular let's
concentrate on the mean Uhlmann curvature. We will show that the MUC
is sensitive to the criticality, but only in the sense of a truly
diverging correlation length. Indeed one can show that the Uhlmann
curvature is insensitive to the vanishing of the dissipative gap, if
the latter, as it may happen, is not accompanied by a diverging
correlation length. In this sense, the Uhlmann curvature confirms
its role as a witness of the purely kinematic aspects of the
criticality, and it is only indirectly affected by the dynamical
features of the NESS-QPTs. 

Thanks to the translational symmetry, one can exploit the formalism
of Fourier transform and derive a quite compact expression of the
MUC. By applying the convolution theorem on 
Eq.~(\ref{MUCCov}), one obtains the following expression for
the MUC \emph{per site}
\begin{equation}\label{iota}
\bar{\mathcal{U}}_{\mu\nu}:=\lim_{n\to\infty}
\frac{\mathcal{U}_{\mu\nu}}{n}=\frac{1}{(2\pi)}\int_{-\pi}^{\pi}d\phi
\,\,u_{\mu\nu}(\vv{\phi}),
\end{equation}
where
\begin{equation}\label{Uphi}
u_{\mu\nu}(\phi):=\frac{i}{4}\Tr\{\tilde{\gamma}{(\vv{\phi})}[\kappa_{\mu}{(\vv{\phi})},\kappa_{\nu}{(\vv{\phi})}]\}=\frac{i}{4}\Tr\{\kappa_{\nu}{(\vv{\phi})}[\tilde{\gamma}{(\vv{\phi})},\kappa_{\mu}{(\vv{\phi})}]\},
\end{equation}
In the above expression,  $\kappa_{\mu}{(\vv{\phi})}$  is the symbol
function of $K_{\mu}$, and it can be found as the operator solution
of the $2\times2$ discrete Lyapunov equation
\begin{equation}\label{DLEphi}
\partial_{\mu}\tilde{\gamma}{(\vv{\phi})}=\tilde{\gamma}{(\vv{\phi})}\kappa_{\mu}{(\vv{\phi})} \tilde{\gamma}{(\vv{\phi})} - \kappa_{\mu}{(\vv{\phi})}.
\end{equation}
In the eigenbasis of $\tilde{\gamma}(\phi)$, with eigenvalues
$\tilde{\gamma}_{j}$, the explicit solution of~(\ref{DLEphi}) reads
\begin{equation}
(\kappa_{\mu}(\phi))_{jk}=\frac{(\partial_{\mu}\tilde{\gamma}(\phi))_{jk}}{1-\tilde{\gamma}_{j}\tilde{\gamma}_{k}}.
\end{equation}
Notice that the diagonal terms $(\kappa_{\mu}(\phi))_{jj}$ provide
vanishing contributions to eq.~(\ref{Uphi}), bacause they commute with
$\tilde{\gamma}{(\vv{\phi})}$. Hence, Eq.~(\ref{Uphi}) can be cast
in the following basis independent form
\begin{equation}\label{Uphi1}
u_{\mu\nu}(\phi) =\left\{ \begin{array}{ll }
      \frac{i}{4}\frac{\Tr\{\tilde{\gamma}{(\vv{\phi})}[\partial_{\mu}\tilde{\gamma}{(\vv{\phi})},\partial_{\nu}\tilde{\gamma}{(\vv{\phi})}]\}}{(1-\Det{\tilde{\gamma}{(\vv{\phi})}})^{2}}& \Det{\tilde{\gamma}{(\vv{\phi})}}\neq 1\\
       0 & \Det{\tilde{\gamma}{(\vv{\phi})}}=1
\end{array} \right. .
\end{equation}
Notice that the condition $\Det{\tilde{\gamma}{(\vv{\phi})}}=1$ is
equivalent to having two eigenvalues of correlation matrix equal to
$(\gamma_{i},\gamma_{k})=\pm(1,1)$. This corresponds to the
situation, already discussed in section~\ref{sec:MUCGFS}, in which
two eigenmodes of the Gaussian state are pure. As already mentioned with regard to Eq.~(\ref{eq:QGT1}), such extremal
values cause no singularity in MUC, but they rather result in a
vanishing contribution to the MUC.

In the following, we will demonstrate that a singularity of
$\bar{\mathcal{U}}$ signals the occurrence of a criticality.
Specifically, employing the analytical extension in the complex
plane of $u_{\mu\nu}{(\phi)}$ leads to
\begin{equation}\label{eq:Uz}
\bar{\mathcal{U}}_{\mu\nu}=\sum_{\bar{z}'\in S_{1}}
 \textrm{Res}_{\bar{z}'}[z^{-1} u_{\mu\nu}(z)].
\end{equation}\indent
Notice that $u_{\mu\nu}(z)$ has at most isolated poles, due to its
rational dependence on $z$. Assume that, as
$\lambda\to\lambda_{0}\in\mathcal{M}$, a pole $\bar{z}_{0}$ of
$u_{\mu\nu}(z)$ approaches the unit circle from inside,  which is
the only condition under which $\bar{\mathcal{U}}$ is singular in
$\lambda_{0}$. One can show that, whenever a pole $\bar{z}_{0}$ of
$u_{\mu\nu}(z)$ approaches the unit circle, also a pole $\bar{z}$ of
$\tilde{\gamma}(z)$ approaches the same value, causing the
correlation length to diverge. Therefore the singular behaviour of
the Uhlmann phase necessarily represents a sufficient criterion for
a NESS-QPTs.  Notice also that such criticalities are necessarily
accompanied by the closure of the dissipative gap, however, the
converse is in general not true. Indeed, a singularity in the MUC
may only arise as the result of criticality and are otherwise
insensitive to a vanishing dissipative gaps.

Let's now prove, that in translationally invariant models a
vanishing dissipative gap is a \emph{necessary condition} for
criticality.
\begin{Proposition}
If there exists a pole $\bar{z}_{0}(\lambda)$ of
$\tilde{\gamma}(z)$, smoothly dependent of system parameters
$\lambda\in\mathcal M$,   such that $\lim_{\lambda\to\lambda_{0}}
|\bar{z}_{0}|=1$, then \\ $\Delta:=2 \min_{|z|=1,j}\Re{x_{j}(z)}=0$
for $\lambda=\lambda_{0}$.
\end{Proposition}
\begin{proof}
Under the vectorising isomorphism, $A=a_{jk}\ket{j}\bra{k}\to
 \textrm{vec}{(A)}:=a_{jk}\ket{j}\otimes\ket{k}$, the continuous
Lyapunov Eq.~(\ref{CLEphi}) can be written as \begin{equation} \hat{X}(z)
\V{\tilde{\gamma}(\vv{z})}=\V{y(\vv{z})}, \end{equation} where
$\hat{X}(z):=x(z)\otimes\one+\one\otimes x(z^{-1})$. When
$\Det{\hat{X}(z)}\neq0$, the unique solution of the symbol function
is found simply as
\begin{equation}\label{sol}
\V{\gamma(z)}=\frac{\V{\eta(z)}}{d(z)},\quad \textrm{ where }
\V{\eta}:= \textrm{adj}(\hat{X}) \V{y}.
\end{equation}
Here $ \textrm{adj}(\hat{X})$ stands for the adjugate matrix of
$\hat{X}$ and $d(z):=\Det{\hat{X}(z)}$. The point in writing the
solution in this form is that, by construction, $x(z)$ and $y(z)$
are polynomials in $z$ and $z^{-1}$ with coefficients smoothly
dependent on the system parameters. Moreover, since determinant and adjugate
matrix are always polynomial functions of a matrix coefficients, it
results that also $\eta(z)$ and $d(z)$ will be two polynomials in
$z$ and $z^{-1}$.  Hence, $\tilde{\gamma}{(z)}$'s poles are to be
found among the roots $\bar{z}$ of $d(z)=0$. Thus, a
\emph{necessary} condition for criticality is that, for
$\lambda\to\lambda_{0}$, a given root $\bar{z}$ approaches the unit
circle $S_{1}$. This clearly means that for $\lambda=\lambda_{0}$
there must exist $\bar{z}_{0}$ such that $|\bar{z}_{0}|=1$ and
$d(z)=\Det{\hat{X}(\bar{z}_{0})}=0$, which implies a vanishing
dissipative gap $\Delta:=2 \min_{|z|=1,j}\Re{x_{j}({z})}$, where
$x_{j}({z})$'s are the eigenvalues of
$\tilde{x}(z)$~\cite{Prosen2010a}.
\end{proof}

On the other hand, the converse of the above proposition is not
true: \emph{a vanishing dissipative gap is not a sufficient
condition for criticality}, but only necessary. Indeed, it may well
be the case that all those roots $\bar{z}$ which approach the unit
circle as $\lambda\to\lambda_{0}$ are removable singularities of
(\ref{sol}). This would result in a finite correlation length, even
when $\Delta\to0$. The fact that this actually happens can be
readily checked with the example in section~\ref{sec:Example}.

We will next show that a singular behaviour of $\mathcal{U}$ with
respect to the parameters is a sufficient condition for criticality.
First of all, notice, from the Eq.~(\ref{Uphi1}), that
$u(\phi)$ may depend on the dynamics only through $\tilde{\gamma}$,
hence any closure of the gap which does not affect the analytical
properties of $\tilde{\gamma}$ cannot result in a singular behaviour
of $\mathcal{U} $(see also Lemma 2 in the following). We will just
need to show that a necessary condition for a singular behaviour of
$u(\phi)$ is $\Delta=0$.

Indeed, let's now show that the poles of $u_{\mu\nu}(z)$ with
$|z|=1$ are to be found only among the roots of $d(z)$. Assuming
$d(z)\neq0$, and plugging the unique solution~(\ref{sol}) into
Eq.~$(\ref{Uphi1})$ leads to
\[
u_{\mu\nu}(z) = \frac{N(z)}{D(z)}=\frac{i}{4}
\frac{d(z)\Tr\{\eta{(z)}[\partial_{\mu}\eta{(z)},\partial_{\nu}\eta{(z)}]\}}{(d(z)^{2}- \textrm{Det}{\eta{(z)}})^{2}},
\]
where the numerator $N(z)$ and denominator $D(z)$ are polynomials in
$z$ and $z^{-1}$ with smooth dependence on $\lambda$'s.

We will demonstrate the following: \emph{(i)} that all roots of
$d(z)$ such that $|z|=1$ are also roots of $D{(z)}$, and \emph{(ii)}
that any other roots of $D(z)$, such that $|z|=1$, are not poles of
$u_{\mu\nu}(z)$.  For the statement \emph{(i)}, it is just enough to
prove the following lemma.

\begin{Lemma}
If $d(z)=0$ with $|z|=1$, then $\eta(z)=0$.
\end{Lemma}
\begin{proof}
For $|z|$=1, let's write explicitly $z=e^{i\phi}$. It is not hard to
show that from its definition, the matrix $\tilde{x}(\phi)$ enjoys
the following property
$\tilde{x}(\phi)^{\dagger}=\tilde{x}(-\phi)^{T}$. Correspondingly,
the eigenvalues of $\hat{X}$ are $x_{j}+x_{k}^{*}$ with $j,k=1,2$,
where $x_{j}$  are the eigenvalues of $\tilde{x}(\phi)$. Since $\Re
x_{j}\ge 0$, $\Det{\hat{X}}=0$ implies that there must exist an
eigenvalue $x_{0}$ of  $\tilde{x}(\phi)$ with vanishing real part,
hence $\Delta=2\min_{j}\Re\, x_{j} = 2\Re\, x_{0} = 0$.  If
$\ket{0}$ is the eigenstate of $\tilde{x}(\phi)$ with eigenvalue
$x_{0}$, then 
\begin{equation}
 x_{0}+x_{0}^{*}=\bra{0}\tilde{x}(\phi)+\tilde{x}(-\phi)^{T}\ket{0}=4\bra{0}\tilde{m}(\phi)+\tilde{m}(-\phi)^{T}\ket{0}=0~,
\end{equation}
  where in the second equality we used the definition $\tilde{x}(\vv{\phi}):=2[ 2i \tilde{h}(\vv{\phi})+\tilde{m}(\vv{\phi})+\tilde{m}^{T}(-\vv{\phi})]$ and the antisymmetry $\tilde{h}(\phi)=-\tilde{h}(-\phi)^{T}$.
 From the non-negativity of the $\tilde{m}(\phi)$ matrices, it follows that $\bra{0}\tilde{y}(\phi)\ket{0}=-4\bra{0}\tilde{m}(\phi)-\tilde{m}(-\phi)^{T}\ket{0}=0$.

In lemma~\ref{lemma:stability} of~\ref{app:Spec} it is
shown that when $2\Re x_{0}=0$, the geometric multiplicity of
$x_{0}$ is equal to its algebraic multiplicity, hence the $2\times2$
matrix $\tilde{x}(\phi)$ is diagonalisable. Then, let $\ket{j}$ be
the set of eigenstates with eigenvalues $x_{j}$. In the eigenbasis
$\ket{j}\otimes\ket{k}$, $j,k=0,1$ the adjugate matrix has the
following diagonal form, \be\nonumber  \textrm{adj}(\hat{X})=2
\begin{pmatrix}
|x_{0}+x_{1}^{*}|^{2} \Re(x_{1}) &0&0&0\\
0& 2 (x_{0}+x_{1}^{*})\Re(x_{1} x_{0})&0&0\\0&0& 2
(x_{1}+x_{0}^{*})\Re(x_{1} x_{0})&0\\0&0&0&|x_{1}+x_{0}^{*}|^{2}
\Re(x_{0})
\end{pmatrix}
\end{equation} and due to $\Re\, x_{0} =0$ all elements, but
$\bra{0,0} \textrm{adj}(\hat{X})\ket{0,0}$, vanish. On the other hand,
the element $\V{\tilde{y}}_{00}:=\bra{0}\tilde{y}\ket{0}=0$,
implying $\V{\eta}= \textrm{adj}(\hat{X}) \V{y}=0$.
\end{proof}
To prove statement \emph{(ii)}, we just need the following
proposition.
\begin{Proposition}
If $\bar{z}_{0}$ is a root of $D(z)$ with $|\bar{z}_{0}|=1$, and
$d(\bar{z}_{0})\ne 0$, then $u_{\mu\nu}(z)$ is analytic in $z_{0}$.
\end{Proposition}
\begin{proof}
Let $\bar{z}_{0}$ be a root of $D(z)$ with $|\bar{z}_{0}|=1$, with
the assumption that $d(\bar{z}_{0})\neq 0$. Notice that whenever
$d(z)\neq0$, $\tilde{\gamma}(z)$ in~(\ref{sol}) is the unique
solution of the Lyapunov Eq.~(\ref{CLEphi}). As such, it is
analytic in $z$ (and smoothly dependent on $\lambda$'s). Since 
\begin{equation}
 D(z):= (d(z)^{2}- \textrm{Det}{\eta{(z)}})^{2}
=d(z)^{4}[1-\Det\tilde{\gamma}(z)]^{2}~, \end{equation} we obviously have
$\Det{\tilde{\gamma}(\bar{z}_{0})}=1$. Just observe that if
$\gamma(z)$ is an analytic, smoothly dependent on the system
parameters $\lambda\in \mathcal{M}$ , $u_{\mu\nu}(z)$ may be
singular in $\bar{z}_{0}$ only if
$\Det{\tilde{\gamma}(\bar{z}_{0})}=1$. Assume then
$\Det{\tilde{\gamma}(\bar{z}_{0})}=1$, then either
$\gamma(\bar{z}_{0})=\pm\one$. Without loss of generality, we can
write $\tilde{\gamma}(z)=\one+T \,(z-\bar{z}_{0})^{2n}+
\mathcal{O}(z-\bar{z}_{0})^{2n}$, $n\in\mathbb{N}$, where
$T=T^{\dagger}$ is the first non-vanishing term of the Taylor
expansion of $\tilde{\gamma}(z)-\one$. The fact that this term must
be of even order ($2n$) is due to the positive semi-definiteness of
the $\one-\tilde{\gamma}(z) $ for $z\in S_{1}$. By expressing the
$2\times2$ matrix $T$ in terms of Pauli matrices, $T=t_{0}\one+
\bm{t}\cdot\bm{\sigma}$, where
$\bm{\sigma}:=(\sigma_{x},\sigma_{y},\sigma_{z})^{T}$, $t_{0}\in
\mathbb{R}$ and $\bm{t}\in \mathbb{R}^{3}$, the positive
semi-definiteness condition above reads: $t_{0}<0$ and
$||\bm{t}||\le |t_{0}|$. Plugging the Taylor expansion
in~(\ref{Uphi1}) and retaining only the first non-vanishing terms,
yields
\[
u_{\mu\nu}(z)=-\frac{1}{4}\frac{\bm{t}\cdot(\partial_{\mu}\bm{t}\wedge\partial_{\nu}\bm{t})}{t_{0}^{2}}(z-\bar{z}_{0})^{2n}
+ o(z-\bar{z}_{0})^{2n}.
\]
\end{proof}

We have thus proven that a non-analycity of $u_{\mu\nu}(z)$ in
$\bar{z}_{0}\in S_{1}$ is necessarily due to a pole $\bar{z}$ of
$\tilde{\gamma}(z)$ approaching $\bar{z}_{0}$, as
$\lambda\to\lambda_{0}$, resulting in a diverging correlation
length. Therefore, a singular behaviour of $\bar{\mathcal{U}}$ in
the manifold $\mathcal{M}$ is a sufficient criterion for
criticality.

\section{Applications}
\subsection{Vanishing dissipative gap without criticality}\label{sec:Example}
The primary scope of this subsection is not discussing a model which
may be relevant per se, rather it serves to illustrate in a simple
translationally invariant case the ambiguous relation between
criticality and vanishing dissipative gap. As a byproduct, one may
also appreciate the sensitivity of the Uhlmann curvature to the
criticality and its unresponsiveness to the gap. Specifically,
 in this section we will describe an example of a 1D Fermionic system in
which the closure of the dissipative gap does not necessarily lead
to a diverging correlation length. Consider a chain of Fermions on a
ring geometry, driven uniquely by an engineered reservoir, i.e. with
\emph{no Hamiltonian}. The reservoir is described by the following
set of jump operators
\[
\Lambda(r)=[(1+\lambda)\bm{l}_{0}^{T}\bm{\omega}_{r}+\bm{l}_{1}^{T}\bm{\omega}_{r+1}+\lambda\bm{l}_{2}^{T}\bm{\omega}_{r+2}]/n(\lambda),
\]
where  $r=1,\dots,n$, $\bm{l}_{0}=(\cos{\theta},-\sin{\theta})^{T}$, $\bm{l}_{1}=\bm{l}_{2}=i(\sin{\theta},\cos{\theta})^{T}$, and $n(\lambda)=4(\lambda^{2}+\lambda+1)$, with $\lambda\in\mathbb{R}$, $\theta=[0,2\pi)$. 
This is a simple extension of a model introduced
in~\cite{Bardyn2013}, which, under open boundary conditions, shows a
dissipative topological phase transitions for $\lambda=\pm1$. In the
thermodynamic limit $n\to\infty$, the eigenvalues of
$\tilde{x}(\phi)$ are $x_{1}=4(1+\lambda)^{2}/n(\lambda)^{2}$, and
$x_{2}=4(1+2\lambda\cos{\phi}+\lambda^{2})/n(\lambda)^{2}$, showing
a closure of the dissipative gap at $\lambda=\pm1$. For
$|\lambda|\neq 1$ the unique NESS is found by solving the continuous
Lyapunov Eq.~(\ref{CLEphi}). The symbol function, in a Pauli
matrix representation, results
$\tilde{\gamma}(\phi)=\bm{\gamma}^{T}\cdot\bm{\sigma}$, where
$\bm{\sigma}:=(\sigma_{x},\sigma_{y},\sigma_{z})^{T}$, and
\begin{figure*}[t]
\begin{center}
\subfigure{\includegraphics[width=0.32\textwidth]{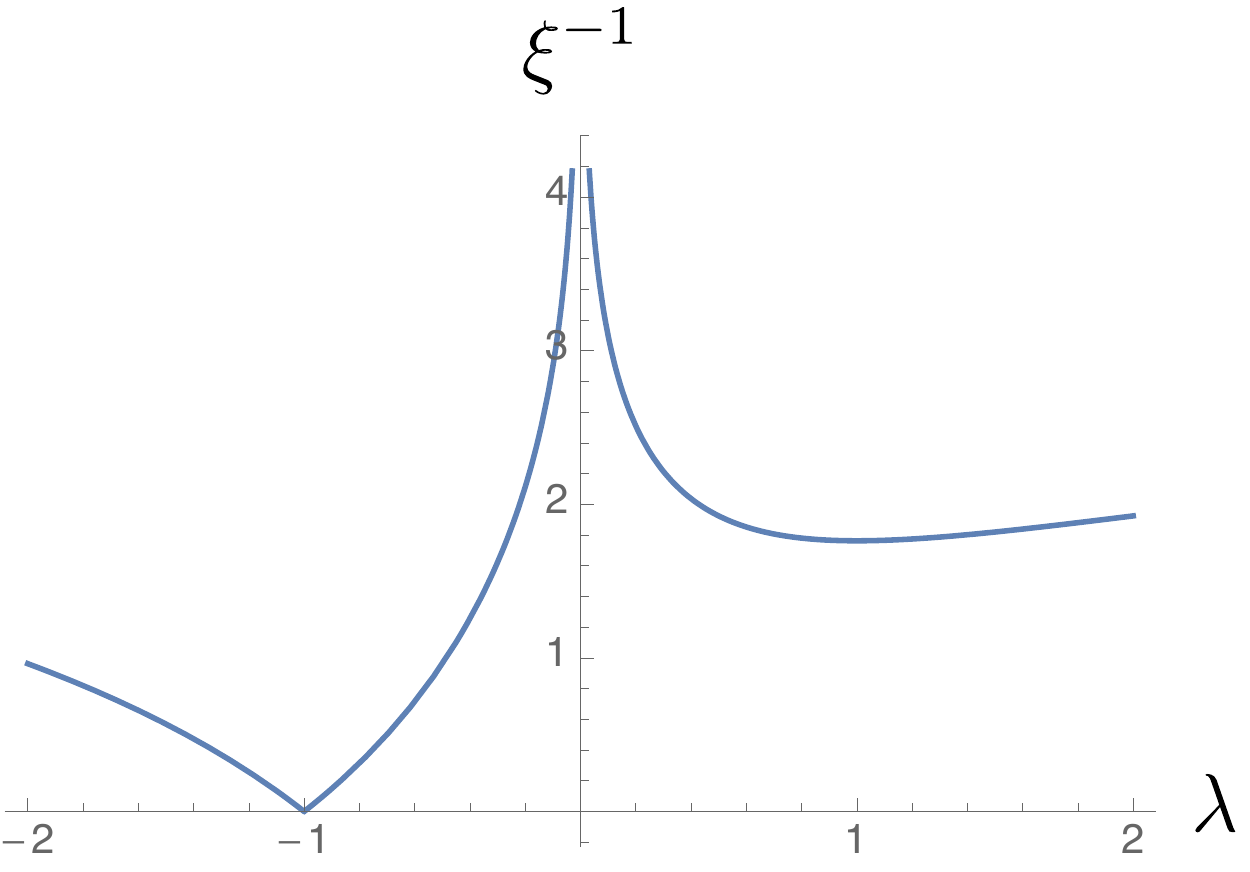}}
\subfigure{\includegraphics[width=0.32\textwidth]{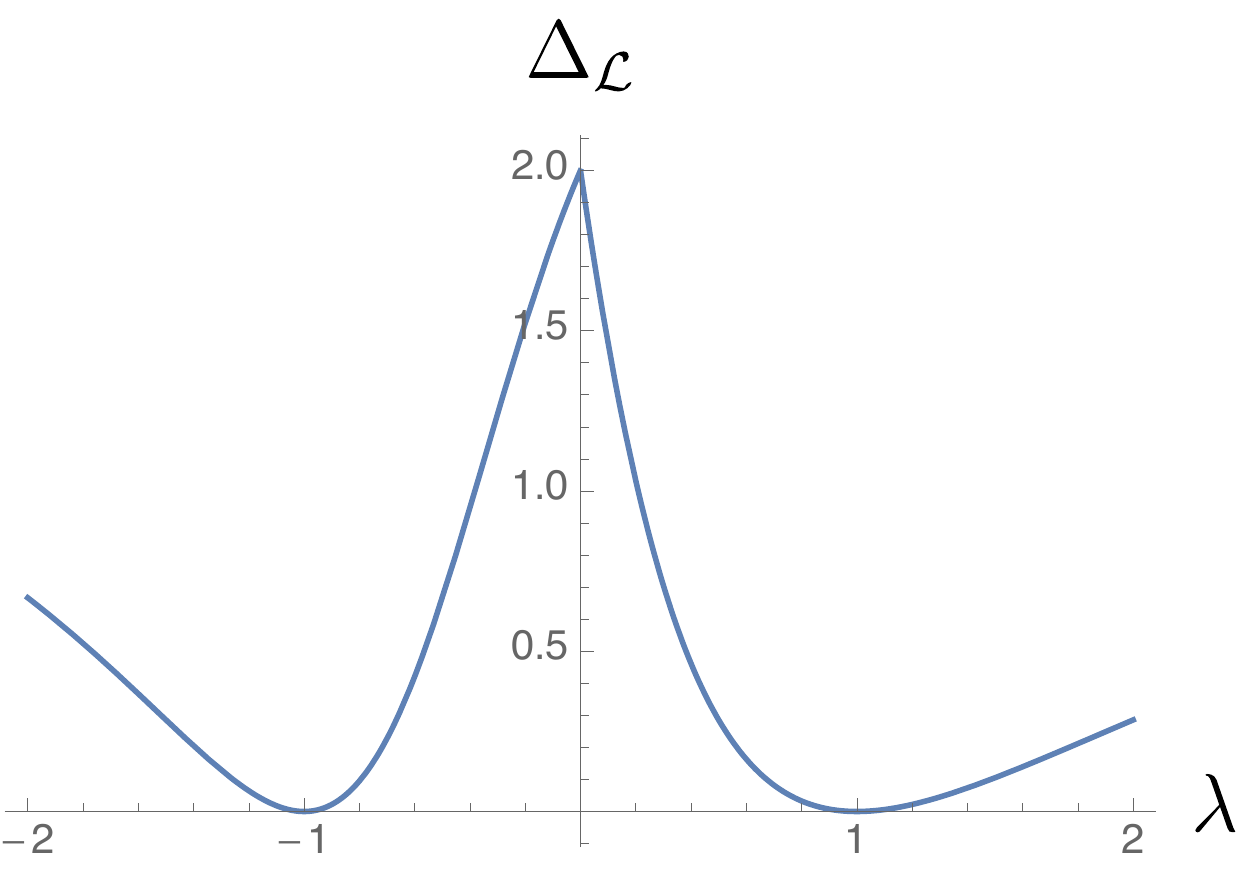}}
\subfigure{\includegraphics[width=0.32\textwidth]{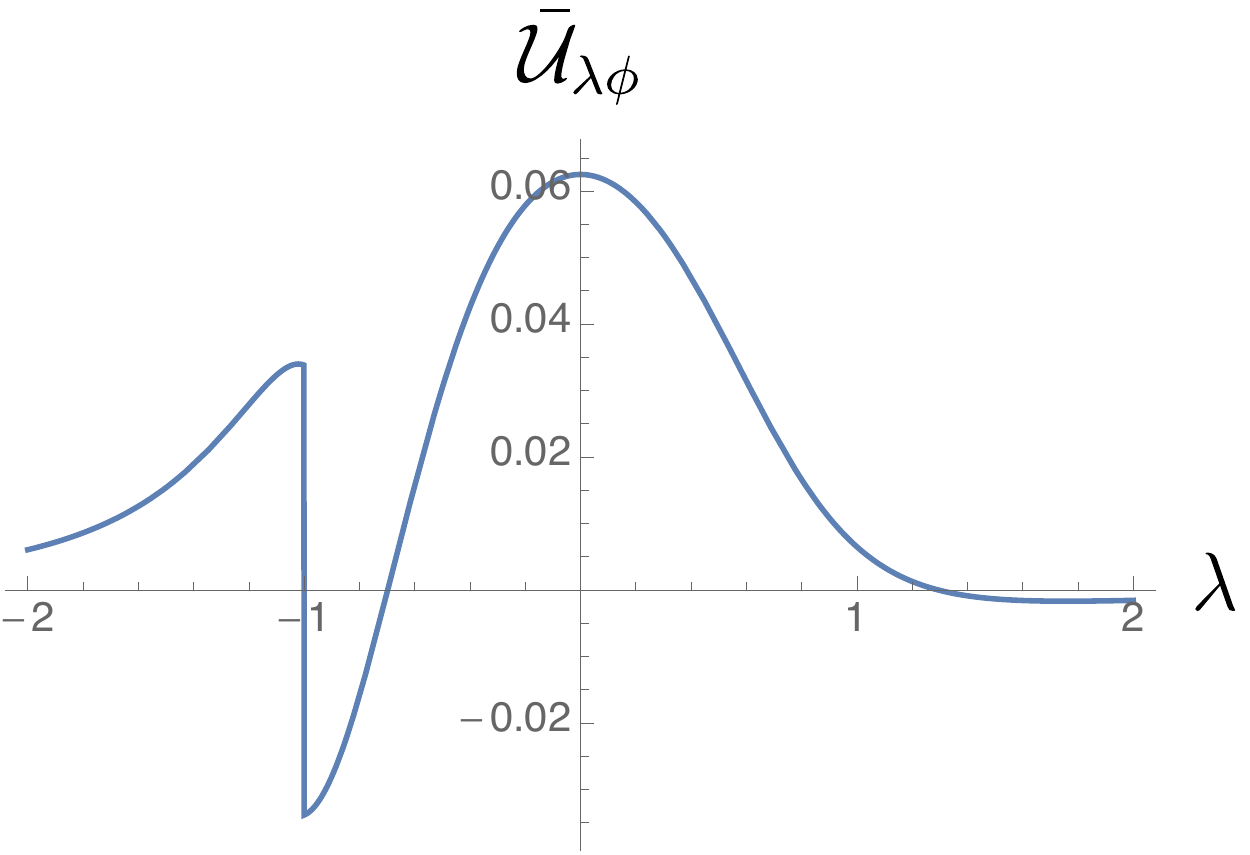}}
\end{center}
\vspace{-0.2 cm} \caption{Model of a 1D Fermionic chain on a ring
showing a closing dissipative gap that does not imply a diverging
correlation length. This is the model discussed in
section~\ref{sec:Example} which is an extension of a model
introduced in~\cite{Bardyn2013}. The inverse correlation length, the
dissipative gap and the MUC are shown, respectively, from the left
to the right panel. The model is critical only for $\lambda=-1$,
while the gap closes for both $\lambda=\pm 1$. As expected, the
discontinuity of MUC captures the criticality, and it is otherwise
insensitive to a vanishing dissipative gap.}\label{TBD}
\end{figure*}
\[
    \bm{\gamma}=g(\phi)\left[\begin{array}{c}(\sin\phi+\lambda\sin{2\phi})\cos 2\theta \\(\cos\phi+\lambda\cos{2\phi}) \\-(\sin\phi+\lambda\sin{2\phi})\sin 2\theta\end{array}\right],
\]
where $g(\phi)=(1+\lambda)/(1+\lambda+\lambda \cos{\phi} +
\lambda^{2})$, with eigenvalues $\pm
g(\phi)\sqrt{1+\lambda^{2}+2\lambda \cos\phi}$. This shows that
$\tilde{\gamma}$ is critical in the sense of diverging correlation,
only for $\lambda=-1$ and not for $\lambda=1$, even if the
dissipative gap closes in both cases. Figure~\ref{TBD} shows the
dependence of the inverse correlation length of the bulk, the
dissipative gap and the mean Uhlmann curvature
$\bar{\mathcal{U}}_{\lambda \phi}$ on the parameter $\lambda$.
Notice a discontinuity of the Uhlmann phase corresponding to the
critical point $\lambda_{0}=-1$, while it does not show any
singularity for $\lambda=1$ where the gap closes.

\subsection{Rotated XY-model with Local Dissipation}
Let's now turn to a prototypical example of a translationally
invariant one-dimensional model. The features described above are
exemplified in the rotated XY model with periodic boundary
conditions~\cite{Carollo2005,Pachos2006},
$H=R(\theta)H_{XY}R(\theta)^{\dagger}$, with
$R(\theta)=e^{-i\frac{\theta}{2}\sum_{j}\sigma_{j}^{z}}$ and
\begin{equation}\label{RotatedXY}
H_{XY}\!=\!\sum_{j=1}^{n}
\!\left(\frac{1\!+\!\delta}{2}\sigma_j^x\,\,
\sigma_{j+1}^x\!+\!\frac{1\!-\!\delta}{2}\sigma_j^y \sigma_{j+1}^y
\!+\! h \sigma^z_j\right),
\end{equation}
where each site $j$ is coupled to two local reservoirs with Lindblad
operators $\Lambda_{j}^{\pm}=\epsilon \mu_{\pm} \sigma_{j}^{\pm}$.
The spin-system is converted into a quadratic Fermionic model via
Jordan-Wigner transformations. The Liouvillian spectrum can be
solved exactly~\cite{Prosen2008,Prosen2010a,Horstmann2013} and it is
independent of $\theta$. In the weak coupling limit $\epsilon\to 0$,
the symbol function of the NESS correlation matrix reads
$\tilde{\gamma}(\phi)=\bm{\gamma}^{T}\cdot\bm{\sigma}$, where 
\begin{equation}
 \bm{\gamma}= g(\phi)\begin{pmatrix}
 t(\phi) \cos\theta\\-1\\ t(\phi) \sin\theta
\end{pmatrix},
\end{equation}
with $g(\phi)= \frac{\mu_{-}^{2}-\mu_{+}^{2}}{\mu_{-}^{2}+\mu_{+}^{2}}
\frac{1}{1+t(\phi)^{2}}$ and $t(\phi) :=\delta \sin\phi/(\cos\phi
-h)$. The system shows criticality in the same regions as
the $XY$ hamiltonian model~\cite{Horstmann2013}. By using
expression~(\ref{eq:Uz}) one can calculate the exact values of the
mean Uhlmann curvature. One finds that $\bar{\mathcal{U}}_{\delta
h}$ vanishes identically, while $\bar{\mathcal{U}}_{\delta\theta}$
and $\bar{\mathcal{U}}_{h\theta}$ are plotted in Fig.~\ref{UXY}. As
predicted, the Uhlmann curvature shows a singular behaviour only
across criticality. In particular, $\mathcal{U}_{h\theta}$ is
discontinuous in the $XY$ critical points $|h|=1$, while
$\mathcal{U}_{\delta\theta}$ is discontinuous in the $XX$ type
criticalities $\delta=0$, $h< 1$. This shows that the mean Uhlmann
curvature captures faithfully the critical behaviour of the underlying physical
model. In the following we will see a model with a richer phase
diagram, in which the geometric features of criticality will be even
more apparent.
\begin{figure}[t]
\includegraphics[width=\textwidth]{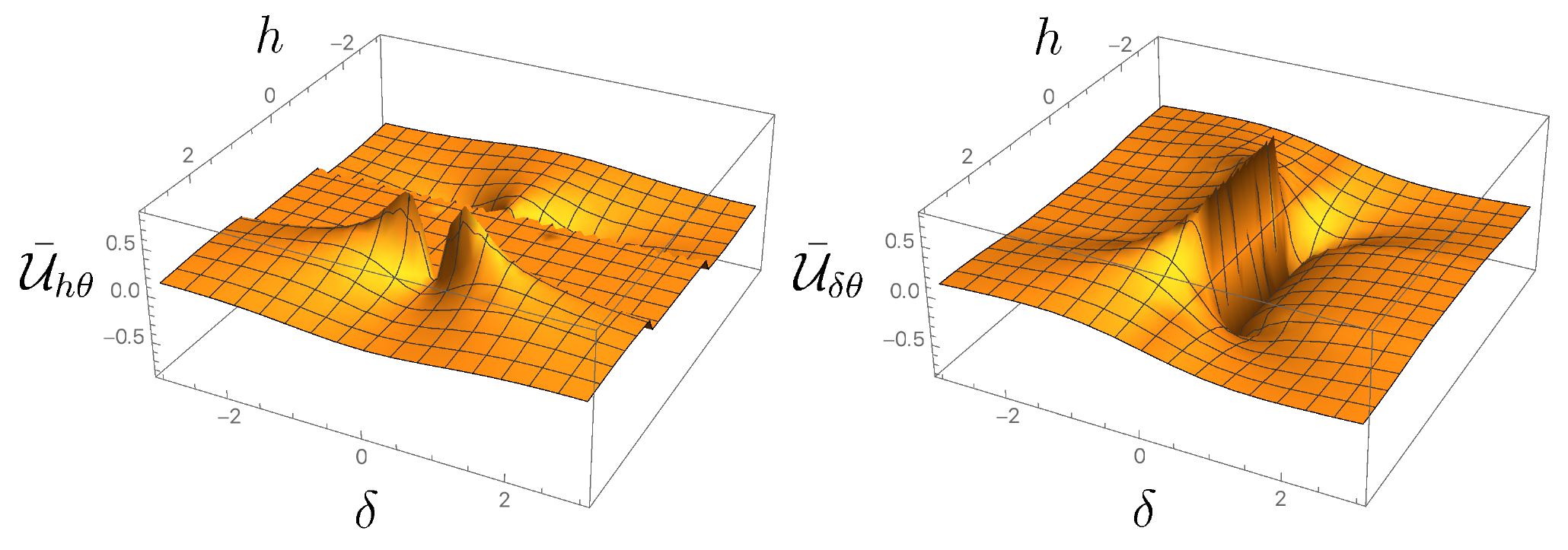}
\caption{The mean Uhlmann curvature per number of sites
$\bar{\mathcal{U}}$ for the rotated XY model with local reservoirs.
The dependence of $\bar{\mathcal{U}}_{h\theta}$ (left) and of
$\bar{\mathcal{U}}_{\delta\theta}$ (right) on the parameters
$\delta$ e $h$. The mean Uhlmann curvature shows a singular
behaviour in the critical regions of the model.
$\mathcal{U}_{h\theta}$ is discontinuous in the $XY$ critical points
$|h|=1$, and $\mathcal{U}_{\delta\theta}$ is discontinuous in the
$XX$ type criticalities $\delta=0$, $h< 1$. }\label{UXY}
\end{figure}

\subsection{Boundary driven XY-model}
Let's apply the above analysis to a specific model, the
boundary-driven spin-1/2 XY chain~\cite{Prosen2008}. In this model,
an open chain of spin-1/2 particles interacts via the
$XY$-Hamiltonian,
\begin{equation}\label{eq:HXYdiss}
H_{XY}\!\!=\!\sum_{j=1}^{n-1}\!
\left(\frac{1\!+\!\delta}{2}\sigma_j^x
\sigma_{j+1}^x\!+\!\frac{1\!-\!\delta}{2}\sigma_j^y
\sigma_{j+1}^y\right)\! +\!\sum_{j=1}^{n} h \sigma^z_j,
\end{equation}
where the $\sigma_{j}^{x,y,z}$ are Pauli operators acting on the
spin of the $j$-th site. At each boundary, the chain is in contact
with two different reservoirs, described by Lindblad operators 
\begin{equation}
 \Lambda^{\pm}_{L}=\sqrt{\kappa_{L}^{\pm}}(\sigma_{j}^{x}\pm
i\sigma_{j}^{y})/2 \quad \textrm{ and }\quad
\Lambda^{\pm}_{R}=\sqrt{\kappa_{R}^{\pm}}(\sigma_{j}^{x}\pm
i\sigma_{j}^{y})/2. 
\end{equation}
  A Jordan-Wigner transform converts the system into a quadratic Fermionic dissipative model with Gaussian NESS~\cite{Prosen2008,Prosen2010}.
\begin{figure}
\begin{center}
\includegraphics[width=0.4\linewidth]{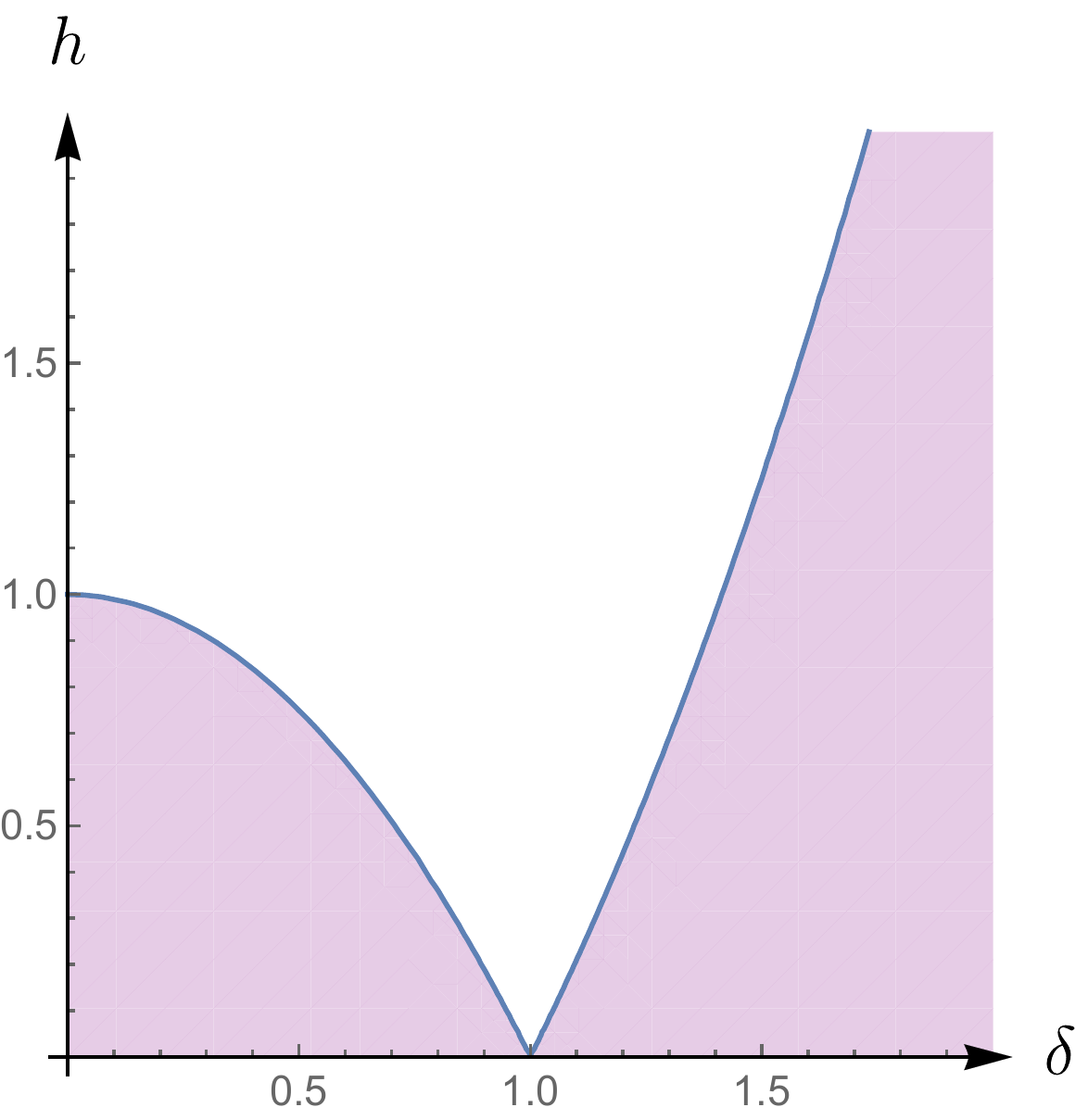}
\caption{Phase diagram of the boundary driven XY model.
$h<h_{c}:=|1-\delta^{2}|$ the chain exhibits long-range magnetic
correlations (LRMC). For $h>h_{c}$ and along the lines $h=0$ and
$\delta=0$ the model shows short-range magnetic correlations (SRMC).
The qualitative features of the phase diagram do not depend on the
values of the environmental parameters $\kappa_{L}^{\pm}$ and
$\kappa_{R}^{\pm}$.}\label{fig:BDXYpd}
\end{center}
\end{figure}
 The system experiences different phases as the anisotropy $\delta$ and magnetic field $h$ are varied. For $h<h_{c}:=|1-\delta^{2}|$ the chain exhibits long-range magnetic correlations (LRMC) and high sensitivity to external parameter variations~(see Fig~\ref{fig:BDXYpd}). For $h>h_{c}$ and along the lines $h=0$ and $\delta=0$ the model shows short-range magnetic correlations (SRMC), with correlation function $C_{jk}:=\langle\sigma^{z}_{j}\sigma^{z}_{k}\rangle - \langle\sigma^{z}_{j}\rangle\langle\sigma^{z}_{k}\rangle$ exponentially decaying: $C_{jk}\propto\exp{-|j-k|/\xi}$, with $\xi^{-1}\simeq 4\sqrt{2(h-h_{c})/h_{c}}$. In both long and short range phases, the dissipative gap closes as $\Delta=\mathcal{O}(n^{-3})$ in the thermodynamic limit $n\to\infty$. The critical line $h=h_{c}$, is characterised by power-law decaying correlations $C_{jk}\propto|j-k|^{-4}$, and $\Delta=\mathcal{O}(n^{-5})$. Therefore, the scaling law of $\Delta$ cannot distinguish long and short range phases, and can only detect the actual critical line $h=h_{c}$. Likewise, $\Delta$ does not identify the transitions from the LRMC phase to the $\delta=0$ and $h=0$ lines.\\\indent
In Table~\ref{table}, the scaling laws of $|\mathcal{U}|$,
$||g||_{\infty}$, $\Det{(g)}$ and $\Delta$ are compared in each
region of the phase diagram. Fig.~\ref{gDBXY} and Fig.~\ref{UDBXY}
clearly show that both $||g||_{\infty}$, and $|\mathcal{U}_{\delta
h}|$ map faithfully the phase diagram of Fig.~\ref{fig:BDXYpd}.
\begin{figure}
\begin{center}
    \includegraphics[width=0.7\linewidth]{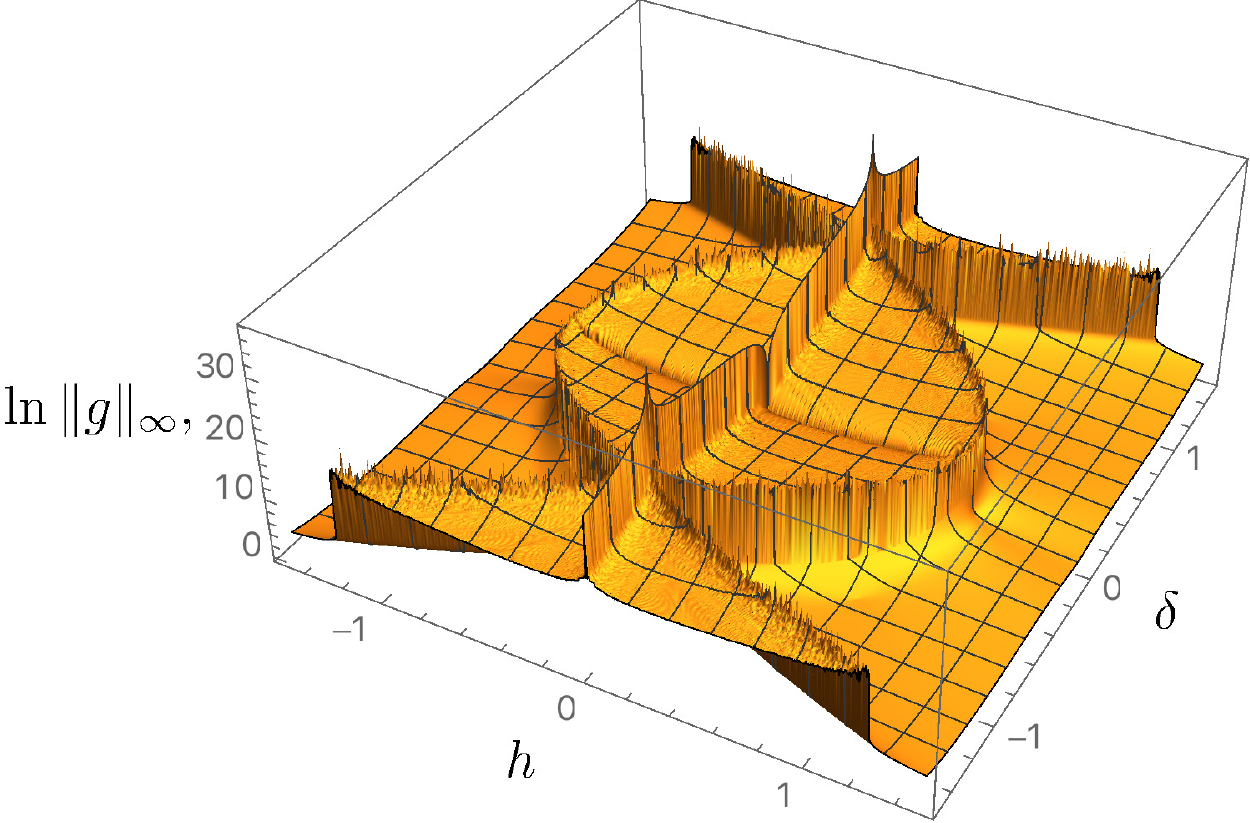}
    \caption{The largest eigenvalue of the Bures metric $||g||_{\infty}$ for the boudary driven XY model, for $n=300$. The qualitative behaviour of the metric maps the phase diagram quite faithfully. It is evident a larger value of $||g||_{\infty}$ close to the phase transitions $h=h_{c}:=|1-\delta^{2}|$ between LRMC and short range phases. $\kappa_{L}^{+}=0.3$, $\kappa_{L}^{-}=0.5$,$\kappa_{R}^{+}=0.1$, $\kappa_{R}^{-}=0.5$. The qualitative features remains unchanged for different values of $\kappa_{L,R}^{\pm}.$}\label{gDBXY}
\end{center}
\end{figure}
\begin{figure}
\begin{center}
    \includegraphics[width=0.7\linewidth]{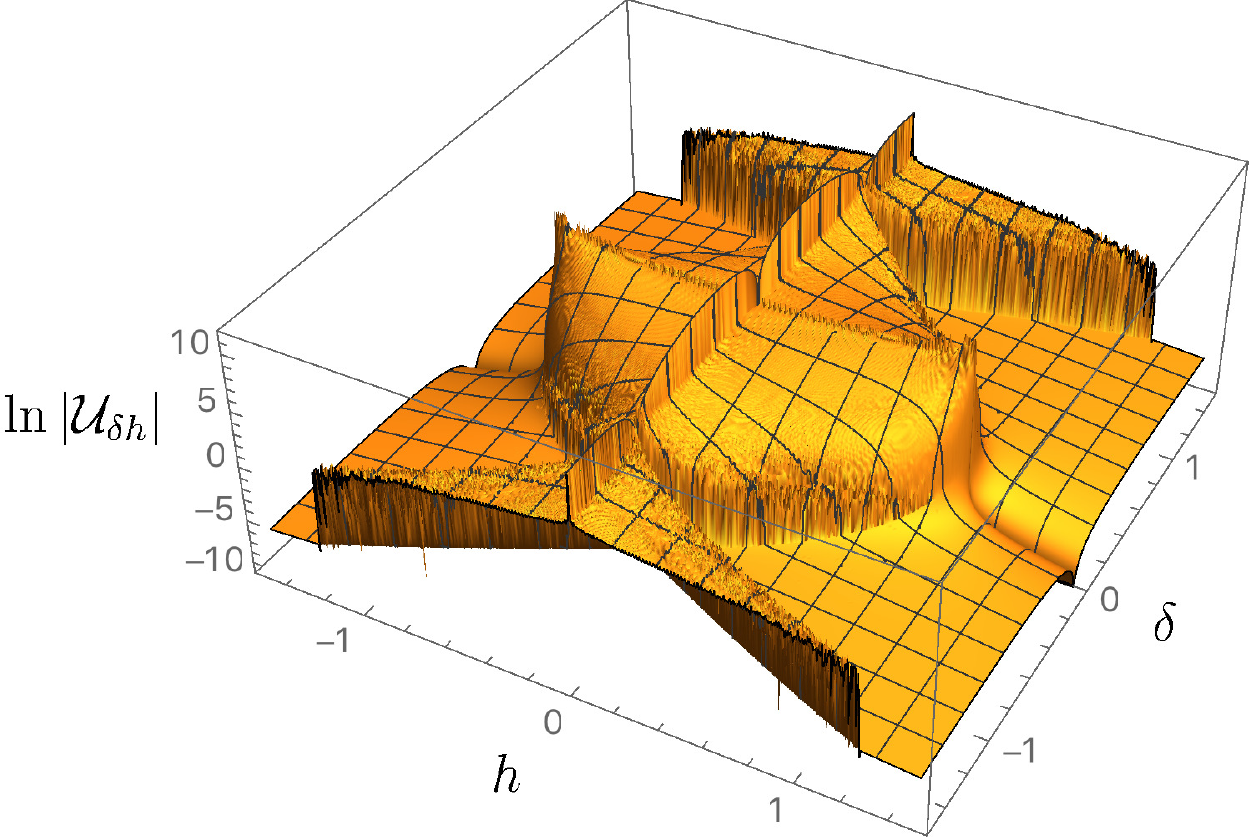}
    \caption{The MUC $|\mathcal{U}_{\delta h}|$ for the boudary driven XY model, for $n=300$. Here the parameters are the same as in figure~\ref{gDBXY}. As in the case of the metric, also the qualitative behaviour of MUC maps quite well the phase diagram. The striking difference with figure~\ref{gDBXY} is the nature of the discontinuity accross the critical line $h=h_{c}:=|1-\delta^{2}|$, which still signals the transitions between LRMC and short range phases. Here the lack of a greater divergence of the MUC at the critical line is a manifestation of the classical nature in the fluctuations driving the NESS-QPTs. }\label{UDBXY}
\end{center}
\end{figure}
\begin{figure}
\begin{center}
\includegraphics[width=0.7\linewidth]{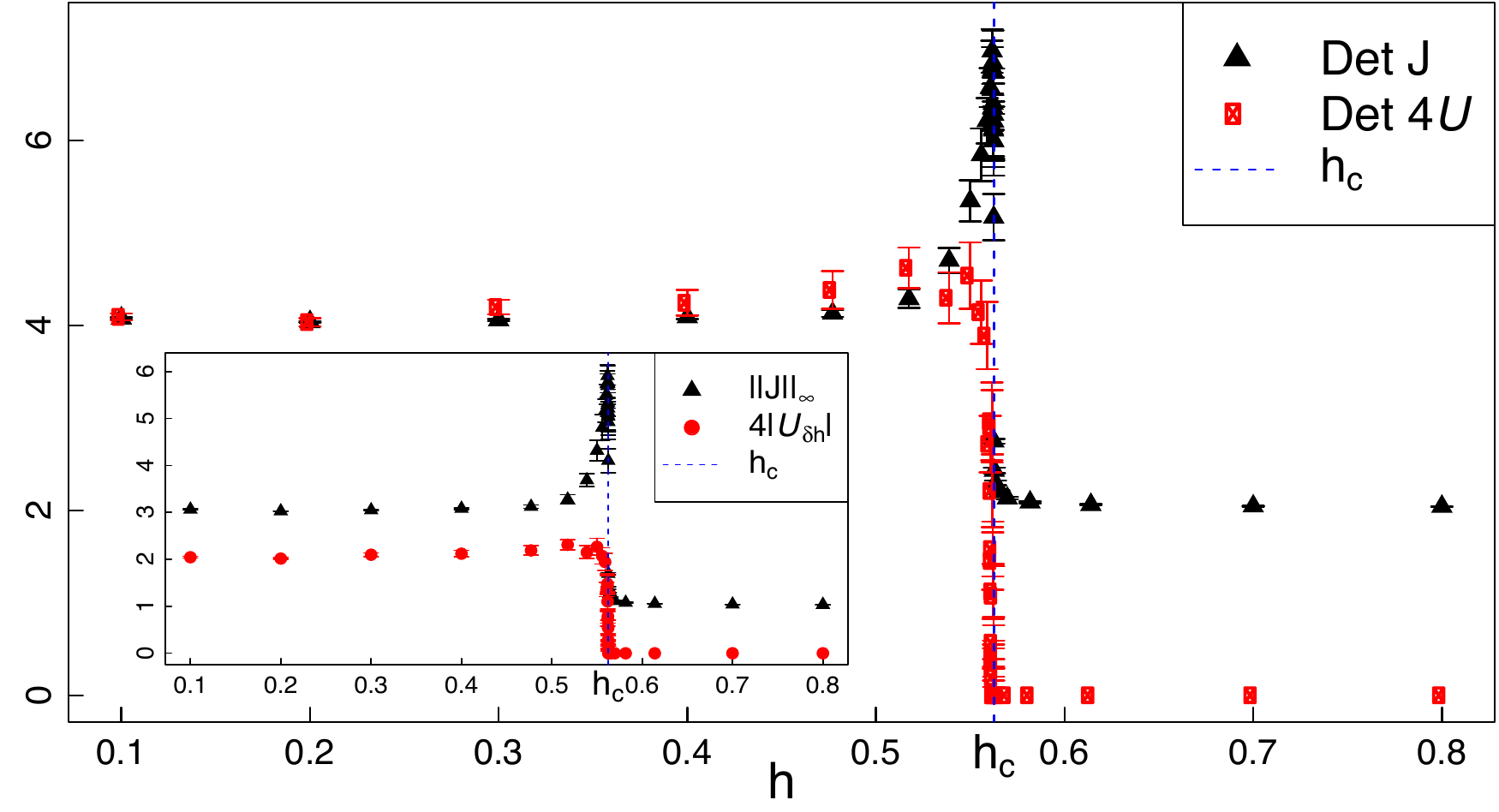}
\caption{Boundary driven XY model. Scaling laws of the determinants
(main) and maximal eigenvalues (inset) of the Fisher information
matrix $J$ and mean Uhlmann curvature $\mathcal{U}$ for different
values of $h$, with $\delta=1.25$ and $h_{c}=|1-\delta^{2}|$. The
laws do not depend on the particular values of the
$\kappa^{\pm}_{R,L}$. The scalings are the results of fits on
numerical data, with size ranging in
$n\in[20,2000]$.}\label{fig:Scaling}
\end{center}
\end{figure}
The results of table~\ref{table} show that the Liouvillian gap, the
metric and the MUC encode different information. Indeed, unlike the
Hamiltonian gap ruling the ground-state QPT, the Liouvillian gap
$\Delta$ closes for $n\to\infty$ at the critical point and for
$h\neq h_c$, both in the LRMC and SRMC phase.
\begin{table}[htp]
\begin{center}
\begin{tabular}{lcccccc}
\hline
Phase&Parameters&$\Delta$  & $||g||_{\infty}$&\Det{\,g }&$|\mathcal{U}_{\delta h}|$& R\\
\hline
Critical &$ h=0$ & $n^{-3}$ & $n^{6}$ & $n^{7}$&$n^{3}$ &$n^{-1}$ \\
Long range&$0<|h|<h_{c}$& $n^{-3}$& $n^{3}$&$n^{4}$ &$n^{2}$& $n^{0}$\\
Critical &$h\simeq h_{c}$& $n^{-5}$& $n^{6}$&$n^{7}$ &$n^{0}$&$n^{-7}$\\
Short range&$h > h_{c}$& $n^{-3}$& $n$&$n^{2}$ &$n^{0}$&$n^{-2}$\\
Critical&$\delta=0,|h|<h_{c}$& $n^{-3}$& $n^{2}$&$n^{8}$ &$n^{3}$&$n^{-2}$\\
\hline
\end{tabular}
\end{center}
\caption{Here we show a comparison between the scaling laws for: the
dissipative gap $\Delta$~\cite{Prosen2008}, the largest eigenvalue
$||g||_{\infty}$ of the metric~\cite{Banchi2014}, the determinant of
$g$ and the largest eigenvalue
$||\mathcal{U}||_{\infty}=|\mathcal{U}_{\delta
h}|=\sqrt{\Det\,\mathcal{U}}$ of the MUC, and the of ratio
$R:=|\Det{2 \mathcal{ U}}|/\Det{J}\propto \Det{2 \mathcal{
U}}|/\Det{g}$  for each phase of the boundary driven XY
model~\cite{Prosen2008}. The ratio $R\le 1$ when $R\sim n^{0}$ marks
the condition of maximal asymptotic incompatibility~\cite{Carollo2019} introduced in
section~\ref{sec:Fisher}.}\label{table}
\end{table}
As the reservoirs act only at the boundaries of the spin chain the
eigenvalues $x_k$ of the matrix $X$ for $n\gg 1$ are a small
perturbation of the $n{\to}\infty$ case, where
$x_k=\pm4i\varepsilon_k$, being 
\begin{equation} \varepsilon_k = \sqrt{(\cos k-h)^2+\gamma^2\sin^2k}, 
\end{equation}
the quasi-particle energy dispersion
relation of the Hamiltonian \eqref{eq:HXYdiss}. In particular $x_k$
gains a small real part and one finds a gap $\Delta =O(n^{-3})$ for
$h\neq h_c$ and $\Delta = O(n^{-5})$ for $h=h_c$. Therefore the
scaling of the Liouvillian gap allows one to identify the transitions
from the SRMC phase to the LRMC phase only along the critical line
$h=h_c$, while the transitions occurring  at the $h=0$ (or
$\gamma=0$) line can only be appreciated by evaluating the
long-rangeness of the magnetic correlations. The question that
naturally arises is how the different phases and transitions can be
precisely characterised in a way similar to what happens for ground
state quantum phase transitions.
This question becomes more compelling if one compares the above
results with the scaling of the Bures metric  $g_{\mu\nu}$ and mean
Uhlmann curvature $\mathcal{U}$ (more precisely their largest
eigenvalue $\|g\|_{\infty}$ and $||\mathcal{U}||_{\infty}=|U_{\delta
h}|$\footnote{$\mathcal{U}$ is an antisymmetric $2\times 2$ matrix
for this two-parameter model. Therefore, it only has two opposite
eigenvalues $\pm|U_{\delta h}|$.}) for specific values of the parameters, see Table~\ref{table},
Fig.~\ref{gDBXY} and Fig.~\ref{UDBXY}.

A first important result is that the geometric properties $g$ and
$\mathcal{U}$ are able to identify the transitions between SRMC and
LRMC phases.
On the "transitions lines" $h=0$ and $h=h_c$ one has that
$\|g\|_{\infty}\sim\mathcal O(n^6)$, while in the rest of the phase
diagram $\|g\|_{infty}<\mathcal O(n^6)$. Furthermore, a closer
inspection of the elements of $g$ shows that, while $g_{hh}(h=0,
\gamma) = \mathcal O(n^6)$,  one has that
$g_{\gamma\gamma}(h=0,\gamma) = \mathcal O(n)$: the scaling is
superextensive only if one moves away from the line $h=0$ ($g_{hh}$)
and enters in the LRMC phase, while, if one moves along the $h=0$
line ($g_{\gamma\gamma}$), i.e. if one remains in the SRMC phase,
the scaling is simply extensive and it matches the scaling displayed
in the other SRMC phase $h > h_c$. On the other hand, the transitions
occurring at $\gamma=0$
has a different scaling: 
$g_{\gamma\gamma} = \mathcal O(n^2)$ while $g_{hh}\approx 0$.

Another important result shown in Table~\ref{table} is that both the
metric tensor and the MUC are able to signal the presence of
long-range correlations: within the LRMC phase $g_{\mu\nu}$ scales
superextensively with $\|g\|_{\infty}\sim\mathcal O(n^3)$, and
$|\mathcal{U}_{h\delta}|\sim\mathcal O(n^2)$. This
super-extensive behaviour is different from that displayed at the
transitions lines.
Thus, differently from $\Delta$, the MUC discriminates these phases,
with no need of crossing the critical line $h=h_{c}$.

However, a more compelling result concerns the quality of the phase
transitions in each region. As discussed in
section~\ref{sec:quant}, the MUC marks the role played by the
quantum nature of the model in the parameter estimation problem. In
other words, it signals the ``quantumness '' of the underlying
physical model.
Table~\ref{table} displays the scaling law of the ratio $R=\sqrt{|\Det{2 \mathcal{ U}}|/\Det{J}}$ in different regions of the phase diagram. In particular, its asymptotic behaviour provides insight into the character of fluctuations which drive the NESS-QPTs.  Indeed, in the limit of $R\sim const$ this ratio signals a condition of maximal asymptotic incompatibility, in which the role of the quantum fluctuations in the criticality cannot be neglected in the thermodynamic limit.\\
Fig.~\ref{fig:Scaling} shows that, in the LRMC phase, the scaling law
of the MUC saturates the upper bound~(\ref{UBound}), in contrast to
the short range phase. This shows the striking different nature of
the two phases. In the LRMC region, the system behaves as an
inherently two-parameter quantum estimation model, where the
parameter incompatibility cannot be neglected even in the
thermodynamic limit. On the short-range phase, instead, the system
is asymptotically quasi-classical. The critical line $\delta=0$
(with $|h|\le h_{c}$) and the critical line $h=0$, which mark
regions of short range correlations embedded in a LRMC phase, show a
MUC which grows super-extensively, with scaling
$\mathcal{O}(n^{3})$, and a nearly saturated
inequality~(\ref{UBound}). In the critical line $h\simeq h_{c}$,
despite the spectacular divergence of
$||g||_{\infty}\simeq\mathcal{O}(n^{6})$, the scaling law of
$|\mathcal{U}_{\delta h}|$ drops to a constant, revealing an
asymptotic quasi-classical behaviour of the model at the phase
transitions.

\section{Conclusions}
Condensed matter physics is a rich framework where a variety of
interesting phenomena arise in association with geometrical
properties of the interactions. Topological behaviour of quantum
interactions are particularly evident for those systems near
QPTs~\cite{Thouless98,Sachdev2011}. It is a well known fact that
quantum criticalities are accompanied by a qualitative change in the
nature of correlations in the ground state of a quantum system, and
describing these changes is clearly one of the major interests in
condensed matter physics. Typical examples are metal-insulator
transitions, or paramagnetic-ferromagnetic transitions for spin
chains, where the two phases are associated with distinct local vs.
global properties of the quantum state.

It is, therefore, expected that such drastic changes in the
properties of the ground states are reflected in the geometry of the
Hilbert space explored by the system across the criticality.
Geometric phase is known to be a signature of the curvature, and in
general of the geometry of the state manifold, and therefore it is a
useful means to investigate the properties of systems near QPTs.  The
heuristic explanation for the non-trivial behaviour of Berry
curvature in the proximity of criticality relies on the idea of level
crossings, occurring at the thermodynamic limit, which involve
ground state and low lying part of the energy spectrum. Level
crossings can be identified as the origin of the curvature in the
phase space manifold, pretty much the same way energy singularity
bends the geometry of the space-time, or more like a Dirac monopole
twists the topology of the field configuration around
it~\cite{Nakahara1990}. Driving the system close to or around these
singularities results in dramatic effects on the state geometry,
picked up by its kinematics in the form of geometric phase
instabilities~\cite{Carollo2005,Hamma2006}. Whether this intuition
may be adapted to an entirely different type of quantum many-body
phenomena is an absolutely non-trivial question.

Novel type of non-equilibrium phase transitions have recently
emerged in quantum mechanical systems, as phenomena that may
underpin new forms of criticality, departing significantly from
transitions that are observed in the equilibrium settings. A
particularly intriguing type of non-equilibrium critical phenomena
arise in the context of open quantum systems, where the
non-equilibrium character is induced by coupling between system and
several external reservoirs. Theoretical investigation of open
quantum systems is ultimately motivated by the inherently open
nature of several current experimental platforms, which are typically
subject to external drive, dissipation and dephasing.

The occurrence of quantum phase transitions in non-equilibrium
steady states, which are the results of complex many-body
dissipative evolutions, is far from being understood. Harnessing the
investigation of such an uncharted scenario with a powerful new set
of tools is certainly desirable. Geometric properties has proven
successful in unravelling general quantitative and qualitative
informations in equilibrium criticalities. It is therefore expected
that they may glean new insights is these novel scenario, providing
a comprehensive framework for their understanding.

The crucial focus of this review was indeed to present the general
framework of geometric methods and to demonstrate the applicability
of these ideas to the entirely novel category of non-equilibrium
phase transitions. What it may seem at first glance a pedantic
application of known concepts to yet another instance of critical
phenomena, is instead quite a major shift of paradigm. It is no
coincidence that nearly 10 years after their introduction, models
such as the one proposed in Ref.~\cite{Prosen2008} are still subject
of active investigations through well established tools of quantum
information and information geometry. Yet, a full understanding of
their main characteristics is still lacking. Most of all, what is
missing is an intuitive comprehensive framework, within which
comparing the non-equilibrium-QPTs  with what is known from equilibrium
phase transitions.

Equilibrium phase transitions fall invariably into two markedly
non-overlapping categories: classical phase transitions and quantum
phase transitions. NESS-QPTs  offer a unique arena where such a
distinction indeed fades off. The coexistence between quantum and
classical fluctuations in these models may vaguely be reminiscent of
what happens at quantum to classical crossovers in equilibrium
phenomena, with a major striking difference: the remarkably sharp
character of truly critical phenomena.

Among other things, we described a recently developed approach which allows one
to quantitatively assess the ``quantum-ness'' of a critical
phenomena. This method resorts to ideas borrowed from quantum
estimation theory, which endow the geometric phase approach with an
operationally well defined character. This approach brings insight
into the interplay between quantum and classical fluctuation in
critical phenomena. Quadratic Fermionic Liouvillian models are
perhaps the simplest significant example where this interaction
plays a non-trivial role, where one finds quite unusual interplay
between  markedly classical and quantum features associated to the
same phase transitions.

A source of confusion within the class of dissipative NESS-QPTs  is
that the concept of criticality has indeed a variety of inequivalent
definitions. Already, in the physically relevant subclass of
quadratic quasi-free Fermionic models, there are two non-equivalent
widely used definitions in literature. They rely on the idea of
diverging correlation length and critical slow down, respectively.
While in the usual setting of equilibrium QPTs  these two definitions
generally coincide, in NESS-QPTs  this is not quite the case. The
first definition introduced by Eisert and Prosen
in~\cite{Eisert2010} is the one mostly adopted in this review. The
second, inequivalent, one is used for example in a series of works
related to dissipatively induced topological order
(e.g.~\cite{Bardyn2013}). One can prove analytically that a singular
geometric phase curvature is an unequivocal signature of a critical
behaviour associated to a diverging correlation length. In
retrospect, it might not come as a surprise that such a connection
exists, as that is indeed what the intuition built up from the
experience on equilibrium phenomena would suggest.

However, a similarly heuristic argument should point towards an
analogous conclusion in the case of criticality defined by a closing
dissipative gap. After all, this is the NESS-QPTs  analogue of a
vanishing Hamiltonian gap. One should legitimately expect that a
closing dissipative gap, i.e. a critical slowing down, should result
in a singular behaviour of the mean Uhlmann curvature. That this is
a reasonable guess is further suggested by several studies on
dissipative topological phase transitions~\cite{Bardyn2013}, where
topologically inequivalent regions of the phase diagram are
separated by critical points with vanishing dissipative gap, which
are not necessarily accompanied by a diverging correlation length in
the bulk.  However, one can analytically demonstrate that this
second heuristic argument does not hold, showing that the Uhlmann
curvature is sensitive to the criticality, but only in the sense of
a truly diverging correlation length. In this sense, the Uhlmann
curvature confirms its role as a witness of the purely kinematic
aspects of the criticality, and it is only indirectly affected by
the dynamical features of the NESS-QPTs. 

Although the main focus of the review is on the physically relevant
class of Fermionic quadratic models, this is by no means the only
context in which this idea is applicable. This approach, for its
generality, immediately extends to any equilibrium and
non-equilibrium QPT, with and without order parameters, with or
without symmetry breaking, including non-equilibrium dynamical
phase-transitions, topological dissipative phase transitions,
and cluster state phase transitions. Moreover, this idea is a promising
tool which may glean insight on the interplay between competing
orders both in equilibrium and non-equilibrium QPTs.  It is hard to
grasp the full extent of the implications that such a general
approach may provide.

Going beyond the geometrical aspect mentioned in this review, the
mean Uhlmann curvature and the Uhlmann geometric phase in general
offer the possibility of studying topological structures on the
manifold of density matrices. One can indeed formulate, under
suitable assumption, a topologically nontrivial gauge structure
based on the notion of mean Uhlmann curvature. In this framework,
topological invariants that are protected by suitable symmetries may
be identified for mixed states. One may define class of
topologically inequivalent mappings from a parameter space into the
density matrices which can be continuously deformed into each other
only if the underlying symmetry assumptions are violated.

In a lattice translation-invariant system, one may think of
identifying the parameter space with the Brillouin zone, thereby
providing a possible way of generalising topological band structure
invariants to the domain of mixed states. The possible applications
of this conceptual framework is to obtain a topologically nontrivial
Chern insulator or topological superconductors as a steady state of
a non-equilibrium open quantum system.

Topological invariants, that may be constructed through the mean
Uhlmann curvature, are in principle experimentally accessible via
state tomography. However, a possible route of investigation would
be to be able to construct a relation to natural observables such as
response functions. A prototypical quantity to look at is the
quantised Hall conductance. However, the formulation of a mixed
state quantity, which under statistical mixture does not cause
deviations from an integer quantisation of the Hall conductance, is
quite an open challenge.

\vspace{6pt}

\appendix

\section{Fermionic Gaussian States}\label{app:FGS}
We review here the main properties of Fermionic Gaussian States.
Let's consider a system of $n$ Fermionic particles described by
creation and annihilation operators $c_{j}^{\dag}$ and $c_{j}$.
These operators obey the canonical anticommutation relations, 
\begin{equation}
 \{c_{j},c_{k}\}=0\qquad \{c_{j},c^{\dag}_{k}\}=\delta_{jk}\,. 
\end{equation}
 Let's define the Hermitian Majorana operators as 
\begin{equation}
 \omega_{2j-1}:=c_{j}+c^{\dagger}_{j}\,,\qquad
\omega_{2j}:=i(c_{j}-c_{j}^{\dagger})\,, \end{equation} which are generators of
a Clifford algebra, and satisfy the following anti-commutation
relations \begin{equation} \{\omega_{j},\omega_{k}\}=2\delta_{jk}\,. 
\end{equation}
 Fermionic Gaussian states are defined as states that can be
expressed as
 \begin{equation}\label{GS}
\rho=\frac{e^{-\frac{i}{4}\bm{\omega}^{T} \Omega
\bm{\omega}}}{Z}\,,\qquad Z:=\Tr[ e^{-\frac{i}{4}\bm{\omega}^{T}
\Omega \bm{\omega}}]
 \end{equation}
 where $\Omega$ is a $2n\times 2n$ real antisymmetric matrix.
As any antisymmetric real matrix, $\Omega$ can be cast in the
following canonical form by an orthogonal matrix $Q$, i.e.
\begin{equation} \begin{aligned}
\Omega &= Q^T\, \bigoplus_{k=1}^n
\begin{pmatrix} 0&\Omega_k\\-\Omega_k&0 \end{pmatrix}
  \, Q & Q^T=Q^{-1}~,
  \label{e.Q}
\end{aligned} \end{equation} 
where $\pm i\Omega_k$ are $\Omega$'s  eigenvalues. Let 
\begin{equation}
 \bm{z}=(z_{1},\dots,z_{2n})^{T}:=Q\bm{\omega} \end{equation} be the vector of
Majorana Fermions in the eigenmode representation. Hence,
\begin{equation} \begin{aligned}
\rho &= \frac1Z\prod_k \left[\cosh\left(\frac{\Omega_k}2\right) -i
\sinh\left(\frac{\Omega_k}2\right)\,z_{2k-1} z_{2k}\right]~,
\\
Z &= \prod_k 2\cosh\left(\frac{\Omega_k}2\right) \,.
\label{e.gaussianZ}
\end{aligned} \end{equation} 

 Gaussian states are completely specified by the two-point correlation matrix

\begin{equation}
  \Gamma_{jk}:=1/2\Tr{(\rho[\omega_{j},\omega_{k}])}\,, \qquad \Gamma=\Gamma^{\dagger}=-\Gamma^{T}\,,
\end{equation}
  which is an imaginary antisymmetric matrix. As

\begin{equation}
  \Gamma_{jk}=\frac{2i}Z\frac{\partial Z}{\partial \Omega_{jk}},
\end{equation}
  one can show that
\begin{align}
\Gamma = \tanh\left(i\frac{\Omega}2\right)~. \label{e.corrG}
\end{align}
The correlation matrix is diagonal in the same basis of $\Omega$ and
its eigenvalues read $\gamma_k = \tanh(\Omega_k/2)$. Hence
\begin{align}
  \rho &= \prod_k \frac{1-i \gamma_k\,z_{2k-1} z_{2k}}2~,
  \label{e.gaussianC}
\end{align}
where $|\gamma_k|\le 1$. Hence the Gaussian Fermionic state can be
factorised into a tensor product $\rho=\bigotimes_{k}\rho_{k}$ of
density matrices of the eigenmodes $\rho_{k}:=\frac{1-i
\gamma_k\,z_{2k-1} z_{2k}}2$. Note that for $\gamma_k=\pm 1$, one
has $\Omega_k = \pm\infty$, making the definition~(\ref{GS}) of
Gaussian state not well defined, unlike Eq.~\eqref{e.gaussianC}.
Indeed, the latter offers an appropriate parameterisation even
in those extremal points. Notice that $|\gamma_{k}|=1$ corresponds
to a Fermionic mode $\tilde{c}_{k}=1/2(z_{2k-1}+z_{2k})$ being in a
pure state, as  it is clear from the following explicit expression
for the purity of the states $\rho_{k}$:
\begin{equation} \begin{aligned}
  \Tr[\rho_{k}^2] =
  \frac{\det\left[2\cosh\left(\,\Omega_{k}\right)\right]^{\frac12}}{
  \det\left[2\cosh\left(\frac{\Omega_{k}}2\right)\right]} ~.
  \label{e.purity}
\end{aligned} \end{equation} 

\begin{equation} \begin{aligned}
  \Tr[\rho^2] =
  \frac{\det\left[2\cosh\left(i\,\Omega\right)\right]^{\frac12}}{
  \det\left[2\cosh\left(i\frac{\Omega}2\right)\right]} =
  \sqrt{\det\left(\frac{1+\Gamma^2}2\right)}~.
  \label{e.purity}
\end{aligned} \end{equation} 

\subsection{Symmetric Logarithmic derivative for Fermionic Gaussian States}\label{app:FDS}
We will review here useful expressions, adapted from
Ref.~\cite{Jiang2014}, which are instrumental for the derivation
of the symmetric logarithmic derivative of density matrices in the
exponential form 
\begin{equation}\label{eq:exp} \rho=e^{D(\bm{\lambda})}. 
\end{equation}
 Clearly, a Gaussian Fermionic state can be expressed in the
exponential form~(\ref{eq:exp}) by identifying 
\begin{equation}
 D=-\frac{i}{4}\bm{w}^{\dagger}\cdot \Omega\cdot\bm{w} - \one \ln{Z}.
\end{equation} Notice, that the above parameterisation is well defined in the
case of full-rank density matrices. As usual, the case of extremal
conditions $|\gamma_{k}|=1$, where $\gamma_{k}$ is an eigenvalue of the
correlation function, should be carried out as a limiting procedure.

The starting point is Eq.~(2.1) of
Ref.~\cite{Wilcox1967} for the derivative of density operators
\begin{equation} \begin{aligned}\label{eq:wilcox}
 \p_\mu  \rho = \int_0^1 e^{sD}\,\p_\mu  D\, e^{(1-s)D}\,d s\;.
\end{aligned} \end{equation} 
One can use the nested-commutator relation
\begin{equation} \begin{aligned}
e^D A e^{-D}&=A+[D,A]+\frac{1}{2!}\,\big[D,[D,A]\,\big]+\cdots\nonumber\\
&=\sum_{n=0}^\infty\frac{1}{n!}\,\mathcal{C}^n(A)=e^{\mathcal{C}}(A)\;,
\end{aligned} \end{equation} 
where $\mathcal{C}^n(A)$, a linear operation on $A$, denotes the
$n$th-order nested commutator
$\commutb{D}{\nsp\ldots\,,\commut{D}{A}}$, with
$\mathcal{C}^0(A)=A$.  Applying this relation to the
expression~(\ref{eq:wilcox}) leads to
\begin{equation} \begin{aligned}\label{eq:rhodot}
\begin{split}
 \p_\mu  \rho \rho^{-1}
 &= \p_\mu  D +\frac{1}{2!}\,\commut{D}{\p_\mu  D}+\frac{1}{3!}\,\commutb{D}{\commut{D}
 {\p_\mu D}}+\cdots\\
 &= \sum_{n=0}^\infty \frac{1}{(n+1)!}\, \mathcal{C}^n(\p_\mu  D)
 = h(\mathcal{C})(\p_\mu D) \;,
\end{split}
\end{aligned} \end{equation} 
where $h$ is the generating function of the expansion coefficients
in Eq.~(\ref{eq:rhodot}),
\begin{equation} \begin{aligned}\label{eq:generating_g}
 h(t) = 1+\frac{t}{2!}+\frac{t^2}{3!}+\cdots = \frac{e^t - 1}{t}\;.
\end{aligned} \end{equation} 

Using the definition of symmetric logarithmic derivative, i.e. 
\begin{equation}
 \p_{\mu} \rho = \frac{1}{2}\left( L_{\mu}\rho + \rho
L_{\mu}\right)~, \end{equation} and that of density matrix in exponential
form~(\ref{eq:exp}), one gets
\begin{equation} \begin{aligned}\label{eq:left_derivative_b}
\begin{split}
\p_\mu  \rho\, \rho^{-1}  &= \frac{1}{2}\big(L + e^D L e^{-D}\big)\\[2pt]
&= \frac{1}{2} \bigg(L + \sum_{n=0}^\infty \frac{1}{n!}\,
\mathcal{C}^n(L) \bigg) = r(\mathcal{C})(L)\;,
\end{split}
\end{aligned} \end{equation} 
where the generating function is $r(t)=(e^t+1)/2$.  Suppose that the
SLD adopts the form,
\begin{equation} \begin{aligned}\label{eq:expansion_L_a}
L_{\mu} &= \sum_{n=0}^\infty f_n\, \mathcal{C}^n(\p_\mu  D)
=f(\mathcal{C})(\p_\mu D)\;,
\end{aligned} \end{equation} 
with the generating function
\begin{equation} \begin{aligned}\label{eq:f_expansion_a}
f(t) = f_0+f_1 t+ f_2 t^2+\cdots
\end{aligned} \end{equation} 
to be determined. Plugging Eq.~(\ref{eq:expansion_L_a}) into
Eq.~(\ref{eq:left_derivative_b}) yields
\begin{equation} \begin{aligned}\label{eq:combining_expansions}
 \p_\mu  \rho\, \rho^{-1}= r(\mathcal{C})\bigl[f(\mathcal{C})(\p_\mu D)\bigr] =r\circ f(\mathcal{C})(\p_\mu D)=r \cdot f(\mathcal{C})(\p_\mu D)~,
\end{aligned} \end{equation} 
where the identity $r\circ f = r\cdot f $ between the combination
and the simple product of the two  functions arises from
$\mathcal{C}^n(\mathcal{C}^m(A))=\mathcal{C}^{n+m}(A)$. Comparing
Eq.~(\ref{eq:combining_expansions}) with Eq.~(\ref{eq:rhodot}) leads
to the following relation between generating functions,
\begin{equation} \begin{aligned}\label{eq:f_expansion_b}
 f(t)= \frac{h(t)}{r(t)}=\frac{\tanh(t/2)}{t/2}
 = \sum_{n=0}^\infty \frac{4\, (4^{n+1}-1) B_{2n+2}}{(2n+2)!}\,t^{2n}\;,
\end{aligned} \end{equation} 
where $B_{2n+2}$ is the $(2n+2)$th Bernoulli number.  Comparing
Eqs.~(\ref{eq:f_expansion_a}) with (\ref{eq:f_expansion_b}), we have
\begin{equation} \begin{aligned}
 f_n =
\begin{cases}
 \displaystyle{\frac{4\, (4^{n/2+1}-1) B_{n+2} }{(n+2)!}}\,,&\mbox{for even $n$}\,,\\[3pt]
  0\,,&\mbox{for odd $n$}\,.
\end{cases}
\end{aligned} \end{equation} 
The vanishing of the odd-order $f_n$'s is a consequence of the
Hermiticity of $L$, which makes $f(t)$ an even function.

For a Gaussian Fermionic state the operator $D$ in terms of the
Majorana Fermions of the eigenmodes is written 
\begin{equation}
 D=-\frac{i}{4}\sum_{k} \Omega_{k}[z_{2k-1}, z_{2k}] - \one
\ln{Z}=\sum_{k} \Omega_{k}\Big(\tilde c_{k}^{\dagger} \tilde
c_{k}-\frac12\Big) - \one \ln{Z}, 
\end{equation} 
where  $\tilde
c_{k}:=\frac{1}{2}(z_{2k-1}+iz_{2k})$, $\tilde
c_{k}^{\dagger}:=\frac{1}{2}(z_{2k-1}-iz_{2k})$ are the ordinary
annihilation and creation operators of the eigenmodes. It is
straightforward to derive the commutation relations between $D$ and
Fermionic operators,
\begin{equation} \begin{aligned}\label{eq:commutation_relations}
 \commutb{D}{\tilde c_k} =- \Omega_k \tilde c_k\,,\quad \commutb{D}{\tilde c_k^\dagger} = \ssp \Omega_k \tilde c_k^\dagger\;,
\end{aligned} \end{equation} 
and for quadratic operators, we get also
\begin{equation} \begin{aligned}
\commutb{D}{\tilde c_j^{\dagger}\tilde c_k} =(\Omega_j - \Omega_k)
\tilde c_j^{\dagger}\ \tilde c_k,\quad \commutb{D}{\tilde c_j \tilde
c_k} =(\Omega_j - \Omega_k) \tilde c_j\ \tilde c_k.\end{aligned} \end{equation} 
Consequently, one finds
\begin{equation} \begin{aligned}
 f(\mathcal C)\ssp(\tilde c_k^\dagger \tilde c_k) = f(\Omega_k-\Omega_k)\ssp \tilde c_k^\dagger \tilde c_k\;,\\[2pt]
 f(\mathcal C)\ssp(\tilde c_k \tilde c_k) = f(\Omega_k+\Omega_k)\ssp \tilde c_k \tilde c_k\;.\label{eq:relation_b}
\end{aligned} \end{equation} 
Most generally, the derivative of $D$ takes the form
\begin{equation} \begin{aligned}\label{eq:derivative_G_a}
 \p_{\mu}D=-\frac12\ssp\ssp {\bm c}\, {\p_{\mu}\Omega}^{\prime}\ssp \bm c -\frac{\p_{\mu} Z}{Z}\;,
\end{aligned} \end{equation} 
which plugged into formula~(\ref{eq:expansion_L_a}) shows that
$L_{\mu}$ is at most quadratic in Fermionic operators.

\section[Spectral Properties of Quadratic Liouvillians]{Spectral Properties of Quadratic Liouvillians}\label{app:Spec}
We will review the main results on the spectral properties of the 
quadratic Fermionic Liouvillian~(\ref{eq:Lindblad}). Note that the
real matrix $X$ defined in~(\ref{eq:X}) has no general structure
apart from the fact that $X + X^T = 8\Re{M}\ge 0$, where $M$ is the
bath matrix defined in~(\ref{eq:BathM}), whose positive
semi-definiteness  implies $\Re{M} \ge 0$. We will drop here the
assumption made in section~\ref{sec:QuadLind}, about the
diagonalisability of $X$, and will show that the qualitative aspects
of the results derived in section~\ref{sec:GeoNESS} still hold.

 The matrix $X$ can always be put in the Jordan canonical form,
i.e. $X=U\,D_{X}^J\,U^{-1}$ with $D_{X}^J = \oplus_m
J_{\ell_m}(x_m)$,
\begin{align}
J_{\ell_m}(x_m) = \begin{pmatrix}
x_m & 1 & \\
& x_m & 1 & \\
&& x_m & 1 & \\
&&& \ddots & \ddots &
\end{pmatrix},
\end{align}
where $x_m$ are (possibly equal) eigenvalues of $X$ and $\ell_m$ is the
dimension of the Jordan block: each block is composed of $\ell_m$
degenerate eigenvalues of $X$. The form of the transformation
\eqref{e.bogohat} remains the same (although with a new matrix $U$)
while \eqref{e.lindiag} becomes
\begin{equation} \begin{aligned}
\mathcal L = -\sum_{j=1}^{2n} x_j \, \hat{b}_j^\times \hat{b}_j -
\sum_m\sum_{k=1}^{\ell_m-1} \hat{b}^\times_{m_k+1} \hat{b}_{m_k}~,
\label{e.linjor}
\end{aligned} \end{equation} 
where $m_k$ refers to the index of the $k$-th element in the $m$-th
Jordan block. It is clear that the state \eqref{e.superss} is still
a stationary state.

\begin{Lemma}
\label{lemma:stability} Let ${X}$ be a  {\em real} square matrix,
such that ${X} + {X}^T \ge 0$. Then:
\begin{enumerate}
\item[(i)] Any eigenvalue $x_{j}$ of ${X}$ satisfies $\Re x_{j} \ge 0$.
\item[(ii)] For any eigenvalue $x_{j}$ of ${X}$ on the imaginary line, $\Re x_{j}= 0$, its algebraic and geometric multiplicities coincide.
\end{enumerate}
\end{Lemma}
\begin{proof}
\emph{(i)} Let $x_{j}$ be an eigenvalue and let $\bm{u}_{j}$ be its
corresponding eigenvector. One can write ${X}\bm{u}_{j} = x_{j}
\bm{u}_{j}$, and the complex conjugate of this equation
${X}\bm{u}^{*} = x_{j}^{*} \bm{u}^{*}$. Then take a scalar product
of the first equation with $\bm{u}^{\dagger}$ and the scalar product
of the second equation with $\bm{u}$ and sum up:
\begin{equation}
\bm{u}^{\dagger}\cdot ({X}+{X}^T)\bm{u} = (2 \Re x_j)
\bm{u}^{\dagger}\cdot \bm{u}.
\end{equation}
Strict positivity of the eigenvector norm, $\bm{u}^{\dagger} \cdot
\bm{u} > 0$, and the positive semi-definiteness, $({X} + {X}^T)\ge 0$, imply $\Re
x_j \ge 0$.

\emph{(ii)} Consider a linear system of differential equations,
\begin{equation}
(d/d t) \bm{u}(t) = -{X} \bm{u}(t). \label{eq:ode}
\end{equation}
Positive semi-definiteness of ${X} + {X}^T$ is equivalent to
Lyapunov stability in control theory, namely \begin{equation} (d/d
t)||\bm{u}||^2_2 = -\bar{\bm{u}}\cdot ({X} + {X}^T)\bm{u} \le 0
\quad  \textrm{ iff }\quad {X} + {X}^T \ge 0. \end{equation} Then, one can show
that, in order for the system~(\ref{eq:ode}) to be Lyapunov stable, a
purely imaginary (or vanishing) eigenvalue $x_m = i b$ cannot
correspond to a Jordan block of dimension $\ell_m > 1$ in the Jordan
canonical form of ${X}$. This is indeed obvious since, if we take
the initial vector $\bm{u}(0)$ for (\ref{eq:ode}) from ${\rm
ker\,}({X} - x_m \one)^{\ell_m} \ominus {\rm ker\,}({X} - x_m
\one)^{\ell_m-1}$, then $\bm{u}(t) \propto t^{\ell_m-1} e^{-i b t}$.
Hence, if the system~\eqref{eq:ode} is not Lyapunov stable, ${X} + {X}^T \not\ge 0$.
\end{proof}

In \cite{Prosen2010a} it has been shown that the spectrum of the
Liouvillian is
\begin{equation} \begin{aligned}
{\rm{Sp}}({\cal L})= -\{x_{\mathbf{n}}:=\sum_{m} x_m n_m\,/\,
n_m=0,\cdots,\ell_m\}. \label{e.spectrumnd}
\end{aligned} \end{equation} 
Accordingly, $\Delta_{\mathcal L} = \Delta \equiv 2\min_m \Re[x_m]$.
If $\Delta>0$ the steady state \eqref{e.superss} is unique
\cite{Prosen2010a}.

In the non-diagonalizable case the last equation in~\eqref{eq:prop1} is not satisfied. On the other hand one can
obtain the following~\cite{Banchi2014}
\begin{Proposition}\begin{equation}
\|\hat{X}^{-1}\|_\infty < \frac{1+p(\Delta^{-1})}\Delta~,
\label{e.deltadeltand}
\end{equation}
for a certain polynomial $p()$.
\end{Proposition}
\begin{proof}
We start by writing
\begin{equation} \begin{aligned}
\hat X &= \bigoplus_m J_{\ell_m}(x_m)\otimes \one + \bigoplus_m \one
\otimes J_{\ell_m}(x_m) \nonumber\\&= \bigoplus_{m,n}
\left[J_{\ell_m}(x_m)\otimes \one_{\ell_n} + \one_{\ell_m} \otimes
J_{\ell_n}(x_n)\right] \nonumber\\&=\hat x + \bigoplus_{m,n}
\left[J_{\ell_m}(0)\otimes \one_{\ell_n} + \one_{\ell_m} \otimes
J_{\ell_n}(0)\right]~,
\end{aligned} \end{equation} 
where $D_{\hat{X}}$ is the diagonal matrix with entries $x_i + x_j$
and where we used the decomposition $\one = \oplus_m 1_{\ell_m}$.
Moreover, thanks to Lemma 3.1 of Ref.~\cite{Prosen2010a},
\begin{equation} \begin{aligned}
\hat X &= D_{\hat{X}} + \bigoplus_{m,n}
\bigoplus_{r=1}^{\min\{\ell_m,\ell_n\}} J_{\ell_m+\ell_n -2r+1}(0)
\nonumber\\&= D_{\hat{X}}\left[\one + \bigoplus_{m,n}
\bigoplus_{r=1}^{\min\{\ell_m,\ell_n\}} \frac{J_{\ell_m+\ell_n
-2r+1}(0)}{x_m+x_n}\right]~.
\end{aligned} \end{equation} 
As $J$ is nilpotent,
\begin{align*}
\hat{X}^{-1} &= D_{\hat{X}}^{-1} \left[\one + \bigoplus_{m,n}
\bigoplus_{r=1}^{\min\{\ell_m,\ell_n\}}
\sum_{m=1}^{\ell_m+\ell_n-2r}\left(- \frac{J_{\ell_m+\ell_n
-2r+1}(0)}{x_m+x_n}\right)^m\right]~,
\end{align*}
and
\begin{align}
\|\hat{X}^{-1}\|_\infty &\le \|D_{\hat{X}}^{-1}\|_\infty
\left[1 + \max_{m,n} \max_r 
\sum_{m=1}^{\ell_m+\ell_n-2r} \frac{1}{|x_m+x_n|^m}\right]
\nonumber\\&= \|D_{\hat{X}}^{-1}\|_\infty \left[1 + \max_{m,n}
\sum_{m=1}^{\ell_m+\ell_n-2} \frac{1}{|x_m+x_n|^m}\right]
\nonumber\\&\le\frac1\Delta \left[1 + \max_{m,n}
\sum_{m=1}^{\ell_m+\ell_n-2} \frac{1}{\Delta^m}\right]~.
\end{align} 
\end{proof}

\bibliographystyle{elsarticle-num}
\bibliography{ref}

\begin{thebibliography}{100}
\expandafter\ifx\csname url\endcsname\relax
  \def\url#1{\texttt{#1}}\fi
\expandafter\ifx\csname urlprefix\endcsname\relax\def\urlprefix{URL }\fi
\expandafter\ifx\csname href\endcsname\relax
  \def\href#1#2{#2} \def\path#1{#1}\fi

\bibitem{Greiner2002}
M.~Greiner, O.~Mandel, T.~Esslinger, T.~W. H{\"{a}}nsch, I.~Bloch,
  \href{http://www.nature.com/doifinder/10.1038/415039a}{{Quantum phase
  transition from a superfluid to a Mott insulator in a gas of ultracold
  atoms}}, Nature 415~(6867) (2002) 39--44.
\newblock \href {http://dx.doi.org/10.1038/415039a}
  {\path{doi:10.1038/415039a}}.
\newline\urlprefix\url{http://www.nature.com/doifinder/10.1038/415039a}

\bibitem{Nishimori2010}
H.~Nishimori, G.~Ortiz,
  \href{http://www.oxfordscholarship.com/view/10.1093/acprof:oso/
  9780199577224.001.0001/acprof-9780199577224}{{Elements of Phase Transitions
  and Critical Phenomena}}, Oxford University Press, 2010.
\newblock \href {http://dx.doi.org/10.1093/acprof:oso/9780199577224.001.0001}
  {\path{doi:10.1093/acprof:oso/9780199577224.001.0001}}.
\newline\urlprefix\url{http://www.oxfordscholarship.com/view/10.1093/acprof:oso/
  9780199577224.001.0001/acprof-9780199577224}

\bibitem{Mussardo2010}
G.~Mussardo, \href{https://global.oup.com/academic/product/statistical-field-
  theory-9780199547586?cc=it&lang=en&}{{Statistical field theory : an
  introduction to exactly solved models in statistical physics}}, Oxford
  University Press, 2010.
\newline\urlprefix\url{https://global.oup.com/academic/product/statistical-field-
  theory-9780199547586?cc=it&lang=en&}

\bibitem{Chaikin1995}
P.~M. Chaikin, T.~C. Lubensky, {Principles of condensed matter physics},
  Cambridge University Press, 1995.

\bibitem{Goldenfeld1992}
N.~Goldenfeld, {Lectures on phase transitions and the renormalization group},
  Addison-Wesley, Advanced Book Program, 1992, 1992.

\bibitem{Stanley1987}
H.~E. H.~E. Stanley, {Introduction to phase transitions and critical
  phenomena}, Oxford University Press, 1987.

\bibitem{Sachdev2011}
S.~Sachdev, {Quantum Phase Transitions}, Cambridge University press, 2011.

\bibitem{Sondhi1997}
S.~L. Sondhi, S.~M. Girvin, J.~P. Carini, D.~Shahar,
  \href{https://link.aps.org/doi/10.1103/RevModPhys.69.315}{{Continuous quantum
  phase transitions}}, Rev. Mod. Phys. 69~(1) (1997) 315--333.
\newblock \href {http://dx.doi.org/10.1103/RevModPhys.69.315}
  {\path{doi:10.1103/RevModPhys.69.315}}.
\newline\urlprefix\url{https://link.aps.org/doi/10.1103/RevModPhys.69.315}

\bibitem{Vojta2003}
M.~Vojta, \href{http://stacks.iop.org/0034-4885/66/i=12/a=R01?
  key=crossref.c536678c9f7b611291c4445db9c68861}{{Quantum phase transitions}},
  Reports Prog. Phys. 66~(12) (2003) 2069--2110.
\newblock \href {http://dx.doi.org/10.1088/0034-4885/66/12/R01}
  {\path{doi:10.1088/0034-4885/66/12/R01}}.
\newline\urlprefix\url{http://stacks.iop.org/0034-4885/66/i=12/a=R01?
  key=crossref.c536678c9f7b611291c4445db9c68861}

\bibitem{Belitz2005}
D.~Belitz, T.~Vojta,
  \href{https://link.aps.org/doi/10.1103/RevModPhys.77.579}{{How generic scale
  invariance influences quantum and classical phase transitions}}, Rev. Mod.
  Phys. 77~(2) (2005) 579--632.
\newblock \href {http://dx.doi.org/10.1103/RevModPhys.77.579}
  {\path{doi:10.1103/RevModPhys.77.579}}.
\newline\urlprefix\url{https://link.aps.org/doi/10.1103/RevModPhys.77.579}

\bibitem{Carr2011}
L.~Carr, {Understanding quantum phase transitions}, CRC Press, 2011.

\bibitem{Suzuki2013}
S.~Suzuki, J.-i. Inoue, B.~K. Chakrabarti,
  \href{http://link.springer.com/10.1007/978-3-642-33039-1}{{Quantum Ising
  Phases and Transitions in Transverse Ising Models}}, Vol. 862 of Lecture
  Notes in Physics, Springer Berlin Heidelberg, Berlin, Heidelberg, 2013.
\newblock \href {http://dx.doi.org/10.1007/978-3-642-33039-1}
  {\path{doi:10.1007/978-3-642-33039-1}}.
\newline\urlprefix\url{http://link.springer.com/10.1007/978-3-642-33039-1}

\bibitem{Wilson1974}
K.~Wilson, \href{https://linkinghub.elsevier.com/retrieve/pii/
  0370157374900234}{{The renormalization group and the $\epsilon$ expansion}},
  Phys. Rep. 12~(2) (1974) 75--199.
\newblock \href {http://dx.doi.org/10.1016/0370-1573(74)90023-4}
  {\path{doi:10.1016/0370-1573(74)90023-4}}.
\newline\urlprefix\url{https://linkinghub.elsevier.com/retrieve/pii/
  0370157374900234}

\bibitem{Parisi1988}
G.~Parisi, {Statistical field theory}, Addison-Wesley Pub. Co, 1988.

\bibitem{Zinn-Justin2002}
J.~Zinn-Justin, \href{http://www.oxfordscholarship.com/view/10.1093/acprof:oso/
  9780198509233.001.0001/acprof-9780198509233}{{Quantum Field Theory and
  Critical Phenomena}}, Oxford University Press, 2002.
\newblock \href {http://dx.doi.org/10.1093/acprof:oso/9780198509233.001.0001}
  {\path{doi:10.1093/acprof:oso/9780198509233.001.0001}}.
\newline\urlprefix\url{http://www.oxfordscholarship.com/view/10.1093/acprof:oso/
  9780198509233.001.0001/acprof-9780198509233}

\bibitem{Cardy1996}
J.~L. Cardy, {Scaling and renormalization in statistical physics}, Cambridge
  University Press, 1996.

\bibitem{Uhlmann1976}
A.~Uhlmann, \href{http://linkinghub.elsevier.com/retrieve/pii/
  0034487776900604}{{The "transition probability" in the state space of a
  *-algebra}}, Reports Math. Phys. 9~(2) (1976) 273--279.
\newblock \href {http://dx.doi.org/10.1016/0034-4877(76)90060-4}
  {\path{doi:10.1016/0034-4877(76)90060-4}}.
\newline\urlprefix\url{http://linkinghub.elsevier.com/retrieve/pii/
  0034487776900604}

\bibitem{Alberti1983}
P.~M. Alberti, A.~Uhlmann,
  \href{http://link.springer.com/10.1007/BF00419927}{{Stochastic linear maps
  and transition probability}}, Lett. Math. Phys. 7~(2) (1983) 107--112.
\newblock \href {http://dx.doi.org/10.1007/BF00419927}
  {\path{doi:10.1007/BF00419927}}.
\newline\urlprefix\url{http://link.springer.com/10.1007/BF00419927}

\bibitem{Alberti1983a}
P.~M. Alberti, \href{http://link.springer.com/10.1007/BF00398708}{{A note on
  the transition probability over C*-algebras}}, Lett. Math. Phys. 7~(1) (1983)
  25--32.
\newblock \href {http://dx.doi.org/10.1007/BF00398708}
  {\path{doi:10.1007/BF00398708}}.
\newline\urlprefix\url{http://link.springer.com/10.1007/BF00398708}

\bibitem{Alberti1984}
P.~Alberti, A.~Uhlmann, {Transition Probabilities on W∗- and C∗-Algebras.},
  Proc. Second Int. Conf. Oper. Algebr. Ideals, Their Appl. Theor. Phys. (1984)
  5--11.

\bibitem{Wootters1981}
W.~K. Wootters,
  \href{https://link.aps.org/doi/10.1103/PhysRevD.23.357}{{Statistical distance
  and Hilbert space}}, Phys. Rev. D 23~(2) (1981) 357--362.
\newblock \href {http://dx.doi.org/10.1103/PhysRevD.23.357}
  {\path{doi:10.1103/PhysRevD.23.357}}.
\newline\urlprefix\url{https://link.aps.org/doi/10.1103/PhysRevD.23.357}

\bibitem{Jozsa1994}
R.~Jozsa, \href{http://www.tandfonline.com/doi/abs/10.1080/
  09500349414552171}{{Fidelity for Mixed Quantum States}}, J. Mod. Opt. 41~(12)
  (1994) 2315--2323.
\newblock \href {http://dx.doi.org/10.1080/09500349414552171}
  {\path{doi:10.1080/09500349414552171}}.
\newline\urlprefix\url{http://www.tandfonline.com/doi/abs/10.1080/
  09500349414552171}

\bibitem{Schumacher1995}
B.~Schumacher,
  \href{https://link.aps.org/doi/10.1103/PhysRevA.51.2738}{{Quantum coding}},
  Phys. Rev. A 51~(4) (1995) 2738--2747.
\newblock \href {http://dx.doi.org/10.1103/PhysRevA.51.2738}
  {\path{doi:10.1103/PhysRevA.51.2738}}.
\newline\urlprefix\url{https://link.aps.org/doi/10.1103/PhysRevA.51.2738}

\bibitem{Fuchs1996}
C.~A. Fuchs, \href{http://arxiv.org/abs/quant-ph/9601020}{{Distinguishability
  and Accessible Information in Quantum Theory}}, Ph.D. thesis, University of
  New Mexico, Albuquerque (1996).
\newline\urlprefix\url{http://arxiv.org/abs/quant-ph/9601020}

\bibitem{Bures1969}
D.~Bures, \href{http://www.ams.org/jourcgi/jour-getitem?pii=S0002-9947-1969-
  0236719-2}{{An extension of Kakutani's theorem on infinite product measures
  to the tensor product of semifinite w*-algebras}}, Trans. Am. Math. Soc.
  135~(5) (1969) 199--199.
\newblock \href {http://dx.doi.org/10.1090/S0002-9947-1969-0236719-2}
  {\path{doi:10.1090/S0002-9947-1969-0236719-2}}.
\newline\urlprefix\url{http://www.ams.org/jourcgi/jour-getitem?pii=S0002-9947-1969-
  0236719-2}

\bibitem{Nakahara1990}
M.~Nakahara, {Geometry, Topology and Physics}, Graduate Student Series in
  Physics, Adam Hilger, Bristol and New York, 1990.

\bibitem{Berry1984}
M.~V. Berry, \href{http://rspa.royalsocietypublishing.org/cgi/doi/10.1098/
  rspa.1984.0023}{{Quantal phase factors accompanying adiabatic changes}},
  Proc. R. Soc. London, Ser. A 392~(1802) (1984) 45--57.
\newblock \href {http://dx.doi.org/10.1098/rspa.1984.0023}
  {\path{doi:10.1098/rspa.1984.0023}}.
\newline\urlprefix\url{http://rspa.royalsocietypublishing.org/cgi/doi/10.1098/
  rspa.1984.0023}

\bibitem{Berry1989}
M.~V. Berry, {The Quantum Phase, Five Years After}, World Scientific
  (Singapore), 1989.
\newblock \href {http://dx.doi.org/doi.org/10.1142/9789812798381_0001}
  {\path{doi:doi.org/10.1142/9789812798381_0001}}.

\bibitem{Wilczek1989}
F.~Wilczek, A.~Shapere,
  \href{https://www.worldscientific.com/worldscibooks/10.1142/0613}{{Geometric
  Phases in Physics}}, Vol.~5 of Advanced Series in Mathematical Physics, WORLD
  SCIENTIFIC, 1989.
\newblock \href {http://dx.doi.org/10.1142/0613} {\path{doi:10.1142/0613}}.
\newline\urlprefix\url{https://www.worldscientific.com/worldscibooks/10.1142/0613}

\bibitem{Bohm2003}
A.~Bohm, A.~Mostafazadeh, H.~Koizumi, Q.~Niu, J.~Zwanziger,
  \href{http://link.springer.com/10.1007/978-3-662-10333-3}{{The Geometric
  Phase in Quantum Systems}}, Springer Berlin Heidelberg, Berlin, Heidelberg,
  2003.
\newblock \href {http://dx.doi.org/10.1007/978-3-662-10333-3}
  {\path{doi:10.1007/978-3-662-10333-3}}.
\newline\urlprefix\url{http://link.springer.com/10.1007/978-3-662-10333-3}

\bibitem{Carollo2005}
A.~C.~M. Carollo, J.~K. Pachos,
  \href{https://link.aps.org/doi/10.1103/PhysRevLett.95.157203}{{Geometric
  Phases and Criticality in Spin-Chain Systems}}, Phys. Rev. Lett. 95~(15)
  (2005) 157203.
\newblock \href {http://dx.doi.org/10.1103/PhysRevLett.95.157203}
  {\path{doi:10.1103/PhysRevLett.95.157203}}.
\newline\urlprefix\url{https://link.aps.org/doi/10.1103/PhysRevLett.95.157203}

\bibitem{Hamma2006}
A.~Hamma, {Berry Phases and Quantum Phase Transitions}.

\bibitem{Zhu2008}
S.-L. Zhu, \href{http://www.worldscientific.com/doi/abs/10.1142/
  S0217979208038855}{{Geometric Phases and Quantum Phase Transitions}}, Int. J.
  Mod. Phys. B 22~(06) (2008) 561--581.
\newblock \href {http://dx.doi.org/10.1142/S0217979208038855}
  {\path{doi:10.1142/S0217979208038855}}.
\newline\urlprefix\url{http://www.worldscientific.com/doi/abs/10.1142/
  S0217979208038855}

\bibitem{Thouless1983}
D.~J. Thouless, {Quantization of particle transport}, Phys. Rev. B 27~(10)
  (1983) 6083--6087.
\newblock \href {http://dx.doi.org/10.1103/PhysRevB.27.6083}
  {\path{doi:10.1103/PhysRevB.27.6083}}.

\bibitem{Bernevig2013}
B.~A. Bernevig, T.~L. Hughes, {Topological insulators and topological
  superconductors}, Princeton University Press, 2013.

\bibitem{Chiu2016}
C.-K. Chiu, J.~C.~Y. Teo, A.~P. Schnyder, S.~Ryu,
  \href{https://link.aps.org/doi/10.1103/RevModPhys.88.035005}{{Classification
  of topological quantum matter with symmetries}}, Rev. Mod. Phys. 88~(3)
  (2016) 035005.
\newblock \href {http://dx.doi.org/10.1103/RevModPhys.88.035005}
  {\path{doi:10.1103/RevModPhys.88.035005}}.
\newline\urlprefix\url{https://link.aps.org/doi/10.1103/RevModPhys.88.035005}

\bibitem{Pachos2006}
J.~K. Pachos, A.~C. Carollo,
  \href{http://www.royalsocietypublishing.org/doi/10.1098/
  rsta.2006.1894}{{Geometric phases and criticality in spin systems}}, Philos.
  Trans. R. Soc. A 364~(1849) (2006) 3463--3476.
\newblock \href {http://dx.doi.org/10.1098/rsta.2006.1894}
  {\path{doi:10.1098/rsta.2006.1894}}.
\newline\urlprefix\url{http://www.royalsocietypublishing.org/doi/10.1098/
  rsta.2006.1894}

\bibitem{Plastina2006}
F.~Plastina, G.~Liberti, A.~Carollo,
  \href{http://stacks.iop.org/0295-5075/76/i=2/a=182?
  key=crossref.35948bee527d447741539fbe9a87bddc}{{Scaling of Berry's phase
  close to the Dicke quantum phase transition}}, Europhys. Lett. 76~(2) (2006)
  182--188.
\newblock \href {http://dx.doi.org/10.1209/epl/i2006-10270-x}
  {\path{doi:10.1209/epl/i2006-10270-x}}.
\newline\urlprefix\url{http://stacks.iop.org/0295-5075/76/i=2/a=182?
  key=crossref.35948bee527d447741539fbe9a87bddc}

\bibitem{Zhu2006}
S.-l. Zhu,
  \href{https://link.aps.org/doi/10.1103/PhysRevLett.96.077206}{{Scaling of
  geometric phases close to the quantum phase transition in the XY spin
  chain}}, Phys. Rev. Lett. 96~(7) (2006) 077206.
\newblock \href {http://dx.doi.org/10.1103/PhysRevLett.96.077206}
  {\path{doi:10.1103/PhysRevLett.96.077206}}.
\newline\urlprefix\url{https://link.aps.org/doi/10.1103/PhysRevLett.96.077206}

\bibitem{Reuter2007}
M.~E. Reuter, M.~J. Hartmann, M.~B. Plenio,
  \href{http://www.royalsocietypublishing.org/doi/10.1098/
  rspa.2007.1822}{{Geometric phases and critical phenomena in a chain of
  interacting spins}}, Proc. R. Soc. A Math. Phys. Eng. Sci. 463~(2081) (2007)
  1271--1285.
\newblock \href {http://dx.doi.org/10.1098/rspa.2007.1822}
  {\path{doi:10.1098/rspa.2007.1822}}.
\newline\urlprefix\url{http://www.royalsocietypublishing.org/doi/10.1098/
  rspa.2007.1822}

\bibitem{Peng2010}
X.~Peng, S.~Wu, J.~Li, D.~Suter, J.~Du,
  \href{https://link.aps.org/doi/10.1103/PhysRevLett.105.240405}{{Observation
  of the Ground-State Geometric Phase in a Heisenberg XY Model}}, Phys. Rev.
  Lett. 105~(24) (2010) 240405.
\newblock \href {http://dx.doi.org/10.1103/PhysRevLett.105.240405}
  {\path{doi:10.1103/PhysRevLett.105.240405}}.
\newline\urlprefix\url{https://link.aps.org/doi/10.1103/PhysRevLett.105.240405}

\bibitem{CamposVenuti2007}
L.~{Campos Venuti}, P.~Zanardi,
  \href{https://link.aps.org/doi/10.1103/PhysRevLett.99.095701}{{Quantum
  Critical Scaling of the Geometric Tensors}}, Phys. Rev. Lett. 99~(9) (2007)
  095701.
\newblock \href {http://dx.doi.org/10.1103/PhysRevLett.99.095701}
  {\path{doi:10.1103/PhysRevLett.99.095701}}.
\newline\urlprefix\url{https://link.aps.org/doi/10.1103/PhysRevLett.99.095701}

\bibitem{Cui2006}
H.~Cui, K.~Li, X.~Yi, \href{https://linkinghub.elsevier.com/retrieve/pii/
  S0375960106012795}{{Geometric phase and quantum phase transition in the
  Lipkin–Meshkov–Glick model}}, Phys. Lett. A 360~(2) (2006) 243--248.
\newblock \href {http://dx.doi.org/10.1016/j.physleta.2006.08.040}
  {\path{doi:10.1016/j.physleta.2006.08.040}}.
\newline\urlprefix\url{https://linkinghub.elsevier.com/retrieve/pii/
  S0375960106012795}

\bibitem{Chen2006}
G.~Chen, J.~Li, J.-Q. Liang,
  \href{https://link.aps.org/doi/10.1103/PhysRevA.74.054101}{{Critical property
  of the geometric phase in the Dicke model}}, Phys. Rev. A 74~(5) (2006)
  054101.
\newblock \href {http://dx.doi.org/10.1103/PhysRevA.74.054101}
  {\path{doi:10.1103/PhysRevA.74.054101}}.
\newline\urlprefix\url{https://link.aps.org/doi/10.1103/PhysRevA.74.054101}

\bibitem{Yi2007}
X.~X. Yi, W.~Wang,
  \href{https://link.aps.org/doi/10.1103/PhysRevA.75.032103}{{Geometric phases
  induced in auxiliary qubits by many-body systems near their critical
  points}}, Phys. Rev. A 75~(3) (2007) 032103.
\newblock \href {http://dx.doi.org/10.1103/PhysRevA.75.032103}
  {\path{doi:10.1103/PhysRevA.75.032103}}.
\newline\urlprefix\url{https://link.aps.org/doi/10.1103/PhysRevA.75.032103}

\bibitem{Yuan2007}
Z.-G. Yuan, P.~Zhang, S.-S. Li,
  \href{https://link.aps.org/doi/10.1103/PhysRevA.75.012102}{{Loschmidt echo
  and Berry phase of a quantum system coupled to an XY spin chain: Proximity to
  a quantum phase transition}}, Phys. Rev. A 75~(1) (2007) 012102.
\newblock \href {http://dx.doi.org/10.1103/PhysRevA.75.012102}
  {\path{doi:10.1103/PhysRevA.75.012102}}.
\newline\urlprefix\url{https://link.aps.org/doi/10.1103/PhysRevA.75.012102}

\bibitem{Cui2008}
H.~T. Cui, J.~Yi,
  \href{https://link.aps.org/doi/10.1103/PhysRevA.78.022101}{{Geometric phase
  and quantum phase transition: Two-band model}}, Phys. Rev. A 78~(2) (2008)
  022101.
\newblock \href {http://dx.doi.org/10.1103/PhysRevA.78.022101}
  {\path{doi:10.1103/PhysRevA.78.022101}}.
\newline\urlprefix\url{https://link.aps.org/doi/10.1103/PhysRevA.78.022101}

\bibitem{Furtado2008}
C.~Furtado, F.~Moraes, A.~{de M. Carvalho},
  \href{https://linkinghub.elsevier.com/retrieve/pii/
  S0375960108009249}{{Geometric phases in graphitic cones}}, Phys. Lett. A
  372~(32) (2008) 5368--5371.
\newblock \href {http://dx.doi.org/10.1016/j.physleta.2008.06.029}
  {\path{doi:10.1016/j.physleta.2008.06.029}}.
\newline\urlprefix\url{https://linkinghub.elsevier.com/retrieve/pii/
  S0375960108009249}

\bibitem{Hu2008}
M.-G. Hu, K.~Xue, M.-L. Ge,
  \href{https://link.aps.org/doi/10.1103/PhysRevA.78.052324}{{Exact solution of
  a Yang-Baxter spin- 1/2 chain model and quantum entanglement}}, Phys. Rev. A
  78~(5) (2008) 052324.
\newblock \href {http://dx.doi.org/10.1103/PhysRevA.78.052324}
  {\path{doi:10.1103/PhysRevA.78.052324}}.
\newline\urlprefix\url{https://link.aps.org/doi/10.1103/PhysRevA.78.052324}

\bibitem{Nesterov2008}
A.~I. Nesterov, S.~G. Ovchinnikov,
  \href{https://link.aps.org/doi/10.1103/PhysRevE.78.015202}{{Geometric phases
  and quantum phase transitions in open systems}}, Phys. Rev. E 78~(1) (2008)
  015202.
\newblock \href {http://dx.doi.org/10.1103/PhysRevE.78.015202}
  {\path{doi:10.1103/PhysRevE.78.015202}}.
\newline\urlprefix\url{https://link.aps.org/doi/10.1103/PhysRevE.78.015202}

\bibitem{Paunkovic2008}
N.~Paunkovi{\'{c}}, V.~{Rocha Vieira},
  \href{https://link.aps.org/doi/10.1103/PhysRevE.77.011129}{{Macroscopic
  distinguishability between quantum states defining different phases of
  matter: Fidelity and the Uhlmann geometric phase}}, Phys. Rev. E 77~(1)
  (2008) 011129.
\newblock \href {http://dx.doi.org/10.1103/PhysRevE.77.011129}
  {\path{doi:10.1103/PhysRevE.77.011129}}.
\newline\urlprefix\url{https://link.aps.org/doi/10.1103/PhysRevE.77.011129}

\bibitem{Contreras2008}
H.~Contreras, A.~Reyes-Lega,
  \href{https://linkinghub.elsevier.com/retrieve/pii/ S0921452607011313}{{Berry
  phases, quantum phase transitions and Chern numbers}}, Phys. B Condens.
  Matter 403~(5-9) (2008) 1301--1302.
\newblock \href {http://dx.doi.org/10.1016/j.physb.2007.10.131}
  {\path{doi:10.1016/j.physb.2007.10.131}}.
\newline\urlprefix\url{https://linkinghub.elsevier.com/retrieve/pii/
  S0921452607011313}

\bibitem{Ma2009}
Y.-Q. Ma, S.~Chen,
  \href{https://link.aps.org/doi/10.1103/PhysRevA.79.022116}{{Geometric phase
  and quantum phase transition in an inhomogeneous periodic XY spin-1/2
  model}}, Phys. Rev. A 79~(2) (2009) 022116.
\newblock \href {http://dx.doi.org/10.1103/PhysRevA.79.022116}
  {\path{doi:10.1103/PhysRevA.79.022116}}.
\newline\urlprefix\url{https://link.aps.org/doi/10.1103/PhysRevA.79.022116}

\bibitem{Oh2009}
S.~Oh, \href{https://linkinghub.elsevier.com/retrieve/pii/
  S0375960108017635}{{Geometric phases and entanglement of two qubits with XY
  type interaction}}, Phys. Lett. A 373~(6) (2009) 644--647.
\newblock \href {http://dx.doi.org/10.1016/j.physleta.2008.12.023}
  {\path{doi:10.1016/j.physleta.2008.12.023}}.
\newline\urlprefix\url{https://linkinghub.elsevier.com/retrieve/pii/
  S0375960108017635}

\bibitem{Nesterov2009}
A.~I. Nesterov, S.~G. Ovchinnikov,
  \href{http://link.springer.com/10.1134/S0021364009190072}{{Spin crossover:
  the quantum phase transition induced by high pressure}}, JETP Lett. 90~(7)
  (2009) 530--534.
\newblock \href {http://dx.doi.org/10.1134/S0021364009190072}
  {\path{doi:10.1134/S0021364009190072}}.
\newline\urlprefix\url{http://link.springer.com/10.1134/S0021364009190072}

\bibitem{Cui2009}
H.~T. Cui, Y.~F. Zhang,
  \href{http://www.springerlink.com/index/10.1140/epjd/e2009-00025-
  9}{{Pairwise entanglement and geometric phase in high dimensional
  free-Fermion lattice systems}}, Eur. Phys. J. D 51~(3) (2009) 393--400.
\newblock \href {http://dx.doi.org/10.1140/epjd/e2009-00025-9}
  {\path{doi:10.1140/epjd/e2009-00025-9}}.
\newline\urlprefix\url{http://www.springerlink.com/index/10.1140/epjd/e2009-00025-
  9}

\bibitem{Quan2009}
H.~T. Quan, \href{http://stacks.iop.org/1751-8121/42/i=39/a=395002?
  key=crossref.727dcff978dbbffece51ab7b0bf75b51}{{Finite-temperature scaling of
  magnetic susceptibility and the geometric phase in the XY spin chain}}, J.
  Phys. A Math. Theor. 42~(39) (2009) 395002.
\newblock \href {http://dx.doi.org/10.1088/1751-8113/42/39/395002}
  {\path{doi:10.1088/1751-8113/42/39/395002}}.
\newline\urlprefix\url{http://stacks.iop.org/1751-8121/42/i=39/a=395002?
  key=crossref.727dcff978dbbffece51ab7b0bf75b51}

\bibitem{Wang2010b}
G.~Wang, K.~Xue, C.~Sun, T.~Hu, C.~Zhou, G.~Du,
  \href{http://link.springer.com/10.1007/s10773-010-0435-x}{{Quantum Phase
  Transition Like Phenomenon in a Two-Qubit Yang-Baxter System}}, Int. J.
  Theor. Phys. 49~(10) (2010) 2499--2505.
\newblock \href {http://dx.doi.org/10.1007/s10773-010-0435-x}
  {\path{doi:10.1007/s10773-010-0435-x}}.
\newline\urlprefix\url{http://link.springer.com/10.1007/s10773-010-0435-x}

\bibitem{Wang2010}
L.~C. Wang, X.~X. Yi,
  \href{http://www.springerlink.com/index/10.1140/epjd/e2010-00045-
  4}{{Geometric phase and quantum phase transition in the one-dimensional
  compass model}}, Eur. Phys. J. D 57~(2) (2010) 281--286.
\newblock \href {http://dx.doi.org/10.1140/epjd/e2010-00045-4}
  {\path{doi:10.1140/epjd/e2010-00045-4}}.
\newline\urlprefix\url{http://www.springerlink.com/index/10.1140/epjd/e2010-00045-
  4}

\bibitem{Sjoqvist2010}
E.~Sj{\"{o}}qvist, R.~Rahaman, U.~Basu, B.~Basu,
  \href{http://stacks.iop.org/1751-8121/43/i=35/a=354026?
  key=crossref.e3039e95c0afd658c7946ea2d515a0e5}{{Berry phase and fidelity
  susceptibility of the three-qubit Lipkin–Meshkov–Glick ground state}}, J.
  Phys. A Math. Theor. 43~(35) (2010) 354026.
\newblock \href {http://dx.doi.org/10.1088/1751-8113/43/35/354026}
  {\path{doi:10.1088/1751-8113/43/35/354026}}.
\newline\urlprefix\url{http://stacks.iop.org/1751-8121/43/i=35/a=354026?
  key=crossref.e3039e95c0afd658c7946ea2d515a0e5}

\bibitem{Basu2010}
B.~Basu, P.~Bandyopadhyay,
  \href{http://stacks.iop.org/1751-8121/43/i=35/a=354023?
  key=crossref.8f9a51bb9ed5af39414ee38c5d21f5d4}{{The geometric phase and the
  dynamics of quantum phase transition induced by a linear quench}}, J. Phys. A
  Math. Theor. 43~(35) (2010) 354023.
\newblock \href {http://dx.doi.org/10.1088/1751-8113/43/35/354023}
  {\path{doi:10.1088/1751-8113/43/35/354023}}.
\newline\urlprefix\url{http://stacks.iop.org/1751-8121/43/i=35/a=354023?
  key=crossref.8f9a51bb9ed5af39414ee38c5d21f5d4}

\bibitem{Lu2010}
X.-M. Lu, X.~Wang, \href{http://stacks.iop.org/0295-5075/91/i=3/a=30003?
  key=crossref.a621f08a017adfa346a410d3c33a93ba}{{Operator quantum geometric
  tensor and quantum phase transitions}}, EPL (Europhysics Lett. 91~(3) (2010)
  30003.
\newblock \href {http://dx.doi.org/10.1209/0295-5075/91/30003}
  {\path{doi:10.1209/0295-5075/91/30003}}.
\newline\urlprefix\url{http://stacks.iop.org/0295-5075/91/i=3/a=30003?
  key=crossref.a621f08a017adfa346a410d3c33a93ba}

\bibitem{Li2010}
L.~Zhi-Jian, C.~Lu, W.~Jiao-Jin,
  \href{http://stacks.iop.org/1674-1056/19/i=1/a=010305?
  key=crossref.8fba5aab45e5735584280bec0cfe8b43}{{Critical entanglement and
  geometric phase of a two-qubit model with Dzyaloshinski–Moriya anisotropic
  interaction}}, Chinese Phys. B 19~(1) (2010) 010305--5.
\newblock \href {http://dx.doi.org/10.1088/1674-1056/19/1/010305}
  {\path{doi:10.1088/1674-1056/19/1/010305}}.
\newline\urlprefix\url{http://stacks.iop.org/1674-1056/19/i=1/a=010305?
  key=crossref.8fba5aab45e5735584280bec0cfe8b43}

\bibitem{Basu2010a}
B.~Basu, \href{https://linkinghub.elsevier.com/retrieve/pii/
  S037596011000006X}{{Dynamics of the geometric phase in the adiabatic limit of
  a quench induced quantum phase transition}}, Phys. Lett. A 374~(10) (2010)
  1205--1208.
\newblock \href {http://dx.doi.org/10.1016/j.physleta.2009.12.072}
  {\path{doi:10.1016/j.physleta.2009.12.072}}.
\newline\urlprefix\url{https://linkinghub.elsevier.com/retrieve/pii/
  S037596011000006X}

\bibitem{Zhong2010}
M.~Zhong, P.~Tong, \href{http://stacks.iop.org/1751-8121/43/i=50/a=505302?
  key=crossref.785360de7a2d61cfc0a1b8e8747a50ff}{{The Ising and anisotropy
  phase transitions of the periodic XY model in a transverse field}}, J. Phys.
  A Math. Theor. 43~(50) (2010) 505302.
\newblock \href {http://dx.doi.org/10.1088/1751-8113/43/50/505302}
  {\path{doi:10.1088/1751-8113/43/50/505302}}.
\newline\urlprefix\url{http://stacks.iop.org/1751-8121/43/i=50/a=505302?
  key=crossref.785360de7a2d61cfc0a1b8e8747a50ff}

\bibitem{Cucchietti2010}
F.~M. Cucchietti, J.-F. Zhang, F.~C. Lombardo, P.~I. Villar, R.~Laflamme,
  \href{https://link.aps.org/doi/10.1103/PhysRevLett.105.240406}{{Geometric
  Phase with Nonunitary Evolution in the Presence of a Quantum Critical Bath}},
  Phys. Rev. Lett. 105~(24) (2010) 240406.
\newblock \href {http://dx.doi.org/10.1103/PhysRevLett.105.240406}
  {\path{doi:10.1103/PhysRevLett.105.240406}}.
\newline\urlprefix\url{https://link.aps.org/doi/10.1103/PhysRevLett.105.240406}

\bibitem{Yuan2010}
X.-Z. Yuan, H.-S. Goan, K.-D. Zhu,
  \href{https://link.aps.org/doi/10.1103/PhysRevA.81.034102}{{Geometric phase
  of a central spin coupled to an antiferromagnetic environment}}, Phys. Rev. A
  81~(3) (2010) 034102.
\newblock \href {http://dx.doi.org/10.1103/PhysRevA.81.034102}
  {\path{doi:10.1103/PhysRevA.81.034102}}.
\newline\urlprefix\url{https://link.aps.org/doi/10.1103/PhysRevA.81.034102}

\bibitem{Cheng2010a}
W.~Cheng, C.~Shan, Y.~Huang, T.~Liu, H.~Li,
  \href{https://linkinghub.elsevier.com/retrieve/pii/
  S092145261000880X}{{Geometric phase signature of quantum criticality in the
  XY spin chain with multiple interaction}}, Phys. B Condens. Matter 405~(23)
  (2010) 4821--4824.
\newblock \href {http://dx.doi.org/10.1016/j.physb.2010.09.012}
  {\path{doi:10.1016/j.physb.2010.09.012}}.
\newline\urlprefix\url{https://linkinghub.elsevier.com/retrieve/pii/
  S092145261000880X}

\bibitem{Bandyopadhyay2011}
P.~Bandyopadhyay, \href{http://www.royalsocietypublishing.org/doi/10.1098/
  rspa.2010.0266}{{Anisotropic spin system, quantized Dirac monopole and the
  Berry phase}}, Proc. R. Soc. A Math. Phys. Eng. Sci. 467~(2126) (2011)
  427--438.
\newblock \href {http://dx.doi.org/10.1098/rspa.2010.0266}
  {\path{doi:10.1098/rspa.2010.0266}}.
\newline\urlprefix\url{http://www.royalsocietypublishing.org/doi/10.1098/
  rspa.2010.0266}

\bibitem{Ribeiro2011}
F.~Ribeiro, J.~de~Lima, L.~Gon{\c{c}}alves,
  \href{https://linkinghub.elsevier.com/retrieve/pii/
  S0304885310005664}{{Quantum phase transitions of the extended isotropic XY
  model with long-range interactions}}, J. Magn. Magn. Mater. 323~(1) (2011)
  39--50.
\newblock \href {http://dx.doi.org/10.1016/j.jmmm.2010.08.027}
  {\path{doi:10.1016/j.jmmm.2010.08.027}}.
\newline\urlprefix\url{https://linkinghub.elsevier.com/retrieve/pii/
  S0304885310005664}

\bibitem{Lian2011}
H.~Lian, D.~Tian, \href{https://linkinghub.elsevier.com/retrieve/pii/
  S0375960111009984}{{Quantum phase transition in XY spin chain with three-site
  interaction studied in terms of Loschmidt echo and Berry phase}}, Phys. Lett.
  A 375~(41) (2011) 3604--3609.
\newblock \href {http://dx.doi.org/10.1016/j.physleta.2011.08.025}
  {\path{doi:10.1016/j.physleta.2011.08.025}}.
\newline\urlprefix\url{https://linkinghub.elsevier.com/retrieve/pii/
  S0375960111009984}

\bibitem{Tian2011}
L.-J. Tian, C.-Q. Zhu, H.-B. Zhang, L.-G. Qin,
  \href{http://stacks.iop.org/1674-1056/20/i=4/a=040302?
  key=crossref.000b4c635960dcc1fee133904cfc8ccc}{{Fidelity susceptibility and
  geometric phase in critical phenomenon}}, Chinese Phys. B 20~(4) (2011)
  040302.
\newblock \href {http://dx.doi.org/10.1088/1674-1056/20/4/040302}
  {\path{doi:10.1088/1674-1056/20/4/040302}}.
\newline\urlprefix\url{http://stacks.iop.org/1674-1056/20/i=4/a=040302?
  key=crossref.000b4c635960dcc1fee133904cfc8ccc}

\bibitem{Li2011}
S.-C. Li, L.-B. Fu,
  \href{https://link.aps.org/doi/10.1103/PhysRevA.84.023605}{{Quantum phase
  transition from mixed atom-molecule phase to pure molecule phase:
  Characteristic scaling laws and Berry-curvature signature}}, Phys. Rev. A
  84~(2) (2011) 023605.
\newblock \href {http://dx.doi.org/10.1103/PhysRevA.84.023605}
  {\path{doi:10.1103/PhysRevA.84.023605}}.
\newline\urlprefix\url{https://link.aps.org/doi/10.1103/PhysRevA.84.023605}

\bibitem{Li2011a}
S.~C. Li, J.~Liu, L.~B. Fu,
  \href{https://link.aps.org/doi/10.1103/PhysRevA.83.042107}{{Berry phase and
  Hannay angle of an interacting boson system}}, Phys. Rev. A 83~(4) (2011)
  042107.
\newblock \href {http://dx.doi.org/10.1103/PhysRevA.83.042107}
  {\path{doi:10.1103/PhysRevA.83.042107}}.
\newline\urlprefix\url{https://link.aps.org/doi/10.1103/PhysRevA.83.042107}

\bibitem{Castro2011}
C.~S. Castro, M.~S. Sarandy,
  \href{https://link.aps.org/doi/10.1103/PhysRevA.83.042334}{{Entanglement
  dynamics via geometric phases in quantum spin chains}}, Phys. Rev. A 83~(4)
  (2011) 042334.
\newblock \href {http://dx.doi.org/10.1103/PhysRevA.83.042334}
  {\path{doi:10.1103/PhysRevA.83.042334}}.
\newline\urlprefix\url{https://link.aps.org/doi/10.1103/PhysRevA.83.042334}

\bibitem{Patra2011}
A.~Patra, V.~Mukherjee, A.~Dutta,
  \href{http://stacks.iop.org/1742-5468/2011/i=03/a=P03026?
  key=crossref.24db6a770fefe7e8cdec02cabe975716}{{Path-dependent scaling of
  geometric phase near a quantum multi-critical point}}, J. Stat. Mech. Theory
  Exp. 2011~(03) (2011) P03026.
\newblock \href {http://dx.doi.org/10.1088/1742-5468/2011/03/P03026}
  {\path{doi:10.1088/1742-5468/2011/03/P03026}}.
\newline\urlprefix\url{http://stacks.iop.org/1742-5468/2011/i=03/a=P03026?
  key=crossref.24db6a770fefe7e8cdec02cabe975716}

\bibitem{Zhang2011}
L.-D. Zhang, L.-B. Fu, \href{http://stacks.iop.org/0295-5075/93/i=3/a=30001?
  key=crossref.adc6d3dceae77ed1e3d5049ad649ab72}{{Mean-field Berry phase of an
  interacting spin-1/2 system}}, Europhys. Lett. 93~(3) (2011) 30001.
\newblock \href {http://dx.doi.org/10.1209/0295-5075/93/30001}
  {\path{doi:10.1209/0295-5075/93/30001}}.
\newline\urlprefix\url{http://stacks.iop.org/0295-5075/93/i=3/a=30001?
  key=crossref.adc6d3dceae77ed1e3d5049ad649ab72}

\bibitem{Zhang2012}
X.-x. Zhang, A.-p. Zhang, F.-l. Li,
  \href{https://linkinghub.elsevier.com/retrieve/pii/
  S0375960112005634}{{Detecting the multi-spin interaction of an XY spin chain
  by the geometric phase of a coupled qubit}}, Phys. Lett. A 376~(30-31) (2012)
  2090--2095.
\newblock \href {http://dx.doi.org/10.1016/j.physleta.2012.05.018}
  {\path{doi:10.1016/j.physleta.2012.05.018}}.
\newline\urlprefix\url{https://linkinghub.elsevier.com/retrieve/pii/
  S0375960112005634}

\bibitem{Yuan2012}
Z.-G. Yuan, P.~Zhang, S.-S. Li, J.~Jing, L.-B. Kong,
  \href{https://link.aps.org/doi/10.1103/PhysRevA.85.044102}{{Scaling of the
  Berry phase close to the excited-state quantum phase transition in the Lipkin
  model}}, Phys. Rev. A 85~(4) (2012) 044102.
\newblock \href {http://dx.doi.org/10.1103/PhysRevA.85.044102}
  {\path{doi:10.1103/PhysRevA.85.044102}}.
\newline\urlprefix\url{https://link.aps.org/doi/10.1103/PhysRevA.85.044102}

\bibitem{Requist2012}
R.~Requist,
  \href{https://link.aps.org/doi/10.1103/PhysRevA.86.022117}{{Hamiltonian
  formulation of nonequilibrium quantum dynamics: Geometric structure of the
  Bogoliubov-Born-Green-Kirkwood-Yvon hierarchy}}, Phys. Rev. A 86~(2) (2012)
  022117.
\newblock \href {http://dx.doi.org/10.1103/PhysRevA.86.022117}
  {\path{doi:10.1103/PhysRevA.86.022117}}.
\newline\urlprefix\url{https://link.aps.org/doi/10.1103/PhysRevA.86.022117}

\bibitem{Shan2012}
C.-J. Shan, {Berry phase and quantum phase transition in spin chain system with
  three-site interaction}, Wuli Xuebao/Acta Phys. Sin. 61~(22).

\bibitem{Tomka2012}
M.~Tomka, A.~Polkovnikov, V.~Gritsev,
  \href{https://link.aps.org/doi/10.1103/PhysRevLett.108.080404}{{Geometric
  Phase Contribution to Quantum Nonequilibrium Many-Body Dynamics}}, Phys. Rev.
  Lett. 108~(8) (2012) 080404.
\newblock \href {http://dx.doi.org/10.1103/PhysRevLett.108.080404}
  {\path{doi:10.1103/PhysRevLett.108.080404}}.
\newline\urlprefix\url{https://link.aps.org/doi/10.1103/PhysRevLett.108.080404}

\bibitem{Lian2012}
J.~Lian, J.~Q. Liang, G.~Chen,
  \href{http://www.springerlink.com/index/10.1140/epjb/e2012-20901-
  1}{{Geometric phase in the Kitaev honeycomb model and scaling behaviour at
  critical points}}, Eur. Phys. J. B 85~(6) (2012) 207.
\newblock \href {http://dx.doi.org/10.1140/epjb/e2012-20901-1}
  {\path{doi:10.1140/epjb/e2012-20901-1}}.
\newline\urlprefix\url{http://www.springerlink.com/index/10.1140/epjb/e2012-20901-
  1}

\bibitem{Ma2012}
Y.-Q. Ma, Z.-X. Yu, D.-S. Wang, B.-H. Xie, X.-G. Li,
  \href{http://stacks.iop.org/0295-5075/100/i=6/a=60001?
  key=crossref.d70e83756fa184dcd4152d727b4d34ea}{{Momentum space Z 2 number,
  quantized Berry phase and the quantum phase transitions in spin chain
  systems}}, Europhys. Lett. 100~(6) (2012) 60001.
\newblock \href {http://dx.doi.org/10.1209/0295-5075/100/60001}
  {\path{doi:10.1209/0295-5075/100/60001}}.
\newline\urlprefix\url{http://stacks.iop.org/0295-5075/100/i=6/a=60001?
  key=crossref.d70e83756fa184dcd4152d727b4d34ea}

\bibitem{Ma2013}
Y.-Q. Ma, S.-J. Gu, S.~Chen, H.~Fan, W.-M. Liu,
  \href{http://stacks.iop.org/0295-5075/103/i=1/a=10008?
  key=crossref.dff62933877e7a844bb6359cf0bc8f46}{{The Euler number of Bloch
  states manifold and the quantum phases in gapped fermionic systems}}, EPL
  (Europhysics Lett. 103~(1) (2013) 10008.
\newblock \href {http://dx.doi.org/10.1209/0295-5075/103/10008}
  {\path{doi:10.1209/0295-5075/103/10008}}.
\newline\urlprefix\url{http://stacks.iop.org/0295-5075/103/i=1/a=10008?
  key=crossref.dff62933877e7a844bb6359cf0bc8f46}

\bibitem{Zhang2013}
X.~Z. Zhang, Z.~Song,
  \href{https://link.aps.org/doi/10.1103/PhysRevA.88.042108}{{Geometric phase
  and phase diagram for a non-Hermitian quantum XY model}}, Phys. Rev. A 88~(4)
  (2013) 042108.
\newblock \href {http://dx.doi.org/10.1103/PhysRevA.88.042108}
  {\path{doi:10.1103/PhysRevA.88.042108}}.
\newline\urlprefix\url{https://link.aps.org/doi/10.1103/PhysRevA.88.042108}

\bibitem{Zhang2013a}
A.-p. Zhang, F.-l. Li, \href{https://linkinghub.elsevier.com/retrieve/pii/
  S0375960112013114}{{Geometric phase of a central qubit coupled to a spin
  chain in a thermal equilibrium state}}, Phys. Lett. A 377~(7) (2013)
  528--533.
\newblock \href {http://dx.doi.org/10.1016/j.physleta.2012.12.028}
  {\path{doi:10.1016/j.physleta.2012.12.028}}.
\newline\urlprefix\url{https://linkinghub.elsevier.com/retrieve/pii/
  S0375960112013114}

\bibitem{Liang2013}
S.-D. Liang, G.-Y. Huang,
  \href{https://link.aps.org/doi/10.1103/PhysRevA.87.012118}{{Topological
  invariance and global Berry phase in non-Hermitian systems}}, Phys. Rev. A
  87~(1) (2013) 012118.
\newblock \href {http://dx.doi.org/10.1103/PhysRevA.87.012118}
  {\path{doi:10.1103/PhysRevA.87.012118}}.
\newline\urlprefix\url{https://link.aps.org/doi/10.1103/PhysRevA.87.012118}

\bibitem{AzimiMousolou2013}
V.~{Azimi Mousolou}, C.~M. Canali, E.~Sj{\"{o}}qvist,
  \href{https://link.aps.org/doi/10.1103/PhysRevA.88.012310}{{Unifying
  geometric entanglement and geometric phase in a quantum phase transition}},
  Phys. Rev. A 88~(1) (2013) 012310.
\newblock \href {http://dx.doi.org/10.1103/PhysRevA.88.012310}
  {\path{doi:10.1103/PhysRevA.88.012310}}.
\newline\urlprefix\url{https://link.aps.org/doi/10.1103/PhysRevA.88.012310}

\bibitem{Jafari2013}
R.~Jafari, \href{https://linkinghub.elsevier.com/retrieve/pii/
  S0375960113009870}{{Quantum renormalization group approach to geometric
  phases in spin chains}}, Phys. Lett. A 377~(45-48) (2013) 3279--3282.
\newblock \href {http://dx.doi.org/10.1016/j.physleta.2013.10.034}
  {\path{doi:10.1016/j.physleta.2013.10.034}}.
\newline\urlprefix\url{https://linkinghub.elsevier.com/retrieve/pii/
  S0375960113009870}

\bibitem{Li2013}
S.-C. Li, H.-L. Liu, X.-Y. Zhao,
  \href{http://link.springer.com/10.1140/epjd/e2013-40357-1}{{Quantum phase
  transition and geometric phase in a coupled cavity-BEC system}}, Eur. Phys.
  J. D 67~(12) (2013) 250.
\newblock \href {http://dx.doi.org/10.1140/epjd/e2013-40357-1}
  {\path{doi:10.1140/epjd/e2013-40357-1}}.
\newline\urlprefix\url{http://link.springer.com/10.1140/epjd/e2013-40357-1}

\bibitem{Zhang2013c}
A.-P. Zhang, F.-L. Li, \href{http://stacks.iop.org/1674-1056/22/i=3/a=030308?
  key=crossref.68f0f6dcb27a876014357904c5aa3f10}{{Induced modification of the
  geometric phase of a qubit coupled to an XY spin chain by
  Dzyaloshinsky—Moriya interaction}}, Chinese Phys. B 22~(3) (2013) 030308.
\newblock \href {http://dx.doi.org/10.1088/1674-1056/22/3/030308}
  {\path{doi:10.1088/1674-1056/22/3/030308}}.
\newline\urlprefix\url{http://stacks.iop.org/1674-1056/22/i=3/a=030308?
  key=crossref.68f0f6dcb27a876014357904c5aa3f10}

\bibitem{Sarkar2014}
S.~Sarkar, \href{https://linkinghub.elsevier.com/retrieve/pii/
  S0921452614003639}{{Quantum criticality of geometric phase in coupled optical
  cavity arrays under linear quench}}, Phys. B Condens. Matter 447 (2014)
  42--46.
\newblock \href {http://dx.doi.org/10.1016/j.physb.2014.04.069}
  {\path{doi:10.1016/j.physb.2014.04.069}}.
\newline\urlprefix\url{https://linkinghub.elsevier.com/retrieve/pii/
  S0921452614003639}

\bibitem{Shan2014}
C.-J. Shan, J.-X. Li, W.-W. Cheng, J.-B. Liu, T.-K. Liu,
  \href{http://stacks.iop.org/1612-202X/11/i=3/a=035202?
  key=crossref.62f367e79fafe1cb9f0f27100ea6fefa}{{Scaling of geometric phases
  close to the topological quantum phase transition in Kitaev's quantum wire
  model}}, Laser Phys. Lett. 11~(3) (2014) 035202.
\newblock \href {http://dx.doi.org/10.1088/1612-2011/11/3/035202}
  {\path{doi:10.1088/1612-2011/11/3/035202}}.
\newline\urlprefix\url{http://stacks.iop.org/1612-202X/11/i=3/a=035202?
  key=crossref.62f367e79fafe1cb9f0f27100ea6fefa}

\bibitem{Hickey2014}
J.~M. Hickey, S.~Genway, J.~P. Garrahan,
  \href{https://link.aps.org/doi/10.1103/PhysRevB.89.054301}{{Dynamical phase
  transitions, time-integrated observables, and geometry of states}}, Phys.
  Rev. B 89~(5) (2014) 054301.
\newblock \href {http://dx.doi.org/10.1103/PhysRevB.89.054301}
  {\path{doi:10.1103/PhysRevB.89.054301}}.
\newline\urlprefix\url{https://link.aps.org/doi/10.1103/PhysRevB.89.054301}

\bibitem{Lue2015}
J.~M. L{\"{u}}, X.~P. Li, L.~C. Wang,
  \href{http://www.worldscientific.com/doi/abs/10.1142/
  S0217984915501468}{{Geometric phase and the influence of the
  Dzyaloshinski–Moriya interaction in the one-dimensional quantum compass
  model}}, Mod. Phys. Lett. B 29~(25) (2015) 1550146.
\newblock \href {http://dx.doi.org/10.1142/S0217984915501468}
  {\path{doi:10.1142/S0217984915501468}}.
\newline\urlprefix\url{http://www.worldscientific.com/doi/abs/10.1142/
  S0217984915501468}

\bibitem{Zhu2015}
L.-L. Zhu, Q.~Huang, H.-L. Dai, C.~Kong, Y.-l. Feng, C.-J. Shan,
  \href{http://stacks.iop.org/1612-202X/12/i=1/a=015202?
  key=crossref.550364ef834166d5e6453ea3f6b452cc}{{Detecting topological phase
  transition in 1D superconducting systems with next nearest neighbor
  hopping}}, Laser Phys. Lett. 12~(1) (2015) 015202.
\newblock \href {http://dx.doi.org/10.1088/1612-2011/12/1/015202}
  {\path{doi:10.1088/1612-2011/12/1/015202}}.
\newline\urlprefix\url{http://stacks.iop.org/1612-202X/12/i=1/a=015202?
  key=crossref.550364ef834166d5e6453ea3f6b452cc}

\bibitem{Yuan2015}
Z.-G. Yuan, P.~Zhang, \href{http://stacks.iop.org/0256-307X/32/i=6/a=060301?
  key=crossref.62805bd0019a439f4de6fadea2e3025f}{{Critical Behavior of the
  Energy Gap and Its Relation with the Berry Phase Close to the Excited State
  Quantum Phase Transition in the Lipkin Model}}, Chinese Phys. Lett. 32~(6)
  (2015) 060301.
\newblock \href {http://dx.doi.org/10.1088/0256-307X/32/6/060301}
  {\path{doi:10.1088/0256-307X/32/6/060301}}.
\newline\urlprefix\url{http://stacks.iop.org/0256-307X/32/i=6/a=060301?
  key=crossref.62805bd0019a439f4de6fadea2e3025f}

\bibitem{Li2015}
L.~Li, S.~Chen,
  \href{https://link.aps.org/doi/10.1103/PhysRevB.92.085118}{{Characterization
  of topological phase transitions via topological properties of transition
  points}}, Phys. Rev. B 92~(8) (2015) 085118.
\newblock \href {http://dx.doi.org/10.1103/PhysRevB.92.085118}
  {\path{doi:10.1103/PhysRevB.92.085118}}.
\newline\urlprefix\url{https://link.aps.org/doi/10.1103/PhysRevB.92.085118}

\bibitem{Wu2015}
W.~Wu, D.-W. Luo, J.-B. Xu,
  \href{http://stacks.iop.org/1742-5468/2015/i=1/a=P01025?
  key=crossref.b67267e68236cd0c3d4453c14fad6e14}{{Spin echo and geometric phase
  of a central spin coupled to a compass spin-chain}}, J. Stat. Mech. Theory
  Exp. 2015~(1) (2015) P01025.
\newblock \href {http://dx.doi.org/10.1088/1742-5468/2015/01/P01025}
  {\path{doi:10.1088/1742-5468/2015/01/P01025}}.
\newline\urlprefix\url{http://stacks.iop.org/1742-5468/2015/i=1/a=P01025?
  key=crossref.b67267e68236cd0c3d4453c14fad6e14}

\bibitem{Li2015a}
L.~Li, C.~Yang, S.~Chen, \href{http://stacks.iop.org/0295-5075/112/i=1/a=10004?
  key=crossref.07ec774ff89abaafd3ea7475c91ef52f}{{Winding numbers of phase
  transition points for one-dimensional topological systems}}, Europhys. Lett.
  112~(1) (2015) 10004.
\newblock \href {http://dx.doi.org/10.1209/0295-5075/112/10004}
  {\path{doi:10.1209/0295-5075/112/10004}}.
\newline\urlprefix\url{http://stacks.iop.org/0295-5075/112/i=1/a=10004?
  key=crossref.07ec774ff89abaafd3ea7475c91ef52f}

\bibitem{Ma2015}
Y.-Q. Ma, \href{http://stacks.iop.org/1674-1056/24/i=9/a=090301?
  key=crossref.17383d899f01d658e7b85a092a6f7045}{{Ground-state information
  geometry and quantum criticality in an inhomogeneous spin model}}, Chinese
  Phys. B 24~(9) (2015) 090301.
\newblock \href {http://dx.doi.org/10.1088/1674-1056/24/9/090301}
  {\path{doi:10.1088/1674-1056/24/9/090301}}.
\newline\urlprefix\url{http://stacks.iop.org/1674-1056/24/i=9/a=090301?
  key=crossref.17383d899f01d658e7b85a092a6f7045}

\bibitem{Yang2015}
L.~Yang, Y.-Q. Ma, X.-G. Li,
  \href{https://linkinghub.elsevier.com/retrieve/pii/
  S0921452614007509}{{Geometric tensor and the topological characterization of
  the Bloch band in a two-band lattice model}}, Phys. B Condens. Matter 456
  (2015) 359--364.
\newblock \href {http://dx.doi.org/10.1016/j.physb.2014.09.022}
  {\path{doi:10.1016/j.physb.2014.09.022}}.
\newline\urlprefix\url{https://linkinghub.elsevier.com/retrieve/pii/
  S0921452614007509}

\bibitem{Zvyagin2016}
A.~A. Zvyagin, \href{http://aip.scitation.org/doi/10.1063/1.4969869}{{Dynamical
  quantum phase transitions (Review Article)}}, Low Temp. Phys. 42~(11) (2016)
  971--994.
\newblock \href {http://dx.doi.org/10.1063/1.4969869}
  {\path{doi:10.1063/1.4969869}}.
\newline\urlprefix\url{http://aip.scitation.org/doi/10.1063/1.4969869}

\bibitem{Ya2016}
E.-J. Ye, Z.-D. Hu, W.~Wu, \href{https://linkinghub.elsevier.com/retrieve/pii/
  S0921452616303891}{{Scaling of quantum Fisher information close to the
  quantum phase transition in the XY spin chain}}, Phys. B Condens. Matter 502
  (2016) 151--154.
\newblock \href {http://dx.doi.org/10.1016/j.physb.2016.08.046}
  {\path{doi:10.1016/j.physb.2016.08.046}}.
\newline\urlprefix\url{https://linkinghub.elsevier.com/retrieve/pii/
  S0921452616303891}

\bibitem{Nie2017}
W.~Nie, F.~Mei, L.~Amico, L.~C. Kwek,
  \href{https://link.aps.org/doi/10.1103/PhysRevE.96.020106}{{Scaling of
  geometric phase versus band structure in cluster-Ising models}}, Phys. Rev. E
  96~(2) (2017) 020106.
\newblock \href {http://dx.doi.org/10.1103/PhysRevE.96.020106}
  {\path{doi:10.1103/PhysRevE.96.020106}}.
\newline\urlprefix\url{https://link.aps.org/doi/10.1103/PhysRevE.96.020106}

\bibitem{Zeng2017}
T.-S. Zeng, W.~Zhu, J.-X. Zhu, D.~N. Sheng,
  \href{https://link.aps.org/doi/10.1103/PhysRevB.96.195118}{{Nature of
  continuous phase transitions in interacting topological insulators}}, Phys.
  Rev. B 96~(19) (2017) 195118.
\newblock \href {http://dx.doi.org/10.1103/PhysRevB.96.195118}
  {\path{doi:10.1103/PhysRevB.96.195118}}.
\newline\urlprefix\url{https://link.aps.org/doi/10.1103/PhysRevB.96.195118}

\bibitem{Alvarez2017}
J.~Alvarez-Jimenez, A.~Dector, J.~D. Vergara,
  \href{http://link.springer.com/10.1007/JHEP03(2017)044}{{Quantum information
  metric and Berry curvature from a Lagrangian approach}}, J. High Energy Phys.
  2017~(3) (2017) 44.
\newblock \href {http://dx.doi.org/10.1007/JHEP03(2017)044}
  {\path{doi:10.1007/JHEP03(2017)044}}.
\newline\urlprefix\url{http://link.springer.com/10.1007/JHEP03(2017)044}

\bibitem{Liu2018}
K.~Liu, S.~Yi,
  \href{http://link.springer.com/10.1007/s10773-018-3802-7}{{Geometric Phase
  and Quantum Phase Transition in Charge-Qubit Array}}, Int. J. Theor. Phys.
  57~(9) (2018) 2828--2830.
\newblock \href {http://dx.doi.org/10.1007/s10773-018-3802-7}
  {\path{doi:10.1007/s10773-018-3802-7}}.
\newline\urlprefix\url{http://link.springer.com/10.1007/s10773-018-3802-7}

\bibitem{Carollo2018}
A.~Carollo, B.~Spagnolo, D.~Valenti,
  \href{http://www.nature.com/articles/s41598-018-27362-9}{{Uhlmann curvature
  in dissipative phase transitions}}, Sci. Rep. 8~(1) (2018) 9852.
\newblock \href {http://dx.doi.org/10.1038/s41598-018-27362-9}
  {\path{doi:10.1038/s41598-018-27362-9}}.
\newline\urlprefix\url{http://www.nature.com/articles/s41598-018-27362-9}

\bibitem{Carollo2019}
A.~Carollo, B.~Spagnolo, A.~A. Dubkov, D.~Valenti,
  \href{https://iopscience.iop.org/article/10.1088/1742-5468/ab3ccb}{{On
  quantumness in multi-parameter quantum estimation}}, J. Stat. Mech. Theory
  Exp. 2019~(9) (2019) 094010.
\newblock \href {http://dx.doi.org/10.1088/1742-5468/ab3ccb}
  {\path{doi:10.1088/1742-5468/ab3ccb}}.
\newline\urlprefix\url{https://iopscience.iop.org/article/10.1088/1742-5468/ab3ccb}

\bibitem{Zhang2018}
D.-W. Zhang, Y.-Q. Zhu, Y.~X. Zhao, H.~Yan, S.-L. Zhu,
  \href{https://www.tandfonline.com/doi/full/10.1080/
  00018732.2019.1594094}{{Topological quantum matter with cold atoms}}, Adv.
  Phys. 67~(4) (2018) 253--402.
\newblock \href {http://dx.doi.org/10.1080/00018732.2019.1594094}
  {\path{doi:10.1080/00018732.2019.1594094}}.
\newline\urlprefix\url{https://www.tandfonline.com/doi/full/10.1080/
  00018732.2019.1594094}

\bibitem{Carollo2018a}
A.~Carollo, B.~Spagnolo, D.~Valenti,
  \href{http://www.mdpi.com/1099-4300/20/7/485}{{Symmetric Logarithmic
  Derivative of Fermionic Gaussian States}}, Entropy 20~(7) (2018) 485.
\newblock \href {http://dx.doi.org/10.3390/e20070485}
  {\path{doi:10.3390/e20070485}}.
\newline\urlprefix\url{http://www.mdpi.com/1099-4300/20/7/485}

\bibitem{Henriet2018}
L.~Henriet,
  \href{https://link.aps.org/doi/10.1103/PhysRevB.97.195138}{{Geometrical
  properties of the ground state manifold in the spin boson model}}, Phys. Rev.
  B 97~(19) (2018) 195138.
\newblock \href {http://dx.doi.org/10.1103/PhysRevB.97.195138}
  {\path{doi:10.1103/PhysRevB.97.195138}}.
\newline\urlprefix\url{https://link.aps.org/doi/10.1103/PhysRevB.97.195138}

\bibitem{Cai2019}
X.~Cai, R.~Meng, Y.~Zhang, L.~Wang,
  \href{http://stacks.iop.org/0295-5075/125/i=3/a=30007?
  key=crossref.20799424937744fe706a6706cf834255}{{Geometry of quantum evolution
  in a nonequilibrium environment}}, Europhys. Lett. 125~(3) (2019) 30007.
\newblock \href {http://dx.doi.org/10.1209/0295-5075/125/30007}
  {\path{doi:10.1209/0295-5075/125/30007}}.
\newline\urlprefix\url{http://stacks.iop.org/0295-5075/125/i=3/a=30007?
  key=crossref.20799424937744fe706a6706cf834255}

\bibitem{Lin-Cheng2010}
W.~Lin-Cheng, Y.~Jun-Yan, Y.~Xue-Xi,
  \href{http://stacks.iop.org/1674-1056/19/i=4/a=040512?
  key=crossref.a4695c1ca2961db57ec692db98b946a9}{{Geometric phases and quantum
  phase transitions in inhomogeneous XY spin-chains: Effect of the
  Dzyaloshinski–Moriya interaction}}, Chinese Phys. B 19~(4) (2010) 040512.
\newblock \href {http://dx.doi.org/10.1088/1674-1056/19/4/040512}
  {\path{doi:10.1088/1674-1056/19/4/040512}}.
\newline\urlprefix\url{http://stacks.iop.org/1674-1056/19/i=4/a=040512?
  key=crossref.a4695c1ca2961db57ec692db98b946a9}

\bibitem{Ma2013a}
Y.-Q. Ma, Z.-X. Yu, D.-S. Wang, X.-G. Li,
  \href{https://linkinghub.elsevier.com/retrieve/pii/
  S0375960113002922}{{Quantized Berry phase in twisted Bloch momentum space as
  a topological order parameter for spin chains}}, Phys. Lett. A 377~(18)
  (2013) 1250--1254.
\newblock \href {http://dx.doi.org/10.1016/j.physleta.2013.03.021}
  {\path{doi:10.1016/j.physleta.2013.03.021}}.
\newline\urlprefix\url{https://linkinghub.elsevier.com/retrieve/pii/
  S0375960113002922}

\bibitem{Leonforte2019}
L.~Leonforte, D.~Valenti, B.~Spagnolo, A.~Carollo,
  \href{https://rdcu.be/bHE7n}{{Uhlmann number in translational invariant
  systems}}, Sci. Rep. 9~(1) (2019) 9106.
\newblock \href {http://dx.doi.org/10.1038/s41598-019-45546-9}
  {\path{doi:10.1038/s41598-019-45546-9}}.
\newline\urlprefix\url{https://rdcu.be/bHE7n}

\bibitem{Leonforte2019a}
L.~Leonforte, D.~Valenti, B.~Spagnolo, A.~A. Dubkov, A.~Carollo,
  \href{https://iopscience.iop.org/article/10.1088/1742-5468/ab33f8}{{Haldane
  model at finite temperature}}, J. Stat. Mech. Theory Exp. 2019~(9) (2019)
  094001.
\newblock \href {http://dx.doi.org/10.1088/1742-5468/ab33f8}
  {\path{doi:10.1088/1742-5468/ab33f8}}.
\newline\urlprefix\url{https://iopscience.iop.org/article/10.1088/1742-5468/ab33f8}

\bibitem{Bascone2019}
F.~Bascone, L.~Leonforte, D.~Valenti, B.~Spagnolo, A.~Carollo,
  \href{https://www2.scopus.com/inward/record.uri?eid=2-s2.0-85067188695&doi=10.1103%2FPhysRevB.99.205155&partnerID=40&md5=3b87aa48f78d641e92d826a0ebcabb1e}{{Finite-temperature
  geometric properties of the Kitaev honeycomb model}}, Phys. Rev. B 99~(20).
\newblock \href {http://dx.doi.org/10.1103/PhysRevB.99.205155}
  {\path{doi:10.1103/PhysRevB.99.205155}}.
\newline\urlprefix\url{https://www2.scopus.com/inward/record.uri?eid=2-s2.0-85067188695&doi=10.1103%2FPhysRevB.99.205155&partnerID=40&md5=3b87aa48f78d641e92d826a0ebcabb1e}

\bibitem{Bascone2019a}
F.~Bascone, L.~Leonforte, D.~Valenti, B.~Spagnolo, A.~Carollo,
  \href{https://iopscience.iop.org/article/10.1088/1742-5468/ab35e9}{{On
  critical properties of the Berry curvature in the Kitaev honeycomb model}},
  J. Stat. Mech. Theory Exp. 2019~(9) (2019) 094002.
\newblock \href {http://dx.doi.org/10.1088/1742-5468/ab35e9}
  {\path{doi:10.1088/1742-5468/ab35e9}}.
\newline\urlprefix\url{https://iopscience.iop.org/article/10.1088/1742-5468/ab35e9}

\bibitem{Zhang2013b}
X.-X. Zhang, A.-P. Zhang, J.~Zhang, J.-X. Wang,
  \href{http://www.worldscientific.com/doi/abs/10.1142/
  S0217984913500784}{{Geometric phase of spin chain system in the
  nonequilibrium thermal environments}}, Mod. Phys. Lett. B 27~(11) (2013)
  1350078.
\newblock \href {http://dx.doi.org/10.1142/S0217984913500784}
  {\path{doi:10.1142/S0217984913500784}}.
\newline\urlprefix\url{http://www.worldscientific.com/doi/abs/10.1142/
  S0217984913500784}

\bibitem{Prosen2008}
T.~Prosen, I.~Pi{\v{z}}orn,
  \href{https://link.aps.org/doi/10.1103/PhysRevLett.101.105701}{{Quantum Phase
  Transition in a Far-from-Equilibrium Steady State of an XY Spin Chain}},
  Phys. Rev. Lett. 101~(10) (2008) 105701.
\newblock \href {http://dx.doi.org/10.1103/PhysRevLett.101.105701}
  {\path{doi:10.1103/PhysRevLett.101.105701}}.
\newline\urlprefix\url{https://link.aps.org/doi/10.1103/PhysRevLett.101.105701}

\bibitem{Diehl2008}
S.~Diehl, A.~Micheli, A.~Kantian, B.~Kraus, H.~P. B{\"{u}}chler, P.~Zoller,
  \href{http://www.nature.com/articles/nphys1073}{{Quantum states and phases in
  driven open quantum systems with cold atoms}}, Nat. Phys. 4~(11) (2008)
  878--883.
\newblock \href {http://dx.doi.org/10.1038/nphys1073}
  {\path{doi:10.1038/nphys1073}}.
\newline\urlprefix\url{http://www.nature.com/articles/nphys1073}

\bibitem{DallaTorre2010}
E.~G. {Dalla Torre}, E.~Demler, T.~Giamarchi, E.~Altman,
  \href{http://www.nature.com/articles/nphys1754}{{Quantum critical states and
  phase transitions in the presence of non-equilibrium noise}}, Nat. Phys.
  6~(10) (2010) 806--810.
\newblock \href {http://dx.doi.org/10.1038/nphys1754}
  {\path{doi:10.1038/nphys1754}}.
\newline\urlprefix\url{http://www.nature.com/articles/nphys1754}

\bibitem{Diehl2010a}
S.~Diehl, A.~Tomadin, A.~Micheli, R.~Fazio, P.~Zoller,
  \href{https://link.aps.org/doi/10.1103/PhysRevLett.105.015702}{{Dynamical
  Phase Transitions and Instabilities in Open Atomic Many-Body Systems}}, Phys.
  Rev. Lett. 105~(1) (2010) 015702.
\newblock \href {http://dx.doi.org/10.1103/PhysRevLett.105.015702}
  {\path{doi:10.1103/PhysRevLett.105.015702}}.
\newline\urlprefix\url{https://link.aps.org/doi/10.1103/PhysRevLett.105.015702}

\bibitem{Heyl2013}
M.~Heyl, A.~Polkovnikov, S.~Kehrein,
  \href{https://link.aps.org/doi/10.1103/PhysRevLett.110.135704}{{Dynamical
  Quantum Phase Transitions in the Transverse-Field Ising Model}}, Phys. Rev.
  Lett. 110~(13) (2013) 135704.
\newblock \href {http://dx.doi.org/10.1103/PhysRevLett.110.135704}
  {\path{doi:10.1103/PhysRevLett.110.135704}}.
\newline\urlprefix\url{https://link.aps.org/doi/10.1103/PhysRevLett.110.135704}

\bibitem{LeBoite2013}
A.~{Le Boit{\'{e}}}, G.~Orso, C.~Ciuti,
  \href{https://link.aps.org/doi/10.1103/PhysRevLett.110.233601}{{Steady-State
  Phases and Tunneling-Induced Instabilities in the Driven Dissipative
  Bose-Hubbard Model}}, Phys. Rev. Lett. 110~(23) (2013) 233601.
\newblock \href {http://dx.doi.org/10.1103/PhysRevLett.110.233601}
  {\path{doi:10.1103/PhysRevLett.110.233601}}.
\newline\urlprefix\url{https://link.aps.org/doi/10.1103/PhysRevLett.110.233601}

\bibitem{Carr2013}
C.~Carr, R.~Ritter, C.~G. Wade, C.~S. Adams, K.~J. Weatherill,
  \href{https://link.aps.org/doi/10.1103/PhysRevLett.111.113901}{{Nonequilibrium
  Phase Transition in a Dilute Rydberg Ensemble}}, Phys. Rev. Lett. 111~(11)
  (2013) 113901.
\newblock \href {http://dx.doi.org/10.1103/PhysRevLett.111.113901}
  {\path{doi:10.1103/PhysRevLett.111.113901}}.
\newline\urlprefix\url{https://link.aps.org/doi/10.1103/PhysRevLett.111.113901}

\bibitem{Ajisaka2014}
S.~Ajisaka, F.~Barra, B.~{\v{Z}}unkovi{\v{c}},
  \href{http://stacks.iop.org/1367-2630/16/i=3/a=033028?
  key=crossref.bee87b83616e54f615f3e20f1b2d3684}{{Nonequilibrium quantum phase
  transitions in the XY model: comparison of unitary time evolution and reduced
  density operator approaches}}, New J. Phys. 16~(3) (2014) 033028.
\newblock \href {http://dx.doi.org/10.1088/1367-2630/16/3/033028}
  {\path{doi:10.1088/1367-2630/16/3/033028}}.
\newline\urlprefix\url{http://stacks.iop.org/1367-2630/16/i=3/a=033028?
  key=crossref.bee87b83616e54f615f3e20f1b2d3684}

\bibitem{Marcuzzi2014}
M.~Marcuzzi, E.~Levi, S.~Diehl, J.~P. Garrahan, I.~Lesanovsky,
  \href{https://link.aps.org/doi/10.1103/PhysRevLett.113.210401}{{Universal
  Nonequilibrium Properties of Dissipative Rydberg Gases}}, Phys. Rev. Lett.
  113~(21) (2014) 210401.
\newblock \href {http://dx.doi.org/10.1103/PhysRevLett.113.210401}
  {\path{doi:10.1103/PhysRevLett.113.210401}}.
\newline\urlprefix\url{https://link.aps.org/doi/10.1103/PhysRevLett.113.210401}

\bibitem{Vajna2015}
S.~Vajna, B.~D{\'{o}}ra,
  \href{https://link.aps.org/doi/10.1103/PhysRevB.91.155127}{{Topological
  classification of dynamical phase transitions}}, Phys. Rev. B 91~(15) (2015)
  155127.
\newblock \href {http://dx.doi.org/10.1103/PhysRevB.91.155127}
  {\path{doi:10.1103/PhysRevB.91.155127}}.
\newline\urlprefix\url{https://link.aps.org/doi/10.1103/PhysRevB.91.155127}

\bibitem{Dagvadorj2015}
G.~Dagvadorj, J.~M. Fellows, S.~Matyja{\'{s}}kiewicz, F.~M. Marchetti,
  I.~Carusotto, M.~H. Szyma{\'{n}}ska,
  \href{https://link.aps.org/doi/10.1103/PhysRevX.5.041028}{{Nonequilibrium
  Phase Transition in a Two-Dimensional Driven Open Quantum System}}, Phys.
  Rev. X 5~(4) (2015) 041028.
\newblock \href {http://dx.doi.org/10.1103/PhysRevX.5.041028}
  {\path{doi:10.1103/PhysRevX.5.041028}}.
\newline\urlprefix\url{https://link.aps.org/doi/10.1103/PhysRevX.5.041028}

\bibitem{Weimer2015}
H.~Weimer,
  \href{https://link.aps.org/doi/10.1103/PhysRevLett.114.040402}{{Variational
  Principle for Steady States of Dissipative Quantum Many-Body Systems}}, Phys.
  Rev. Lett. 114~(4) (2015) 040402.
\newblock \href {http://dx.doi.org/10.1103/PhysRevLett.114.040402}
  {\path{doi:10.1103/PhysRevLett.114.040402}}.
\newline\urlprefix\url{https://link.aps.org/doi/10.1103/PhysRevLett.114.040402}

\bibitem{Macieszczak2016}
K.~Macieszczak, M.~Gu\c{t}\u{a}, I.~Lesanovsky, J.~P. Garrahan,
  \href{https://link.aps.org/doi/10.1103/PhysRevLett.116.240404}{{Towards a
  Theory of Metastability in Open Quantum Dynamics}}, Phys. Rev. Lett. 116~(24)
  (2016) 240404.
\newblock \href {http://dx.doi.org/10.1103/PhysRevLett.116.240404}
  {\path{doi:10.1103/PhysRevLett.116.240404}}.
\newline\urlprefix\url{https://link.aps.org/doi/10.1103/PhysRevLett.116.240404}

\bibitem{Jin2016}
J.~Jin, A.~Biella, O.~Viyuela, L.~Mazza, J.~Keeling, R.~Fazio, D.~Rossini,
  \href{https://link.aps.org/doi/10.1103/PhysRevX.6.031011}{{Cluster Mean-Field
  Approach to the Steady-State Phase Diagram of Dissipative Spin Systems}},
  Phys. Rev. X 6~(3) (2016) 031011.
\newblock \href {http://dx.doi.org/10.1103/PhysRevX.6.031011}
  {\path{doi:10.1103/PhysRevX.6.031011}}.
\newline\urlprefix\url{https://link.aps.org/doi/10.1103/PhysRevX.6.031011}

\bibitem{Rose2016}
D.~C. Rose, K.~Macieszczak, I.~Lesanovsky, J.~P. Garrahan,
  \href{https://link.aps.org/doi/10.1103/PhysRevE.94.052132}{{Metastability in
  an open quantum Ising model}}, Phys. Rev. E 94~(5) (2016) 052132.
\newblock \href {http://dx.doi.org/10.1103/PhysRevE.94.052132}
  {\path{doi:10.1103/PhysRevE.94.052132}}.
\newline\urlprefix\url{https://link.aps.org/doi/10.1103/PhysRevE.94.052132}

\bibitem{Bartolo2016}
N.~Bartolo, F.~Minganti, W.~Casteels, C.~Ciuti,
  \href{https://link.aps.org/doi/10.1103/PhysRevA.94.033841}{{Exact steady
  state of a Kerr resonator with one- and two-photon driving and dissipation:
  Controllable Wigner-function multimodality and dissipative phase
  transitions}}, Phys. Rev. A 94~(3) (2016) 033841.
\newblock \href {http://dx.doi.org/10.1103/PhysRevA.94.033841}
  {\path{doi:10.1103/PhysRevA.94.033841}}.
\newline\urlprefix\url{https://link.aps.org/doi/10.1103/PhysRevA.94.033841}

\bibitem{Maghrebi2016}
M.~F. Maghrebi, A.~V. Gorshkov,
  \href{https://link.aps.org/doi/10.1103/PhysRevB.93.014307}{{Nonequilibrium
  many-body steady states via Keldysh formalism}}, Phys. Rev. B 93~(1) (2016)
  014307.
\newblock \href {http://dx.doi.org/10.1103/PhysRevB.93.014307}
  {\path{doi:10.1103/PhysRevB.93.014307}}.
\newline\urlprefix\url{https://link.aps.org/doi/10.1103/PhysRevB.93.014307}

\bibitem{Sieberer2016}
L.~M. Sieberer, M.~Buchhold, S.~Diehl,
  \href{http://stacks.iop.org/0034-4885/79/i=9/a=096001?
  key=crossref.f0f9fd891a0bcd6c69ea234e89466f3a}{{Keldysh field theory for
  driven open quantum systems}}, Reports Prog. Phys. 79~(9) (2016) 096001.
\newblock \href {http://dx.doi.org/10.1088/0034-4885/79/9/096001}
  {\path{doi:10.1088/0034-4885/79/9/096001}}.
\newline\urlprefix\url{http://stacks.iop.org/0034-4885/79/i=9/a=096001?
  key=crossref.f0f9fd891a0bcd6c69ea234e89466f3a}

\bibitem{Roy2017}
S.~Roy, R.~Moessner, A.~Das,
  \href{https://link.aps.org/doi/10.1103/PhysRevB.95.041105}{{Locating
  topological phase transitions using nonequilibrium signatures in local bulk
  observables}}, Phys. Rev. B 95~(4) (2017) 041105.
\newblock \href {http://dx.doi.org/10.1103/PhysRevB.95.041105}
  {\path{doi:10.1103/PhysRevB.95.041105}}.
\newline\urlprefix\url{https://link.aps.org/doi/10.1103/PhysRevB.95.041105}

\bibitem{Fink2017}
J.~M. Fink, A.~Dombi, A.~Vukics, A.~Wallraff, P.~Domokos, {Observation of the
  photon-blockade breakdown phase transition}, Phys. Rev. X 7~(1) (2017)
  011012.
\newblock \href {http://dx.doi.org/10.1103/PhysRevX.7.011012}
  {\path{doi:10.1103/PhysRevX.7.011012}}.

\bibitem{Fitzpatrick2017}
M.~Fitzpatrick, N.~M. Sundaresan, A.~C. Y.~C. Li, J.~Koch, A.~A. Houck,
  \href{https://link.aps.org/doi/10.1103/PhysRevX.7.011016}{{Observation of a
  Dissipative Phase Transition in a One-Dimensional Circuit QED Lattice}},
  Phys. Rev. X 7~(1) (2017) 011016.
\newblock \href {http://dx.doi.org/10.1103/PhysRevX.7.011016}
  {\path{doi:10.1103/PhysRevX.7.011016}}.
\newline\urlprefix\url{https://link.aps.org/doi/10.1103/PhysRevX.7.011016}

\bibitem{Rota2017}
R.~Rota, F.~Storme, N.~Bartolo, R.~Fazio, C.~Ciuti,
  \href{http://link.aps.org/doi/10.1103/PhysRevB.95.134431}{{Critical behavior
  of dissipative two-dimensional spin lattices}}, Phys. Rev. B 95~(13) (2017)
  134431.
\newblock \href {http://dx.doi.org/10.1103/PhysRevB.95.134431}
  {\path{doi:10.1103/PhysRevB.95.134431}}.
\newline\urlprefix\url{http://link.aps.org/doi/10.1103/PhysRevB.95.134431}

\bibitem{Overbeck2017}
V.~R. Overbeck, M.~F. Maghrebi, A.~V. Gorshkov, H.~Weimer,
  \href{http://link.aps.org/doi/10.1103/PhysRevA.95.042133}{{Multicritical
  behavior in dissipative Ising models}}, Phys. Rev. A 95~(4) (2017) 042133.
\newblock \href {http://dx.doi.org/10.1103/PhysRevA.95.042133}
  {\path{doi:10.1103/PhysRevA.95.042133}}.
\newline\urlprefix\url{http://link.aps.org/doi/10.1103/PhysRevA.95.042133}

\bibitem{Foss-Feig2017}
M.~Foss-Feig, J.~T. Young, V.~V. Albert, A.~V. Gorshkov, M.~F. Maghrebi,
  \href{https://link.aps.org/doi/10.1103/PhysRevLett.119.190402}{{Solvable
  Family of Driven-Dissipative Many-Body Systems}}, Phys. Rev. Lett. 119~(19)
  (2017) 190402.
\newblock \href {http://dx.doi.org/10.1103/PhysRevLett.119.190402}
  {\path{doi:10.1103/PhysRevLett.119.190402}}.
\newline\urlprefix\url{https://link.aps.org/doi/10.1103/PhysRevLett.119.190402}

\bibitem{Jin2018}
J.~Jin, A.~Biella, O.~Viyuela, C.~Ciuti, R.~Fazio, D.~Rossini,
  \href{https://link.aps.org/doi/10.1103/PhysRevB.98.241108}{{Phase diagram of
  the dissipative quantum Ising model on a square lattice}}, Phys. Rev. B
  98~(24) (2018) 241108.
\newblock \href {http://dx.doi.org/10.1103/PhysRevB.98.241108}
  {\path{doi:10.1103/PhysRevB.98.241108}}.
\newline\urlprefix\url{https://link.aps.org/doi/10.1103/PhysRevB.98.241108}

\bibitem{Rota2018}
R.~Rota, F.~Minganti, A.~Biella, C.~Ciuti,
  \href{http://stacks.iop.org/1367-2630/20/i=4/a=045003?
  key=crossref.803698b626a12881ce8276ba28cc199e}{{Dynamical properties of
  dissipative XYZ Heisenberg lattices}}, New J. Phys. 20~(4) (2018) 045003.
\newblock \href {http://dx.doi.org/10.1088/1367-2630/aab703}
  {\path{doi:10.1088/1367-2630/aab703}}.
\newline\urlprefix\url{http://stacks.iop.org/1367-2630/20/i=4/a=045003?
  key=crossref.803698b626a12881ce8276ba28cc199e}

\bibitem{Minganti2018}
F.~Minganti, A.~Biella, N.~Bartolo, C.~Ciuti,
  \href{https://link.aps.org/doi/10.1103/PhysRevA.98.042118}{{Spectral theory
  of Liouvillians for dissipative phase transitions}}, Phys. Rev. A 98~(4)
  (2018) 042118.
\newblock \href {http://dx.doi.org/10.1103/PhysRevA.98.042118}
  {\path{doi:10.1103/PhysRevA.98.042118}}.
\newline\urlprefix\url{https://link.aps.org/doi/10.1103/PhysRevA.98.042118}

\bibitem{Vicentini2018}
F.~Vicentini, F.~Minganti, R.~Rota, G.~Orso, C.~Ciuti,
  \href{https://link.aps.org/doi/10.1103/PhysRevA.97.013853}{{Critical slowing
  down in driven-dissipative Bose-Hubbard lattices}}, Phys. Rev. A 97~(1)
  (2018) 013853.
\newblock \href {http://dx.doi.org/10.1103/PhysRevA.97.013853}
  {\path{doi:10.1103/PhysRevA.97.013853}}.
\newline\urlprefix\url{https://link.aps.org/doi/10.1103/PhysRevA.97.013853}

\bibitem{Nagy2018}
A.~Nagy, V.~Savona,
  \href{https://link.aps.org/doi/10.1103/PhysRevA.97.052129}{{Driven-dissipative
  quantum Monte Carlo method for open quantum systems}}, Phys. Rev. A 97~(5)
  (2018) 052129.
\newblock \href {http://dx.doi.org/10.1103/PhysRevA.97.052129}
  {\path{doi:10.1103/PhysRevA.97.052129}}.
\newline\urlprefix\url{https://link.aps.org/doi/10.1103/PhysRevA.97.052129}

\bibitem{Casteels2018}
W.~Casteels, R.~M. Wilson, M.~Wouters,
  \href{https://link.aps.org/doi/10.1103/PhysRevA.97.062107}{{Gutzwiller Monte
  Carlo approach for a critical dissipative spin model}}, Phys. Rev. A 97~(6)
  (2018) 062107.
\newblock \href {http://dx.doi.org/10.1103/PhysRevA.97.062107}
  {\path{doi:10.1103/PhysRevA.97.062107}}.
\newline\urlprefix\url{https://link.aps.org/doi/10.1103/PhysRevA.97.062107}

\bibitem{Rota2019}
R.~Rota, F.~Minganti, C.~Ciuti, V.~Savona,
  \href{https://link.aps.org/doi/10.1103/PhysRevLett.122.110405}{{Quantum
  Critical Regime in a Quadratically Driven Nonlinear Photonic Lattice}}, Phys.
  Rev. Lett. 122~(11) (2019) 110405.
\newblock \href {http://dx.doi.org/10.1103/PhysRevLett.122.110405}
  {\path{doi:10.1103/PhysRevLett.122.110405}}.
\newline\urlprefix\url{https://link.aps.org/doi/10.1103/PhysRevLett.122.110405}

\bibitem{Alicki2007}
R.~Alicki, K.~Lendi, {Quantum dynamical semigroups and applications},
  Springer-Verlag, 2007.

\bibitem{Breuer2007}
H.-P. Breuer, F.~Petruccione, {The Theory of Open Quantum Systems}, Oxford
  University Press, 2007.
\newblock \href {http://dx.doi.org/10.1093/acprof:oso/9780199213900.001.0001}
  {\path{doi:10.1093/acprof:oso/9780199213900.001.0001}}.

\bibitem{Uhlmann1986}
A.~Uhlmann, \href{http://linkinghub.elsevier.com/retrieve/pii/
  0034487786900558}{{Parallel transport and “quantum holonomy” along
  density operators}}, Reports Math. Phys. 24~(2) (1986) 229--240.
\newblock \href {http://dx.doi.org/10.1016/0034-4877(86)90055-8}
  {\path{doi:10.1016/0034-4877(86)90055-8}}.
\newline\urlprefix\url{http://linkinghub.elsevier.com/retrieve/pii/
  0034487786900558}

\bibitem{Braunstein1994}
S.~L. Braunstein, C.~M. Caves,
  \href{https://link.aps.org/doi/10.1103/PhysRevLett.72.3439}{{Statistical
  distance and the geometry of quantum states}}, Phys. Rev. Lett. 72~(22)
  (1994) 3439--3443.
\newblock \href {http://dx.doi.org/10.1103/PhysRevLett.72.3439}
  {\path{doi:10.1103/PhysRevLett.72.3439}}.
\newline\urlprefix\url{https://link.aps.org/doi/10.1103/PhysRevLett.72.3439}

\bibitem{Zanardi2006}
P.~Zanardi, N.~Paunkovi{\'{c}},
  \href{https://link.aps.org/doi/10.1103/PhysRevE.74.031123}{{Ground state
  overlap and quantum phase transitions}}, Phys. Rev. E 74~(3) (2006) 031123.
\newblock \href {http://dx.doi.org/10.1103/PhysRevE.74.031123}
  {\path{doi:10.1103/PhysRevE.74.031123}}.
\newline\urlprefix\url{https://link.aps.org/doi/10.1103/PhysRevE.74.031123}

\bibitem{Zanardi2007}
P.~Zanardi, P.~Giorda, M.~Cozzini,
  \href{https://link.aps.org/doi/10.1103/PhysRevLett.99.100603}{{Information-Theoretic
  Differential Geometry of Quantum Phase Transitions}}, Phys. Rev. Lett.
  99~(10) (2007) 100603.
\newblock \href {http://dx.doi.org/10.1103/PhysRevLett.99.100603}
  {\path{doi:10.1103/PhysRevLett.99.100603}}.
\newline\urlprefix\url{https://link.aps.org/doi/10.1103/PhysRevLett.99.100603}

\bibitem{Gu2010}
S.-J. S.-J. Gu, \href{http://www.worldscientific.com/doi/abs/10.1142/
  S0217979210056335}{{Fidelity Approach to Quantum Phase Transitions}}, Int. J.
  Mod. Phys. B 24~(23) (2010) 4371--4458.
\newblock \href {http://dx.doi.org/10.1142/S0217979210056335}
  {\path{doi:10.1142/S0217979210056335}}.
\newline\urlprefix\url{http://www.worldscientific.com/doi/abs/10.1142/
  S0217979210056335}

\bibitem{Dey2012}
A.~Dey, S.~Mahapatra, P.~Roy, T.~Sarkar,
  \href{https://link.aps.org/doi/10.1103/PhysRevE.86.031137}{{Information
  geometry and quantum phase transitions in the Dicke model}}, Phys. Rev. E
  86~(3) (2012) 031137.
\newblock \href {http://dx.doi.org/10.1103/PhysRevE.86.031137}
  {\path{doi:10.1103/PhysRevE.86.031137}}.
\newline\urlprefix\url{https://link.aps.org/doi/10.1103/PhysRevE.86.031137}

\bibitem{Janyszek1999}
H.~Janyszek, \href{http://stacks.iop.org/0305-4470/23/i=4/a=017?
  key=crossref.d0b7805ad3b73006eb6f93924e228de8}{{Riemannian geometry and
  stability of thermodynamical equilibrium systems}}, J. Phys. A. Math. Gen.
  23~(4) (1990) 477--490.
\newblock \href {http://dx.doi.org/10.1088/0305-4470/23/4/017}
  {\path{doi:10.1088/0305-4470/23/4/017}}.
\newline\urlprefix\url{http://stacks.iop.org/0305-4470/23/i=4/a=017?
  key=crossref.d0b7805ad3b73006eb6f93924e228de8}

\bibitem{Ruppeiner1995}
G.~Ruppeiner,
  \href{https://link.aps.org/doi/10.1103/RevModPhys.67.605}{{Riemannian
  geometry in thermodynamic fluctuation theory}}, Rev. Mod. Phys. 67~(3) (1995)
  605--659.
\newblock \href {http://dx.doi.org/10.1103/RevModPhys.67.605}
  {\path{doi:10.1103/RevModPhys.67.605}}.
\newline\urlprefix\url{https://link.aps.org/doi/10.1103/RevModPhys.67.605}

\bibitem{Quan2009a}
H.~T. Quan, F.~M. Cucchietti,
  \href{https://link.aps.org/doi/10.1103/PhysRevE.79.031101}{{Quantum fidelity
  and thermal phase transitions}}, Phys. Rev. E 79~(3) (2009) 031101.
\newblock \href {http://dx.doi.org/10.1103/PhysRevE.79.031101}
  {\path{doi:10.1103/PhysRevE.79.031101}}.
\newline\urlprefix\url{https://link.aps.org/doi/10.1103/PhysRevE.79.031101}

\bibitem{Zanardi2007a}
P.~Zanardi, L.~{Campos Venuti}, P.~Giorda,
  \href{https://link.aps.org/doi/10.1103/PhysRevA.76.062318}{{Bures metric over
  thermal state manifolds and quantum criticality}}, Phys. Rev. A 76~(6) (2007)
  062318.
\newblock \href {http://dx.doi.org/10.1103/PhysRevA.76.062318}
  {\path{doi:10.1103/PhysRevA.76.062318}}.
\newline\urlprefix\url{https://link.aps.org/doi/10.1103/PhysRevA.76.062318}

\bibitem{Kolodrubetz2013}
M.~Kolodrubetz, V.~Gritsev, A.~Polkovnikov,
  \href{https://link.aps.org/doi/10.1103/PhysRevB.88.064304}{{Classifying and
  measuring geometry of a quantum ground state manifold}}, Phys. Rev. B 88~(6)
  (2013) 064304.
\newblock \href {http://dx.doi.org/10.1103/PhysRevB.88.064304}
  {\path{doi:10.1103/PhysRevB.88.064304}}.
\newline\urlprefix\url{https://link.aps.org/doi/10.1103/PhysRevB.88.064304}

\bibitem{Yang2008}
S.~Yang, S.-J. Gu, C.-P. Sun, H.-Q. Lin,
  \href{https://link.aps.org/doi/10.1103/PhysRevA.78.012304}{{Fidelity
  susceptibility and long-range correlation in the Kitaev honeycomb model}},
  Phys. Rev. A 78~(1) (2008) 012304.
\newblock \href {http://dx.doi.org/10.1103/PhysRevA.78.012304}
  {\path{doi:10.1103/PhysRevA.78.012304}}.
\newline\urlprefix\url{https://link.aps.org/doi/10.1103/PhysRevA.78.012304}

\bibitem{Sjoqvist2000}
E.~Sj{\"{o}}qvist, A.~K. Pati, A.~Ekert, J.~S. Anandan, M.~Ericsson, D.~K.~L.
  Oi, V.~Vedral,
  \href{https://link.aps.org/doi/10.1103/PhysRevLett.85.2845}{{Geometric Phases
  for Mixed States in Interferometry}}, Phys. Rev. Lett. 85~(14) (2000)
  2845--2849.
\newblock \href {http://dx.doi.org/10.1103/PhysRevLett.85.2845}
  {\path{doi:10.1103/PhysRevLett.85.2845}}.
\newline\urlprefix\url{https://link.aps.org/doi/10.1103/PhysRevLett.85.2845}

\bibitem{Tong2004}
D.~M. Tong, E.~Sj{\"{o}}qvist, L.~C. Kwek, C.~H. Oh,
  \href{https://link.aps.org/doi/10.1103/PhysRevLett.93.080405}{{Kinematic
  Approach to the Mixed State Geometric Phase in Nonunitary Evolution}}, Phys.
  Rev. Lett. 93~(8) (2004) 080405.
\newblock \href {http://dx.doi.org/10.1103/PhysRevLett.93.080405}
  {\path{doi:10.1103/PhysRevLett.93.080405}}.
\newline\urlprefix\url{https://link.aps.org/doi/10.1103/PhysRevLett.93.080405}

\bibitem{Chaturvedi2004}
S.~Chaturvedi, E.~Ercolessi, G.~Marmo, G.~Morandi, N.~Mukunda, R.~Simon,
  \href{https://link.springer.com/article/10.1140/epjc/s2004-01814-
  5}{{Geometric phase for mixed states: a differential geometric approach}},
  Eur. Phys. J. C 35~(3) (2004) 413--423.
\newblock \href {http://dx.doi.org/10.1140/epjc/s2004-01814-5}
  {\path{doi:10.1140/epjc/s2004-01814-5}}.
\newline\urlprefix\url{https://link.springer.com/article/10.1140/epjc/s2004-01814-
  5}

\bibitem{Marzlin2004}
K.-P. Marzlin, S.~Ghose, B.~C. Sanders,
  \href{http://link.aps.org/doi/10.1103/PhysRevLett.93.260402}{{Geometric Phase
  Distributions for Open Quantum Systems}}, Phys. Rev. Lett. 93~(26) (2004)
  260402.
\newblock \href {http://dx.doi.org/10.1103/PhysRevLett.93.260402}
  {\path{doi:10.1103/PhysRevLett.93.260402}}.
\newline\urlprefix\url{http://link.aps.org/doi/10.1103/PhysRevLett.93.260402}

\bibitem{Carollo2005a}
A.~Carollo, \href{http://www.worldscientific.com/doi/abs/10.1142/
  S0217732305017718}{{The Quantum Trajectory Approach To Geometric Phase For
  Open Systems}}, Mod. Phys. Lett. A 20~(22) (2005) 1635--1654.
\newblock \href {http://dx.doi.org/10.1142/S0217732305017718}
  {\path{doi:10.1142/S0217732305017718}}.
\newline\urlprefix\url{http://www.worldscientific.com/doi/abs/10.1142/
  S0217732305017718}

\bibitem{Buric2009}
N.~Buri{\'{c}}, M.~Radonji{\'{c}},
  \href{https://link.aps.org/doi/10.1103/PhysRevA.80.014101}{{Uniquely defined
  geometric phase of an open system}}, Phys. Rev. A 80~(1) (2009) 014101.
\newblock \href {http://dx.doi.org/10.1103/PhysRevA.80.014101}
  {\path{doi:10.1103/PhysRevA.80.014101}}.
\newline\urlprefix\url{https://link.aps.org/doi/10.1103/PhysRevA.80.014101}

\bibitem{Sinitsyn2009}
N.~a. Sinitsyn, \href{http://stacks.iop.org/1751-8121/42/i=19/a=193001?
  key=crossref.a38e93a5bcb293540ced1ecbed8ffd30}{{The stochastic pump effect
  and geometric phases in dissipative and stochastic systems}}, J. Phys. A
  Math. Theor. 42~(19) (2009) 193001.
\newblock \href {http://dx.doi.org/10.1088/1751-8113/42/19/193001}
  {\path{doi:10.1088/1751-8113/42/19/193001}}.
\newline\urlprefix\url{http://stacks.iop.org/1751-8121/42/i=19/a=193001?
  key=crossref.a38e93a5bcb293540ced1ecbed8ffd30}

\bibitem{Matsumoto1997}
K.~Matsumoto, \href{http://www.worldscientific.com/doi/abs/10.1142/
  9789812563071_0021}{{A Geometrical Approach to Quantum Estimation Theory}},
  in: Asymptot. Theory Quantum Stat. Inference, World Scientific, 2005, pp.
  305--350.
\newblock \href {http://dx.doi.org/10.1142/9789812563071_0021}
  {\path{doi:10.1142/9789812563071_0021}}.
\newline\urlprefix\url{http://www.worldscientific.com/doi/abs/10.1142/
  9789812563071_0021}

\bibitem{Hayashi2017}
M.~Hayashi,
  \href{http://link.springer.com/10.1007/978-3-662-49725-8_6}{{Quantum
  Information Geometry and Quantum Estimation}}, in: Quantum Inf. theory Math.
  Found., Springer, Berlin, Heidelberg, 2017, pp. 253--322.
\newblock \href {http://dx.doi.org/10.1007/978-3-662-49725-8_6}
  {\path{doi:10.1007/978-3-662-49725-8_6}}.
\newline\urlprefix\url{http://link.springer.com/10.1007/978-3-662-49725-8_6}

\bibitem{Banchi2014}
L.~Banchi, P.~Giorda, P.~Zanardi,
  \href{https://link.aps.org/doi/10.1103/PhysRevE.89.022102}{{Quantum
  information-geometry of dissipative quantum phase transitions}}, Phys. Rev. E
  89~(2) (2014) 022102.
\newblock \href {http://dx.doi.org/10.1103/PhysRevE.89.022102}
  {\path{doi:10.1103/PhysRevE.89.022102}}.
\newline\urlprefix\url{https://link.aps.org/doi/10.1103/PhysRevE.89.022102}

\bibitem{Marzolino2017}
U.~Marzolino, T.~Prosen,
  \href{https://link.aps.org/doi/10.1103/PhysRevB.96.104402}{{Fisher
  information approach to nonequilibrium phase transitions in a quantum XXZ
  spin chain with boundary noise}}, Phys. Rev. B 96~(10) (2017) 104402.
\newblock \href {http://dx.doi.org/10.1103/PhysRevB.96.104402}
  {\path{doi:10.1103/PhysRevB.96.104402}}.
\newline\urlprefix\url{https://link.aps.org/doi/10.1103/PhysRevB.96.104402}

\bibitem{Huang2014}
Z.~Huang, D.~P. Arovas,
  \href{https://link.aps.org/doi/10.1103/PhysRevLett.113.076407}{{Topological
  Indices for Open and Thermal Systems Via Uhlmann's Phase}}, Phys. Rev. Lett.
  113~(7) (2014) 076407.
\newblock \href {http://dx.doi.org/10.1103/PhysRevLett.113.076407}
  {\path{doi:10.1103/PhysRevLett.113.076407}}.
\newline\urlprefix\url{https://link.aps.org/doi/10.1103/PhysRevLett.113.076407}

\bibitem{Viyuela2014}
O.~Viyuela, A.~Rivas, M.~A. Martin-Delgado,
  \href{https://link.aps.org/doi/10.1103/PhysRevLett.112.130401}{{Uhlmann Phase
  as a Topological Measure for One-Dimensional Fermion Systems}}, Phys. Rev.
  Lett. 112~(13) (2014) 130401.
\newblock \href {http://dx.doi.org/10.1103/PhysRevLett.112.130401}
  {\path{doi:10.1103/PhysRevLett.112.130401}}.
\newline\urlprefix\url{https://link.aps.org/doi/10.1103/PhysRevLett.112.130401}

\bibitem{Andersson2016}
O.~Andersson, I.~Bengtsson, M.~Ericsson, E.~Sj{\"{o}}qvist,
  \href{http://rsta.royalsocietypublishing.org/lookup/doi/10.1098/
  rsta.2015.0231}{{Geometric phases for mixed states of the Kitaev chain}},
  Philos. Trans. R. Soc. A Math. Phys. Eng. Sci. 374~(2068) (2016) 20150231.
\newblock \href {http://dx.doi.org/10.1098/rsta.2015.0231}
  {\path{doi:10.1098/rsta.2015.0231}}.
\newline\urlprefix\url{http://rsta.royalsocietypublishing.org/lookup/doi/10.1098/
  rsta.2015.0231}

\bibitem{Viyuela2014a}
O.~Viyuela, A.~Rivas, M.~A. Martin-Delgado,
  \href{http://link.aps.org/doi/10.1103/PhysRevLett.113.076408}{{Two-Dimensional
  Density-Matrix Topological Fermionic Phases: Topological Uhlmann Numbers}},
  Phys. Rev. Lett. 113~(7) (2014) 076408.
\newblock \href {http://dx.doi.org/10.1103/PhysRevLett.113.076408}
  {\path{doi:10.1103/PhysRevLett.113.076408}}.
\newline\urlprefix\url{http://link.aps.org/doi/10.1103/PhysRevLett.113.076408}

\bibitem{Budich2015a}
J.~C. Budich, S.~Diehl,
  \href{https://link.aps.org/doi/10.1103/PhysRevB.91.165140}{{Topology of
  density matrices}}, Phys. Rev. B 91~(16) (2015) 165140.
\newblock \href {http://dx.doi.org/10.1103/PhysRevB.91.165140}
  {\path{doi:10.1103/PhysRevB.91.165140}}.
\newline\urlprefix\url{https://link.aps.org/doi/10.1103/PhysRevB.91.165140}

\bibitem{Kempkes2016}
S.~N. Kempkes, A.~Quelle, C.~M. Smith,
  \href{http://www.nature.com/articles/srep38530}{{Universalities of
  thermodynamic signatures in topological phases}}, Sci. Rep. 6~(1) (2016)
  38530.
\newblock \href {http://dx.doi.org/10.1038/srep38530}
  {\path{doi:10.1038/srep38530}}.
\newline\urlprefix\url{http://www.nature.com/articles/srep38530}

\bibitem{Mera2017}
B.~Mera, C.~Vlachou, N.~Paunkovi{\'{c}}, V.~R. Vieira,
  \href{http://link.aps.org/doi/10.1103/PhysRevLett.119.015702}{{Uhlmann
  Connection in Fermionic Systems Undergoing Phase Transitions}}, Phys. Rev.
  Lett. 119~(1) (2017) 015702.
\newblock \href {http://dx.doi.org/10.1103/PhysRevLett.119.015702}
  {\path{doi:10.1103/PhysRevLett.119.015702}}.
\newline\urlprefix\url{http://link.aps.org/doi/10.1103/PhysRevLett.119.015702}

\bibitem{Tidstrom2003}
J.~Tidstr{\"{o}}m, E.~Sj{\"{o}}qvist,
  \href{https://link.aps.org/doi/10.1103/PhysRevA.67.032110}{{Uhlmann's
  geometric phase in presence of isotropic decoherence}}, Phys. Rev. A 67~(3)
  (2003) 032110.
\newblock \href {http://dx.doi.org/10.1103/PhysRevA.67.032110}
  {\path{doi:10.1103/PhysRevA.67.032110}}.
\newline\urlprefix\url{https://link.aps.org/doi/10.1103/PhysRevA.67.032110}

\bibitem{Aberg2007}
J.~{\AA}berg, D.~Kult, E.~Sj{\"{o}}qvist, D.~K.~L. Oi,
  \href{https://link.aps.org/doi/10.1103/PhysRevA.75.032106}{{Operational
  approach to the Uhlmann holonomy}}, Phys. Rev. A 75~(3) (2007) 032106.
\newblock \href {http://dx.doi.org/10.1103/PhysRevA.75.032106}
  {\path{doi:10.1103/PhysRevA.75.032106}}.
\newline\urlprefix\url{https://link.aps.org/doi/10.1103/PhysRevA.75.032106}

\bibitem{Viyuela2016}
O.~Viyuela, A.~Rivas, S.~Gasparinetti, A.~Wallraff, S.~Filipp, M.~A.
  Martin-Delgado,
  \href{http://www.nature.com/articles/s41534-017-0056-9}{{Observation of
  topological Uhlmann phases with superconducting qubits}}, npj Quantum Inf.
  4~(1) (2018) 10.
\newblock \href {http://dx.doi.org/10.1038/s41534-017-0056-9}
  {\path{doi:10.1038/s41534-017-0056-9}}.
\newline\urlprefix\url{http://www.nature.com/articles/s41534-017-0056-9}

\bibitem{Zhu2011}
J.~Zhu, M.~Shi, V.~Vedral, X.~Peng, D.~Suter, J.~Du,
  \href{http://stacks.iop.org/0295-5075/94/i=2/a=20007?
  key=crossref.a5d79ff4f3bf3269582f84fe8ac69272}{{Experimental demonstration of
  a unified framework for mixed-state geometric phases}}, Europhys. Lett.
  94~(2) (2011) 20007.
\newblock \href {http://dx.doi.org/10.1209/0295-5075/94/20007}
  {\path{doi:10.1209/0295-5075/94/20007}}.
\newline\urlprefix\url{http://stacks.iop.org/0295-5075/94/i=2/a=20007?
  key=crossref.a5d79ff4f3bf3269582f84fe8ac69272}

\bibitem{Ragy2016}
S.~Ragy, M.~Jarzyna, R.~Demkowicz-Dobrza{\'{n}}ski,
  \href{https://link.aps.org/doi/10.1103/PhysRevA.94.052108}{{Compatibility in
  multiparameter quantum metrology}}, Phys. Rev. A 94~(5) (2016) 052108.
\newblock \href {http://dx.doi.org/10.1103/PhysRevA.94.052108}
  {\path{doi:10.1103/PhysRevA.94.052108}}.
\newline\urlprefix\url{https://link.aps.org/doi/10.1103/PhysRevA.94.052108}

\bibitem{Eisert2010}
J.~Eisert, T.~Prosen, \href{http://arxiv.org/abs/1012.5013}{{Noise-driven
  quantum criticality}}, http://arxiv.org/abs/1012.5013.
\newline\urlprefix\url{http://arxiv.org/abs/1012.5013}

\bibitem{Marzolino2014}
U.~Marzolino, T.~Prosen,
  \href{https://link.aps.org/doi/10.1103/PhysRevA.90.062130}{{Quantum metrology
  with nonequilibrium steady states of quantum spin chains}}, Phys. Rev. A
  90~(6) (2014) 062130.
\newblock \href {http://dx.doi.org/10.1103/PhysRevA.90.062130}
  {\path{doi:10.1103/PhysRevA.90.062130}}.
\newline\urlprefix\url{https://link.aps.org/doi/10.1103/PhysRevA.90.062130}

\bibitem{You2007}
W.-L. You, Y.-W. Li, S.-J. Gu,
  \href{https://link.aps.org/doi/10.1103/PhysRevE.76.022101}{{Fidelity, dynamic
  structure factor, and susceptibility in critical phenomena}}, Phys. Rev. E
  76~(2) (2007) 022101.
\newblock \href {http://dx.doi.org/10.1103/PhysRevE.76.022101}
  {\path{doi:10.1103/PhysRevE.76.022101}}.
\newline\urlprefix\url{https://link.aps.org/doi/10.1103/PhysRevE.76.022101}

\bibitem{Zhou2008}
H.-Q. Zhou, J.~P. Barjaktarevi{\v{c}},
  \href{http://stacks.iop.org/1751-8121/41/i=41/a=412001?
  key=crossref.b880de02325394802321f25b3f94dde6}{{Fidelity and quantum phase
  transitions}}, J. Phys. A Math. Theor. 41~(41) (2008) 412001.
\newblock \href {http://dx.doi.org/10.1088/1751-8113/41/41/412001}
  {\path{doi:10.1088/1751-8113/41/41/412001}}.
\newline\urlprefix\url{http://stacks.iop.org/1751-8121/41/i=41/a=412001?
  key=crossref.b880de02325394802321f25b3f94dde6}

\bibitem{Schwandt2009}
D.~Schwandt, F.~Alet, S.~Capponi,
  \href{https://link.aps.org/doi/10.1103/PhysRevLett.103.170501}{{Quantum Monte
  Carlo Simulations of Fidelity at Magnetic Quantum Phase Transitions}}, Phys.
  Rev. Lett. 103~(17) (2009) 170501.
\newblock \href {http://dx.doi.org/10.1103/PhysRevLett.103.170501}
  {\path{doi:10.1103/PhysRevLett.103.170501}}.
\newline\urlprefix\url{https://link.aps.org/doi/10.1103/PhysRevLett.103.170501}

\bibitem{Quan2006}
H.~T. Quan, Z.~Song, X.~F. Liu, P.~Zanardi, C.~P. Sun,
  \href{https://link.aps.org/doi/10.1103/PhysRevLett.96.140604}{{Decay of
  Loschmidt Echo Enhanced by Quantum Criticality}}, Phys. Rev. Lett. 96~(14)
  (2006) 140604.
\newblock \href {http://dx.doi.org/10.1103/PhysRevLett.96.140604}
  {\path{doi:10.1103/PhysRevLett.96.140604}}.
\newline\urlprefix\url{https://link.aps.org/doi/10.1103/PhysRevLett.96.140604}

\bibitem{Anderson1967}
P.~W. Anderson,
  \href{https://link.aps.org/doi/10.1103/PhysRevLett.18.1049}{{Infrared
  Catastrophe in Fermi Gases with Local Scattering Potentials}}, Phys. Rev.
  Lett. 18~(24) (1967) 1049--1051.
\newblock \href {http://dx.doi.org/10.1103/PhysRevLett.18.1049}
  {\path{doi:10.1103/PhysRevLett.18.1049}}.
\newline\urlprefix\url{https://link.aps.org/doi/10.1103/PhysRevLett.18.1049}

\bibitem{Zanardi2007c}
P.~Zanardi, M.~Cozzini, P.~Giorda,
  \href{http://stacks.iop.org/1742-5468/2007/i=02/a=L02002?
  key=crossref.afaa941b230791afe9af65fb0d5ab3a3}{{Ground state fidelity and
  quantum phase transitions in free Fermi systems}}, J. Stat. Mech. Theory Exp.
  2007~(02) (2007) L02002--L02002.
\newblock \href {http://dx.doi.org/10.1088/1742-5468/2007/02/L02002}
  {\path{doi:10.1088/1742-5468/2007/02/L02002}}.
\newline\urlprefix\url{http://stacks.iop.org/1742-5468/2007/i=02/a=L02002?
  key=crossref.afaa941b230791afe9af65fb0d5ab3a3}

\bibitem{Cozzini2007}
M.~Cozzini, P.~Giorda, P.~Zanardi,
  \href{https://link.aps.org/doi/10.1103/PhysRevB.75.014439}{{Quantum phase
  transitions and quantum fidelity in free fermion graphs}}, Phys. Rev. B
  75~(1) (2007) 014439.
\newblock \href {http://dx.doi.org/10.1103/PhysRevB.75.014439}
  {\path{doi:10.1103/PhysRevB.75.014439}}.
\newline\urlprefix\url{https://link.aps.org/doi/10.1103/PhysRevB.75.014439}

\bibitem{Liu2009}
T.~Liu, Y.-Y. Zhang, Q.-H. Chen, K.-L. Wang,
  \href{https://link.aps.org/doi/10.1103/PhysRevA.80.023810}{{Large- N scaling
  behavior of the ground-state energy, fidelity, and the order parameter in the
  Dicke model}}, Phys. Rev. A 80~(2) (2009) 023810.
\newblock \href {http://dx.doi.org/10.1103/PhysRevA.80.023810}
  {\path{doi:10.1103/PhysRevA.80.023810}}.
\newline\urlprefix\url{https://link.aps.org/doi/10.1103/PhysRevA.80.023810}

\bibitem{Cozzini2007a}
M.~Cozzini, R.~Ionicioiu, P.~Zanardi,
  \href{https://link.aps.org/doi/10.1103/PhysRevB.76.104420}{{Quantum fidelity
  and quantum phase transitions in matrix product states}}, Phys. Rev. B
  76~(10) (2007) 104420.
\newblock \href {http://dx.doi.org/10.1103/PhysRevB.76.104420}
  {\path{doi:10.1103/PhysRevB.76.104420}}.
\newline\urlprefix\url{https://link.aps.org/doi/10.1103/PhysRevB.76.104420}

\bibitem{Buonsante2007}
P.~Buonsante, A.~Vezzani,
  \href{https://link.aps.org/doi/10.1103/PhysRevLett.98.110601}{{Ground-State
  Fidelity and Bipartite Entanglement in the Bose-Hubbard Model}}, Phys. Rev.
  Lett. 98~(11) (2007) 110601.
\newblock \href {http://dx.doi.org/10.1103/PhysRevLett.98.110601}
  {\path{doi:10.1103/PhysRevLett.98.110601}}.
\newline\urlprefix\url{https://link.aps.org/doi/10.1103/PhysRevLett.98.110601}

\bibitem{Lacki2014}
M.~{\L}{\c{a}}cki, B.~Damski, J.~Zakrzewski,
  \href{https://link.aps.org/doi/10.1103/PhysRevA.89.033625}{{Numerical studies
  of ground-state fidelity of the Bose-Hubbard model}}, Phys. Rev. A 89~(3)
  (2014) 033625.
\newblock \href {http://dx.doi.org/10.1103/PhysRevA.89.033625}
  {\path{doi:10.1103/PhysRevA.89.033625}}.
\newline\urlprefix\url{https://link.aps.org/doi/10.1103/PhysRevA.89.033625}

\bibitem{Luo2017}
Q.~Luo, S.~Hu, B.~Xi, J.~Zhao, X.~Wang,
  \href{http://link.aps.org/doi/10.1103/PhysRevB.95.165110}{{Ground-state phase
  diagram of an anisotropic spin-1/2 model on the triangular lattice}}, Phys.
  Rev. B 95~(16) (2017) 165110.
\newblock \href {http://dx.doi.org/10.1103/PhysRevB.95.165110}
  {\path{doi:10.1103/PhysRevB.95.165110}}.
\newline\urlprefix\url{http://link.aps.org/doi/10.1103/PhysRevB.95.165110}

\bibitem{Hickey2017}
C.~Hickey, L.~Cincio, Z.~Papi{\'{c}}, A.~Paramekanti,
  \href{https://link.aps.org/doi/10.1103/PhysRevB.96.115115}{{Emergence of
  chiral spin liquids via quantum melting of noncoplanar magnetic orders}},
  Phys. Rev. B 96~(11) (2017) 115115.
\newblock \href {http://dx.doi.org/10.1103/PhysRevB.96.115115}
  {\path{doi:10.1103/PhysRevB.96.115115}}.
\newline\urlprefix\url{https://link.aps.org/doi/10.1103/PhysRevB.96.115115}

\bibitem{Troyer2015}
L.~Wang, Y.-H. Liu, J.~Imri{\v{s}}ka, P.~N. Ma, M.~Troyer,
  \href{https://link.aps.org/doi/10.1103/PhysRevX.5.031007}{{Fidelity
  Susceptibility Made Simple: A Unified Quantum Monte Carlo Approach}}, Phys.
  Rev. X 5~(3) (2015) 031007.
\newblock \href {http://dx.doi.org/10.1103/PhysRevX.5.031007}
  {\path{doi:10.1103/PhysRevX.5.031007}}.
\newline\urlprefix\url{https://link.aps.org/doi/10.1103/PhysRevX.5.031007}

\bibitem{Weber2018}
M.~Weber, M.~Hohenadler,
  \href{https://link.aps.org/doi/10.1103/PhysRevB.98.085405}{{Two-dimensional
  Holstein-Hubbard model: Critical temperature, Ising universality, and
  bipolaron liquid}}, Phys. Rev. B 98~(8) (2018) 085405.
\newblock \href {http://dx.doi.org/10.1103/PhysRevB.98.085405}
  {\path{doi:10.1103/PhysRevB.98.085405}}.
\newline\urlprefix\url{https://link.aps.org/doi/10.1103/PhysRevB.98.085405}

\bibitem{Zhao2010}
J.-H. Zhao, H.-L. Wang, B.~Li, H.-Q. Zhou,
  \href{https://link.aps.org/doi/10.1103/PhysRevE.82.061127}{{Spontaneous
  symmetry breaking and bifurcations in ground-state fidelity for quantum
  lattice systems}}, Phys. Rev. E 82~(6) (2010) 061127.
\newblock \href {http://dx.doi.org/10.1103/PhysRevE.82.061127}
  {\path{doi:10.1103/PhysRevE.82.061127}}.
\newline\urlprefix\url{https://link.aps.org/doi/10.1103/PhysRevE.82.061127}

\bibitem{Su2013}
Y.~H. Su, B.-Q. Hu, S.-H. Li, S.~Y. Cho,
  \href{https://link.aps.org/doi/10.1103/PhysRevE.88.032110}{{Quantum fidelity
  for degenerate ground states in quantum phase transitions}}, Phys. Rev. E
  88~(3) (2013) 032110.
\newblock \href {http://dx.doi.org/10.1103/PhysRevE.88.032110}
  {\path{doi:10.1103/PhysRevE.88.032110}}.
\newline\urlprefix\url{https://link.aps.org/doi/10.1103/PhysRevE.88.032110}

\bibitem{Agarwala2019}
A.~Agarwala, G.~K. Gupta, V.~B. Shenoy, S.~Bhattacharjee,
  \href{https://link.aps.org/doi/10.1103/PhysRevB.99.165125}{{Statistics-tuned
  phases of pseudofermions in one dimension}}, Phys. Rev. B 99~(16) (2019)
  165125.
\newblock \href {http://dx.doi.org/10.1103/PhysRevB.99.165125}
  {\path{doi:10.1103/PhysRevB.99.165125}}.
\newline\urlprefix\url{https://link.aps.org/doi/10.1103/PhysRevB.99.165125}

\bibitem{Giudici2019}
G.~Giudici, A.~Angelone, G.~Magnifico, Z.~Zeng, G.~Giudice, T.~Mendes-Santos,
  M.~Dalmonte,
  \href{https://link.aps.org/doi/10.1103/PhysRevB.99.094434}{{Diagnosing Potts
  criticality and two-stage melting in one-dimensional hard-core boson
  models}}, Phys. Rev. B 99~(9) (2019) 094434.
\newblock \href {http://dx.doi.org/10.1103/PhysRevB.99.094434}
  {\path{doi:10.1103/PhysRevB.99.094434}}.
\newline\urlprefix\url{https://link.aps.org/doi/10.1103/PhysRevB.99.094434}

\bibitem{Vicari2018}
D.~Rossini, E.~Vicari,
  \href{https://link.aps.org/doi/10.1103/PhysRevE.98.062137}{{Ground-state
  fidelity at first-order quantum transitions}}, Phys. Rev. E 98~(6) (2018)
  062137.
\newblock \href {http://dx.doi.org/10.1103/PhysRevE.98.062137}
  {\path{doi:10.1103/PhysRevE.98.062137}}.
\newline\urlprefix\url{https://link.aps.org/doi/10.1103/PhysRevE.98.062137}

\bibitem{Chen2008}
S.~Chen, L.~Wang, Y.~Hao, Y.~Wang,
  \href{https://link.aps.org/doi/10.1103/PhysRevA.77.032111}{{Intrinsic
  relation between ground-state fidelity and the characterization of a quantum
  phase transition}}, Phys. Rev. A 77~(3) (2008) 032111.
\newblock \href {http://dx.doi.org/10.1103/PhysRevA.77.032111}
  {\path{doi:10.1103/PhysRevA.77.032111}}.
\newline\urlprefix\url{https://link.aps.org/doi/10.1103/PhysRevA.77.032111}

\bibitem{Zanardi2007b}
P.~Zanardi, H.~T. Quan, X.~Wang, C.~P. Sun,
  \href{https://link.aps.org/doi/10.1103/PhysRevA.75.032109}{{Mixed-state
  fidelity and quantum criticality at finite temperature}}, Phys. Rev. A 75~(3)
  (2007) 032109.
\newblock \href {http://dx.doi.org/10.1103/PhysRevA.75.032109}
  {\path{doi:10.1103/PhysRevA.75.032109}}.
\newline\urlprefix\url{https://link.aps.org/doi/10.1103/PhysRevA.75.032109}

\bibitem{DeGrandi2010}
C.~{De Grandi}, V.~Gritsev, A.~Polkovnikov,
  \href{https://link.aps.org/doi/10.1103/PhysRevB.81.012303}{{Quench dynamics
  near a quantum critical point}}, Phys. Rev. B 81~(1) (2010) 012303.
\newblock \href {http://dx.doi.org/10.1103/PhysRevB.81.012303}
  {\path{doi:10.1103/PhysRevB.81.012303}}.
\newline\urlprefix\url{https://link.aps.org/doi/10.1103/PhysRevB.81.012303}

\bibitem{Polkovnikov2011}
A.~Polkovnikov, K.~Sengupta, A.~Silva, M.~Vengalattore,
  \href{https://link.aps.org/doi/10.1103/RevModPhys.83.863}{{Colloquium :
  Nonequilibrium dynamics of closed interacting quantum systems}}, Rev. Mod.
  Phys. 83~(3) (2011) 863--883.
\newblock \href {http://dx.doi.org/10.1103/RevModPhys.83.863}
  {\path{doi:10.1103/RevModPhys.83.863}}.
\newline\urlprefix\url{https://link.aps.org/doi/10.1103/RevModPhys.83.863}

\bibitem{Hannukainen2017}
J.~Hannukainen, J.~Larson,
  \href{https://link.aps.org/doi/10.1103/PhysRevA.98.042113}{{Dissipation
  driven quantum phase transitions and symmetry breaking}}, Phys. Rev. A 98~(4)
  (2017) 042113.
\newblock \href {http://dx.doi.org/10.1103/PhysRevA.98.042113}
  {\path{doi:10.1103/PhysRevA.98.042113}}.
\newline\urlprefix\url{https://link.aps.org/doi/10.1103/PhysRevA.98.042113}

\bibitem{Yang2007}
M.-F. Yang,
  \href{https://link.aps.org/doi/10.1103/PhysRevB.76.180403}{{Ground-state
  fidelity in one-dimensional gapless models}}, Phys. Rev. B 76~(18) (2007)
  180403.
\newblock \href {http://dx.doi.org/10.1103/PhysRevB.76.180403}
  {\path{doi:10.1103/PhysRevB.76.180403}}.
\newline\urlprefix\url{https://link.aps.org/doi/10.1103/PhysRevB.76.180403}

\bibitem{Fjaerestad2008}
J.~O. Fj{\ae}restad, \href{http://stacks.iop.org/1742-5468/2008/i=07/a=P07011?
  key=crossref.f3200e9f0e89785fb222b51041921820}{{Ground state fidelity of
  Luttinger liquids: a wavefunctional approach}}, J. Stat. Mech. Theory Exp.
  2008~(07) (2008) P07011.
\newblock \href {http://dx.doi.org/10.1088/1742-5468/2008/07/P07011}
  {\path{doi:10.1088/1742-5468/2008/07/P07011}}.
\newline\urlprefix\url{http://stacks.iop.org/1742-5468/2008/i=07/a=P07011?
  key=crossref.f3200e9f0e89785fb222b51041921820}

\bibitem{Wang2010d}
B.~Wang, M.~Feng, Z.-Q. Chen,
  \href{https://link.aps.org/doi/10.1103/PhysRevA.81.064301}{{Berezinskii-Kosterlitz-Thouless
  transition uncovered by the fidelity susceptibility in the XXZ model}}, Phys.
  Rev. A 81~(6) (2010) 064301.
\newblock \href {http://dx.doi.org/10.1103/PhysRevA.81.064301}
  {\path{doi:10.1103/PhysRevA.81.064301}}.
\newline\urlprefix\url{https://link.aps.org/doi/10.1103/PhysRevA.81.064301}

\bibitem{Wang2012}
H.-L. Wang, A.-M. Chen, B.~Li, H.-Q. Zhou,
  \href{http://stacks.iop.org/1751-8121/45/i=1/a=015306?
  key=crossref.5a3d48ccd5d45dda6c9ffa18729037a7}{{Ground-state fidelity and
  Kosterlitz–Thouless phase transition for the spin-1/2 Heisenberg chain with
  next-to-the-nearest-neighbor interaction}}, J. Phys. A Math. Theor. 45~(1)
  (2012) 015306.
\newblock \href {http://dx.doi.org/10.1088/1751-8113/45/1/015306}
  {\path{doi:10.1088/1751-8113/45/1/015306}}.
\newline\urlprefix\url{http://stacks.iop.org/1751-8121/45/i=1/a=015306?
  key=crossref.5a3d48ccd5d45dda6c9ffa18729037a7}

\bibitem{Sun2015}
G.~Sun, A.~K. Kolezhuk, T.~Vekua,
  \href{https://link.aps.org/doi/10.1103/PhysRevB.91.014418}{{Fidelity at
  Berezinskii-Kosterlitz-Thouless quantum phase transitions}}, Phys. Rev. B
  91~(1) (2015) 014418.
\newblock \href {http://dx.doi.org/10.1103/PhysRevB.91.014418}
  {\path{doi:10.1103/PhysRevB.91.014418}}.
\newline\urlprefix\url{https://link.aps.org/doi/10.1103/PhysRevB.91.014418}

\bibitem{Abasto2008a}
D.~F. Abasto, A.~Hamma, P.~Zanardi,
  \href{https://link.aps.org/doi/10.1103/PhysRevA.78.010301}{{Fidelity analysis
  of topological quantum phase transitions}}, Phys. Rev. A 78~(1) (2008)
  010301.
\newblock \href {http://dx.doi.org/10.1103/PhysRevA.78.010301}
  {\path{doi:10.1103/PhysRevA.78.010301}}.
\newline\urlprefix\url{https://link.aps.org/doi/10.1103/PhysRevA.78.010301}

\bibitem{Zhao2009}
J.-H. Zhao, H.-Q. Zhou,
  \href{https://link.aps.org/doi/10.1103/PhysRevB.80.014403}{{Singularities in
  ground-state fidelity and quantum phase transitions for the Kitaev model}},
  Phys. Rev. B 80~(1) (2009) 014403.
\newblock \href {http://dx.doi.org/10.1103/PhysRevB.80.014403}
  {\path{doi:10.1103/PhysRevB.80.014403}}.
\newline\urlprefix\url{https://link.aps.org/doi/10.1103/PhysRevB.80.014403}

\bibitem{Garnerone2009}
S.~Garnerone, D.~Abasto, S.~Haas, P.~Zanardi,
  \href{https://link.aps.org/doi/10.1103/PhysRevA.79.032302}{{Fidelity in
  topological quantum phases of matter}}, Phys. Rev. A 79~(3) (2009) 032302.
\newblock \href {http://dx.doi.org/10.1103/PhysRevA.79.032302}
  {\path{doi:10.1103/PhysRevA.79.032302}}.
\newline\urlprefix\url{https://link.aps.org/doi/10.1103/PhysRevA.79.032302}

\bibitem{Rigol2009}
M.~Rigol, B.~S. Shastry, S.~Haas,
  \href{https://link.aps.org/doi/10.1103/PhysRevB.80.094529}{{Fidelity and
  superconductivity in two-dimensional t-J models}}, Phys. Rev. B 80~(9) (2009)
  094529.
\newblock \href {http://dx.doi.org/10.1103/PhysRevB.80.094529}
  {\path{doi:10.1103/PhysRevB.80.094529}}.
\newline\urlprefix\url{https://link.aps.org/doi/10.1103/PhysRevB.80.094529}

\bibitem{Jia2011}
C.~J. Jia, B.~Moritz, C.-C. Chen, B.~S. Shastry, T.~P. Devereaux,
  \href{https://link.aps.org/doi/10.1103/PhysRevB.84.125113}{{Fidelity study of
  the superconducting phase diagram in the two-dimensional single-band Hubbard
  model}}, Phys. Rev. B 84~(12) (2011) 125113.
\newblock \href {http://dx.doi.org/10.1103/PhysRevB.84.125113}
  {\path{doi:10.1103/PhysRevB.84.125113}}.
\newline\urlprefix\url{https://link.aps.org/doi/10.1103/PhysRevB.84.125113}

\bibitem{Albuquerque2010}
A.~F. Albuquerque, F.~Alet, C.~Sire, S.~Capponi,
  \href{https://link.aps.org/doi/10.1103/PhysRevB.81.064418}{{Quantum critical
  scaling of fidelity susceptibility}}, Phys. Rev. B 81~(6) (2010) 064418.
\newblock \href {http://dx.doi.org/10.1103/PhysRevB.81.064418}
  {\path{doi:10.1103/PhysRevB.81.064418}}.
\newline\urlprefix\url{https://link.aps.org/doi/10.1103/PhysRevB.81.064418}

\bibitem{Rams2011}
M.~M. Rams, B.~Damski,
  \href{https://link.aps.org/doi/10.1103/PhysRevLett.106.055701}{{Quantum
  Fidelity in the Thermodynamic Limit}}, Phys. Rev. Lett. 106~(5) (2011)
  055701.
\newblock \href {http://dx.doi.org/10.1103/PhysRevLett.106.055701}
  {\path{doi:10.1103/PhysRevLett.106.055701}}.
\newline\urlprefix\url{https://link.aps.org/doi/10.1103/PhysRevLett.106.055701}

\bibitem{Zanardi2008}
P.~Zanardi, M.~G.~A. Paris, L.~{Campos Venuti},
  \href{https://link.aps.org/doi/10.1103/PhysRevA.78.042105}{{Quantum
  criticality as a resource for quantum estimation}}, Phys. Rev. A 78~(4)
  (2008) 042105.
\newblock \href {http://dx.doi.org/10.1103/PhysRevA.78.042105}
  {\path{doi:10.1103/PhysRevA.78.042105}}.
\newline\urlprefix\url{https://link.aps.org/doi/10.1103/PhysRevA.78.042105}

\bibitem{Zhang2008}
J.~Zhang, X.~Peng, N.~Rajendran, D.~Suter,
  \href{https://www.scopus.com/inward/record.uri?eid=2-s2.0-
  40849106560&doi=10.1103%2FPhysRevLett.100.100501&
  partnerID=40&md5=2aca07340f2f562aad524a1967f7223e}{{Detection of quantum
  critical points by a probe qubit}}, Phys. Rev. Lett. 100~(10).
\newblock \href {http://dx.doi.org/10.1103/PhysRevLett.100.100501}
  {\path{doi:10.1103/PhysRevLett.100.100501}}.
\newline\urlprefix\url{https://www.scopus.com/inward/record.uri?eid=2-s2.0-
  40849106560&doi=10.1103%2FPhysRevLett.100.100501&
  partnerID=40&md5=2aca07340f2f562aad524a1967f7223e}

\bibitem{Gu2014}
S.-J. Gu, W.~C. Yu, \href{http://stacks.iop.org/0295-5075/108/i=2/a=20002?
  key=crossref.3cefa42889fe82c4c7635fae7fa9e020}{{Spectral function and
  fidelity susceptibility in quantum critical phenomena}}, EPL (Europhysics
  Lett. 108~(2) (2014) 20002.
\newblock \href {http://dx.doi.org/10.1209/0295-5075/108/20002}
  {\path{doi:10.1209/0295-5075/108/20002}}.
\newline\urlprefix\url{http://stacks.iop.org/0295-5075/108/i=2/a=20002?
  key=crossref.3cefa42889fe82c4c7635fae7fa9e020}

\bibitem{Tran2017}
D.~T. Tran, A.~Dauphin, A.~G. Grushin, P.~Zoller, N.~Goldman,
  \href{http://advances.sciencemag.org/lookup/doi/10.1126/
  sciadv.1701207}{{Probing topology by “heating”: Quantized circular
  dichroism in ultracold atoms}}, Sci. Adv. 3~(8) (2017) e1701207.
\newblock \href {http://dx.doi.org/10.1126/sciadv.1701207}
  {\path{doi:10.1126/sciadv.1701207}}.
\newline\urlprefix\url{http://advances.sciencemag.org/lookup/doi/10.1126/
  sciadv.1701207}

\bibitem{Amit1984}
D.~J. Amit, {Field theory, the renormalization group, and critical phenomena},
  World Scientific, 1984.

\bibitem{Herzberg1963}
G.~Herzberg, H.~C. Longuet-Higgins, {Intersection of potential energy surfaces
  in polyatomic molecules}, Discuss. Faraday Soc. 35 (1963) 77--82.
\newblock \href {http://dx.doi.org/10.1039/df9633500077}
  {\path{doi:10.1039/df9633500077}}.

\bibitem{Stone1976}
A.~J. Stone, \href{http://rspa.royalsocietypublishing.org/cgi/doi/10.1098/
  rspa.1976.0134}{{Spin-Orbit Coupling and the Intersection of Potential Energy
  Surfaces in Polyatomic Molecules}}, Proc. R. Soc. A Math. Phys. Eng. Sci.
  351~(1664) (1976) 141--150.
\newblock \href {http://dx.doi.org/10.1098/rspa.1976.0134}
  {\path{doi:10.1098/rspa.1976.0134}}.
\newline\urlprefix\url{http://rspa.royalsocietypublishing.org/cgi/doi/10.1098/
  rspa.1976.0134}

\bibitem{Johansson2004}
N.~Johansson, E.~Sj{\"{o}}qvist,
  \href{https://link.aps.org/doi/10.1103/PhysRevLett.92.060406}{{Optimal
  Topological Test for Degeneracies of Real Hamiltonians}}, Phys. Rev. Lett.
  92~(6) (2004) 060406.
\newblock \href {http://dx.doi.org/10.1103/PhysRevLett.92.060406}
  {\path{doi:10.1103/PhysRevLett.92.060406}}.
\newline\urlprefix\url{https://link.aps.org/doi/10.1103/PhysRevLett.92.060406}

\bibitem{Johansson2005}
N.~Johansson, E.~Sj{\"{o}}qvist,
  \href{https://link.aps.org/doi/10.1103/PhysRevA.71.012106}{{Searching for
  degeneracies of real Hamiltonians using homotopy classification of loops in
  SO(n)}}, Phys. Rev. A 71~(1) (2005) 012106.
\newblock \href {http://dx.doi.org/10.1103/PhysRevA.71.012106}
  {\path{doi:10.1103/PhysRevA.71.012106}}.
\newline\urlprefix\url{https://link.aps.org/doi/10.1103/PhysRevA.71.012106}

\bibitem{Lieb1961}
E.~H. Lieb, T.~Schultz, D.~Mattis, {Two soluble models of an antiferromagnetic
  chain}, Ann. Phys. (N. Y). 16~(3) (1961) 407--466.
\newblock \href {http://dx.doi.org/10.1016/0003-4916(61)90115-4}
  {\path{doi:10.1016/0003-4916(61)90115-4}}.

\bibitem{Katsura1962}
S.~Katsura,
  \href{https://link.aps.org/doi/10.1103/PhysRev.127.1508}{{Statistical
  Mechanics of the Anisotropic Linear Heisenberg Model}}, Phys. Rev. 127~(5)
  (1962) 1508--1518.
\newblock \href {http://dx.doi.org/10.1103/PhysRev.127.1508}
  {\path{doi:10.1103/PhysRev.127.1508}}.
\newline\urlprefix\url{https://link.aps.org/doi/10.1103/PhysRev.127.1508}

\bibitem{DeGennes1963}
P.~de~Gennes, \href{http://linkinghub.elsevier.com/retrieve/pii/
  0038109863902126}{{Collective motions of hydrogen bonds}}, Solid State
  Commun. 1~(6) (1963) 132--137.
\newblock \href {http://dx.doi.org/10.1016/0038-1098(63)90212-6}
  {\path{doi:10.1016/0038-1098(63)90212-6}}.
\newline\urlprefix\url{http://linkinghub.elsevier.com/retrieve/pii/
  0038109863902126}

\bibitem{Dicke1954}
R.~H. Dicke, \href{https://link.aps.org/doi/10.1103/PhysRev.93.99}{{Coherence
  in Spontaneous Radiation Processes}}, Phys. Rev. 93~(1) (1954) 99--110.
\newblock \href {http://dx.doi.org/10.1103/PhysRev.93.99}
  {\path{doi:10.1103/PhysRev.93.99}}.
\newline\urlprefix\url{https://link.aps.org/doi/10.1103/PhysRev.93.99}

\bibitem{Hepp1973}
K.~Hepp, E.~H. Lieb, \href{http://linkinghub.elsevier.com/retrieve/pii/
  0003491673900390}{{On the superradiant phase transition for molecules in a
  quantized radiation field: the dicke maser model}}, Ann. Phys. (N. Y). 76~(2)
  (1973) 360--404.
\newblock \href {http://dx.doi.org/10.1016/0003-4916(73)90039-0}
  {\path{doi:10.1016/0003-4916(73)90039-0}}.
\newline\urlprefix\url{http://linkinghub.elsevier.com/retrieve/pii/
  0003491673900390}

\bibitem{Hepp1973a}
K.~Hepp, E.~H. Lieb,
  \href{https://link.aps.org/doi/10.1103/PhysRevA.8.2517}{{Equilibrium
  Statistical Mechanics of Matter Interacting with the Quantized Radiation
  Field}}, Phys. Rev. A 8~(5) (1973) 2517--2525.
\newblock \href {http://dx.doi.org/10.1103/PhysRevA.8.2517}
  {\path{doi:10.1103/PhysRevA.8.2517}}.
\newline\urlprefix\url{https://link.aps.org/doi/10.1103/PhysRevA.8.2517}

\bibitem{Wang1973}
Y.~K. Wang, F.~T. Hioe,
  \href{https://link.aps.org/doi/10.1103/PhysRevA.7.831}{{Phase Transition in
  the Dicke Model of Superradiance}}, Phys. Rev. A 7~(3) (1973) 831--836.
\newblock \href {http://dx.doi.org/10.1103/PhysRevA.7.831}
  {\path{doi:10.1103/PhysRevA.7.831}}.
\newline\urlprefix\url{https://link.aps.org/doi/10.1103/PhysRevA.7.831}

\bibitem{Duncan1974}
G.~C. Duncan, \href{https://link.aps.org/doi/10.1103/PhysRevA.9.418}{{Effect of
  antiresonant atom-field interactions on phase transitions in the Dicke
  model}}, Phys. Rev. A 9~(1) (1974) 418--421.
\newblock \href {http://dx.doi.org/10.1103/PhysRevA.9.418}
  {\path{doi:10.1103/PhysRevA.9.418}}.
\newline\urlprefix\url{https://link.aps.org/doi/10.1103/PhysRevA.9.418}

\bibitem{Gilmore1976}
R.~Gilmore, C.~M. Bowden,
  \href{https://link.aps.org/doi/10.1103/PhysRevA.13.1898}{{Coupled
  order-parameter treatment of the Dicke Hamiltonian}}, Phys. Rev. A 13~(5)
  (1976) 1898--1907.
\newblock \href {http://dx.doi.org/10.1103/PhysRevA.13.1898}
  {\path{doi:10.1103/PhysRevA.13.1898}}.
\newline\urlprefix\url{https://link.aps.org/doi/10.1103/PhysRevA.13.1898}

\bibitem{Orszag1977}
M.~Orszag, \href{http://stacks.iop.org/0305-4470/10/i=11/a=025?
  key=crossref.14ca5176b21b0f0e45bf9c60cd47d703}{{Phase transition of a system
  of two-level atoms}}, J. Phys. A. Math. Gen. 10~(11) (1977) 1995--2005.
\newblock \href {http://dx.doi.org/10.1088/0305-4470/10/11/025}
  {\path{doi:10.1088/0305-4470/10/11/025}}.
\newline\urlprefix\url{http://stacks.iop.org/0305-4470/10/i=11/a=025?
  key=crossref.14ca5176b21b0f0e45bf9c60cd47d703}

\bibitem{Sivasubramanian2001}
S.~Sivasubramanian, A.~Widom, Y.~Srivastava,
  \href{https://linkinghub.elsevier.com/retrieve/pii/ S0378437101003843}{{Gauge
  invariant formulations of Dicke –Preparata super-radiant models}}, Phys. A
  Stat. Mech. its Appl. 301~(1-4) (2001) 241--254.
\newblock \href {http://dx.doi.org/10.1016/S0378-4371(01)00384-3}
  {\path{doi:10.1016/S0378-4371(01)00384-3}}.
\newline\urlprefix\url{https://linkinghub.elsevier.com/retrieve/pii/
  S0378437101003843}

\bibitem{Liberti2004}
G.~Liberti, R.~L. Zaffino,
  \href{https://link.aps.org/doi/10.1103/PhysRevA.70.033808}{{Critical
  properties of two-level atom systems interacting with a radiation field}},
  Phys. Rev. A 70~(3) (2004) 033808.
\newblock \href {http://dx.doi.org/10.1103/PhysRevA.70.033808}
  {\path{doi:10.1103/PhysRevA.70.033808}}.
\newline\urlprefix\url{https://link.aps.org/doi/10.1103/PhysRevA.70.033808}

\bibitem{Liberti2005}
G.~Liberti, R.~L. Zaffino,
  \href{http://www.springerlink.com/index/10.1140/epjb/e2005-00153-
  0}{{Thermodynamic properties of the Dicke model in the strong-coupling
  regime}}, Eur. Phys. J. B 44~(4) (2005) 535--541.
\newblock \href {http://dx.doi.org/10.1140/epjb/e2005-00153-0}
  {\path{doi:10.1140/epjb/e2005-00153-0}}.
\newline\urlprefix\url{http://www.springerlink.com/index/10.1140/epjb/e2005-00153-
  0}

\bibitem{Schneider2002}
S.~Schneider, G.~J. Milburn,
  \href{https://link.aps.org/doi/10.1103/PhysRevA.65.042107}{{Entanglement in
  the steady state of a collective-angular-momentum (Dicke) model}}, Phys. Rev.
  A 65~(4) (2002) 042107.
\newblock \href {http://dx.doi.org/10.1103/PhysRevA.65.042107}
  {\path{doi:10.1103/PhysRevA.65.042107}}.
\newline\urlprefix\url{https://link.aps.org/doi/10.1103/PhysRevA.65.042107}

\bibitem{Emary2003}
C.~Emary, T.~Brandes,
  \href{https://link.aps.org/doi/10.1103/PhysRevLett.90.044101}{{Quantum Chaos
  Triggered by Precursors of a Quantum Phase Transition: The Dicke Model}},
  Phys. Rev. Lett. 90~(4) (2003) 044101.
\newblock \href {http://dx.doi.org/10.1103/PhysRevLett.90.044101}
  {\path{doi:10.1103/PhysRevLett.90.044101}}.
\newline\urlprefix\url{https://link.aps.org/doi/10.1103/PhysRevLett.90.044101}

\bibitem{Emary2003a}
C.~Emary, T.~Brandes,
  \href{https://link.aps.org/doi/10.1103/PhysRevE.67.066203}{{Chaos and the
  quantum phase transition in the Dicke model}}, Phys. Rev. E 67~(6) (2003)
  066203.
\newblock \href {http://dx.doi.org/10.1103/PhysRevE.67.066203}
  {\path{doi:10.1103/PhysRevE.67.066203}}.
\newline\urlprefix\url{https://link.aps.org/doi/10.1103/PhysRevE.67.066203}

\bibitem{Frasca2004}
M.~Frasca, \href{http://linkinghub.elsevier.com/retrieve/pii/
  S0003491604000685}{{1/N-Expansion for the Dicke model and the decoherence
  program}}, Ann. Phys. (N. Y). 313~(1) (2004) 26--36.
\newblock \href {http://dx.doi.org/10.1016/j.aop.2004.04.005}
  {\path{doi:10.1016/j.aop.2004.04.005}}.
\newline\urlprefix\url{http://linkinghub.elsevier.com/retrieve/pii/
  S0003491604000685}

\bibitem{Hou2004}
X.-W. Hou, B.~Hu,
  \href{https://link.aps.org/doi/10.1103/PhysRevA.69.042110}{{Decoherence,
  entanglement, and chaos in the Dicke model}}, Phys. Rev. A 69~(4) (2004)
  042110.
\newblock \href {http://dx.doi.org/10.1103/PhysRevA.69.042110}
  {\path{doi:10.1103/PhysRevA.69.042110}}.
\newline\urlprefix\url{https://link.aps.org/doi/10.1103/PhysRevA.69.042110}

\bibitem{Busek2005}
V.~Bu{\v{z}}ek, M.~Orszag, M.~Ro{\v{s}}ko,
  \href{https://link.aps.org/doi/10.1103/PhysRevLett.94.163601}{{Instability
  and Entanglement of the Ground State of the Dicke Model}}, Phys. Rev. Lett.
  94~(16) (2005) 163601.
\newblock \href {http://dx.doi.org/10.1103/PhysRevLett.94.163601}
  {\path{doi:10.1103/PhysRevLett.94.163601}}.
\newline\urlprefix\url{https://link.aps.org/doi/10.1103/PhysRevLett.94.163601}

\bibitem{Brandes2005}
T.~Brandes, \href{http://linkinghub.elsevier.com/retrieve/pii/
  S0370157304005496}{{Coherent and collective quantum optical effects in
  mesoscopic systems}}, Phys. Rep. 408~(5-6) (2005) 315--474.
\newblock \href {http://dx.doi.org/10.1016/j.physrep.2004.12.002}
  {\path{doi:10.1016/j.physrep.2004.12.002}}.
\newline\urlprefix\url{http://linkinghub.elsevier.com/retrieve/pii/
  S0370157304005496}

\bibitem{Lambert2004}
N.~Lambert, C.~Emary, T.~Brandes,
  \href{https://link.aps.org/doi/10.1103/PhysRevLett.92.073602}{{Entanglement
  and the Phase Transition in Single-Mode Superradiance}}, Phys. Rev. Lett.
  92~(7) (2004) 073602.
\newblock \href {http://dx.doi.org/10.1103/PhysRevLett.92.073602}
  {\path{doi:10.1103/PhysRevLett.92.073602}}.
\newline\urlprefix\url{https://link.aps.org/doi/10.1103/PhysRevLett.92.073602}

\bibitem{Reslen2005}
J.~Reslen, L.~Quiroga, N.~F. Johnson,
  \href{http://stacks.iop.org/0295-5075/69/i=1/a=008?
  key=crossref.df7af9c3e3dc146bcbce9a6c54aa0f3d}{{Direct equivalence between
  quantum phase transition phenomena in radiation-matter and magnetic systems:
  Scaling of entanglement}}, Europhys. Lett. 69~(1) (2005) 8--14.
\newblock \href {http://dx.doi.org/10.1209/epl/i2004-10313-4}
  {\path{doi:10.1209/epl/i2004-10313-4}}.
\newline\urlprefix\url{http://stacks.iop.org/0295-5075/69/i=1/a=008?
  key=crossref.df7af9c3e3dc146bcbce9a6c54aa0f3d}

\bibitem{Vidal2006}
J.~Vidal, S.~Dusuel, \href{http://stacks.iop.org/0295-5075/74/i=5/a=817?
  key=crossref.21d1dda5940a2437d3781acf6a68eea0}{{Finite-size scaling exponents
  in the Dicke model}}, Europhys. Lett. 74~(5) (2006) 817--822.
\newblock \href {http://dx.doi.org/10.1209/epl/i2006-10041-9}
  {\path{doi:10.1209/epl/i2006-10041-9}}.
\newline\urlprefix\url{http://stacks.iop.org/0295-5075/74/i=5/a=817?
  key=crossref.21d1dda5940a2437d3781acf6a68eea0}

\bibitem{Liberti2006}
G.~Liberti, R.~L. Zaffino, F.~Piperno, F.~Plastina,
  \href{https://link.aps.org/doi/10.1103/PhysRevA.73.032346}{{Entanglement of a
  qubit coupled to a resonator in the adiabatic regime}}, Phys. Rev. A 73~(3)
  (2006) 032346.
\newblock \href {http://dx.doi.org/10.1103/PhysRevA.73.032346}
  {\path{doi:10.1103/PhysRevA.73.032346}}.
\newline\urlprefix\url{https://link.aps.org/doi/10.1103/PhysRevA.73.032346}

\bibitem{Liberti2006a}
G.~Liberti, F.~Plastina, F.~Piperno,
  \href{https://link.aps.org/doi/10.1103/PhysRevA.74.022324}{{Scaling behavior
  of the adiabatic Dicke model}}, Phys. Rev. A 74~(2) (2006) 022324.
\newblock \href {http://dx.doi.org/10.1103/PhysRevA.74.022324}
  {\path{doi:10.1103/PhysRevA.74.022324}}.
\newline\urlprefix\url{https://link.aps.org/doi/10.1103/PhysRevA.74.022324}

\bibitem{Simon1970}
B.~Simon, A.~Dicke, \href{http://linkinghub.elsevier.com/retrieve/pii/
  000349167090240X}{{Coupling constant analyticity for the anharmonic
  oscillator}}, Ann. Phys. (N. Y). 58~(1) (1970) 76--136.
\newblock \href {http://dx.doi.org/10.1016/0003-4916(70)90240-X}
  {\path{doi:10.1016/0003-4916(70)90240-X}}.
\newline\urlprefix\url{http://linkinghub.elsevier.com/retrieve/pii/
  000349167090240X}

\bibitem{Fubini1904}
G.~Fubini, {Sulle metriche definite da una forma hermitiana}, Atti Istit.
  Veneto LXIII~(2) (1904) 501--511.

\bibitem{Study1905}
{Eduard Study}, {K{\"{u}}rzeste Wege im komplexen Gebiet}, Math. Ann. 60 (1905)
  321--378.

\bibitem{Provost1980}
J.~P. Provost, G.~Vallee,
  \href{http://link.springer.com/10.1007/BF02193559}{{Riemannian structure on
  manifolds of quantum states}}, Commun. Math. Phys. 76~(3) (1980) 289--301.
\newblock \href {http://dx.doi.org/10.1007/BF02193559}
  {\path{doi:10.1007/BF02193559}}.
\newline\urlprefix\url{http://link.springer.com/10.1007/BF02193559}

\bibitem{Hastings2004}
M.~B. Hastings,
  \href{https://link.aps.org/doi/10.1103/PhysRevLett.93.140402}{{Locality in
  Quantum and Markov Dynamics on Lattices and Networks}}, Phys. Rev. Lett.
  93~(14) (2004) 140402.
\newblock \href {http://dx.doi.org/10.1103/PhysRevLett.93.140402}
  {\path{doi:10.1103/PhysRevLett.93.140402}}.
\newline\urlprefix\url{https://link.aps.org/doi/10.1103/PhysRevLett.93.140402}

\bibitem{Magazzu2015}
L.~Magazz{\`{u}}, D.~Valenti, A.~Carollo, B.~Spagnolo,
  \href{http://www.mdpi.com/1099-4300/17/4/2341}{{Multi-State Quantum
  Dissipative Dynamics in Sub-Ohmic Environment: The Strong Coupling Regime}},
  Entropy 17~(4) (2015) 2341--2354.
\newblock \href {http://dx.doi.org/10.3390/e17042341}
  {\path{doi:10.3390/e17042341}}.
\newline\urlprefix\url{http://www.mdpi.com/1099-4300/17/4/2341}

\bibitem{Spagnolo2018}
B.~Spagnolo, A.~Carollo, D.~Valenti,
  \href{http://www.mdpi.com/1099-4300/20/4/226}{{Enhancing Metastability by
  Dissipation and Driving in an Asymmetric Bistable Quantum System}}, Entropy
  20~(4) (2018) 226.
\newblock \href {http://dx.doi.org/10.3390/e20040226}
  {\path{doi:10.3390/e20040226}}.
\newline\urlprefix\url{http://www.mdpi.com/1099-4300/20/4/226}

\bibitem{Valenti2018}
D.~Valenti, A.~Carollo, B.~Spagnolo,
  \href{https://link.aps.org/doi/10.1103/PhysRevA.97.042109}{{Stabilizing
  effect of driving and dissipation on quantum metastable states}}, Phys. Rev.
  A 97~(4) (2018) 042109.
\newblock \href {http://dx.doi.org/10.1103/PhysRevA.97.042109}
  {\path{doi:10.1103/PhysRevA.97.042109}}.
\newline\urlprefix\url{https://link.aps.org/doi/10.1103/PhysRevA.97.042109}

\bibitem{Guarcello2016}
C.~Guarcello, D.~Valenti, A.~Carollo, B.~Spagnolo,
  \href{http://stacks.iop.org/1742-5468/2016/i=5/a=054012?
  key=crossref.48e7f7f504b019713822c481ca8bac97}{{Effects of L{\'{e}}vy noise
  on the dynamics of sine-Gordon solitons in long Josephson junctions}}, J.
  Stat. Mech. Theory Exp. 2016~(5) (2016) 054012.
\newblock \href {http://dx.doi.org/10.1088/1742-5468/2016/05/054012}
  {\path{doi:10.1088/1742-5468/2016/05/054012}}.
\newline\urlprefix\url{http://stacks.iop.org/1742-5468/2016/i=5/a=054012?
  key=crossref.48e7f7f504b019713822c481ca8bac97}

\bibitem{Spagnolo2017}
B.~Spagnolo, C.~Guarcello, L.~Magazz{\`{u}}, A.~Carollo, D.~{Persano Adorno},
  D.~Valenti, \href{http://www.mdpi.com/1099-4300/19/1/20}{{Nonlinear
  Relaxation Phenomena in Metastable Condensed Matter Systems}}, Entropy 19~(1)
  (2017) 20.
\newblock \href {http://dx.doi.org/10.3390/e19010020}
  {\path{doi:10.3390/e19010020}}.
\newline\urlprefix\url{http://www.mdpi.com/1099-4300/19/1/20}

\bibitem{Spagnolo2015}
B.~Spagnolo, D.~Valenti, C.~Guarcello, A.~Carollo, D.~{Persano Adorno},
  S.~Spezia, N.~Pizzolato, B.~{Di Paola},
  \href{http://dx.doi.org/10.1016/j.chaos.2015.07.023}{{Noise-induced effects
  in nonlinear relaxation of condensed matter systems}}, Chaos, Solitons \&
  Fractals 81 (2015) 412--424.
\newblock \href {http://dx.doi.org/10.1016/j.chaos.2015.07.023}
  {\path{doi:10.1016/j.chaos.2015.07.023}}.
\newline\urlprefix\url{http://dx.doi.org/10.1016/j.chaos.2015.07.023}

\bibitem{Spagnolo2019}
B.~Spagnolo, A.~Carollo, D.~Valenti,
  \href{http://link.springer.com/10.1140/epjst/e2018-00121-x}{{Stabilization by
  dissipation and stochastic resonant activation in quantum metastable
  systems}}, Eur. Phys. J. Spec. Top. 227~(3-4) (2019) 379--420.
\newblock \href {http://dx.doi.org/10.1140/epjst/e2018-00121-x}
  {\path{doi:10.1140/epjst/e2018-00121-x}}.
\newline\urlprefix\url{http://link.springer.com/10.1140/epjst/e2018-00121-x}

\bibitem{Dittmann1999}
J.~Dittmann, A.~Uhlmann,
  \href{http://aip.scitation.org/doi/10.1063/1.532884}{{Connections and metrics
  respecting purification of quantum states}}, J. Math. Phys. 40~(7) (1999)
  3246--3267.
\newblock \href {http://dx.doi.org/10.1063/1.532884}
  {\path{doi:10.1063/1.532884}}.
\newline\urlprefix\url{http://aip.scitation.org/doi/10.1063/1.532884}

\bibitem{Kiselev2018}
A.~D. Kiselev, V.~V. Kesaev,
  \href{https://link.aps.org/doi/10.1103/PhysRevA.98.033816}{{Interferometric
  and Uhlmann phases of mixed polarization states}}, Phys. Rev. A 98~(3) (2018)
  033816.
\newblock \href {http://dx.doi.org/10.1103/PhysRevA.98.033816}
  {\path{doi:10.1103/PhysRevA.98.033816}}.
\newline\urlprefix\url{https://link.aps.org/doi/10.1103/PhysRevA.98.033816}

\bibitem{Uhlmann1989}
A.~Uhlmann, \href{http://doi.wiley.com/10.1002/andp.19895010108}{{On Berry
  Phases Along Mixtures of States}}, Ann. Phys. 501~(1) (1989) 63--69.
\newblock \href {http://dx.doi.org/10.1002/andp.19895010108}
  {\path{doi:10.1002/andp.19895010108}}.
\newline\urlprefix\url{http://doi.wiley.com/10.1002/andp.19895010108}

\bibitem{Uhlmann1991}
A.~Uhlmann, \href{http://link.springer.com/10.1007/BF00420373}{{A gauge field
  governing parallel transport along mixed states}}, Lett. Math. Phys. 21~(3)
  (1991) 229--236.
\newblock \href {http://dx.doi.org/10.1007/BF00420373}
  {\path{doi:10.1007/BF00420373}}.
\newline\urlprefix\url{http://link.springer.com/10.1007/BF00420373}

\bibitem{Araki1982}
H.~Araki, G.~A. Raggio, \href{http://link.springer.com/10.1007/BF00403278}{{A
  remark on transition probability}}, Lett. Math. Phys. 6~(3) (1982) 237--240.
\newblock \href {http://dx.doi.org/10.1007/BF00403278}
  {\path{doi:10.1007/BF00403278}}.
\newline\urlprefix\url{http://link.springer.com/10.1007/BF00403278}

\bibitem{Uhlmann1992}
A.~Uhlmann, \href{http://www.springerlink.com/index/10.1007/978-94-011-2801-
  8_23}{{The Metric of Bures and the Geometric Phase}}, in: P.~Z. {Gielerak R.,
  Lukierski J.} (Ed.), Groups Relat. Top., no.~13 in Mathematical Physics
  Studies, Springer, Dordrecht, Dordrecht, 1992, pp. 267--274.
\newblock \href {http://dx.doi.org/10.1007/978-94-011-2801-8_23}
  {\path{doi:10.1007/978-94-011-2801-8_23}}.
\newline\urlprefix\url{http://www.springerlink.com/index/10.1007/978-94-011-2801-
  8_23}

\bibitem{Fuchs1995}
C.~A. Fuchs, C.~M. Caves,
  \href{http://link.springer.com/10.1007/BF02228997}{{Mathematical techniques
  for quantum communication theory}}, Open Syst. Inf. Dyn. 3~(3) (1995)
  345--356.
\newblock \href {http://dx.doi.org/10.1007/BF02228997}
  {\path{doi:10.1007/BF02228997}}.
\newline\urlprefix\url{http://link.springer.com/10.1007/BF02228997}

\bibitem{Nielsen2000}
M.~A. Nielsen, I.~L. Chuang, {Quantum Computation and Quantum Information},
  Cambridge University Press, Cambridge, 2000.

\bibitem{Dabrowski1989}
L.~Dabrowski, A.~Jadczyk, \href{http://stacks.iop.org/0305-4470/22/i=15/a=032?
  key=crossref.c86a4673fdfbfc06a65fffd8d3aa0120}{{Quantum statistical
  holonomy}}, J. Phys. A. Math. Gen. 22~(15) (1989) 3167--3170.
\newblock \href {http://dx.doi.org/10.1088/0305-4470/22/15/032}
  {\path{doi:10.1088/0305-4470/22/15/032}}.
\newline\urlprefix\url{http://stacks.iop.org/0305-4470/22/i=15/a=032?
  key=crossref.c86a4673fdfbfc06a65fffd8d3aa0120}

\bibitem{Dabrowski1990}
L.~Dabrowski, H.~Grosse, \href{http://link.springer.com/10.1007/BF01039313}{{On
  quantum holonomy for mixed states}}, Lett. Math. Phys. 19~(3) (1990)
  205--210.
\newblock \href {http://dx.doi.org/10.1007/BF01039313}
  {\path{doi:10.1007/BF01039313}}.
\newline\urlprefix\url{http://link.springer.com/10.1007/BF01039313}

\bibitem{Sommers2003}
H.-J. rgen Sommers, K.~Zyczkowski,
  \href{http://stacks.iop.org/0305-4470/36/i=39/a=308?
  key=crossref.f55a540ab9f625599586caed4b66bf3a}{{Bures volume of the set of
  mixed quantum states}}, J. Phys. A. Math. Gen. 36~(39) (2003) 10083--10100.
\newblock \href {http://dx.doi.org/10.1088/0305-4470/36/39/308}
  {\path{doi:10.1088/0305-4470/36/39/308}}.
\newline\urlprefix\url{http://stacks.iop.org/0305-4470/36/i=39/a=308?
  key=crossref.f55a540ab9f625599586caed4b66bf3a}

\bibitem{Safranek2017}
D.~{\v{S}}afr{\'{a}}nek,
  \href{http://link.aps.org/doi/10.1103/PhysRevA.95.052320}{{Discontinuities of
  the quantum Fisher information and the Bures metric}}, Phys. Rev. A 95~(5)
  (2017) 052320.
\newblock \href {http://dx.doi.org/10.1103/PhysRevA.95.052320}
  {\path{doi:10.1103/PhysRevA.95.052320}}.
\newline\urlprefix\url{http://link.aps.org/doi/10.1103/PhysRevA.95.052320}

\bibitem{Audenaert2007}
K.~M.~R. Audenaert, J.~Calsamiglia, R.~Mu{\~{n}}oz-Tapia, E.~Bagan, L.~Masanes,
  A.~Acin, F.~Verstraete,
  \href{https://link.aps.org/doi/10.1103/PhysRevLett.98.160501}{{Discriminating
  States: The Quantum Chernoff Bound}}, Phys. Rev. Lett. 98~(16) (2007) 160501.
\newblock \href {http://dx.doi.org/10.1103/PhysRevLett.98.160501}
  {\path{doi:10.1103/PhysRevLett.98.160501}}.
\newline\urlprefix\url{https://link.aps.org/doi/10.1103/PhysRevLett.98.160501}

\bibitem{Chernoff1952}
H.~Chernoff, \href{http://www.jstor.org/stable/2236576}{{A measure of
  asymptotic efficiency for tests of a hypothesis based on the sum of
  observations}}, Ann. Math. Stat. 23~(4) (1952) 493--507.
\newblock \href {http://dx.doi.org/10.1214/aoms/1177729330}
  {\path{doi:10.1214/aoms/1177729330}}.
\newline\urlprefix\url{http://www.jstor.org/stable/2236576}

\bibitem{Cover2006}
T.~M. Cover, J.~A. Thomas, {Elements of information theory},
  Wiley-Interscience, 2006.

\bibitem{Helstrom1976}
C.~W. Helstrom, {Quantum detection and estimation theory}, Academic Press,
  1976.

\bibitem{Udem2002}
T.~Udem, R.~Holzwarth, T.~W. H{\"{a}}nsch,
  \href{http://www.nature.com/articles/416233a}{{Optical frequency metrology}},
  Nature 416~(6877) (2002) 233--237.
\newblock \href {http://dx.doi.org/10.1038/416233a}
  {\path{doi:10.1038/416233a}}.
\newline\urlprefix\url{http://www.nature.com/articles/416233a}

\bibitem{Katori2011}
H.~Katori, \href{http://www.nature.com/articles/nphoton.2011.45}{{Optical
  lattice clocks and quantum metrology}}, Nat. Photonics 5~(4) (2011) 203--210.
\newblock \href {http://dx.doi.org/10.1038/nphoton.2011.45}
  {\path{doi:10.1038/nphoton.2011.45}}.
\newline\urlprefix\url{http://www.nature.com/articles/nphoton.2011.45}

\bibitem{Giovannetti2004}
V.~Giovannetti,
  \href{http://www.sciencemag.org/cgi/doi/10.1126/science.1104149}{{Quantum-Enhanced
  Measurements: Beating the Standard Quantum Limit}}, {Science} 306~(5700)
  (2004) 1330--1336.
\newblock \href {http://dx.doi.org/10.1126/science.1104149}
  {\path{doi:10.1126/science.1104149}}.
\newline\urlprefix\url{http://www.sciencemag.org/cgi/doi/10.1126/science.1104149}

\bibitem{Aspachs2010}
M.~Aspachs, G.~Adesso, I.~Fuentes,
  \href{https://link.aps.org/doi/10.1103/PhysRevLett.105.151301}{{Optimal
  Quantum Estimation of the Unruh-Hawking Effect}}, Phys. Rev. Lett. 105~(15)
  (2010) 151301.
\newblock \href {http://dx.doi.org/10.1103/PhysRevLett.105.151301}
  {\path{doi:10.1103/PhysRevLett.105.151301}}.
\newline\urlprefix\url{https://link.aps.org/doi/10.1103/PhysRevLett.105.151301}

\bibitem{Ahmadi2014}
M.~Ahmadi, D.~E. Bruschi, I.~Fuentes,
  \href{https://link.aps.org/doi/10.1103/PhysRevD.89.065028}{{Quantum metrology
  for relativistic quantum fields}}, Phys. Rev. D 89~(6) (2014) 065028.
\newblock \href {http://dx.doi.org/10.1103/PhysRevD.89.065028}
  {\path{doi:10.1103/PhysRevD.89.065028}}.
\newline\urlprefix\url{https://link.aps.org/doi/10.1103/PhysRevD.89.065028}

\bibitem{Schnabel2010}
R.~Schnabel, N.~Mavalvala, D.~E. McClelland, P.~K. Lam,
  \href{http://dx.doi.org/10.1038/ncomms1122}{{Quantum metrology for
  gravitational wave astronomy}}, Nat. Commun. 1~(8) (2010) 110--121.
\newblock \href {http://dx.doi.org/10.1038/ncomms1122}
  {\path{doi:10.1038/ncomms1122}}.
\newline\urlprefix\url{http://dx.doi.org/10.1038/ncomms1122}

\bibitem{Aasi2013}
J.~Aasi, \href{http://www.nature.com/articles/nphoton.2013.177}{{Enhanced
  sensitivity of the LIGO gravitational wave detector by using squeezed states
  of light}}, Nat. Photonics 7~(8) (2013) 613--619.
\newblock \href {http://dx.doi.org/10.1038/nphoton.2013.177}
  {\path{doi:10.1038/nphoton.2013.177}}.
\newline\urlprefix\url{http://www.nature.com/articles/nphoton.2013.177}

\bibitem{Correa2015}
L.~A. Correa, M.~Mehboudi, G.~Adesso, A.~Sanpera,
  \href{https://link.aps.org/doi/10.1103/PhysRevLett.114.220405}{{Individual
  Quantum Probes for Optimal Thermometry}}, Phys. Rev. Lett. 114~(22) (2015)
  220405.
\newblock \href {http://dx.doi.org/10.1103/PhysRevLett.114.220405}
  {\path{doi:10.1103/PhysRevLett.114.220405}}.
\newline\urlprefix\url{https://link.aps.org/doi/10.1103/PhysRevLett.114.220405}

\bibitem{DePasquale2016}
A.~{De Pasquale}, D.~Rossini, R.~Fazio, V.~Giovannetti,
  \href{http://www.nature.com/articles/ncomms12782}{{Local quantum thermal
  susceptibility}}, Nat. Commun. 7~(1) (2016) 12782.
\newblock \href {http://dx.doi.org/10.1038/ncomms12782}
  {\path{doi:10.1038/ncomms12782}}.
\newline\urlprefix\url{http://www.nature.com/articles/ncomms12782}

\bibitem{Schmitt2017}
S.~Schmitt, T.~Gefen, F.~M. St{\"{u}}rner, T.~Unden, G.~Wolff, C.~M{\"{u}}ller,
  J.~Scheuer, B.~Naydenov, M.~Markham, S.~Pezzagna, J.~Meijer, I.~Schwarz,
  M.~Plenio, A.~Retzker, L.~P. McGuinness, F.~Jelezko,
  \href{http://www.sciencemag.org/lookup/doi/10.1126/
  science.aam5532}{{Submillihertz magnetic spectroscopy performed with a
  nanoscale quantum sensor}}, Science 356~(6340) (2017) 832--837.
\newblock \href {http://dx.doi.org/10.1126/science.aam5532}
  {\path{doi:10.1126/science.aam5532}}.
\newline\urlprefix\url{http://www.sciencemag.org/lookup/doi/10.1126/
  science.aam5532}

\bibitem{Boss2017}
J.~M. Boss, K.~S. Cujia, J.~Zopes, C.~L. Degen, {Quantum sensing with arbitrary
  frequency resolution}, Science 356~(6340) (2017) 837--840.
\newblock \href {http://dx.doi.org/10.1126/science.aam7009}
  {\path{doi:10.1126/science.aam7009}}.

\bibitem{Tsang2016}
M.~Tsang, R.~Nair, X.-M. Lu,
  \href{https://link.aps.org/doi/10.1103/PhysRevX.6.031033}{{Quantum Theory of
  Superresolution for Two Incoherent Optical Point Sources}}, Phys. Rev. X
  6~(3) (2016) 031033.
\newblock \href {http://dx.doi.org/10.1103/PhysRevX.6.031033}
  {\path{doi:10.1103/PhysRevX.6.031033}}.
\newline\urlprefix\url{https://link.aps.org/doi/10.1103/PhysRevX.6.031033}

\bibitem{Nair2016}
R.~Nair, M.~Tsang,
  \href{https://link.aps.org/doi/10.1103/PhysRevLett.117.190801}{{Far-Field
  Superresolution of Thermal Electromagnetic Sources at the Quantum Limit}},
  Phys. Rev. Lett. 117~(19) (2016) 190801.
\newblock \href {http://dx.doi.org/10.1103/PhysRevLett.117.190801}
  {\path{doi:10.1103/PhysRevLett.117.190801}}.
\newline\urlprefix\url{https://link.aps.org/doi/10.1103/PhysRevLett.117.190801}

\bibitem{Lupo2016}
C.~Lupo, S.~Pirandola,
  \href{https://link.aps.org/doi/10.1103/PhysRevLett.117.190802}{{Ultimate
  Precision Bound of Quantum and Subwavelength Imaging}}, Phys. Rev. Lett.
  117~(19) (2016) 190802.
\newblock \href {http://dx.doi.org/10.1103/PhysRevLett.117.190802}
  {\path{doi:10.1103/PhysRevLett.117.190802}}.
\newline\urlprefix\url{https://link.aps.org/doi/10.1103/PhysRevLett.117.190802}

\bibitem{Caves1981}
C.~M. Caves, {Quantum-mechanical noise in an interferometer}, Phys. Rev. D
  23~(8) (1981) 1693--1708.
\newblock \href {http://dx.doi.org/10.1103/PhysRevD.23.1693}
  {\path{doi:10.1103/PhysRevD.23.1693}}.

\bibitem{Huelga1997}
S.~F. Huelga, C.~Macchiavello, T.~Pellizzari, A.~K. Ekert, M.~B. Plenio, J.~I.
  Cirac,
  \href{https://link.aps.org/doi/10.1103/PhysRevLett.79.3865}{{Improvement of
  Frequency Standards with Quantum Entanglement}}, Phys. Rev. Lett. 79~(20)
  (1997) 3865--3868.
\newblock \href {http://dx.doi.org/10.1103/PhysRevLett.79.3865}
  {\path{doi:10.1103/PhysRevLett.79.3865}}.
\newline\urlprefix\url{https://link.aps.org/doi/10.1103/PhysRevLett.79.3865}

\bibitem{Giovannetti2006}
V.~Giovannetti, S.~Lloyd, L.~Maccone,
  \href{https://link.aps.org/doi/10.1103/PhysRevLett.96.010401}{{Quantum
  Metrology}}, Phys. Rev. Lett. 96~(1) (2006) 010401.
\newblock \href {http://dx.doi.org/10.1103/PhysRevLett.96.010401}
  {\path{doi:10.1103/PhysRevLett.96.010401}}.
\newline\urlprefix\url{https://link.aps.org/doi/10.1103/PhysRevLett.96.010401}

\bibitem{Paris2009}
M.~G.~A. Paris, \href{http://www.worldscientific.com/doi/abs/10.1142/
  S0219749909004839}{{Quantum Estimation For Quantum Technology}}, Int. J.
  Quantum Inf. 07~(supp01) (2009) 125--137.
\newblock \href {http://dx.doi.org/10.1142/S0219749909004839}
  {\path{doi:10.1142/S0219749909004839}}.
\newline\urlprefix\url{http://www.worldscientific.com/doi/abs/10.1142/
  S0219749909004839}

\bibitem{Giovannetti2011}
V.~Giovannetti, S.~Lloyd, L.~Maccone,
  \href{http://www.nature.com/articles/nphoton.2011.35}{{Advances in quantum
  metrology}}, Nat. Photonics 5~(4) (2011) 222--229.
\newblock \href {http://dx.doi.org/10.1038/nphoton.2011.35}
  {\path{doi:10.1038/nphoton.2011.35}}.
\newline\urlprefix\url{http://www.nature.com/articles/nphoton.2011.35}

\bibitem{Toth2014}
G.~T{\'{o}}th, I.~Apellaniz,
  \href{https://iopscience.iop.org/article/10.1088/1751-8113/47/42/
  424006/pdf}{{Quantum metrology from a quantum information science
  perspective}}, J. Phys. A Math. Theor. 47~(42) (2014) 424006.
\newblock \href {http://dx.doi.org/10.1088/1751-8113/47/42/424006}
  {\path{doi:10.1088/1751-8113/47/42/424006}}.
\newline\urlprefix\url{https://iopscience.iop.org/article/10.1088/1751-8113/47/42/
  424006/pdf}

\bibitem{Szczykulska2016}
M.~Szczykulska, T.~Baumgratz, A.~Datta,
  \href{https://www.tandfonline.com/doi/full/10.1080/
  23746149.2016.1230476}{{Multi-parameter quantum metrology}}, Adv. Phys. X
  1~(4) (2016) 621--639.
\newblock \href {http://dx.doi.org/10.1080/23746149.2016.1230476}
  {\path{doi:10.1080/23746149.2016.1230476}}.
\newline\urlprefix\url{https://www.tandfonline.com/doi/full/10.1080/
  23746149.2016.1230476}

\bibitem{Pezze2016}
L.~Pezz{\`{e}}, A.~Smerzi, M.~K. Oberthaler, R.~Schmied, P.~Treutlein,
  \href{https://link.aps.org/doi/10.1103/RevModPhys.90.035005}{{Quantum
  metrology with nonclassical states of atomic ensembles}}, Rev. Mod. Phys.
  90~(3) (2018) 035005.
\newblock \href {http://dx.doi.org/10.1103/RevModPhys.90.035005}
  {\path{doi:10.1103/RevModPhys.90.035005}}.
\newline\urlprefix\url{https://link.aps.org/doi/10.1103/RevModPhys.90.035005}

\bibitem{Nichols2018}
R.~Nichols, P.~Liuzzo-Scorpo, P.~A. Knott, G.~Adesso,
  \href{https://link.aps.org/doi/10.1103/PhysRevA.98.012114}{{Multiparameter
  Gaussian quantum metrology}}, Phys. Rev. A 98~(1) (2018) 012114.
\newblock \href {http://dx.doi.org/10.1103/PhysRevA.98.012114}
  {\path{doi:10.1103/PhysRevA.98.012114}}.
\newline\urlprefix\url{https://link.aps.org/doi/10.1103/PhysRevA.98.012114}

\bibitem{Braun2018}
D.~Braun, G.~Adesso, F.~Benatti, R.~Floreanini, U.~Marzolino, M.~W. Mitchell,
  S.~Pirandola,
  \href{https://link.aps.org/doi/10.1103/RevModPhys.90.035006}{{Quantum-enhanced
  measurements without entanglement}}, Rev. Mod. Phys. 90~(3) (2018) 035006.
\newblock \href {http://dx.doi.org/10.1103/RevModPhys.90.035006}
  {\path{doi:10.1103/RevModPhys.90.035006}}.
\newline\urlprefix\url{https://link.aps.org/doi/10.1103/RevModPhys.90.035006}

\bibitem{Humphreys2013}
P.~C. Humphreys, M.~Barbieri, A.~Datta, I.~A. Walmsley,
  \href{https://link.aps.org/doi/10.1103/PhysRevLett.111.070403}{{Quantum
  Enhanced Multiple Phase Estimation}}, Phys. Rev. Lett. 111~(7) (2013) 070403.
\newblock \href {http://dx.doi.org/10.1103/PhysRevLett.111.070403}
  {\path{doi:10.1103/PhysRevLett.111.070403}}.
\newline\urlprefix\url{https://link.aps.org/doi/10.1103/PhysRevLett.111.070403}

\bibitem{Baumgratz2016}
T.~Baumgratz, A.~Datta,
  \href{https://link.aps.org/doi/10.1103/PhysRevLett.116.030801}{{Quantum
  Enhanced Estimation of a Multidimensional Field}}, Phys. Rev. Lett. 116~(3)
  (2016) 030801.
\newblock \href {http://dx.doi.org/10.1103/PhysRevLett.116.030801}
  {\path{doi:10.1103/PhysRevLett.116.030801}}.
\newline\urlprefix\url{https://link.aps.org/doi/10.1103/PhysRevLett.116.030801}

\bibitem{Pezze2017}
L.~Pezz{\`{e}}, M.~A. Ciampini, N.~Spagnolo, P.~C. Humphreys, A.~Datta, I.~A.
  Walmsley, M.~Barbieri, F.~Sciarrino, A.~Smerzi,
  \href{https://link.aps.org/doi/10.1103/PhysRevLett.119.130504}{{Optimal
  Measurements for Simultaneous Quantum Estimation of Multiple Phases}}, Phys.
  Rev. Lett. 119~(13) (2017) 130504.
\newblock \href {http://dx.doi.org/10.1103/PhysRevLett.119.130504}
  {\path{doi:10.1103/PhysRevLett.119.130504}}.
\newline\urlprefix\url{https://link.aps.org/doi/10.1103/PhysRevLett.119.130504}

\bibitem{Apellaniz2018}
I.~Apellaniz, I.~Urizar-Lanz, Z.~Zimbor{\'{a}}s, P.~Hyllus, G.~T{\'{o}}th,
  \href{https://link.aps.org/doi/10.1103/PhysRevA.97.053603}{{Precision bounds
  for gradient magnetometry with atomic ensembles}}, Phys. Rev. A 97~(5) (2018)
  053603.
\newblock \href {http://dx.doi.org/10.1103/PhysRevA.97.053603}
  {\path{doi:10.1103/PhysRevA.97.053603}}.
\newline\urlprefix\url{https://link.aps.org/doi/10.1103/PhysRevA.97.053603}

\bibitem{CamposVenuti2008}
L.~{Campos Venuti}, M.~Cozzini, P.~Buonsante, F.~Massel, N.~Bray-Ali,
  P.~Zanardi,
  \href{https://link.aps.org/doi/10.1103/PhysRevB.78.115410}{{Fidelity approach
  to the Hubbard model}}, Phys. Rev. B 78~(11) (2008) 115410.
\newblock \href {http://dx.doi.org/10.1103/PhysRevB.78.115410}
  {\path{doi:10.1103/PhysRevB.78.115410}}.
\newline\urlprefix\url{https://link.aps.org/doi/10.1103/PhysRevB.78.115410}

\bibitem{Garnerone2009a}
S.~Garnerone, N.~T. Jacobson, S.~Haas, P.~Zanardi,
  \href{https://link.aps.org/doi/10.1103/PhysRevLett.102.057205}{{Fidelity
  approach to the disordered quantum XY model}}, Phys. Rev. Lett. 102~(5)
  (2009) 057205.
\newblock \href {http://dx.doi.org/10.1103/PhysRevLett.102.057205}
  {\path{doi:10.1103/PhysRevLett.102.057205}}.
\newline\urlprefix\url{https://link.aps.org/doi/10.1103/PhysRevLett.102.057205}

\bibitem{Rezakhani2010}
A.~T. Rezakhani, D.~F. Abasto, D.~A. Lidar, P.~Zanardi,
  \href{https://link.aps.org/doi/10.1103/PhysRevA.82.012321}{{Intrinsic
  geometry of quantum adiabatic evolution and quantum phase transitions}},
  Phys. Rev. A 82~(1) (2010) 012321.
\newblock \href {http://dx.doi.org/10.1103/PhysRevA.82.012321}
  {\path{doi:10.1103/PhysRevA.82.012321}}.
\newline\urlprefix\url{https://link.aps.org/doi/10.1103/PhysRevA.82.012321}

\bibitem{Magazzu2016}
L.~Magazz{\`{u}}, A.~Carollo, B.~Spagnolo, D.~Valenti,
  \href{http://stacks.iop.org/1742-5468/2016/i=5/a=054016?
  key=crossref.6a3bd3909efdcfdfd7b7d8c46d642d01}{{Quantum dissipative dynamics
  of a bistable system in the sub-Ohmic to super-Ohmic regime}}, J. Stat. Mech.
  Theory Exp. 2016~(5) (2016) 054016.
\newblock \href {http://dx.doi.org/10.1088/1742-5468/2016/05/054016}
  {\path{doi:10.1088/1742-5468/2016/05/054016}}.
\newline\urlprefix\url{http://stacks.iop.org/1742-5468/2016/i=5/a=054016?
  key=crossref.6a3bd3909efdcfdfd7b7d8c46d642d01}

\bibitem{Guarcello2015}
C.~Guarcello, D.~Valenti, A.~Carollo, B.~Spagnolo,
  \href{http://www.mdpi.com/1099-4300/17/5/2862}{{Stabilization Effects of
  Dichotomous Noise on the Lifetime of the Superconducting State in a Long
  Josephson Junction}}, Entropy 17~(5) (2015) 2862--2875.
\newblock \href {http://dx.doi.org/10.3390/e17052862}
  {\path{doi:10.3390/e17052862}}.
\newline\urlprefix\url{http://www.mdpi.com/1099-4300/17/5/2862}

\bibitem{Consiglio2016}
A.~Consiglio, A.~Carollo, S.~A. Zenios,
  \href{http://www.tandfonline.com/doi/full/10.1080/14697688.2015.1114359}{{A
  parsimonious model for generating arbitrage-free scenario trees}}, Quant.
  Financ. 16~(2) (2016) 201--212.
\newblock \href {http://dx.doi.org/10.1080/14697688.2015.1114359}
  {\path{doi:10.1080/14697688.2015.1114359}}.
\newline\urlprefix\url{http://www.tandfonline.com/doi/full/10.1080/14697688.2015.1114359}

\bibitem{Kolodrubetz2017}
M.~Kolodrubetz, D.~Sels, P.~Mehta, A.~Polkovnikov,
  \href{https://linkinghub.elsevier.com/retrieve/pii/
  S0370157317301989}{{Geometry and non-adiabatic response in quantum and
  classical systems}}, Phys. Rep. 697 (2017) 1--87.
\newblock \href {http://dx.doi.org/10.1016/j.physrep.2017.07.001}
  {\path{doi:10.1016/j.physrep.2017.07.001}}.
\newline\urlprefix\url{https://linkinghub.elsevier.com/retrieve/pii/
  S0370157317301989}

\bibitem{Holevo2011}
A.~Holevo,
  \href{http://link.springer.com/10.1007/978-88-7642-378-9}{{Probabilistic and
  Statistical Aspects of Quantum Theory}}, Edizioni della Normale, Pisa, 2011.
\newblock \href {http://dx.doi.org/10.1007/978-88-7642-378-9}
  {\path{doi:10.1007/978-88-7642-378-9}}.
\newline\urlprefix\url{http://link.springer.com/10.1007/978-88-7642-378-9}

\bibitem{Hayashi2008}
M.~Hayashi, K.~Matsumoto,
  \href{http://aip.scitation.org/doi/10.1063/1.2988130}{{Asymptotic performance
  of optimal state estimation in qubit system}}, J. Math. Phys. 49~(10) (2008)
  102101.
\newblock \href {http://dx.doi.org/10.1063/1.2988130}
  {\path{doi:10.1063/1.2988130}}.
\newline\urlprefix\url{http://aip.scitation.org/doi/10.1063/1.2988130}

\bibitem{Kahn2009}
J.~Kahn, M.~Gu\c{t}\u{a},
  \href{http://link.springer.com/10.1007/s00220-009-0787-3}{{Local Asymptotic
  Normality for Finite Dimensional Quantum Systems}}, Commun. Math. Phys.
  289~(2) (2009) 597--652.
\newblock \href {http://dx.doi.org/10.1007/s00220-009-0787-3}
  {\path{doi:10.1007/s00220-009-0787-3}}.
\newline\urlprefix\url{http://link.springer.com/10.1007/s00220-009-0787-3}

\bibitem{Gill2013}
R.~D. Gill, M.~I. Gu\c{t}\u{a},
  \href{http://projecteuclid.org/euclid.imsc/1362751183}{{On asymptotic quantum
  statistical inference}}, in: M.~H. {Banerjee, M., Bunea, F., Huang, J.,
  Koltchinskii, V., and Maathuis} (Ed.), From Probab. to Stat. Back
  High-Dimensional Model. Process. -- A Festschrift Honor Jon A. Wellner,
  Vol.~9, Institute of Mathematical Statistics, Beachwood, Ohio, USA, 2013, pp.
  105--127.
\newblock \href {http://dx.doi.org/10.1214/12-IMSCOLL909}
  {\path{doi:10.1214/12-IMSCOLL909}}.
\newline\urlprefix\url{http://projecteuclid.org/euclid.imsc/1362751183}

\bibitem{Yamagata2013}
K.~Yamagata, A.~Fujiwara, R.~D. Gill,
  \href{http://projecteuclid.org/euclid.aos/1382547518}{{Quantum local
  asymptotic normality based on a new quantum likelihood ratio}}, Ann. Stat.
  41~(4) (2013) 2197--2217.
\newblock \href {http://dx.doi.org/10.1214/13-AOS1147}
  {\path{doi:10.1214/13-AOS1147}}.
\newline\urlprefix\url{http://projecteuclid.org/euclid.aos/1382547518}

\bibitem{Cramer1946}
H.~Cram\'er, {Mathematical methods of statistics}, Princeton University Press,
  1946.

\bibitem{Kay1993}
S.~M. Kay, {Fundamentals of statistical signal processing}, Prentice-Hall PTR,
  1993.

\bibitem{Cox1987}
D.~R. Cox, N.~Reid, \href{http://www.jstor.org/stable/2345476}{{Parameter
  Orthogonality and Approximate Conditional Inference}}, J. R. Stat. Soc. Ser.
  B 49 (1987) 1--39.
\newline\urlprefix\url{http://www.jstor.org/stable/2345476}

\bibitem{Robertson1929}
H.~P. Robertson, \href{https://link.aps.org/doi/10.1103/PhysRev.34.163}{{The
  Uncertainty Principle}}, Phys. Rev. 34~(1) (1929) 163--164.
\newblock \href {http://dx.doi.org/10.1103/PhysRev.34.163}
  {\path{doi:10.1103/PhysRev.34.163}}.
\newline\urlprefix\url{https://link.aps.org/doi/10.1103/PhysRev.34.163}

\bibitem{Horn2013}
R.~A. Horn, C.~R. Johnson, {Matrix analysis}, Cambridge University Press, 2013.

\bibitem{Brody1995}
D.~Brody, N.~Rivier, {Geometrical aspects of statistical mechanics}, Phys. Rev.
  E 51~(2) (1995) 1006--1011.
\newblock \href {http://dx.doi.org/10.1103/PhysRevE.51.1006}
  {\path{doi:10.1103/PhysRevE.51.1006}}.

\bibitem{Altland2006}
A.~Altland, B.~Simons,
  \href{http://ebooks.cambridge.org/ref/id/CBO9780511804236}{{Condensed Matter
  Field Theory}}, Cambridge University Press, Cambridge, 2006.
\newblock \href {http://dx.doi.org/10.1017/CBO9780511804236}
  {\path{doi:10.1017/CBO9780511804236}}.
\newline\urlprefix\url{http://ebooks.cambridge.org/ref/id/CBO9780511804236}

\bibitem{Yim2006}
S.~Yim, T.~S. Jones,
  \href{https://link.aps.org/doi/10.1103/PhysRevB.73.161305}{{Anomalous scaling
  behavior and surface roughening in molecular thin-film deposition}}, Phys.
  Rev. B 73~(16) (2006) 161305.
\newblock \href {http://dx.doi.org/10.1103/PhysRevB.73.161305}
  {\path{doi:10.1103/PhysRevB.73.161305}}.
\newline\urlprefix\url{https://link.aps.org/doi/10.1103/PhysRevB.73.161305}

\bibitem{Coldea2010}
R.~Coldea, D.~A. Tennant, E.~M. Wheeler, E.~Wawrzynska, D.~Prabhakaran,
  M.~Telling, K.~Habicht, P.~Smeibidl, K.~Kiefer,
  \href{http://www.sciencemag.org/lookup/doi/10.1126/ science.1180085}{{Quantum
  Criticality in an Ising Chain: Experimental Evidence for Emergent E 8
  Symmetry}}, {Science} 327~(5962) (2010) 177--180.
\newblock \href {http://dx.doi.org/10.1126/science.1180085}
  {\path{doi:10.1126/science.1180085}}.
\newline\urlprefix\url{http://www.sciencemag.org/lookup/doi/10.1126/
  science.1180085}

\bibitem{Lake2010}
B.~Lake, A.~M. Tsvelik, S.~Notbohm, D.~{Alan Tennant}, T.~G. Perring,
  M.~Reehuis, C.~Sekar, G.~Krabbes, B.~B{\"{u}}chner, {Confinement of
  fractional quantum number particles in a condensed-matter system}, Nat.
  Phys.\href {http://dx.doi.org/10.1038/nphys1462}
  {\path{doi:10.1038/nphys1462}}.

\bibitem{Han2012}
T.-H. Han, J.~S. Helton, S.~Chu, D.~G. Nocera, J.~A. Rodriguez-Rivera,
  C.~Broholm, Y.~S. Lee,
  \href{http://www.nature.com/articles/nature11659}{{Fractionalized excitations
  in the spin-liquid state of a kagome-lattice antiferromagnet}}, Nature
  492~(7429) (2012) 406--410.
\newblock \href {http://dx.doi.org/10.1038/nature11659}
  {\path{doi:10.1038/nature11659}}.
\newline\urlprefix\url{http://www.nature.com/articles/nature11659}

\bibitem{Dai2015}
P.~Dai,
  \href{https://link.aps.org/doi/10.1103/RevModPhys.87.855}{{Antiferromagnetic
  order and spin dynamics in iron-based superconductors}}, Rev. Mod. Phys.
  87~(3) (2015) 855--896.
\newblock \href {http://dx.doi.org/10.1103/RevModPhys.87.855}
  {\path{doi:10.1103/RevModPhys.87.855}}.
\newline\urlprefix\url{https://link.aps.org/doi/10.1103/RevModPhys.87.855}

\bibitem{Halg2015}
M.~H{\"{a}}lg, D.~H{\"{u}}vonen, N.~P. Butch, F.~Demmel, A.~Zheludev,
  \href{https://link.aps.org/doi/10.1103/PhysRevB.92.104416}{{Finite-temperature
  scaling of spin correlations in a partially magnetized Heisenberg S=1/2
  chain}}, Phys. Rev. B 92~(10) (2015) 104416.
\newblock \href {http://dx.doi.org/10.1103/PhysRevB.92.104416}
  {\path{doi:10.1103/PhysRevB.92.104416}}.
\newline\urlprefix\url{https://link.aps.org/doi/10.1103/PhysRevB.92.104416}

\bibitem{Woodcock1985}
L.~V. Woodcock,
  \href{https://link.aps.org/doi/10.1103/PhysRevLett.54.1513}{{Origins of
  Thixotropy}}, Phys. Rev. Lett. 54~(14) (1985) 1513--1516.
\newblock \href {http://dx.doi.org/10.1103/PhysRevLett.54.1513}
  {\path{doi:10.1103/PhysRevLett.54.1513}}.
\newline\urlprefix\url{https://link.aps.org/doi/10.1103/PhysRevLett.54.1513}

\bibitem{Chrzan1992}
D.~C. Chrzan, M.~J. Mills,
  \href{https://link.aps.org/doi/10.1103/PhysRevLett.69.2795}{{Criticality in
  the Plastic Deformation of Ni3AI}}, Phys. Rev. Lett. 69~(19) (1992)
  2795--2798.
\newblock \href {http://dx.doi.org/10.1103/PhysRevLett.69.2795}
  {\path{doi:10.1103/PhysRevLett.69.2795}}.
\newline\urlprefix\url{https://link.aps.org/doi/10.1103/PhysRevLett.69.2795}

\bibitem{Schweigert1998}
V.~A. Schweigert, I.~V. Schweigert, A.~Melzer, A.~Homann, A.~Piel,
  \href{http://link.aps.org/doi/10.1103/PhysRevLett.80.5345}{{Plasma Crystal
  Melting: A Nonequilibrium Phase Transition}}, Phys. Rev. Lett. 80~(24) (1998)
  5345--5348.
\newblock \href {http://dx.doi.org/10.1103/PhysRevLett.80.5345}
  {\path{doi:10.1103/PhysRevLett.80.5345}}.
\newline\urlprefix\url{http://link.aps.org/doi/10.1103/PhysRevLett.80.5345}

\bibitem{Blythe2002}
R.~A. Blythe, M.~R. Evans,
  \href{https://link.aps.org/doi/10.1103/PhysRevLett.89.080601}{{Lee-Yang Zeros
  and Phase Transitions in Nonequilibrium Steady States}}, Phys. Rev. Lett.
  89~(8) (2002) 080601.
\newblock \href {http://dx.doi.org/10.1103/PhysRevLett.89.080601}
  {\path{doi:10.1103/PhysRevLett.89.080601}}.
\newline\urlprefix\url{https://link.aps.org/doi/10.1103/PhysRevLett.89.080601}

\bibitem{Whitelam2014}
S.~Whitelam, L.~O. Hedges, J.~D. Schmit,
  \href{https://link.aps.org/doi/10.1103/PhysRevLett.112.155504}{{Self-Assembly
  at a Nonequilibrium Critical Point}}, Phys. Rev. Lett. 112~(15) (2014)
  155504.
\newblock \href {http://dx.doi.org/10.1103/PhysRevLett.112.155504}
  {\path{doi:10.1103/PhysRevLett.112.155504}}.
\newline\urlprefix\url{https://link.aps.org/doi/10.1103/PhysRevLett.112.155504}

\bibitem{Zhang2015a}
X.~Zhang, M.~van Hulzen, D.~P. Singh, A.~Brownrigg, J.~P. Wright, N.~H. van
  Dijk, M.~Wagemaker, \href{http://dx.doi.org/10.1038/ncomms9333
  http://www.nature.com/ doifinder/10.1038/ncomms9333}{{Direct view on the
  phase evolution in individual LiFePO4 nanoparticles during Li-ion battery
  cycling}}, Nat. Commun. 6~(May) (2015) 8333.
\newblock \href {http://dx.doi.org/10.1038/ncomms9333}
  {\path{doi:10.1038/ncomms9333}}.
\newline\urlprefix\url{http://dx.doi.org/10.1038/ncomms9333
  http://www.nature.com/ doifinder/10.1038/ncomms9333}

\bibitem{Egelhaaf1999}
S.~U. Egelhaaf, P.~Schurtenberger,
  \href{https://link.aps.org/doi/10.1103/PhysRevLett.82.2804}{{Micelle-to-Vesicle
  Transition: A Time-Resolved Structural Study}}, Phys. Rev. Lett. 82~(13)
  (1999) 2804--2807.
\newblock \href {http://dx.doi.org/10.1103/PhysRevLett.82.2804}
  {\path{doi:10.1103/PhysRevLett.82.2804}}.
\newline\urlprefix\url{https://link.aps.org/doi/10.1103/PhysRevLett.82.2804}

\bibitem{Marenduzzo2001}
D.~Marenduzzo, S.~M. Bhattacharjee, A.~Maritan, E.~Orlandini, F.~Seno,
  \href{https://link.aps.org/doi/10.1103/PhysRevLett.88.028102}{{Dynamical
  Scaling of the DNA Unzipping Transition}}, Phys. Rev. Lett. 88~(2) (2001)
  028102.
\newblock \href {http://dx.doi.org/10.1103/PhysRevLett.88.028102}
  {\path{doi:10.1103/PhysRevLett.88.028102}}.
\newline\urlprefix\url{https://link.aps.org/doi/10.1103/PhysRevLett.88.028102}

\bibitem{Barrett-Freeman2008}
C.~Barrett-Freeman, M.~R. Evans, D.~Marenduzzo, W.~C.~K. Poon,
  \href{https://link.aps.org/doi/10.1103/PhysRevLett.101.100602}{{Nonequilibrium
  Phase Transition in the Sedimentation of Reproducing Particles}}, Phys. Rev.
  Lett. 101~(10) (2008) 100602.
\newblock \href {http://dx.doi.org/10.1103/PhysRevLett.101.100602}
  {\path{doi:10.1103/PhysRevLett.101.100602}}.
\newline\urlprefix\url{https://link.aps.org/doi/10.1103/PhysRevLett.101.100602}

\bibitem{Woo2011}
H.-J. Woo, A.~Wallqvist,
  \href{https://link.aps.org/doi/10.1103/PhysRevLett.106.060601}{{Nonequilibrium
  Phase Transitions Associated with DNA Replication}}, Phys. Rev. Lett. 106~(6)
  (2011) 060601.
\newblock \href {http://dx.doi.org/10.1103/PhysRevLett.106.060601}
  {\path{doi:10.1103/PhysRevLett.106.060601}}.
\newline\urlprefix\url{https://link.aps.org/doi/10.1103/PhysRevLett.106.060601}

\bibitem{Mak2016}
M.~Mak, M.~H. Zaman, R.~D. Kamm, T.~Kim,
  \href{http://dx.doi.org/10.1038/ncomms10323 http://www.nature.com/
  doifinder/10.1038/ncomms10323}{{Interplay of active processes modulates
  tension and drives phase transition in self-renewing, motor-driven
  cytoskeletal networks}}, Nat. Commun. 7~(May 2015) (2016) 10323.
\newblock \href {http://dx.doi.org/10.1038/ncomms10323}
  {\path{doi:10.1038/ncomms10323}}.
\newline\urlprefix\url{http://dx.doi.org/10.1038/ncomms10323
  http://www.nature.com/ doifinder/10.1038/ncomms10323}

\bibitem{Battle2016}
C.~Battle, C.~P. Broedersz, N.~Fakhri, V.~F. Geyer, J.~Howard, C.~F. Schmidt,
  F.~C. MacKintosh, \href{http://www.sciencemag.org/lookup/doi/10.1126/
  science.aac8167}{{Broken detailed balance at mesoscopic scales in active
  biological systems}}, {Science} 352~(6285) (2016) 604--607.
\newblock \href {http://dx.doi.org/10.1126/science.aac8167}
  {\path{doi:10.1126/science.aac8167}}.
\newline\urlprefix\url{http://www.sciencemag.org/lookup/doi/10.1126/
  science.aac8167}

\bibitem{Llas2003}
M.~Llas, P.~M. Gleiser, J.~M. L{\'{o}}pez, A.~D{\'{i}}az-Guilera,
  \href{https://link.aps.org/doi/10.1103/PhysRevE.68.066101}{{Nonequilibrium
  phase transition in a model for the propagation of innovations among economic
  agents}}, Phys. Rev. E 68~(6) (2003) 066101.
\newblock \href {http://dx.doi.org/10.1103/PhysRevE.68.066101}
  {\path{doi:10.1103/PhysRevE.68.066101}}.
\newline\urlprefix\url{https://link.aps.org/doi/10.1103/PhysRevE.68.066101}

\bibitem{Baronchelli2007}
A.~Baronchelli, L.~Dall'Asta, A.~Barrat, V.~Loreto,
  \href{https://link.aps.org/doi/10.1103/PhysRevE.76.051102}{{Nonequilibrium
  phase transition in negotiation dynamics}}, Phys. Rev. E 76~(5) (2007)
  051102.
\newblock \href {http://dx.doi.org/10.1103/PhysRevE.76.051102}
  {\path{doi:10.1103/PhysRevE.76.051102}}.
\newline\urlprefix\url{https://link.aps.org/doi/10.1103/PhysRevE.76.051102}

\bibitem{Scheffer2012}
M.~Scheffer, S.~R. Carpenter, T.~M. Lenton, J.~Bascompte, W.~Brock, V.~Dakos,
  J.~van~de Koppel, I.~A. van~de Leemput, S.~A. Levin, E.~H. van Nes,
  M.~Pascual, J.~Vandermeer,
  \href{http://www.sciencemag.org/cgi/doi/10.1126/science.1225244}{{Anticipating
  Critical Transitions}}, Science 338~(6105) (2012) 344--348.
\newblock \href {http://dx.doi.org/10.1126/science.1225244}
  {\path{doi:10.1126/science.1225244}}.
\newline\urlprefix\url{http://www.sciencemag.org/cgi/doi/10.1126/science.1225244}

\bibitem{Odor2004}
G.~{\'{O}}dor,
  \href{https://link.aps.org/doi/10.1103/RevModPhys.76.663}{{Universality
  classes in nonequilibrium lattice systems}}, Rev. Mod. Phys. 76~(3) (2004)
  663--724.
\newblock \href {http://dx.doi.org/10.1103/RevModPhys.76.663}
  {\path{doi:10.1103/RevModPhys.76.663}}.
\newline\urlprefix\url{https://link.aps.org/doi/10.1103/RevModPhys.76.663}

\bibitem{Lubeck2004}
S.~L{\"{u}}beck, \href{http://www.worldscientific.com/doi/abs/10.1142/
  S0217979204027748}{{Universal Scaling Behavior Of Non-Equilibrium Phase
  Transitions}}, Int. J. Mod. Phys. B 18~(31n32) (2004) 3977--4118.
\newblock \href {http://dx.doi.org/10.1142/S0217979204027748}
  {\path{doi:10.1142/S0217979204027748}}.
\newline\urlprefix\url{http://www.worldscientific.com/doi/abs/10.1142/
  S0217979204027748}

\bibitem{Prosen2010}
T.~Prosen, B.~{\v{Z}}unkovi{\v{c}},
  \href{http://stacks.iop.org/1367-2630/12/i=2/a=025016?
  key=crossref.27d4e8da79883b9ea2cc0dde04dc1523}{{Exact solution of Markovian
  master equations for quadratic Fermi systems: thermal baths, open XY spin
  chains and non-equilibrium phase transition}}, New J. Phys. 12~(2) (2010)
  025016.
\newblock \href {http://dx.doi.org/10.1088/1367-2630/12/2/025016}
  {\path{doi:10.1088/1367-2630/12/2/025016}}.
\newline\urlprefix\url{http://stacks.iop.org/1367-2630/12/i=2/a=025016?
  key=crossref.27d4e8da79883b9ea2cc0dde04dc1523}

\bibitem{Znidaric2015}
M.~{\v{Z}}nidari{\v{c}},
  \href{https://link.aps.org/doi/10.1103/PhysRevE.92.042143}{{Relaxation times
  of dissipative many-body quantum systems}}, Phys. Rev. E 92~(4) (2015)
  042143.
\newblock \href {http://dx.doi.org/10.1103/PhysRevE.92.042143}
  {\path{doi:10.1103/PhysRevE.92.042143}}.
\newline\urlprefix\url{https://link.aps.org/doi/10.1103/PhysRevE.92.042143}

\bibitem{Bloch2008}
I.~Bloch, J.~Dalibard, W.~Zwerger,
  \href{https://link.aps.org/doi/10.1103/RevModPhys.80.885}{{Many-body physics
  with ultracold gases}}, Rev. Mod. Phys. 80~(3) (2008) 885--964.
\newblock \href {http://dx.doi.org/10.1103/RevModPhys.80.885}
  {\path{doi:10.1103/RevModPhys.80.885}}.
\newline\urlprefix\url{https://link.aps.org/doi/10.1103/RevModPhys.80.885}

\bibitem{Barreiro2011}
J.~T. Barreiro, M.~M{\"{u}}ller, P.~Schindler, D.~Nigg, T.~Monz, M.~Chwalla,
  M.~Hennrich, C.~F. Roos, P.~Zoller, R.~Blatt,
  \href{http://dx.doi.org/10.1038/nature09801%5Cnhttp://
  www.ncbi.nlm.nih.gov/pubmed/21350481}{{An open-system quantum simulator with
  trapped ions.}}, Nature 470~(7335) (2011) 486--91.
\newblock \href {http://dx.doi.org/10.1038/nature09801}
  {\path{doi:10.1038/nature09801}}.
\newline\urlprefix\url{http://dx.doi.org/10.1038/nature09801%5Cnhttp://
  www.ncbi.nlm.nih.gov/pubmed/21350481}

\bibitem{Schindler2012}
P.~Schindler, M.~M{\"{u}}ller, D.~Nigg, J.~T. Barreiro, E.~A. Martinez,
  M.~Hennrich, T.~Monz, S.~Diehl, P.~Zoller, R.~Blatt,
  \href{http://www.nature.com/articles/nphys2630}{{Quantum simulation of
  dynamical maps with trapped ions}}, Nat. Phys. 9~(6) (2013) 361--367.
\newblock \href {http://dx.doi.org/10.1038/nphys2630}
  {\path{doi:10.1038/nphys2630}}.
\newline\urlprefix\url{http://www.nature.com/articles/nphys2630}

\bibitem{Hartmann2006}
M.~J. Hartmann, F.~G. S.~L. Brandao, M.~B. Plenio, {Strongly interacting
  polaritons in coupled arrays of cavities}, Nat. Phys. 2~(12) (2006) 849--855.
\newblock \href {http://dx.doi.org/10.1038/nphys462}
  {\path{doi:10.1038/nphys462}}.

\bibitem{Greentree2006}
A.~D. Greentree, C.~Tahan, J.~H. Cole, L.~C.~L. Hollenberg,
  \href{http://www.nature.com/articles/nphys466}{{Quantum phase transitions of
  light}}, Nat. Phys. 2~(12) (2006) 856--861.
\newblock \href {http://dx.doi.org/10.1038/nphys466}
  {\path{doi:10.1038/nphys466}}.
\newline\urlprefix\url{http://www.nature.com/articles/nphys466}

\bibitem{Angelakis2007}
D.~G. Angelakis, M.~F. Santos, S.~Bose,
  \href{https://link.aps.org/doi/10.1103/PhysRevA.76.031805}{{Photon-blockade-induced
  Mott transitions and XY spin models in coupled cavity arrays}}, Phys. Rev. A
  76~(3) (2007) 031805.
\newblock \href {http://dx.doi.org/10.1103/PhysRevA.76.031805}
  {\path{doi:10.1103/PhysRevA.76.031805}}.
\newline\urlprefix\url{https://link.aps.org/doi/10.1103/PhysRevA.76.031805}

\bibitem{Underwood2012}
D.~L. Underwood, W.~E. Shanks, J.~Koch, A.~A. Houck,
  \href{https://link.aps.org/doi/10.1103/PhysRevA.86.023837}{{Low-disorder
  microwave cavity lattices for quantum simulation with photons}}, Phys. Rev. A
  86~(2) (2012) 023837.
\newblock \href {http://dx.doi.org/10.1103/PhysRevA.86.023837}
  {\path{doi:10.1103/PhysRevA.86.023837}}.
\newline\urlprefix\url{https://link.aps.org/doi/10.1103/PhysRevA.86.023837}

\bibitem{Houck2012}
A.~A. Houck, H.~E. T{\"{u}}reci, J.~Koch,
  \href{http://www.nature.com/articles/nphys2251}{{On-chip quantum simulation
  with superconducting circuits}}, Nat. Phys. 8~(4) (2012) 292--299.
\newblock \href {http://dx.doi.org/10.1038/nphys2251}
  {\path{doi:10.1038/nphys2251}}.
\newline\urlprefix\url{http://www.nature.com/articles/nphys2251}

\bibitem{Raftery2014}
J.~Raftery, D.~Sadri, S.~Schmidt, H.~E. T{\"{u}}reci, A.~A. Houck,
  \href{https://link.aps.org/doi/10.1103/PhysRevX.4.031043}{{Observation of a
  Dissipation-Induced Classical to Quantum Transition}}, Phys. Rev. X 4~(3)
  (2014) 031043.
\newblock \href {http://dx.doi.org/10.1103/PhysRevX.4.031043}
  {\path{doi:10.1103/PhysRevX.4.031043}}.
\newline\urlprefix\url{https://link.aps.org/doi/10.1103/PhysRevX.4.031043}

\bibitem{Weimer2010}
H.~Weimer, M.~M{\"{u}}ller, I.~Lesanovsky, P.~Zoller, H.~P. B{\"{u}}chler,
  \href{http://www.nature.com/articles/nphys1614}{{A Rydberg quantum
  simulator}}, Nat. Phys. 6~(5) (2010) 382--388.
\newblock \href {http://dx.doi.org/10.1038/nphys1614}
  {\path{doi:10.1038/nphys1614}}.
\newline\urlprefix\url{http://www.nature.com/articles/nphys1614}

\bibitem{Dudin2012}
Y.~O. Dudin, L.~Li, F.~Bariani, A.~Kuzmich,
  \href{http://www.nature.com/articles/nphys2413}{{Observation of coherent
  many-body Rabi oscillations}}, Nat. Phys. 8~(11) (2012) 790--794.
\newblock \href {http://dx.doi.org/10.1038/nphys2413}
  {\path{doi:10.1038/nphys2413}}.
\newline\urlprefix\url{http://www.nature.com/articles/nphys2413}

\bibitem{Eisert2015}
J.~Eisert, M.~Friesdorf, C.~Gogolin,
  \href{http://www.nature.com/doifinder/10.1038/nphys3215}{{Quantum many-body
  systems out of equilibrium}}, Nat. Phys. 11~(2) (2015) 124--130.
\newblock \href {http://dx.doi.org/10.1038/nphys3215}
  {\path{doi:10.1038/nphys3215}}.
\newline\urlprefix\url{http://www.nature.com/doifinder/10.1038/nphys3215}

\bibitem{Lindblad1976}
G.~Lindblad, {On the generators of quantum dynamical semigroups}, Comm. Math.
  Phys. 48 (1976) 119--130.

\bibitem{Verstraete2009}
F.~Verstraete, M.~M. Wolf, J.~{Ignacio Cirac},
  \href{http://www.nature.com/articles/nphys1342}{{Quantum computation and
  quantum-state engineering driven by dissipation}}, Nat. Phys. 5~(9) (2009)
  633--636.
\newblock \href {http://dx.doi.org/10.1038/nphys1342}
  {\path{doi:10.1038/nphys1342}}.
\newline\urlprefix\url{http://www.nature.com/articles/nphys1342}

\bibitem{Honing2012}
M.~H{\"{o}}ning, M.~Moos, M.~Fleischhauer,
  \href{https://link.aps.org/doi/10.1103/PhysRevA.86.013606}{{Critical
  exponents of steady-state phase transitions in fermionic lattice models}},
  Phys. Rev. A 86~(1) (2012) 013606.
\newblock \href {http://dx.doi.org/10.1103/PhysRevA.86.013606}
  {\path{doi:10.1103/PhysRevA.86.013606}}.
\newline\urlprefix\url{https://link.aps.org/doi/10.1103/PhysRevA.86.013606}

\bibitem{Horstmann2013}
B.~Horstmann, J.~I. Cirac, G.~Giedke,
  \href{https://link.aps.org/doi/10.1103/PhysRevA.87.012108}{{Noise-driven
  dynamics and phase transitions in fermionic systems}}, Phys. Rev. A 87~(1)
  (2013) 012108.
\newblock \href {http://dx.doi.org/10.1103/PhysRevA.87.012108}
  {\path{doi:10.1103/PhysRevA.87.012108}}.
\newline\urlprefix\url{https://link.aps.org/doi/10.1103/PhysRevA.87.012108}

\bibitem{Bardyn2013}
C.-E. Bardyn, M.~A. Baranov, C.~V. Kraus, E.~Rico, A.~İmamoğlu, P.~Zoller,
  S.~Diehl, \href{http://stacks.iop.org/1367-2630/15/i=8/a=085001?
  key=crossref.dbfc132c3a50871fe00070d6baa253c4}{{Topology by dissipation}},
  New J. Phys. 15~(8) (2013) 085001.
\newblock \href {http://dx.doi.org/10.1088/1367-2630/15/8/085001}
  {\path{doi:10.1088/1367-2630/15/8/085001}}.
\newline\urlprefix\url{http://stacks.iop.org/1367-2630/15/i=8/a=085001?
  key=crossref.dbfc132c3a50871fe00070d6baa253c4}

\bibitem{Feynman1982}
R.~P. Feynman, \href{http://link.springer.com/10.1007/BF02650179}{{Simulating
  Physics with Computers}}, Int. J. Theor. Phys. 21~(6/7) (1982) 467.
\newblock \href {http://dx.doi.org/10.1007/BF02650179}
  {\path{doi:10.1007/BF02650179}}.
\newline\urlprefix\url{http://link.springer.com/10.1007/BF02650179}

\bibitem{Lloyd1996}
S.~Lloyd, \href{http://www.sciencemag.org/cgi/doi/10.1126/
  science.273.5278.1073}{{Universal Quantum Simulators}}, Science (80-. ).
  273~(5278) (1996) 1073--1078.
\newblock \href {http://dx.doi.org/10.1126/science.273.5278.1073}
  {\path{doi:10.1126/science.273.5278.1073}}.
\newline\urlprefix\url{http://www.sciencemag.org/cgi/doi/10.1126/
  science.273.5278.1073}

\bibitem{Bloch2012}
I.~Bloch, J.~Dalibard, S.~Nascimb{\`{e}}ne,
  \href{http://www.nature.com/articles/nphys2259}{{Quantum simulations with
  ultracold quantum gases}}, Nat. Phys. 8~(4) (2012) 267--276.
\newblock \href {http://dx.doi.org/10.1038/nphys2259}
  {\path{doi:10.1038/nphys2259}}.
\newline\urlprefix\url{http://www.nature.com/articles/nphys2259}

\bibitem{Blatt2012}
R.~Blatt, C.~F. Roos, \href{http://www.nature.com/articles/nphys2252}{{Quantum
  simulations with trapped ions}}, Nat. Phys. 8~(4) (2012) 277--284.
\newblock \href {http://dx.doi.org/10.1038/nphys2252}
  {\path{doi:10.1038/nphys2252}}.
\newline\urlprefix\url{http://www.nature.com/articles/nphys2252}

\bibitem{Aspuru-Guzik2012}
A.~Aspuru-Guzik, P.~Walther,
  \href{http://dx.doi.org/10.1038/nphys2253}{{Photonic quantum simulators}},
  Nat. Phys. 8~(4) (2012) 285--291.
\newblock \href {http://dx.doi.org/10.1038/nphys2253}
  {\path{doi:10.1038/nphys2253}}.
\newline\urlprefix\url{http://dx.doi.org/10.1038/nphys2253}

\bibitem{Bach1994}
V.~Bach, E.~H. Lieb, J.~P. Solovej,
  \href{http://link.springer.com/10.1007/BF02188656}{{Generalized Hartree-Fock
  theory and the Hubbard model}}, J. Stat. Phys. 76~(1-2) (1994) 3--89.
\newblock \href {http://dx.doi.org/10.1007/BF02188656}
  {\path{doi:10.1007/BF02188656}}.
\newline\urlprefix\url{http://link.springer.com/10.1007/BF02188656}

\bibitem{Prosen2008a}
T.~Prosen, \href{http://stacks.iop.org/1367-2630/10/i=4/a=043026?
  key=crossref.b000d5ff6bcc288ed499ccbcf202d73d}{{Third quantization: a general
  method to solve master equations for quadratic open Fermi systems}}, New J.
  Phys. 10~(4) (2008) 043026.
\newblock \href {http://dx.doi.org/10.1088/1367-2630/10/4/043026}
  {\path{doi:10.1088/1367-2630/10/4/043026}}.
\newline\urlprefix\url{http://stacks.iop.org/1367-2630/10/i=4/a=043026?
  key=crossref.b000d5ff6bcc288ed499ccbcf202d73d}

\bibitem{Prosen2010a}
T.~Prosen, \href{http://stacks.iop.org/1742-5468/2010/i=07/a=P07020?
  key=crossref.dbdde3357621be1587fb254ddd3f9886}{{Spectral theorem for the
  Lindblad equation for quadratic open fermionic systems}}, J. Stat. Mech.
  Theory Exp. 2010~(07) (2010) P07020.
\newblock \href {http://dx.doi.org/10.1088/1742-5468/2010/07/P07020}
  {\path{doi:10.1088/1742-5468/2010/07/P07020}}.
\newline\urlprefix\url{http://stacks.iop.org/1742-5468/2010/i=07/a=P07020?
  key=crossref.dbdde3357621be1587fb254ddd3f9886}

\bibitem{Zunkovic2010}
B.~{\v{Z}}unkovi{\v{c}}, T.~Prosen,
  \href{http://stacks.iop.org/1742-5468/2010/i=08/a=P08016?
  key=crossref.f270fcbdc5b226375cc426faf937b2f8}{{Explicit solution of the
  Lindblad equation for nearly isotropic boundary driven XY spin 1/2 chain}},
  J. Stat. Mech. Theory Exp. 2010~(08) (2010) P08016.
\newblock \href {http://dx.doi.org/10.1088/1742-5468/2010/08/P08016}
  {\path{doi:10.1088/1742-5468/2010/08/P08016}}.
\newline\urlprefix\url{http://stacks.iop.org/1742-5468/2010/i=08/a=P08016?
  key=crossref.f270fcbdc5b226375cc426faf937b2f8}

\bibitem{Blaizot1986}
J.-P. Blaizot, G.~Ripka, {Quantum theory of finite systems}, MIT Press, 1986.

\bibitem{Znidaric2011}
M.~{\v{Z}}nidari{\v{c}},
  \href{https://link.aps.org/doi/10.1103/PhysRevE.83.011108}{{Solvable quantum
  nonequilibrium model exhibiting a phase transition and a matrix product
  representation}}, Phys. Rev. E 83~(1) (2011) 011108.
\newblock \href {http://dx.doi.org/10.1103/PhysRevE.83.011108}
  {\path{doi:10.1103/PhysRevE.83.011108}}.
\newline\urlprefix\url{https://link.aps.org/doi/10.1103/PhysRevE.83.011108}

\bibitem{Cai2013}
Z.~Cai, T.~Barthel,
  \href{https://link.aps.org/doi/10.1103/PhysRevLett.111.150403}{{Algebraic
  versus Exponential Decoherence in Dissipative Many-Particle Systems}}, Phys.
  Rev. Lett. 111~(15) (2013) 150403.
\newblock \href {http://dx.doi.org/10.1103/PhysRevLett.111.150403}
  {\path{doi:10.1103/PhysRevLett.111.150403}}.
\newline\urlprefix\url{https://link.aps.org/doi/10.1103/PhysRevLett.111.150403}

\bibitem{Carollo2003}
A.~Carollo, I.~Fuentes-Guridi, M.~F. Santos, V.~Vedral,
  \href{https://link.aps.org/doi/10.1103/PhysRevLett.90.160402}{{Geometric
  Phase in Open Systems}}, Phys. Rev. Lett. 90~(16) (2003) 160402.
\newblock \href {http://dx.doi.org/10.1103/PhysRevLett.90.160402}
  {\path{doi:10.1103/PhysRevLett.90.160402}}.
\newline\urlprefix\url{https://link.aps.org/doi/10.1103/PhysRevLett.90.160402}

\bibitem{Carollo2004}
A.~Carollo, I.~Fuentes-Guridi, M.~F. Santos, V.~Vedral,
  \href{https://link.aps.org/doi/10.1103/PhysRevLett.92.020402}{{Spin-1/2
  Geometric Phase Driven by Decohering Quantum Fields}}, Phys. Rev. Lett.
  92~(2) (2004) 020402.
\newblock \href {http://dx.doi.org/10.1103/PhysRevLett.92.020402}
  {\path{doi:10.1103/PhysRevLett.92.020402}}.
\newline\urlprefix\url{https://link.aps.org/doi/10.1103/PhysRevLett.92.020402}

\bibitem{Ercolessi2013}
E.~Ercolessi, M.~Schiavina, \href{https://linkinghub.elsevier.com/retrieve/pii/
  S0375960113005872}{{Symmetric logarithmic derivative for general n-level
  systems and the quantum Fisher information tensor for three-level systems}},
  Phys. Lett. A 377~(34-36) (2013) 1996--2002.
\newblock \href {http://dx.doi.org/10.1016/j.physleta.2013.06.012}
  {\path{doi:10.1016/j.physleta.2013.06.012}}.
\newline\urlprefix\url{https://linkinghub.elsevier.com/retrieve/pii/
  S0375960113005872}

\bibitem{Monras2013}
A.~Monras, \href{http://arxiv.org/abs/1303.3682}{{Phase space formalism for
  quantum estimation of Gaussian states}}.
\newline\urlprefix\url{http://arxiv.org/abs/1303.3682}

\bibitem{Jiang2014}
Z.~Jiang, {Quantum Fisher information for states in exponential form}, Phys.
  Rev. A 89~(3) (2014) 1--6.
\newblock \href {http://dx.doi.org/10.1103/PhysRevA.89.032128}
  {\path{doi:10.1103/PhysRevA.89.032128}}.

\bibitem{Thouless98}
D.~J. Thouless, {Topological Quantum Numbers in Nonrelativistic Physics}, World
  Scientific, Singapore, London, 1998.

\bibitem{Wilcox1967}
R.~M. Wilcox,
  \href{http://aip.scitation.org/doi/10.1063/1.1705306}{{Exponential Operators
  and Parameter Differentiation in Quantum Physics}}, J. Math. Phys. 8~(4)
  (1967) 962--982.
\newblock \href {http://dx.doi.org/10.1063/1.1705306}
  {\path{doi:10.1063/1.1705306}}.
\newline\urlprefix\url{http://aip.scitation.org/doi/10.1063/1.1705306}

\end{thebibliography}
\end{document}